\documentclass[noinfoline]{imsart}
\RequirePackage[OT1]{fontenc}
\RequirePackage{amsthm,amsmath}
\RequirePackage[authoryear]{natbib}
\RequirePackage[colorlinks,citecolor=blue,urlcolor=blue]{hyperref}

\usepackage{MyStylefile}

\startlocaldefs
\numberwithin{equation}{section}
\theoremstyle{plain}
\endlocaldefs

\begin{document}
\begin{frontmatter}

\title{{\large CovNet: Covariance Networks for Functional Data\\ on Multidimensional Domains}}

\runtitle{Covariance Networks for Functional Data on Multidimensional Domains}

\begin{aug}
\author{\fnms{Soham} \snm{Sarkar}\ead[label=e1]{soham.sarkar@epfl.ch}}
\and
\author{\fnms{Victor M.} \snm{Panaretos}\ead[label=e2]{victor.panaretos@epfl.ch}}

\thankstext{t1}{Research supported by a Swiss National Science Foundation grant.}

\runauthor{S. Sarkar \& V.M. Panaretos}

\affiliation{Ecole Polytechnique F\'ed\'erale de Lausanne}

\address{Institut de Math\'ematiques\\
Ecole Polytechnique F\'ed\'erale de Lausanne \\
e-mail: 
\href{mailto:soham.sarkar@epfl.ch}{soham.sarkar@epfl.ch}, 
\href{mailto:victor.panaretos@epfl.ch}{victor.panaretos@epfl.ch}}

\end{aug}

\begin{abstract}
Covariance estimation is ubiquitous in functional data analysis. Yet, the case of functional observations over multidimensional domains introduces computational and statistical challenges, rendering the standard methods effectively inapplicable. To address this problem, we introduce \emph{Covariance Networks} (CovNet) as a modeling and estimation tool. The CovNet model is  \emph{universal} -- it can be used to approximate any covariance up to desired precision. Moreover, the model can be fitted efficiently to the data and its neural network architecture allows us to employ modern computational tools in the implementation. The CovNet model also admits a closed-form eigendecomposition, which can be computed efficiently, without constructing the covariance itself. This facilitates easy storage and subsequent manipulation of a covariance in the context of the CovNet. We establish consistency of the proposed estimator and derive its rate of convergence. The usefulness of the proposed method is demonstrated by means of an extensive simulation study and an application to resting state fMRI data.
\end{abstract}



\end{frontmatter}

\tableofcontents

\section{Introduction}\label{sec:intro}

We consider the problem of covariance estimation from a collection of functional observations defined over a multidimensional domain. To be precise, let $\X = \{X(\uvec):\uvec \in \Q\}$ be a compactly supported random field, i.e., a real-valued second-order stochastic process on a compact set $\Q \subset \R^d$, with covariance kernel $c(\uvec,\vvec) = \cov(X(\uvec),X(\vvec))$. We want to estimate $c$ based on an independent and identically distributed (i.i.d.) sample $\X_1,\ldots,\X_N \sim \X$. In particular, we work in the framework of \emph{functional data analysis} \citep[FDA, see][]{ramsay2002,hsing2015}, where we assume that $\X$ takes values in $\L_2(\Q)$, the space of all real-valued square-integrable functions on $\Q$.

Covariance estimation, along with mean estimation, is a fundamental problem in functional data analysis and has multifaceted applications, e.g., in regression, prediction, classification. This problem has been studied extensively for observations over one-dimensional domains (i.e., $d=1$) or \emph{curve data} (see \cite{wang2016} for a detailed review). The same is, however, not true for observations over multidimensional domains. Although, in principal, these two regimes are similar, and most methods for curve data can be ``readily used" for data over multidimensional domains, in practice, the dimensionality of the problem draws a clear distinction between the two paradigms. To appreciate this, suppose that we observe the random fields on a grid of size $K \times \cdots \times K$ in $\Q \subset \R^d$. In this case, the estimation of the empirical covariance requires $\O(K^{2d})$ computations. The storage cost for this estimator is also of the order $\O(K^{2d})$ which, for $d=2$, can be prohibitive even for $K \approx 100$. The problem becomes even more severe when $d$ is larger ($d\geq 3$), which is increasingly common, e.g., for observations over spatial volumes or of a spatio-temporal nature. Moreover, subsequent manipulation, e.g., inversion, needed in applications, requires computation in the order of $\O(K^{3d})$, leading to a prohibitive computational burden.

To put things into perspective, consider the \textsl{$1000$ Functional Connectomes Project}\footnote{\url{https://www.nitrc.org/projects/fcon_1000/}} which contains functional magnetic resonance imaging (fMRI) of brains for more than $1200$ individuals. For each individual, the data consist of 3D brain-scans on a grid of size $64 \times 64 \times 33$ taken at $2$ second intervals over $225$ time points. Covariance estimation is of utmost importance in fMRI studies as it captures the connectivity patterns in the brain \citep{aston2012,stoehr2021}. At the same time, it is extremely difficult to do so nonparametrically because of the dimensionality of the problem. For instance, the empirical covariance estimator for the 3D fMRI data would be an object of size $64 \times 64 \times 33 \times 64 \times 64 \times 33$ which requires $68$ Gb memory to compute and store (at $32$-byte precision). This is impossible with a regular computer which usually has $16$ or $32$ Gb of memory. Also, apart from looking at the connectivity pattern of the brain as a 3D object, it is also of importance to check how these patterns evolve over time, i.e., to consider the full 4D data on a grid of size $64 \times 64 \times 33 \times 225$. The problem becomes even more severe in this case, where the empirical covariance would require approximately $3.4 \times 10^6$ Gb of memory during computation and for storage \citep[see also][]{aston2012,stoehr2021}.

To alleviate this \emph{curse of dimensionality}, further modeling assumptions are often made on the underlying covariance, the most popular being that of \emph{separability}. A \emph{separable model} assumes that the true covariance over the multidimensional domain can be factored into several covariances over one-dimensional domains, i.e., $c(\uvec,\vvec) = c_1(u_1,v_1) \times \cdots \times c_d(u_d,v_d)$ for $\uvec, \vvec \in \Q$. This greatly simplifies the problem and entails enormous computational savings. For instance, in the case of observations on a grid, a separable model can be estimated with $\O(dK^2)$ computations and has the same order of storage requirements. The gain during application of the model is even better -- the inversion of the model requires $\O(dK^3)$ computations compared to the $\O(K^{3d})$ for the empirical covariance. Despite all these advantages, separability is merely a modeling assumption, which is highly restrictive and often violated in practice \citep{aston2017,constantinou2017,rougier2017,bagchi2020}. Still, it is often the preferred choice in practice, not because it is believed to hold, but rather for the savings that it entails \citep{gneiting2006,pigoli2018}. Perhaps it is safe to say that the popularity of the separable model stems from the non-availability of a better alternative. It is worth clarifying here that when the data are \emph{sparse} (each random field is observed at a few randomly scattered locations) there do exist methods applicable to multidimensional domains without assuming separability (for example, \cite{wang2020} used a penalized method, leading to the use of tensor-products of splines, and yielding a low-rank approximation to the covariance). However, such approaches are infeasible in the \emph{dense} regime, where each random field is measured on the same dense grid (e.g., the fMRI data), and which is our main interest in this paper. It is the denseness of the measurements that gives rise to the severe computational challenges mentioned above (e.g., in a dense regime, the approach of \cite{wang2020} is infeasible since it requires computation of the ``raw covariances" -- equivalent to the computation of the empirical covariance).

In an effort to deliver a more general yet tractable approach, we propose a new model for covariance estimation using neural networks. Neural networks have long been successfully used in nonparametric function estimation, and recently, they have been shown to successfully overcome the curse of dimensionality in nonparametric regression \citep{bauer2019,schmidt-hieber2020}. Also, they have been used for mean estimation of functional data over multidimensional domains \citep{wang2021}. Motivated by the success of neural networks, we propose \emph{Covariance Networks} (CovNet) as a framework for the estimation of the covariance of multidimensional random fields. A CovNet is a positive semi-definite function on $\mathcal{Q}\times \mathcal{Q}$ described by a neural network architecture. In particular, we define and study three variants: the \emph{shallow} CovNet model and the \emph{deep} CovNet model which differ with respect to the \emph{depth} of the network; and the \emph{deepshared} CovNet model which is a restricted (regularised) version of the deep CovNet model.\\

\noindent Our framework features several advantages, namely:
\begin{enumerate}
\item It is genuinely \emph{nonparametric} -- any covariance can be approximated up to arbitrary precision via a CovNet structure. We establish this so-called \emph{universal approximation property} of the CovNet models in Theorems~\ref{thm:universal_approximation_shallow} and \ref{thm:universal_approximation_deep}. Moreover, the proposed model has an explicit functional form. This functional form has its own advantage in applications such as kriging, where we do not need to interpolate or smooth the estimated covariance before use.

\item Fitting a CovNet to the data is computationally tractable. The models we introduce can be fitted at the level of the data, without the need to compute or store any high-order objects. Moreover, the neural network structure allows us to exploit modern machine learning tools during the estimation. These are discussed in Section~\ref{sec:implementation}.

\item The special structure of the CovNet models ensures that the eigendecomposition of the associated operator can be obtained without the need to explicitly form the operator itself. Thus, we can access the eigensystem of the fitted CovNet very easily without ever forming any higher order objects. This allows us to store and subsequently manipulate the fitted model very easily (see Section~\ref{sec:eigendecomposition}).

\item The CovNet estimators come with theoretical guarantees. In particular, we establish their consistency and derive their rates of convergence (Section~\ref{sec:asymptotics}). Our analyses are \emph{fully nonparametric} -- we make no structural assumption on the underlying covariance $\C$ for our derivations.
\end{enumerate}

The rest of the article is organized as follows. We lay out our methodology in the next section. In particular, we begin by describing the shallow CovNet model in Section~\ref{sec:shallow_model} and establish its \emph{universality}. In Section~\ref{sec:deep_learning}, we extend the shallow CovNet model by using \emph{deep} architectures, leading to the \emph{deep} and the \emph{deepshared} CovNet models. In Section~\ref{sec:implementation}, we demonstrate how the CovNet models can be efficiently estimated in practice. The eigendecomposition of the CovNet operator is discussed in Section~\ref{sec:eigendecomposition}. In Section~\ref{sec:empirical_demonstration}, we demonstrate the usefulness of the proposed CovNet models by means of a detailed simulation study and an application to the fMRI data. We establish the theoretical properties of these models in Section~\ref{sec:asymptotics}. Some concluding remarks are made in Section~\ref{sec:conclusion}. The proofs of our asymptotic results are provided in the appendices. The appendices also cover some related mathematical ideas, as well as some further numerical results.

\section{Covariance networks}\label{sec:method}

We start with some background concepts, more details can be found in \cite{hsing2015} and Appendix~\ref{supp:math_background}. Let $\Q$ be a compact subset of $\R^d$ and let $\X = \{X(\uvec):\uvec \in \Q\}$ be a random element of $\L_2(\Q)$. For $d=1$, $\X$ is usually referred to as a \emph{random curve}, whereas for $d>1$, it is referred to as a \emph{random field}. We assume that $\X$ has finite second moment, i.e., $\E(\|\X\|^2) < \infty$, which ensures the existence of its mean $m = \E(\X)$ and covariance $\C = \E\{(\X-m) \otimes (\X-m)\}$ (both the expectations are in the Bochner sense, see \citealt{hsing2015}). Here, $\|\cdot\|$ is the $\L_2$-norm associated with the inner-product $\langle f,g \rangle = \int_{\Q} f(\uvec) g(\uvec) \diff \uvec$ for $f,g \in \L_2(\Q)$, and the tensor product $h\otimes g $ denotes the rank 1 operator $f\mapsto \langle f,g\rangle h$.  The covariance $\C$ is a linear operator from $\L_2(\Q)$ onto itself, given by
\[
\C f(\uvec) = \int_{\Q} c(\uvec,\vvec)\,f(\vvec) \diff\vvec,~~ f \in \L_2(\Q).
\]
Here, $c \in \L_2(\Q \times \Q)$ defined as $c(\uvec,\vvec) = \cov(X(\uvec),X(\vvec))$ is the \emph{covariance kernel} associated with $\X$. We also say that $\C$ is the \emph{integral operator} associated with the kernel $c$. The operator $\C$ is positive semi-definite and the kernel $c$ is non-negative definite. The Hilbert-Schmidt norm $\vertj{\cdot}_2$ of $\C$ is finite, and $\vertj{\C}_2 = \|c\|_{\L_2(\Q \times \Q)}$. Thus, the covariance operator $\C$ and the covariance kernel $c$ are linked by an isometric isomorphism. Consequently, we can use $\C$ and $c$ interchangeably, and the estimation of the covariance $\C$ is equivalent to the estimation of the kernel $c$. Since the object of interest is the covariance rather than the mean, we work under the assumption that $m=\E(\X) = 0$, unless specifically mentioned.

\subsection{Shallow architecture}\label{sec:shallow_model}

We propose to estimate the covariance kernel $c$ using the following neural network structure:
\begin{equation}\label{eq:shallow_covnet_kernel}
c_{\rm sh}(\uvec,\vvec) = \sum_{r=1}^R \sum_{s=1}^R \lambda_{r,s}\, \sigma(\mathbf w_r^\top\uvec+b_r)\,\sigma(\mathbf w_s^\top\vvec+b_s),\qquad \uvec,\vvec \in \Q,
\end{equation}
where $R \in \N$ is the \emph{width}, $\sigma:\R \to \R$ is an {\em activation} function, and $\Lambda := ((\lambda_{r,s}))$ is a positive semi-definite matrix. The parameters $\mathbf w_r \in \R^d$ and $b_r \in \R$ for $r=1,\ldots,R$ are the \emph{weights} and the \emph{biases} of the model \eqref{eq:shallow_covnet_kernel}. Positive semi-definiteness of $\Lambda$ readily implies that $c_{\rm sh}$ is a non-negative definite kernel. For a given activation function $\sigma$ and width $R \in \N$, we define
\begin{equation}\label{eq:shallow_covnet_kernel_class}
\F^{\rm sh}_{R,\sigma} = \left\{c_{\rm sh} \text{ of the form } \eqref{eq:shallow_covnet_kernel} : \Lambda = ((\lambda_{r,s})) \succeq 0, \mathbf w_1,\ldots,\mathbf w_R \in \R^d, b_1,\ldots,b_R \in \R\right\},
\end{equation}
to be the class of {\em shallow Covariance Network} kernels or \emph{shallow CovNet} kernels. We also define
\begin{equation}\label{eq:shallow_covnet_operator_class}
\widetilde\F^{\rm sh}_{R,\sigma} = \left\{\G : \G \text{ is an integral operator with kernel } g \in \F^{\rm sh}_{R,\sigma}\right\},
\end{equation}
to be the class of shallow covariance network operators or shallow CovNet operators.

We call the structure \eqref{eq:shallow_covnet_kernel} \emph{shallow} because each of the constituents $\sigma(\mathbf w_r^\top\cdot+\,b_r)$ for $r=1,\ldots,R$ of the kernel \eqref{eq:shallow_covnet_kernel} is a shallow neural network, i.e., a neural network with a single hidden layer. The special structure of the kernel \eqref{eq:shallow_covnet_kernel} allows us to visualize it as a neural network with two hidden layers, as depicted in Figure~\ref{fig:shallow_covnet}. In the first layer, starting from the inputs $\uvec$ and $\vvec$, single-layer perceptrons $\sigma(\mathbf w_r^\top\uvec+b_r)$ and $\sigma(\mathbf w_r^\top\vvec+b_r), r=1,\ldots,R$ are computed. In the next layer, these outputs are cross-multiplied with the weights $\lambda_{r,s}$ to produce the final result $c_{\rm sh}(\uvec,\vvec)$. As one can see, this is a feed-forward neural network \citep[see][Chapter~6]{anthony1999}, which is completely determined (for fixed $\sigma$ and $R$) by the parameters $\mathbf w_1,\ldots,\mathbf w_R$, $b_1,\ldots,b_R$ and $\Lambda = ((\lambda_{r,s}))$ (with the added restriction on $\Lambda$).

\begin{figure}[t]
\centering
\includegraphics[height=0.35\textheight]{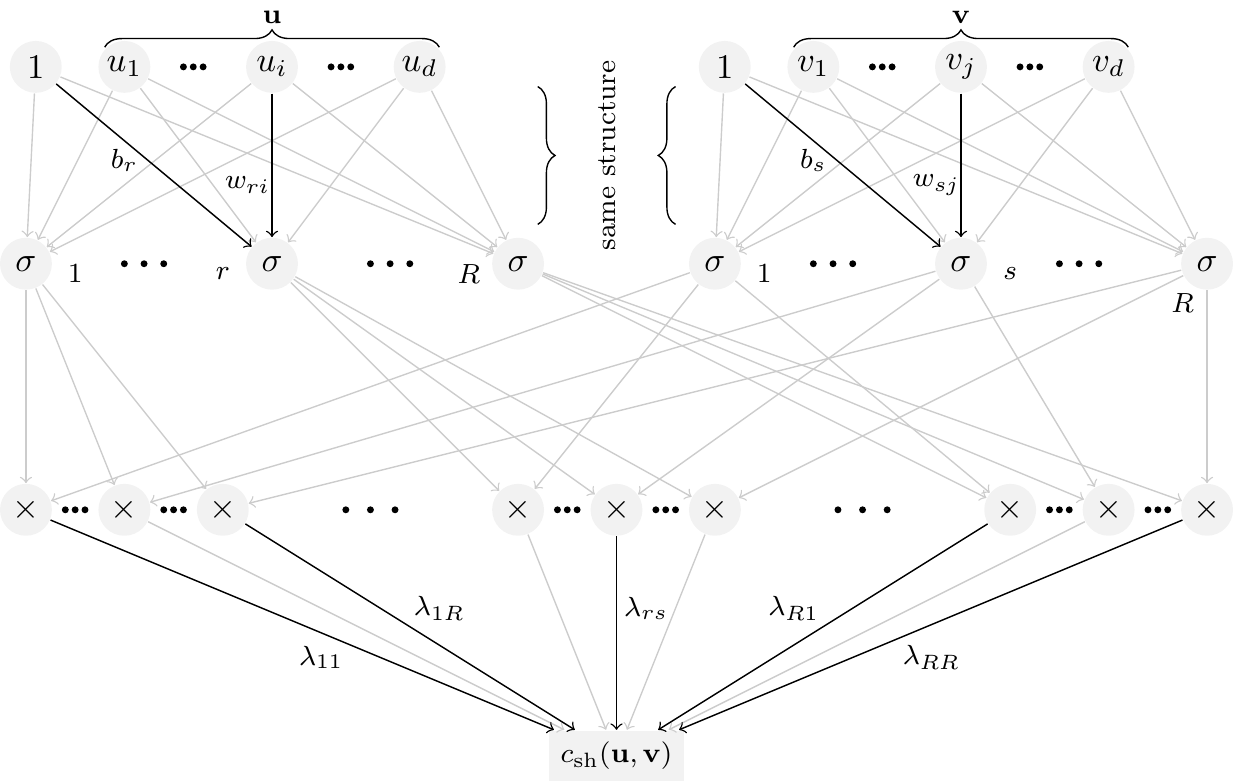}
\caption{\label{fig:shallow_covnet}A schematic representation of the shallow CovNet structure. The top layer corresponds to the inputs $\uvec$ and $\vvec$. In the first hidden layer (second from the top), the inputs are projected by weights, shifted by bias and transformed by the activation function $\sigma$ to produce single-layer perceptrons $\sigma(\mathbf w_r^\top\uvec+b_r), \sigma(\mathbf w_r^\top\vvec+b_r),r=1,\ldots,R$. The weights and biases of the two sides (left and right), corresponding to $\uvec$ and $\vvec$, are the same. In the second hidden layer (third from the top), the outputs of the first hidden layer are cross-multiplied. These are then multiplied by weights $\lambda_{r,s}$ and added to produce the output $c_{\rm sh}(\uvec,\vvec)$.}
\end{figure}

As mentioned in the introduction, the shallow CovNet structure \eqref{eq:shallow_covnet_kernel} is a general model, in the sense that any covariance kernel can be approximated with arbitrary precision using a shallow CovNet kernel of the form \eqref{eq:shallow_covnet_kernel}. Thus, we do not need to make any assumption on the underlying covariance $c$, resulting in a completely nonparametric procedure. However, we do need a particular condition on the activation function $\sigma$ of the network.
\begin{definition}[Sigmoidal activation\label{def:sigmoidal_function}]
An activation function $\sigma: \R \to [0,1]$ is said to be sigmoidal if it is non-decreasing with $\lim_{x \to \infty} \sigma(x) = 1 \text{ and } \lim_{x \to -\infty} \sigma(x) = 0$.
\end{definition}
In probabilistic terms, a sigmoidal function is a cumulative distribution function. Sigmoidal activations are very common in the literature of neural networks. One of the most popular activation functions,  the \emph{logistic function} $\sigma(t) = 1/(1+\exp(-t))$ is a sigmoidal function often also referred to as the \emph{sigmoid}. Many other popularly used activation functions are also sigmoidal \citep[see][Chapter~16, for a plethora of examples]{gyorfi2002}. It is worthwhile to note that the definition of sigmoidal functions is not universal. In this article, whenever we refer to a sigmoidal function, we mean it in the sense of Definition~\ref{def:sigmoidal_function}.

If we use a sigmoidal activation function, then any covariance kernel can be approximated up to arbitrary precision using a shallow CovNet kernel of the form \eqref{eq:shallow_covnet_kernel}. Such a property is often referred to as the \emph{universal approximation property} in the computer science literature.

\begin{theorem}[Shallow CovNet is Universal Approximator\label{thm:universal_approximation_shallow}]
Let $c: \Q \times \Q \to \R$ be the kernel of the covariance operator $\C$. Also, assume that the activation function $\sigma$ is sigmoidal. Then, for every $\epsilon > 0$, there exists $R \in \N$ and $c_{\rm sh} \in \F^{\rm sh}_{R,\sigma}$ such that 
\[
\int_{\Q \times \Q} \big|c(\uvec,\vvec) - c_{\rm sh}(\uvec,\vvec)\big|^2 \diff\uvec\diff\vvec \le \epsilon.
\]
If in addition $c$ is continuous, then the same conclusion holds uniformly. That is, for every $\epsilon > 0$, we can find $R \in \N$ and $c_{\rm sh} \in \F^{\rm sh}_{R,\sigma}$ such that 
\[
\sup_{\uvec,\vvec \in \Q} \big|c(\uvec,\vvec) - c_{\rm sh}(\uvec,\vvec)\big| \le \epsilon.
\]
\end{theorem}

\begin{remark}
The proof of the theorem rests on the universal approximation property of single hidden layer neural networks on the class of square integrable functions on $\Q$. The sigmoidal condition on the activation function ensures this property, but is not necessary \citep[see, e.g.,][]{pinkus1999}.
\end{remark}

The construction of the shallow CovNet in Figure~\ref{fig:shallow_covnet} shows an immediate way to make the structure \emph{deep} by augmenting more layers. It is well-known that shallow networks may require a rather large width to approximate a function, whereas the same precision can be achieved by using a deep and less wide network \citep{eldan2016,liang2017,poggio2017}. Deep networks can capture more complex structures than the shallow networks with much fewer parameters. Moreover, during training, shallow networks are more prone to get stuck at bad local minima, which are usually avoided by deep networks \citep{choromanska2015}. In the next section, we extend the CovNet structure by incorporating deep networks instead of perceptrons in the construction.

\subsection{Deep architectures}\label{sec:deep_learning}
We start with a brief description of deep neural networks. For an integer $L>1$, an integer-tuple $\mathbf p = (p_1,\ldots,p_L)$, matrices $\mathrm W_1 \in \R^{p_1 \times d}$, $\mathrm W_2 \in \R^{p_2 \times p_1},\ldots,\mathrm W_L \in \R^{p_L \times p_{L-1}}$, vectors $\mathbf b_1 \in \R^{p_1},\ldots,\mathbf b_L \in \R^{p_L}, \mathbf w_{L+1} \in \R^{p_L}$, and $b_{L+1} \in \R$, define the function $g: \R^d \to \R$ which maps $\uvec \mapsto g(\uvec)$ recursively as follows:
\begin{align}\label{eq:deep_neural_network}
\uvec_1 &= \sigma(\mathrm W_1 \uvec + \mathbf b_1) \nonumber\\
\uvec_{l+1} &= \sigma(\mathrm W_{l+1} \uvec_l + \mathbf b_{l+1}) \quad \text{ for } l=1,\ldots,L-1, \nonumber\\
g(\uvec) &= \sigma(\mathbf w_{L+1}^\top \mathbf u_L + b_{L+1}).
\end{align}
Here, for a vector $\mathbf z \in \R^p$, $\sigma(\mathbf z)$ represents the component-wise application of the function $\sigma$. The function $g$ is a deep neural network, where $L$ is the number of hidden layers or \emph{depth}, $p_1,\ldots,p_L$ are the \emph{widths} of the hidden layers ($p_{\rm max} = \max\{p_1,\ldots,p_L\}$ is sometimes referred to as the width of the network) and $\mathrm W_1,\ldots,\mathrm W_L,\mathbf w_{L+1},\mathbf b_1,\ldots,\mathbf b_L, b_{L+1}$ are the network parameters. A schematic representation of the deep neural network is shown in Figure~\ref{fig:deep_covnet}(a). Starting from the input $\uvec$, we go to the first hidden layer by multiplying it with the weight matrix $\mathrm W_1$, adding the bias $\mathbf b_1$ and applying the activation function $\sigma$ component-wise on the resultant. The same structure is repeated for all the subsequent layers. We define the class
\begin{align}\label{eq:deep_neural_network_class}
\mathcal D_{L,\mathbf p} = \big\{g: \R^d \to \R \text{ of the form \eqref{eq:deep_neural_network} with } &\mathrm W_1 \in \R^{p_1 \times d}, \mathbf b_1 \in \R^{p_1}, \mathrm W_2 \in \R^{p_2 \times p_1}, \mathbf b_2 \in \R^{p_2},\ldots, \nonumber\\
&\kern2ex\mathrm W_L \in \R^{p_L \times p_{L-1}}, \mathbf b_L \in \R^{p_L}, \mathbf w_{L+1} \in \R^{p_L}, b_{L+1} \in \R\big\},
\end{align}
to be the class of all possible deep neural networks with depth $L$ and widths $p_1,\ldots,p_L$. A network from the class $\mathcal D_{L,\mathbf p}$ has $\sum_{l=0}^L (p_l+1) p_{l+1}$ parameters, where $p_0 = d$ and $p_{L+1} = 1$.

We define the deep CovNet kernel as
\begin{equation}\label{eq:deep_covnet_kernel}
c_{\rm d}(\uvec,\vvec) = \sum_{r=1}^R \sum_{s=1}^R \lambda_{r,s}\,g_r(\uvec)\,g_s(\vvec),\qquad \uvec,\vvec \in \Q,
\end{equation}
where $g_1,\ldots,g_R \in \mathcal D_{L,\mathbf p}$ and $\Lambda:=((\lambda_{r,s}))$ is positive semi-definite. This is similar to the shallow CovNet kernel \eqref{eq:shallow_covnet_kernel}, except the constituents $g_r(\uvec)$ are deep networks of the form \eqref{eq:deep_neural_network} instead of the perceptrons $\sigma(\mathbf w_r^\top\uvec+b_r)$. A schematic representation of the deep CovNet kernel is shown in Figure~\ref{fig:deep_covnet}(b). We define the class of deep CovNet kernels and the corresponding class of operators as
\begin{align}\label{eq:deep_covnet_class}
\F^{\rm d}_{R,L,\mathbf p,\sigma} &= \big\{c_{\rm d} \text{ of the form \eqref{eq:deep_covnet_kernel}}: \Lambda = ((\lambda_{r,s})) \succeq \mathrm 0, g_1,\ldots,g_R \in \mathcal D_{L,\mathbf p}\big\} \nonumber\\
\widetilde\F^{\rm d}_{R,L,\mathbf p,\sigma} &= \big\{\G: \G \text{ is the integral operator associated with kernel } g \in \F^{\rm d}_{R,L,\mathbf p,\sigma}\big\}.
\end{align}

\begin{remark}
As can be seen from the construction of the deep CovNet structure, it is possible to allow the individual networks $g_1,\ldots,g_R$ to have different depths and widths, allowing for more flexible models. However, this complicates the analysis, so we do not pursue this model in this paper.
\end{remark}

\begin{figure}[t!]
\centering
\begin{tabular}{c}
(a) Schematic representation of a deep neural network \\
\hspace{0.1in}
\includegraphics[height=1in]{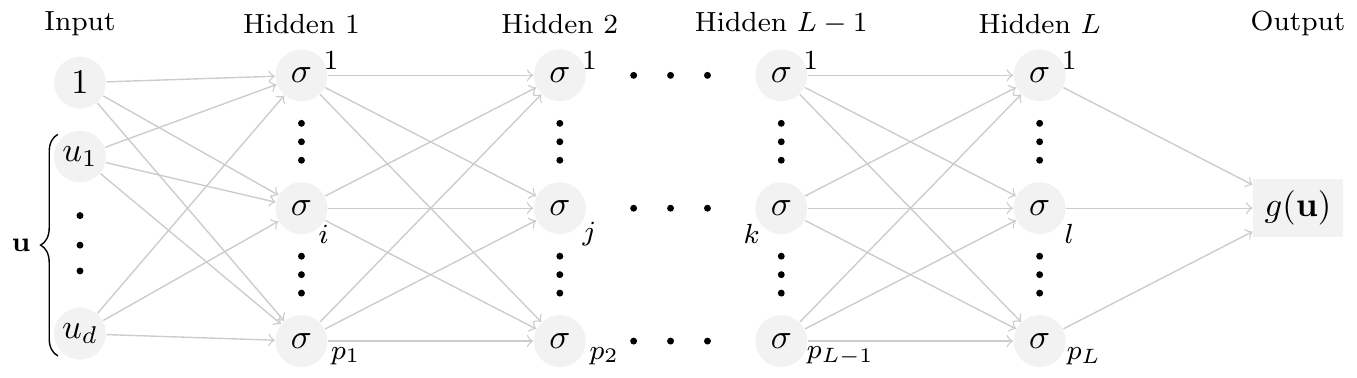} \\
\\
(b) Schematic representation of a deep CovNet structure \\
\hspace{0.1in}
\includegraphics[height=0.5\textheight]{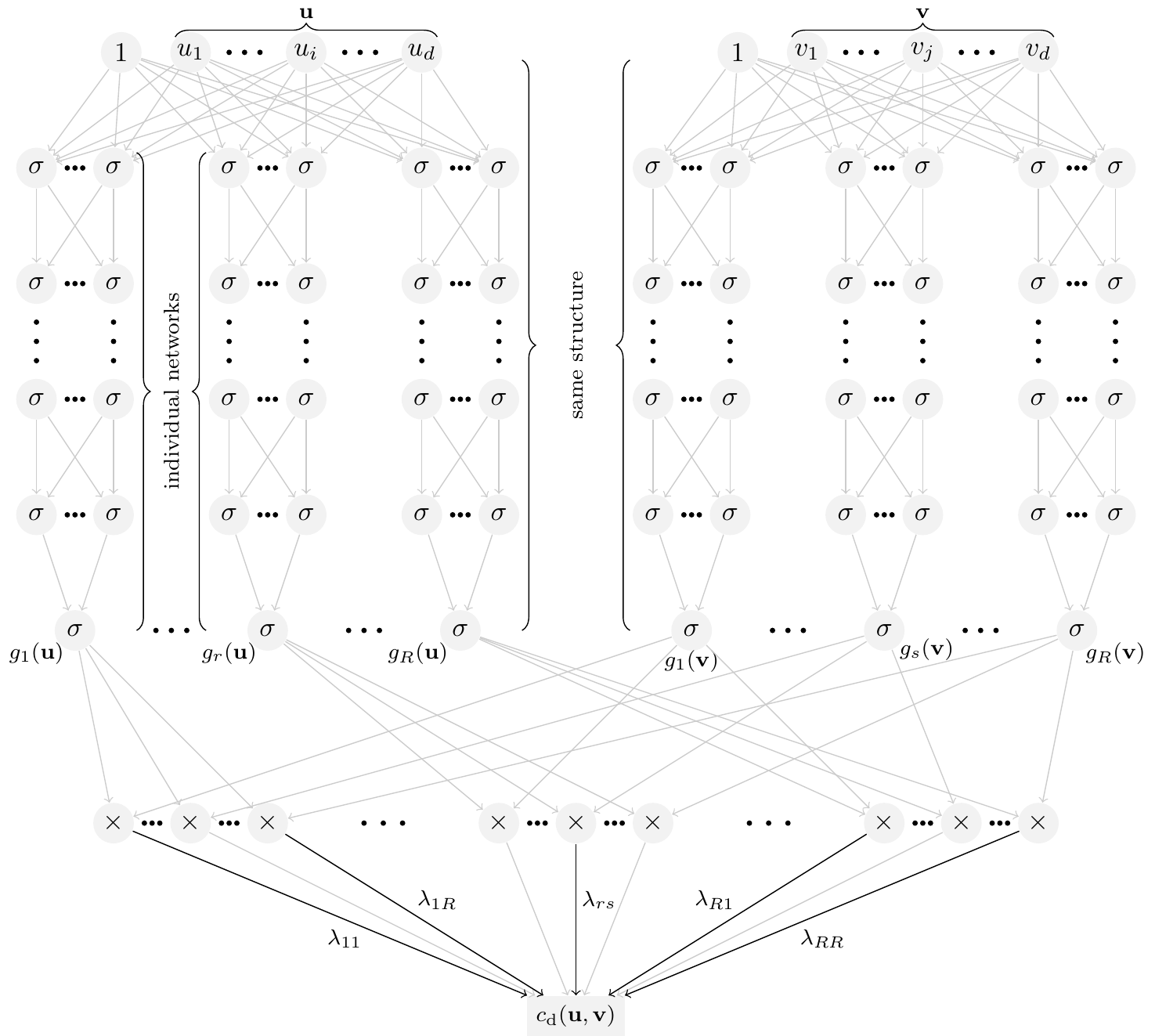} \\
\end{tabular}
\caption{\label{fig:deep_covnet}(a) A schematic representation of the deep neural network. (b) A schematic representation of the deep CovNet structure. Starting from the inputs $\uvec$ and $\vvec$, $R$ individual deep neural networks are fitted to get the outputs $g_1(\uvec),\ldots,g_R(\uvec)$ and $g_1(\vvec),\ldots,g_R(\vvec)$. The networks on left and right (corr.\ to $\uvec$ and $\vvec$) are the same. Cross-products of the outputs of the individual networks layer are taken in the next layer, which are then multiplied by weights $\lambda_{r,s}$ and added to produce the output $c_{\rm d}(\uvec,\vvec)$.}
\end{figure}

The deep CovNet model is quite rich. Moreover, it retains the universal approximation property of the shallow CovNet model (Theorem~\ref{thm:universal_approximation_deep}). But the number of parameters of the deep CovNet model can be quite large, making it prone to overfitting. Thus, some sort of regularization is needed for the deep CovNet structure to make it more stable. We do this by enforcing \emph{weight sharing} among the constituents as follows. For an integer $L>1$, integer-tuple $\mathbf p=(p_1,\ldots,p_L)$, matrices $\mathrm W_1 \in \R^{p_1 \times d}$, $\mathrm W_2 \in \R^{p_2 \times p_1},\ldots,\mathrm W_L \in \R^{p_L \times p_{L-1}}$, and vectors $\mathbf b_1 \in \R^{p_1},\ldots,\mathbf b_L \in \R^{p_L}$, we define the networks $g_1,\ldots,g_R$ jointly as 
\begin{align}\label{eq:deepshared_neural_network}
\uvec_1 &= \sigma(W_1\uvec+\mathbf b_1), \nonumber\\
\uvec_{l+1} &= \sigma(W_{l+1}\uvec_l + \mathbf b_{l+1}) \quad \text{ for } l=1,\ldots,L-1,\nonumber\\
g_r(\uvec) &= \sigma(\bm\omega_r^\top \uvec_L + \beta_r), \quad r=1,\ldots,R,
\end{align}
where $\bm\omega_r \in \R^{p_L}$ and $\beta_r \in \R$ for $r=1,\ldots,R$. Individually, each of the networks $g_1,\ldots,g_R$ is an element of the class $\mathcal D_{L,\mathbf p}$. But collectively, they share certain patterns among themselves, specifically they share all their parameters except for the ones in the final layer (see Figure~\ref{fig:deepshared_covnet}). We formally define the \emph{deepshared} CovNet kernel as
\begin{equation}\label{eq:deepshared_covnet_kernel}
c_{\rm ds}(\uvec,\vvec) = \sum_{r=1}^R\sum_{s=1}^R \lambda_{r,s}\,g_r(\uvec)\,g_s(\vvec),\qquad \uvec,\vvec \in \Q,
\end{equation}
where $g_1,\ldots,g_R$ are networks with shared structures as defined in \eqref{eq:deepshared_neural_network}. A schematic representation of the structure \eqref{eq:deepshared_covnet_kernel} is shown in Figure~\ref{fig:deepshared_covnet}. We also define the class of deepshared CovNet kernels and the corresponding operators as
\begin{align}\label{eq:deepshared_covnet_class}
\F^{\rm ds}_{R,L,\mathbf p,\sigma} &= \big\{c_{\rm ds} \text{ of the form \eqref{eq:deepshared_covnet_kernel}}: \Lambda = ((\lambda_{r,s})) \succeq \mathrm 0, g_1,\ldots,g_R \text{ of the form \eqref{eq:deepshared_neural_network}}\big\} \nonumber\\
\widetilde\F^{\rm ds}_{R,L,\mathbf p,\sigma} &= \big\{\G: \G \text{ is the integral operator associated with kernel } g \in \F^{\rm ds}_{R,L,\mathbf p,\sigma}\big\}.
\end{align}

\begin{figure}[t]
\centering
\includegraphics[height=0.5\textheight]{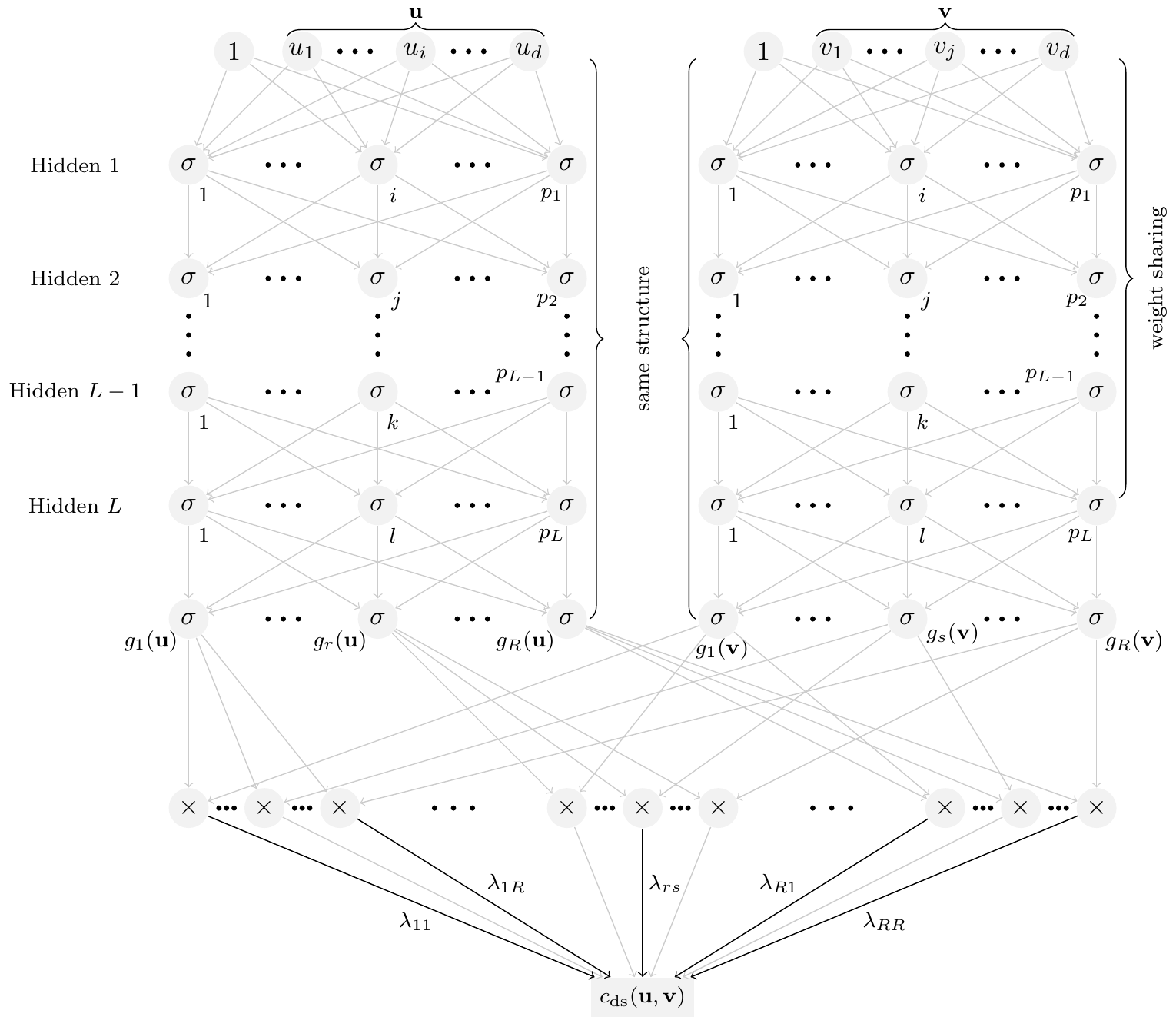}
\caption{\label{fig:deepshared_covnet}A schematic representation of the deepshared CovNet structure. Starting from the inputs $\uvec$ and $\vvec$, $R$ deep neural networks with shared weights are fitted to get the outputs $g_1(\uvec),\ldots,g_R(\uvec)$ and $g_1(\vvec),\ldots,g_R(\vvec)$. These outputs are combined by cross-multiplication, and finally addition after multiplication by weights $\lambda_{r,s}$, to produce $c_{\rm ds}(\uvec,\vvec)$.}
\end{figure}

The shared structure drastically reduces the number of parameters of the model. A deepshared CovNet kernel from the class $\F^{\rm ds}_{L,\mathbf p,R,\sigma}$ requires $\sum_{l=0}^{L-1} (p_l+1)p_{l+1} + R(p_L+1) + R(R+1)/2$ parameters, compared to $R \big(\sum_{l=0}^L (p_l+1) p_{l+1}\big) + R(R+1)/2$ for a deep CovNet kernel from the class $\F^{\rm d}_{L,\mathbf p,R,\sigma}$. To appreciate this, assume that each of the hidden layers have the same width $R$, i.e., $p_1=\cdots=p_L=R$. In this case, the deepshared kernel contains $\O(R^2)$ parameters, compared to $\O(R^3)$ parameters for the deep kernel.

\medskip

Similar to the shallow CovNet model, both the deep and the deepshared models are universal approximators, i.e., they can approximate any covariance kernel up to any desired accuracy.
\begin{theorem}\label{thm:universal_approximation_deep}
Let $\C$ be a covariance operator on $\L_2(\Q)$ with kernel $c$. Also, assume that the activation function $\sigma$ is sigmoidal. Then, for every $\epsilon > 0$, we can find a deep CovNet kernel $c_{\rm d}$ and a deepshared CovNet kernel $c_{\rm ds}$ such that 
\[
\int_{\Q \times \Q} \big|c(\uvec,\vvec) - c_{\rm d}(\uvec,\vvec)\big|^2 \diff\uvec\diff\vvec \le \epsilon \text{ and } \int_{\Q \times \Q} \big|c(\uvec,\vvec) - c_{\rm ds}(\uvec,\vvec)\big|^2 \diff\uvec\diff\vvec \le \epsilon.
\]
If in addition $c$ is continuous, then the same conclusion holds uniformly. That is, for every $\epsilon > 0$, we can find $c_{\rm d}$ and $c_{\rm ds}$ such that 
\[
\sup_{\uvec,\vvec \in \Q} \big|c(\uvec,\vvec) - c_{\rm d}(\uvec,\vvec)\big| \le \epsilon \text{ and } \sup_{\uvec,\vvec \in \Q} \big|c(\uvec,\vvec) - c_{\rm ds}(\uvec,\vvec)\big| \le \epsilon.
\]
\end{theorem}

Although both the models are universal approximators, the depth, width and the number of components $R$ of the approximator for the two models can be different. Moreover, although this result is similar to the one for the shallow model, we expect the deep and the deepshared approximator to have a much smaller number of parameters compared to the shallow approximator \citep{poggio2017}.
\begin{remark}\label{remark:deep_activation_choice}
The proof of the theorem again depends on the universal approximation property of deep neural networks, and the sigmoidal condition on the activation function is not necessary. From the proof of the theorem, one can see that the universal approximation property of the deep (and the deepshared) CovNet structure is guaranteed whenever the associated class of deep neural networks has the universal approximation property. In particular, deep and deepshared CovNet models with the highly popular ReLU activation: $\sigma(t) = \max\{t,0\}$ are also universal approximators.
\end{remark}

\medskip

Theorems~\ref{thm:universal_approximation_shallow} and \ref{thm:universal_approximation_deep} justify the use of covariance networks (shallow, deep or deepshared) for modeling the covariance kernel $c$. Given a fixed width $R$ and activation function $\sigma$, approximation by a shallow CovNet amounts to determining a $\G$ in the class $\widetilde\F^{\rm sh}_{R,\sigma}$ that is closest to $\C$ in terms of the Hilbert-Schmidt norm, i.e.,
\begin{equation}\label{eq:best_shallow_covnet}
\Chat^{\rm sh}_{R,\sigma} \in \argmin_{\G \in \widetilde\F^{\rm sh}_{R,\sigma}} \verti{\C-\G}_2^2.
\end{equation}
Given a sample of random fields $\X_1,\ldots,\X_N \overset{\iid}{\sim} \X$ in $\L_2(\Q)$, with covariance $\C$, we can replace $\C$ in \eqref{eq:best_shallow_covnet} by the empirical covariance operator $\Chat_N = N^{-1} \sum_{n=1}^N \X_n \otimes \X_n$ to obtain an estimator:
\begin{equation}\label{eq:estimator_shallow_covnet}
\Chat_{R,N}^{\rm sh} \in \argmin_{\G \in \widetilde\F^{\rm sh}_{R,\sigma}} \vertj{\Chat_N - \G}^2.
\end{equation}
We call this the shallow CovNet estimator. Similarly, for given width $R$, depth $L$ and activation $\sigma$, we can define the estimators based on the deep and deepshared models as
\begin{align}\label{eq:deep_and_deepshared_estimator}
\Chat^{\rm d}_{L,R,N} \in \argmin_{\G \in \widetilde\F^{\rm d}_{R,L,\mathbf p,\sigma}} \vertj{\Chat_N - \G}_2^2 \quad \text{ and } \quad \Chat^{\rm ds}_{L,R,N} \in \argmin_{\G \in \widetilde\F^{\rm ds}_{R,L,\mathbf p,\sigma}} \vertj{\Chat_N - \G}_2^2,
\end{align}
which we call the deep CovNet and the deepshared CovNet estimators, respectively. The estimators $\Chat^{\rm sh}_{R,N}$, $\Chat^{\rm d}_{L,R,N}$, and $\Chat^{\rm ds}_{L,R,N}$ can be seen as regularized versions of the empirical covariance, by projection into the corresponding CovNet classes. Nevertheless, it is crucial to note here that, although the definition of the estimators involve $\Chat_N$, we never actually need to form the empirical covariance in order to construct them. The estimators $\Chat_{R,N}^{\rm sh}$, $\Chat^{\rm d}_{L,R,N}$ and $\Chat^{\rm ds}_{L,R,N}$ can be computed directly at the level of the data, without the need to ever store or access the $2d$-dimensional object $\Chat_N$. We discuss the implementation details in the next section.

Note that although these estimators depend on the widths $p_1,\ldots,p_L$, we have suppressed it in the notation for ease of exposition. Also, in our numerical experiments, we have used $p_1=\cdots=p_L = R$, which justifies this notation. This choice is motivated by the empirical evidence that suggests using the same width for all the hidden layers \citep[see][Section~19.3.2]{bengio2012}.

\begin{remark}
Observe the notation in \eqref{eq:best_shallow_covnet}--\eqref{eq:deep_and_deepshared_estimator}. In either of the equations, we cannot guarantee that the minimizer is unique. Here, and throughout the article, by the notation $\widehat\G \in \argmin_{\G \in \widetilde\F} \vertj{\G - \C}_2^2$, we mean that $\widehat\G$ is an element (out of possibly many) of the class $\widetilde\F$ satisfying $\vertj{\widehat\G - \C}_2^2 \le \vertj{\G - \C}_2^2$ for all $\G \in \widetilde\F$. Note that this non-uniqueness does not affect the subsequent development, in particular the asymptotic theory for the estimator.
\end{remark}

\section{Practical implementation}\label{sec:implementation}

Note that for all three CovNet models \eqref{eq:shallow_covnet_kernel}, \eqref{eq:deep_covnet_kernel} and \eqref{eq:deepshared_covnet_kernel}, the covariance kernel is of the form
\begin{equation}\label{eq:covnet_kernel_generic}
\sum_{r=1}^R\sum_{s=1}^R \lambda_{r,s}\,g_r(\uvec)\,g_s(\vvec), \qquad \uvec,\vvec \in \Q,
\end{equation}
where $\Lambda:=((\lambda_{r,s}))$ is positive semi-definite, and $g_1,\ldots,g_R$ are allowed to vary keeping up to the model under consideration. In particular, $g_r(\uvec) = \sigma(\mathbf w_r^\top \uvec + b_r)$ for shallow CovNet, $g_r$'s are the individual deep neural networks from the class $\mathcal D_{L,\mathbf p}$ \eqref{eq:deep_neural_network_class} for deep CovNet, and $g_r$'s are jointly defined as in \eqref{eq:deepshared_neural_network} for deepshared CovNet. We denote the generic class of all such kernels (with the additional structures on the functions $g_1,\ldots,g_R$) by $\F_R$ and the corresponding class of operators by $\widetilde\F_R$. Thus, $\F_R$ (resp., $\widetilde\F_R$) can be $\F_{R,\sigma}^{\rm sh}$, $\F_{R,L,\mathbf p,\sigma}^{\rm d}$, or $\F_{R,L,\mathbf p,\sigma}^{\rm ds}$ (resp., $\widetilde\F_{R,\sigma}^{\rm sh}$, $\widetilde\F_{R,L,\mathbf p,\sigma}^{\rm d}$, or $\widetilde\F_{R,L,\mathbf p,\sigma}^{\rm ds}$) depending on the situation. Here, and throughout, we suppress the dependence on $L,p_1,\ldots,p_L$ and $\sigma$ for convenience, unless specifically mentioned.

For a given $R \in \N$, the CovNet structure \eqref{eq:covnet_kernel_generic} is completely determined by the parameters of $g_1,\ldots,g_R$, and the coefficients $\Lambda = ((\lambda_{r,s}))$. In particular, apart from $\Lambda$, these parameters are $\mathbf w_1,\ldots,\mathbf w_R, b_1,\ldots,b_R$ for the shallow CovNet \eqref{eq:shallow_covnet_kernel}, the weights and the biases of the individual deep neural networks $g_1,\ldots,g_R$ for the deep CovNet \eqref{eq:deep_covnet_kernel}, and $W_1,\ldots,W_L,\mathbf b_1,\ldots,\mathbf b_L$ and $\bm\omega_1,\ldots,\bm\omega_R,\beta_1,\ldots,\beta_R$ for the deepshared CovNet \eqref{eq:deepshared_covnet_kernel}. Thus, obtaining the estimators in \eqref{eq:estimator_shallow_covnet} or \eqref{eq:deep_and_deepshared_estimator} is equivalent to finding these parameters minimizing the corresponding criterion
\begin{equation*}
\ell := \ell(\Theta) = \vertj{\Chat_N - \G}_2^2.
\end{equation*}
Here, we use $\Theta$ to denote all the estimable parameters (i.e., the parameters of $g_1,\ldots,g_R$, and $\Lambda$), taking into account the positive-definiteness of $\Lambda$ (which reduces the number of \emph{free parameters}).

As already mentioned, we do not need to form the tensor $\Chat_N$ (or candidate tensor $\G$) to minimize $\ell$. The trick is to not fit the covariance directly, but to instead fit the observed fields $\X_1,\ldots,\X_N$ themselves by \emph{neural networks with shared structures}. To be precise, consider the fields
\begin{equation}\label{eq:observation_network}
\X^{\rm NN}_n(\uvec) = \sum_{r=1}^R \xi_{n,r}\,g_r(\uvec), \qquad n=1,\ldots,N,
\end{equation}
where $\xi_{n,r} \in \R$ for $n=1,\ldots,N$, $r=1,\ldots,R$, and $g_1,\ldots,g_R$ are the constituents of the CovNet model \eqref{eq:covnet_kernel_generic} under consideration. By construction, the fields $\X_1^{\rm NN},\ldots,\X_N^{\rm NN}$ are themselves neural networks with shared components $g_1,\ldots,g_R$ (and hence shared parameters), but potentially different coefficients $\xi_{n,r}$. Define the operator
\begin{equation}\label{eq:covnet_empirical_formulation}
\G_{R,N}^{\rm NN} = \frac{1}{N} \sum_{n=1}^N (\X^{\rm NN}_n - \bar\X^{\rm NN}) \otimes (\X^{\rm NN}_n - \bar\X^{\rm NN}),
\end{equation}
which is the empirical covariance based on the neural networks $\X^{\rm NN}_1,\ldots,\X^{\rm NN}_N$, and let $\widetilde\F^{\rm NN}_{R,N}$ be the class of all such covariance operators:
\begin{align}\label{eq:covnet_empirical_class}
\widetilde\F^{\rm NN}_{R,N} = \big\{\text{all empirical covariance operators of the form \eqref{eq:covnet_empirical_formulation}}\big\}.
\end{align}
Because of the shared structure of the networks $\X^{\rm NN}_n$, the kernel of the operator $\G^{\rm NN}_{R,N}$ has the CovNet structure \eqref{eq:covnet_kernel_generic}:
\begin{equation*}
g^{\rm NN}_{R,N}(\uvec,\vvec) = \sum_{r=1}^R \sum_{s=1}^R \lambda_{r,s}\,g_r(\uvec)\,g_s(\vvec),
\end{equation*}
where $\lambda_{r,s} = N^{-1}\sum_{n=1}^N (\xi_{r,n} - \bar{\xi}_r)(\xi_{s,n} - \bar{\xi}_s)$. At the same time, if $N > R$, any CovNet operator from the class $\widetilde\F_R$ can be written as an empirical covariance operator of the form \eqref{eq:covnet_empirical_formulation}. Specifically, for every $\G \in \widetilde\F_R$, we can find $N$ networks of the form \eqref{eq:observation_network} such that $\G$ is the empirical covariance operator of those $N$ networks. This should be intuitively clear, but we nevertheless state this formally below, and a detailed construction is shown in Appendix~\ref{supp:implementation}.
\begin{proposition}\label{prop:covnet_optimizer_equivalence}
If $N > R$, then $\widetilde\F^{\rm NN}_{R,N} = \widetilde\F_R$.
\end{proposition}

This simple correspondence between the CovNet operator class $\widetilde\F_R$ and the class of empirical covariances of neural networks with shared structure \eqref{eq:covnet_empirical_class} is of great consequence in estimating the CovNet model based on the observed data. Note that the criterion $\ell$ is non-convex in the parameters, so we cannot find the explicit minimizer. Instead, we need to rely on some iterative minimization procedure, e.g., gradient descent or its variants \citep[][Chapters~2 and 4]{buduma2017}. The application of gradient descent requires us to calculate the gradient of the minimization criterion. But, with modern optimization routines, this can be done numerically on a computer, without the need to compute the derivatives analytically. In particular, the special neural network structure of our method allows us to employ \emph{automatic differentiation} techniques to efficiently compute the derivative at machine precision \citep{baydin2018}. In essence, the minimizer can be efficiently obtained if we can compute the criterion efficiently. This is where the empirical covariance formulation \eqref{eq:observation_network} and \eqref{eq:covnet_empirical_formulation} come in handy. With this formulation, we compute the criterion $\ell$ as a function of the parameters of $g_1,\ldots,g_R$, and coefficients $\xi_{n,r}$ instead of $\lambda_{r,s}$. At each step of gradient descent, we obtain the fields $\X_n^{\rm NN}$ as feed-forward neural networks. The minimization criterion $\ell$ can be computed by simply computing inner-products between the observed fields $\X_n$ and the fitted networks $\X^{\rm NN}_n$, as shown below. For simplicity, we assume that the observed fields $\X_1,\ldots,\X_N$ are centered, so that the empirical covariance is $\Chat_N = N^{-1} \sum_{n=1}^N \X_n \otimes \X_n$. We also assume that the fitted networks $\X^{\rm NN}_1,\ldots,\X^{\rm NN}_N$ are centered, so that their empirical covariance is $\G^{\rm NN}_{R,N} = N^{-1} \sum_{n=1}^N \X^{\rm NN}_n \otimes \X^{\rm NN}_n$. With these, we get the following formula for the minimization criterion:
\begin{align*}
\ell := \verti{\Chat_N - \G^{\rm NN}_{R,N}}_2^2 = \frac{1}{N^2}\sum_{n=1}^N\sum_{m=1}^N \langle \X_n,\X_m\rangle^2 + \frac{1}{N^2}\sum_{n=1}^N\sum_{m=1}^N \langle \X^{\rm NN}_n,\X^{\rm NN}_m\rangle^2 - \frac{2}{N^2}\sum_{n=1}^N\sum_{m=1}^N \langle \X_n,\X^{\rm NN}_m\rangle^2.
\end{align*}
The detailed derivations are shown in Appendix~\ref{supp:loss_function_detailed}.

\begin{remark}\label{remark:psd_by_design}
The alternative formulation also helps us in imposing positive semi-definiteness on $\Lambda$. After estimating the parameters from the reformulated problem, we obtain $\lambda_{r,s}$ as $N^{-1} \sum_{n=1}^N (\xi_{n,r}-\bar\xi_r) (\xi_{n,s}-\bar\xi_s)$. By virtue of this construction, the resulting matrix $\Lambda = (\lambda_{r,s})$ is automatically positive semi-definite. This is quite useful, as it circumvents the need to work with a constrained optimisation problem on a cone.
\end{remark}

In practice, we observe the data on a grid of size $D = K_1 \times \cdots \times K_d$, say $\{\uvec_1,\ldots,\uvec_D\}$. Let us denote the $i$-th measurement corresponding to the $n$-th field by $X_{ni}$ for $n=1,\ldots,N$, $i=1,\ldots,D$. So, we can store the observed fields as an $N \times D$ matrix $\mathbf{X} = ((X_{ni}))$. Similarly, the fitted networks $\X^{\rm NN}_1,\ldots,\X^{\rm NN}_N$ can be evaluated at the $D$ grid points and all of these can be stored as an $N \times D$ matrix $\mathbf X^{\rm NN} = ((X^{\rm NN}_{ni}))$, where $X^{\rm NN}_{ni} = \X^{\rm NN}_n(\uvec_i)$. We can approximate $\langle \X_n,\X_m \rangle$ by the average over the grid points, i.e., $\langle \X_n,\X_m\rangle \cong D^{-1} \sum_{i=1}^D X_{ni} X_{mi}$. It is easy to see that this is the $(n,m)$-th element of the $N \times N$ matrix $D^{-1}\mathbf X \mathbf X^\top$. Similarly, we approximate $\langle \X^{\rm NN}_n,\X^{\rm NN}_m \rangle$ and $\langle \X_n, \X^{\rm NN}_m \rangle$ by the corresponding averages $D^{-1}\sum_{i=1}^D X^{\rm NN}_{ni}X^{\rm NN}_{mi}$ and $D^{-1} \sum_{i=1}^D X_{ni} X^{\rm NN}_{mi}$, which are the $(n,m)$-th elements of $D^{-1}\mathbf X^{\rm NN} \mathbf X^{\rm NN,\top}$ and $D^{-1}\mathbf X\mathbf X^{\rm NN,\top}$, respectively. Thus, apart from the computation of $\mathbf X^{\rm NN}$, the computational cost of $\ell$ is $\O(N^2D)$. Moreover, to store the model, we only need to store the parameters of $g_1,\ldots,g_R$, and the coefficient matrix $\Lambda$, which is completely free of the grid size $D$. In particular, this amounts to a storage cost of $\O(R^2 + Rd)$ for the shallow CovNet model, $\O\big(R^2 + Rp_L + R\sum_{l=0}^{L-1} (p_l+1)p_{l+1}\big)$ for the deep CovNet model, and $\O\big(R^2 + Rp_L + \sum_{l=0}^{L-1} (p_l + 1) p_{l+1}\big)$ for the deepshared CovNet model ($p_0 = d$ for the latter two). It is easy to see the savings relative to the empirical covariance, which requires $\O(ND^2)$ computations and $\O(D^2)$ storage.

In the above discussion, we have not addressed the computational requirements for $\mathbf X^{\rm NN}$. It is not difficult to show that for a fixed set of parameters, computation of the matrix $\mathbf X^{\rm NN}$ needs $\O(DR(N+d))$ operations for the shallow CovNet, $\O\big(DR(N+\sum_{l=0}^{L-1} p_l p_{l+1}+p_L)\big)$ operations for the deep CovNet, and $\O\big(D(NR + \sum_{l=0}^{L-1} p_l p_{l+1} + Rp_L)\big)$ operations for the deepshared CovNet (see Appendix~\ref{append:implementation_grid}). Thus, for a fixed set of parameters, the computational cost for the evaluation of $\ell$ remains linear in the grid size $D$ for all three CovNet models. Of course, we need to re-evaluate the criterion for each step of the gradient descent algorithm. But that is also the case for other modern machine learning methods. Moreover, the computation can be sped up by considering other techniques from machine learning, such as the stochastic or mini-batch version of gradient descent and parallel computing \citep{bengio2012,buduma2017}.

\begin{remark}
We have not tried to find analytic expression for the derivative of $\ell$ as a function of the parameters. Instead, we focused more on evaluating the criterion efficiently, and rely on automatic differentiation to compute the gradient. There are three reasons for doing this. Firstly, because of the complex neural network structure, finding analytic expressions for the gradient is cumbersome. This becomes more relevant for the deep and the deepshared CovNet structures. Secondly, we have at our disposal modern optimization routines, which are very efficient in automatic differentiation, especially with neural network structures such as ours. In our code, we have used the \texttt{autograd} feature of \texttt{pytorch} (\url{https://pytorch.org/}). Finally, even if we compute the derivatives analytically, when implementing the method on a computer the accumulation of errors for analytic derivatives may sometimes be quite large, especially for complex structures such as neural networks. Automatic differentiation, on the other hand, produces results which are exact up to machine precision, and thus are preferred to analytic derivatives \citep{baydin2018}.
\end{remark}

\begin{remark}\label{remark:covnet_mean_estimation}
In our derivations, we have assumed that the fields $\X_1,\ldots,\X_N$ as well as $\X_1^{\rm NN},\ldots,\X_N^{\rm NN}$ are centered. In practice, we can center the observed fields by subtracting the mean (empirical or estimated by some other method), with negligible computational overhead. For the fitted fields $\X_1^{\rm NN},\ldots,\X_N^{\rm NN}$, because of their shared structure, the mean turns out to be $\bar\X^{\rm NN}(\uvec) = \sum_{r=1}^R \bar\xi_r\,g_r(\uvec)$. So, centering the fitted fields boils down to centering the coefficients $\xi_{n,r}$.  {We can use another approach, where we do not center the fields (observed or fitted) beforehand and minimize a slightly different criterion. In this case, we also get an estimate of the mean as a by-product (see Appendix\,\ref{sec:loss_function_mean} for details).}
\end{remark}

\section{Eigendecomposition of the estimated covariance operator}\label{sec:eigendecomposition}

Once we estimate the covariance, it is important to be able to \emph{manipulate} it, e.g., for regression, prediction, or even for visualization purposes. For such tasks, typical manipulations involve inverting the covariance operator or obtaining its eigendecomposition, either of which may be quite demanding in practice. For instance, for data observed on a grid of size $D = K_1 \times \cdots \times K_d$, the empirical covariance is stored as a $D \times D$ matrix. The inversion in this case requires $\O(D^3)$ operations, which is highly demanding and sometimes even prohibitive. Even if the inverse is constructed, it is available only at the $D\times D$ pre-specified locations. To evaluate the inverse (or the covariance itself, for that matter) at any other location, as required e.g., in kriging, one needs to apply some sort of interpolation or smoothing on a high-dimensional (in our case, $\R^d \times \R^d$) object, which can be even more demanding than the inversion itself. Finally, the cost of storing the inverse and/or the eigenfunctions adds another layer of burden. 

By contrast, the proposed CovNet estimators enjoy considerable advantage in this respect. The special form of the CovNet operators allow us to easily compute their eigendecomposition. Note that our estimated CovNet kernels are of the form
\begin{equation*}
\widehat c(\uvec,\vvec) = \sum_{r=1}^R \sum_{s=1}^R \widehat\lambda_{r,s}\, \widehat g_r(\uvec)\,\widehat g_s(\vvec), \qquad \uvec,\vvec \in \Q.
\end{equation*}
Thus, for the estimated covariance operator $\Chat$ and for any $f \in \L_2(\Q)$,
\begin{align*}
\Chat f(\uvec) = \int_{\Q} \widehat c(\uvec,\vvec) f(\vvec) \diff\vvec = \sum_{r=1}^R\sum_{s=1}^R \widehat\lambda_{r,s}\,\widehat g_r(\uvec) \int_{\Q} \widehat g_s(\vvec)\,f(\vvec) \diff\vvec = \sum_{r=1}^R a_r\,\widehat g_r(\uvec),
\end{align*}
where $a_r = \sum_{s=1}^R \widehat\lambda_{r,s} \int_{\Q} \widehat g_s(\vvec)\,f(\vvec)\diff\vvec$. This shows that the eigenfunctions of $\Chat$ are of the form $\psi(\uvec) = \sum_{r=1}^R a_r\,\widehat g_r(\uvec)$ for some $a_1,\ldots,a_R \in \R$. Now, for such a function $\psi$,
\begin{equation*}
\|\psi\|^2 = \sum_{r=1}^R\sum_{s=1}^R a_r\,a_s \int_{\Q} \widehat g_r(\uvec)\,\widehat g_s(\uvec)\diff\uvec = \sum_{r=1}^R \sum_{s=1}^R a_r\,a_s\,\widetilde g(r,s) = \mathbf a^\top\widetilde{\mathrm G}\,\mathbf a,
\end{equation*}
where $\widetilde g(r,s) = \int_{\Q} \widehat g_r(\uvec)\,\widehat g_s(\uvec)\,\diff\uvec$, $\mathbf a = (a_1,\ldots,a_R)^\top$ and $\widetilde{\mathrm G} = ((\widetilde g(r,s)))_{1\le r,s\le R}$. Also,
\begin{align*}
\langle \Chat \psi,\psi\rangle &= \iint_{\Q \times \Q} \widehat c(\uvec,\vvec)\,\psi(\uvec)\, \psi(\vvec) \diff\uvec\diff\vvec \\
&= \sum_{r=1}^R\sum_{s=1}^R \widehat\lambda_{r,s} \iint_{\Q \times \Q} \widehat g_r(\uvec)\,\widehat g_s(\vvec)\,\psi(\uvec)\,\psi(\vvec)\,\diff\uvec\,\diff\vvec \\
&= \sum_{r=1}^R\sum_{s=1}^R \widehat\lambda_{r,s} \sum_{i=1}^R\sum_{j=1}^R a_i\,a_j \int_{\Q} \widehat g_r(\uvec)\,\widehat g_i(\uvec)\diff\uvec \int_{\Q} \widehat g_s(\vvec)\,\widehat g_j(\vvec)\diff\vvec \\
&= \sum_{i=1}^R \sum_{j=1}^R a_i\,a_j \bigg(\sum_{r=1}^R \sum_{s=1}^R \widehat\lambda_{r,s}\,\widetilde g(r,i)\, \widetilde g(s,j)\bigg) = \sum_{i=1}^R \sum_{j=1}^R a_i\,a_j (\widetilde{\mathrm G}\Lambda\widetilde{\mathrm G})_{i,j} = \mathbf{a}^\top \widetilde{\mathrm G}\Lambda\widetilde{\mathrm G}\,\mathbf a.
\end{align*}
Thus, finding the leading eigenvalue and eigenfunction of $\Chat$ reduces to maximizing $\mathbf{a}^\top \widetilde{\mathrm G}\Lambda\widetilde{\mathrm G}\,\mathbf a$ subject to $\mathbf a^\top\widetilde{\mathrm G}\,\mathbf a = 1$. This amounts to solving
\begin{equation*}
(\widetilde{\mathrm G}\Lambda\widetilde{\mathrm G} - \eta \widetilde{\mathrm G})\,\mathbf a = \mathbf 0.
\end{equation*}
Again, if $\psi_i(\uvec) = \sum_{r=1}^R a_{i,r}\,\widehat g_r(\uvec)$, then we can similarly show that
\begin{equation*}
\langle \psi_i,\psi_j\rangle = \mathbf a_i^\top \widetilde{\mathrm G}\,\mathbf a_j \text{ and } \langle \Chat\psi_i,\psi_j\rangle = \mathbf a_i^\top\widetilde{\mathrm G}\Lambda\widetilde{\mathrm G}\,\mathbf a_j,
\end{equation*}
where $\mathbf a_i = (a_{i,1},\ldots,a_{i,R})^\top$ is the vector of coefficients of $\psi_i$. Thus, finding the subsequent eigenvalues and eigenfunctions also amounts to solving $(\widetilde{\mathrm G}\Lambda\widetilde{\mathrm G} - \eta\widetilde{\mathrm G})\,\mathbf a = \mathbf 0$, with added orthogonality constraints. In summary, finding the eigensystem of the CovNet operator $\Chat$ boils down to finding the solution of the generalized eigenvalue problem \citep[][Chapter~7]{golub2013} involving the non-negative definite matrices $\widetilde{\mathrm G}\Lambda\widetilde{\mathrm G}$ and $\widetilde{\mathrm G}$. Several optimization routines are available to obtain the solution. Also, this can be done very efficiently since the matrices $\Lambda$ and $\widetilde{\mathrm G}$ involved in the computations are of the order $R \times R$, and typical values of $R$ will be much smaller than $D=K_1\times\hdots\times K_d$. The matrix $\Lambda$ is obtained during the estimation procedure. The only bottleneck is the computation of the matrix $\widetilde{\mathrm G}$, which involves the integrals
\begin{equation*}
\widetilde g(r,s) = \int_{\Q} \widehat g_r(\uvec)\,\widehat g_s(\uvec)\diff\uvec.
\end{equation*}
These are integrals on a compact subset of $\R^d$. When $d$ is moderate, we can approximate the integral using Monte-Carlo methods, while for large $d$, we can resort to using quasi-Monte-Carlo methods \citep{dick2013}. In typical FDA applications, $d$ is $2,3$ or $4$ (corresponding to spatial/spatio-temporal data on $\R^2$ and $\R^3$), and it suffices to use Monte-Carlo integration. For this, we generate independent observations $\uvec_1,\ldots,\uvec_M$ distributed uniformly on $\Q$, and approximate the integral as
\[
\widetilde g(r,s) \simeq \frac{1}{M} \sum_{j=1}^M \widehat g_r(\uvec_j)\,\widehat g_s(\uvec_j).
\]
Also, when the functions $\widehat g_1,\ldots,\widehat g_R$ are bounded (e.g., when the activation $\sigma$ is sigmoidal), we can control the approximation error up to any desired accuracy by selecting $M$ large enough. After generating the observations $\uvec_1,\ldots,\uvec_M$, $\widetilde{\mathrm G}$ can be obtained by passing them through $\widehat g_1,\ldots,\widehat g_R$ to create an $R \times M$ matrix $\mathrm G$, and then computing the outer-product $\mathrm G {\mathrm G}^\top$. It can be verified that the overall computational cost remains linear in $M$ (see Appendix~\ref{append:implementation_grid}). So, even with a large value of $M$, the computational time is quite small. Moreover, after obtaining the eigendecomposition, the complete eigenstructure can be stored using an $R \times R$ matrix of coefficients $\widehat{\mathbf A} = (\widehat{\mathbf a}_1^\top,\ldots,\widehat{\mathbf a}_R^\top)^\top$ and a vector of eigenvalues $\widehat{\bm\eta} = (\widehat\eta_1,\ldots,\widehat\eta_R)$, in addition to the already estimated parameters. The usefulness of the eigendecomposition is shown in Section~\ref{sec:simulation_eigen}.

\section{Empirical study}\label{sec:empirical_demonstration}

We now demonstrate the usefulness of the proposed methods by means of a variety of numerical examples. We start with some simulated examples, where we generate the data from a Gaussian process \citep[][Chapter~1]{adler2007} on $[0,1]^d$ with mean $0$ and variance $\C$. We consider the following five choices for the kernel $c$.

\begin{enumerate}[{Ex}\,1]
\item \emph{Brownian sheet}: $c(\uvec,\vvec) = c_{\rm bm}(u_1,v_1) \times \cdots \times c_{\rm bm}(u_d,v_d)$ for $\uvec,\vvec \in [0,1]^d$, where $c_{\rm bm}(u,v) = \min\{u,v\}$ is the covariance of the standard Brownian motion \citep[][Sec.~1.4.3]{adler2007}.

\item \emph{Rotated Brownian sheet}: $c(\uvec,\vvec) = \tilde c(\mathrm O\uvec,\mathrm O\vvec)$, where $\mathrm O$ is a rotation matrix and $\tilde c$ is the covariance kernel of the Brownian sheet from Ex\,1.

\item \emph{Integrated Brownian sheet}: $c(\uvec,\vvec) = c_{\rm ibm}(u_1,v_1) \times \cdots \times c_{\rm ibm}(u_d,v_d)$ for $\uvec,\vvec \in [0,1]^d$, where $c_{\rm ibm}(u,v) = (u^2/2)\,(v-u/3) \1\{u\le v\} + (v^2/2)\,(u-v/3) \1\{u>v\}$ is the covariance of the integrated Brownian motion.

\item \emph{Rotated integrated Brownian sheet}: $c(\uvec,\vvec) = \tilde c(\mathrm O\uvec,\mathrm O\vvec)$, where $\mathrm O$ is a rotation matrix and $\tilde c$ is the covariance kernel of the integrated Brownian sheet from Ex\,3.

\item \emph{Mat\'ern covariance}: $c_{\nu}(\uvec,\vvec) = 2^{1-\nu}/\Gamma(\nu)\, (\sqrt{2\nu}\,\|\uvec-\vvec\|_d)^\nu\,K_{\nu}(\sqrt{2\nu}\,\|\uvec-\vvec\|_d)$, where $\Gamma$ is the gamma function, $K_{\nu}$ is the modified Bessel function of the second kind and $\|\cdot\|_d$ is the Euclidean distance on $\R^d$ \citep[][Chapter~4]{rasmussen2006}. The Mat\'ern covariance is indexed by the parameter $\nu>0$, which regulates its smoothness.
\end{enumerate}

Note that the covariance kernels in Ex\,1 and 2 are separable. We eliminate the separability in Ex\,3 and 4 by introducing a rotation of the domain. The Mat\'ern covariance in Ex\,5 is stationary and isotropic, but not separable for any finite $\nu$. On the other hand, none of the other covariances are stationary. Ex\,1 and 3 yield continuous but nowhere differentiable random fields, whereas Ex\,2 and 4 yield continuously differentiable random fields. For Ex\,5, the random fields are $\lceil \nu \rceil-1$ times differentiable in the mean-square sense.

We carried out our experiments with $d=2$ and $d=3$, which we refer to as 2D and 3D, respectively. For each experiment, we generated $N$ independent fields at $K \times \cdots \times K$ regular grid points on $[0,1]^d$. Henceforth, we refer to $K$ as the resolution. We used the three CovNet models (shallow, deep and deepshared) on the generated data to estimate $\C$. To facilitate comparison, we also consider the empirical covariance estimator and the best separable covariance estimator \citep[e.g.,][]{dette2020}. For each of these estimators, we compute the relative estimation error $\vertj{\Chat-\C}_2/\vertj{\C}_2$. Note that
\[
\vertj{\Chat - \C}_2^2 = \iint_{[0,1]^d \times [0,1]^d} \big(\widehat c(\uvec,\vvec) - c(\uvec,\vvec)\big)^2 \diff\uvec \diff\vvec ~~\text{ and }~~ \vertj{\C}_2^2 = \iint_{[0,1]^d \times [0,1]^d} c^2(\uvec,\vvec) \diff\uvec \diff\vvec
\]
cannot always be computed analytically. So, we use a Monte-Carlo approximation. We generate $M$ points $(\uvec_1,\vvec_1),\ldots,(\uvec_M,\vvec_M)$ from the uniform distribution on $[0,1]^d \times [0,1]^d$ and approximate
\[
\vertj{\Chat - \C}_2^2 \simeq \frac{1}{M} \sum_{i=1}^M \big(\widehat c(\uvec_i,\vvec_i) - c(\uvec_i,\vvec_i)\big)^2 ~~\text{ and }~~ \vertj{\C}_2^2 \simeq \frac{1}{M} \sum_{i=1}^M c^2(\uvec_i,\vvec_i).
\]
These are then used to approximate the relative errors of the estimators. The advantage of using Monte-Carlo is that we can control the approximation error up to any desired accuracy by selecting $M$ large enough. Also, by evaluating the estimators on a different set of locations than where the data were generated, we avoid committing an \emph{inverse crime} \citep{kaipio2005}. In particular, we used $M=50000$ in 2D and $M=100000$ in 3D.

\begin{figure}[t]
\centering
\begin{tabular}{cc}
(a) Brownian sheet & (b) Rotated Brownian sheet \\
\includegraphics[width=0.45\linewidth]{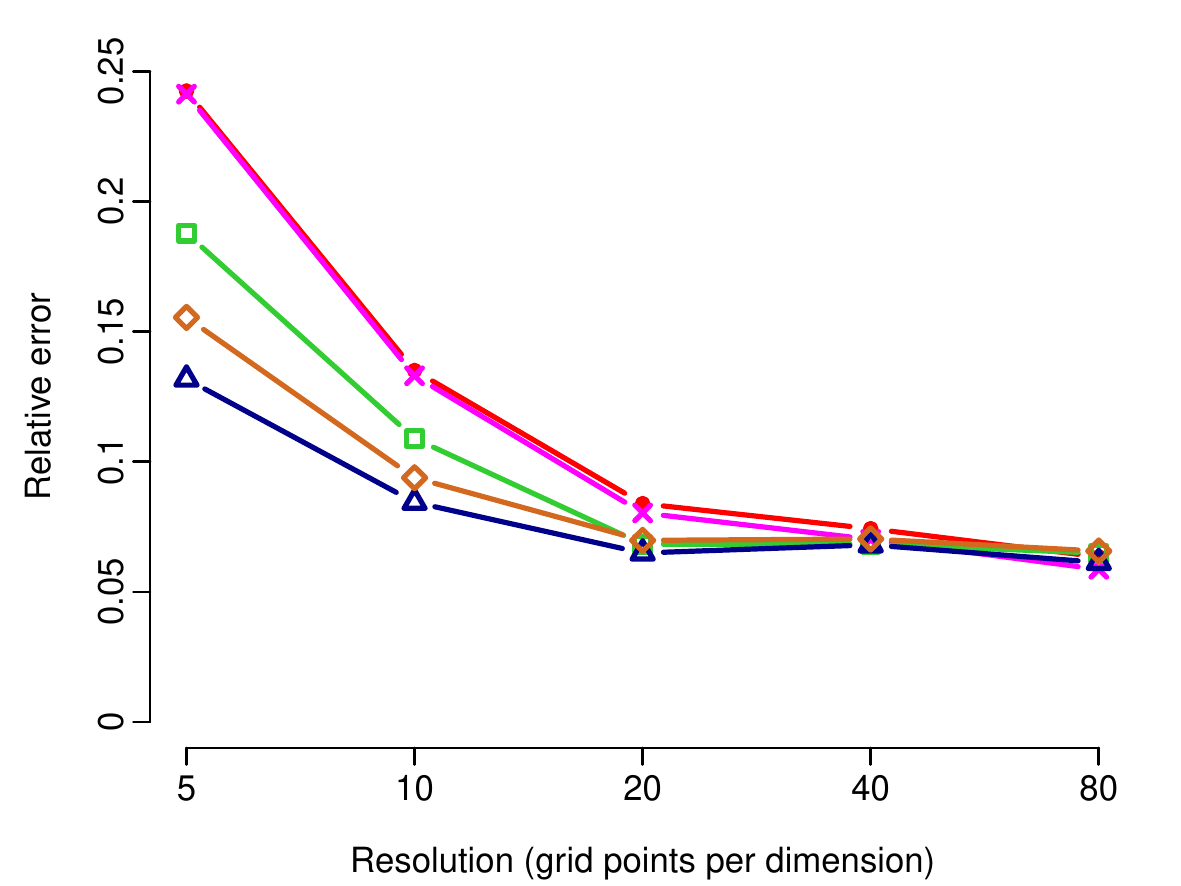} &
\includegraphics[width=0.45\linewidth]{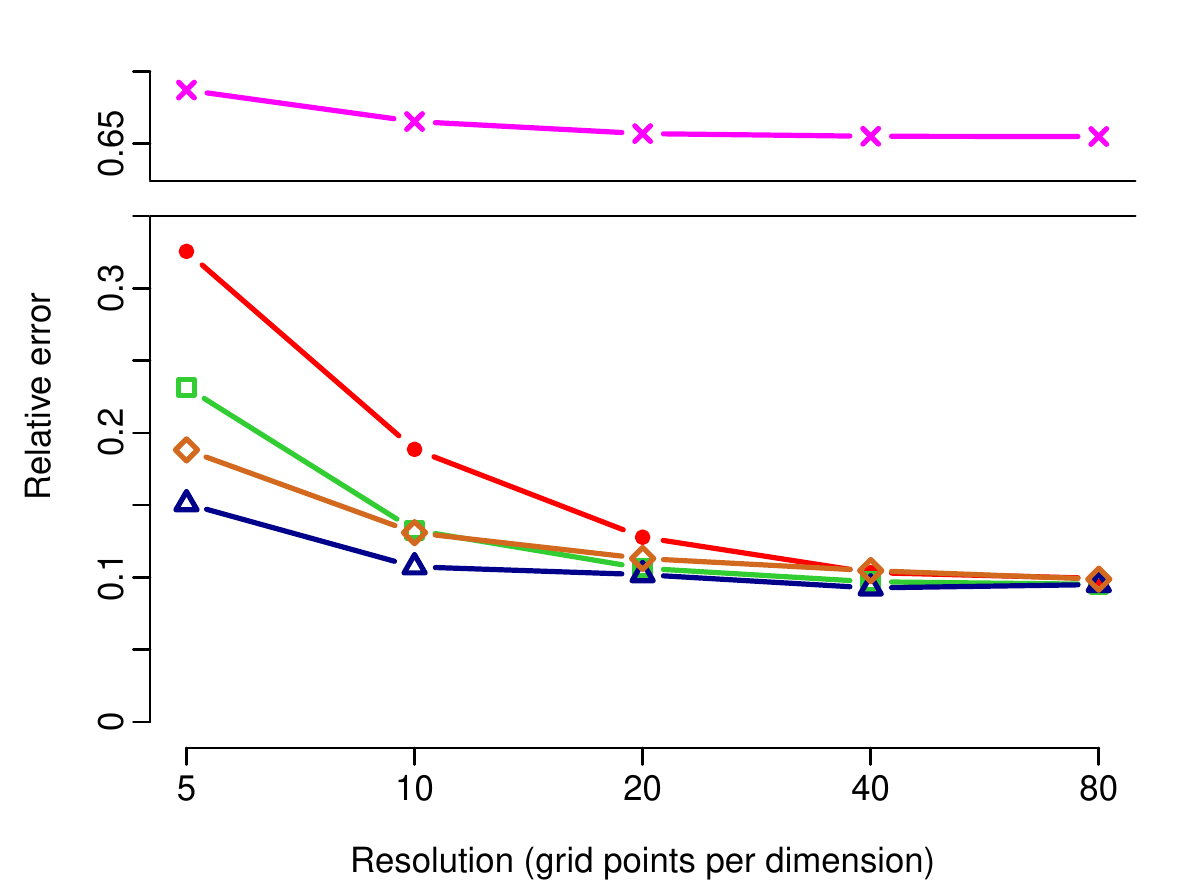} \\ [5pt]
(c) Integrted Brownian sheet & (d) Rotated integrated Brownian sheet \\
\includegraphics[width=0.45\linewidth]{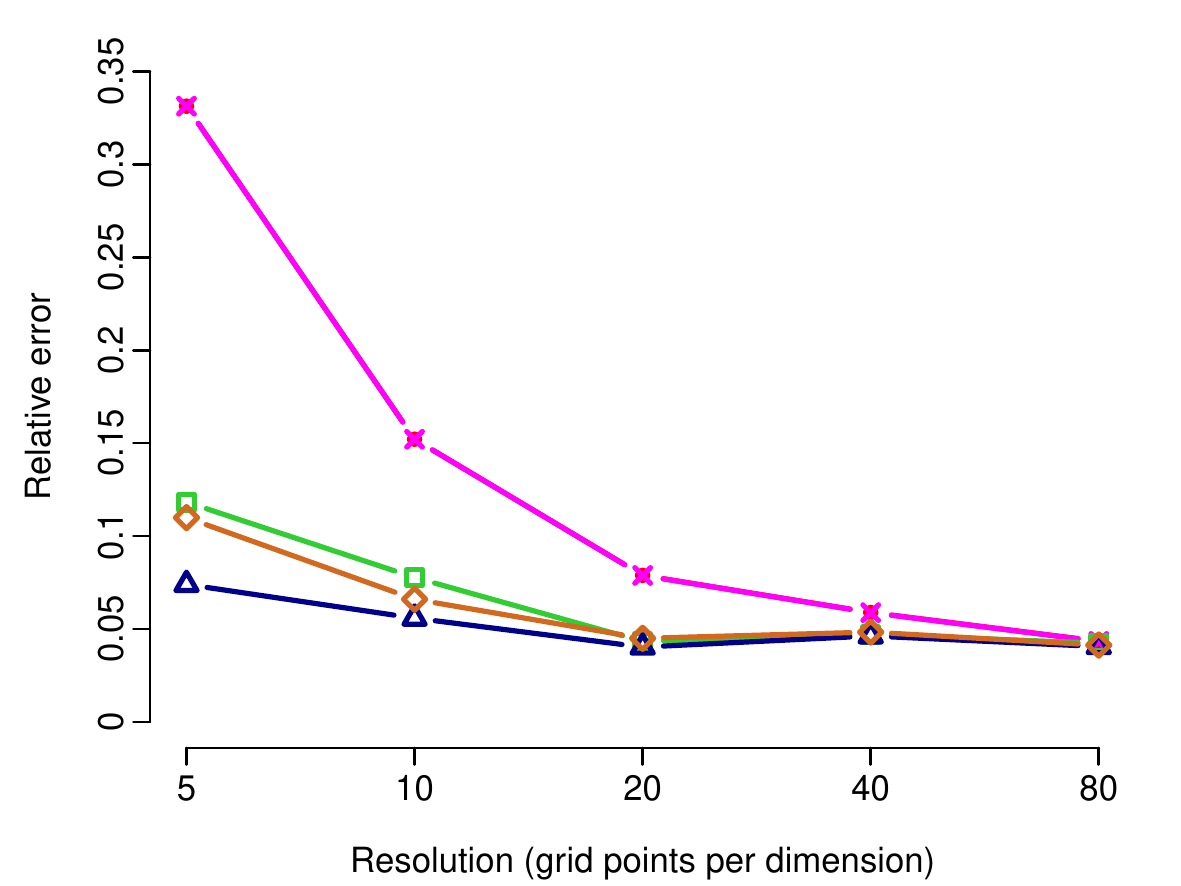} &
\includegraphics[width=0.45\linewidth]{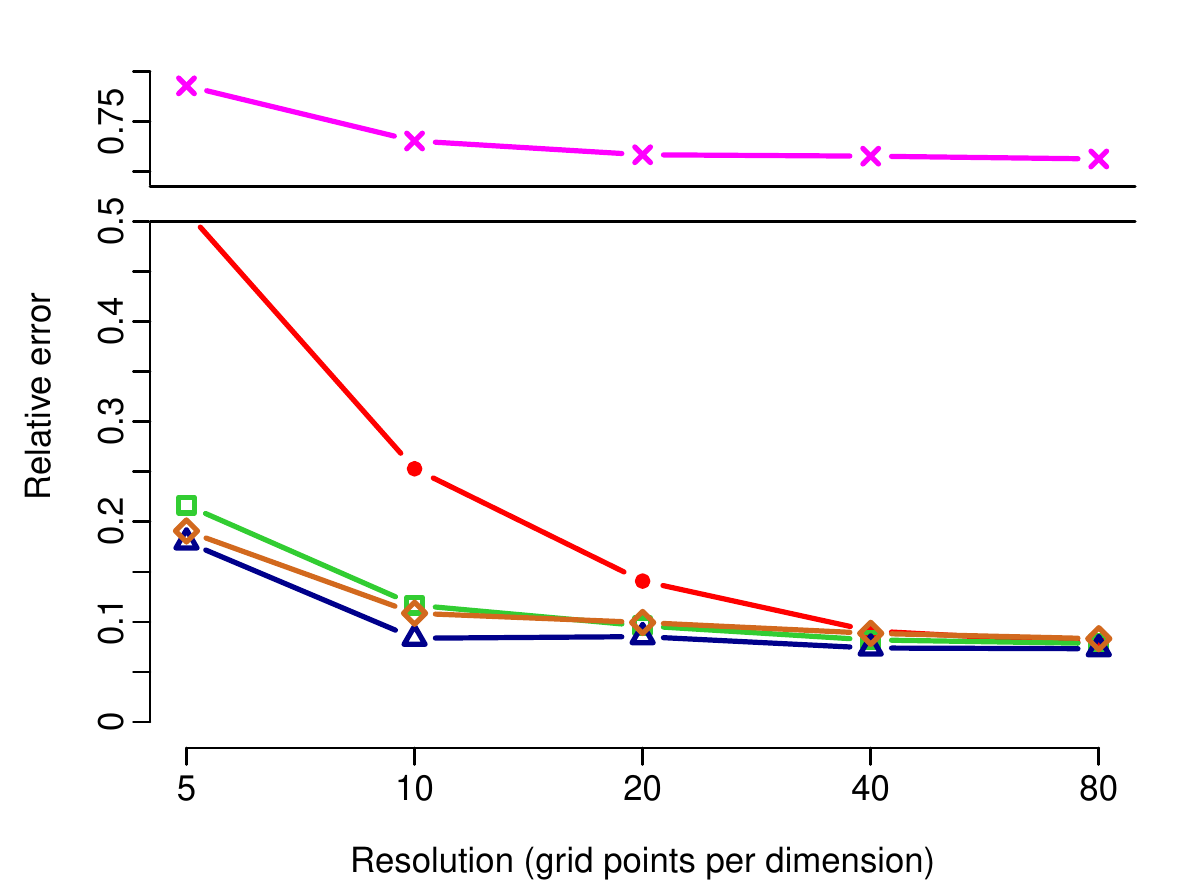} \\
\multicolumn{2}{c}{\textbf{Legend:} Empirical \Line[Emp]\Emp\Line[Emp]~~ Best separable \Line[Sep]\Sep\Line[Sep]~~ Shallow \Line[Sh]\Sh\Line[Sh]~~ Deep \Line[D]\D\Line[D]~~ DeepShared \Line[DS]\DS\Line[DS]}
\end{tabular}
\caption{\label{fig:BM_IBM_2D} Relative errors of different methods for different examples in 2D. Results are reported for a fixed sample size of $500$ and varying resolution. The numbers are averages based on $25$ simulation runs.}
\end{figure}

We also considered two different setups based on the sample size and the resolution: (a) fixed resolution $K$ and varying sample size $N$ and (b) fixed sample size $N$ and varying resolution $K$. Also, for the Mat\'ern example, we considered different values of $\nu$ with fixed sample size and resolution. For setup (a), the results are unremarkable -- the errors of all the estimators decrease as $N$ increases. These are reported in Appendix~\ref{supp:additional_simulations}. The results for setup (b) are rather interesting and exhibit the superiority of the CovNet estimators. We show these results for 2D in Figures~\ref{fig:BM_IBM_2D}--\ref{fig:Matern_2D}, where the reported numbers are the average relative errors based on $25$ simulation runs. The results for 3D are qualitatively similar, and we report them in Appendix~\ref{supp:additional_simulations}.

For the CovNet estimators, the results depend on the choice of hyperparameters $R$ and $L$. We used $R=5,10,20,40,80$ for the shallow CovNet model, and $L=2,3,4,R=5,10,20,40$ for the deep and the deepshared CovNet models in our experiments. In Figures~\ref{fig:BM_IBM_2D}--\ref{fig:Matern_2D}, we report the best result (i.e., minimum average estimation error) obtained by each CovNet model. In Section~\ref{sec:hyperparameter}, we discuss a practical method to select the hyperparameters and exhibit the corresponding results. It is seen there that the selection method yields values comparable to the ``best choice". For all the CovNet models, we used the standard sigmoid activation function $\sigma(t) = 1/(1+\exp(-t))$. For the optimization involved in fitting these models, we used the \texttt{ADAM} optimizer \citep{kingma2014} available in \texttt{pytorch}.

\begin{figure}[t]
\centering
\begin{tabular}{cc}
(a) $N=250$, resolution $25 \times 25$ & (b) $N=500$, $\nu=0.01$ \\
\includegraphics[width=0.45\linewidth]{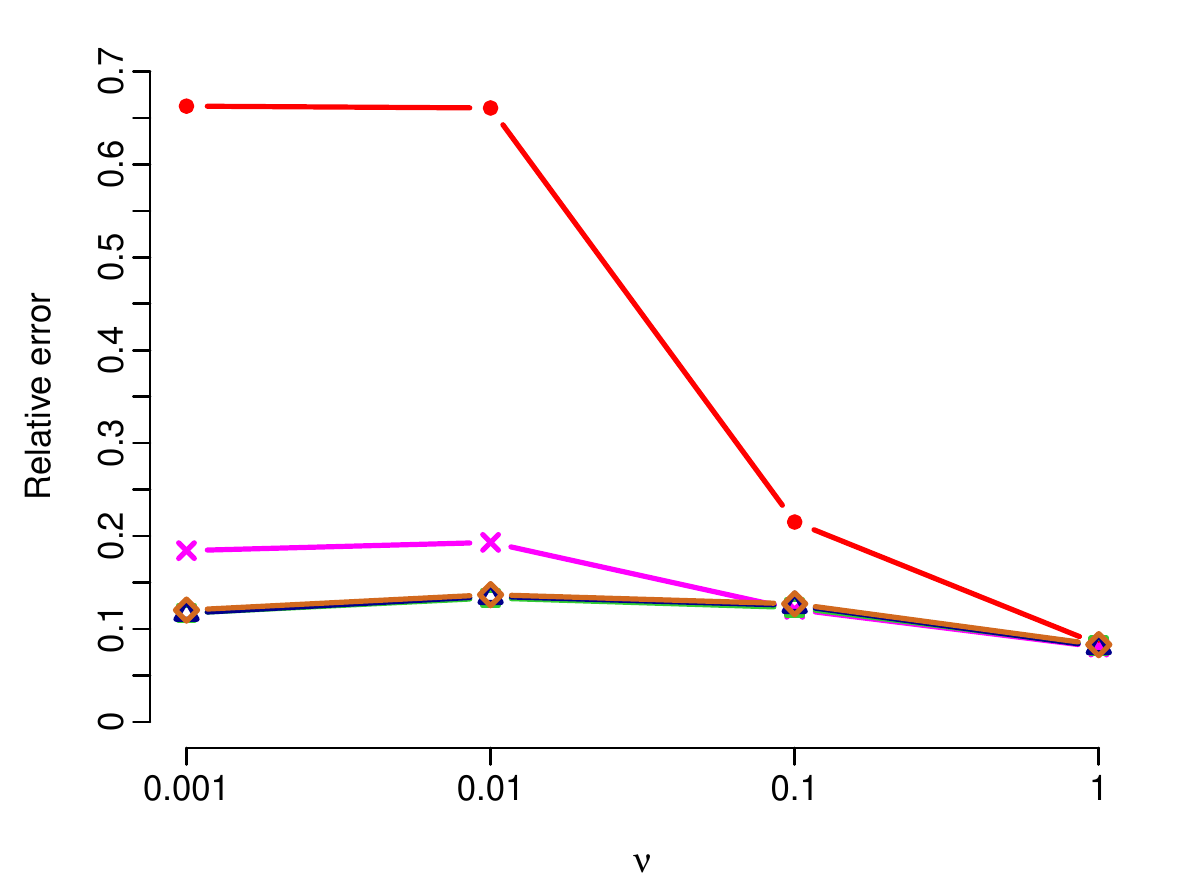} &
\includegraphics[width=0.45\linewidth]{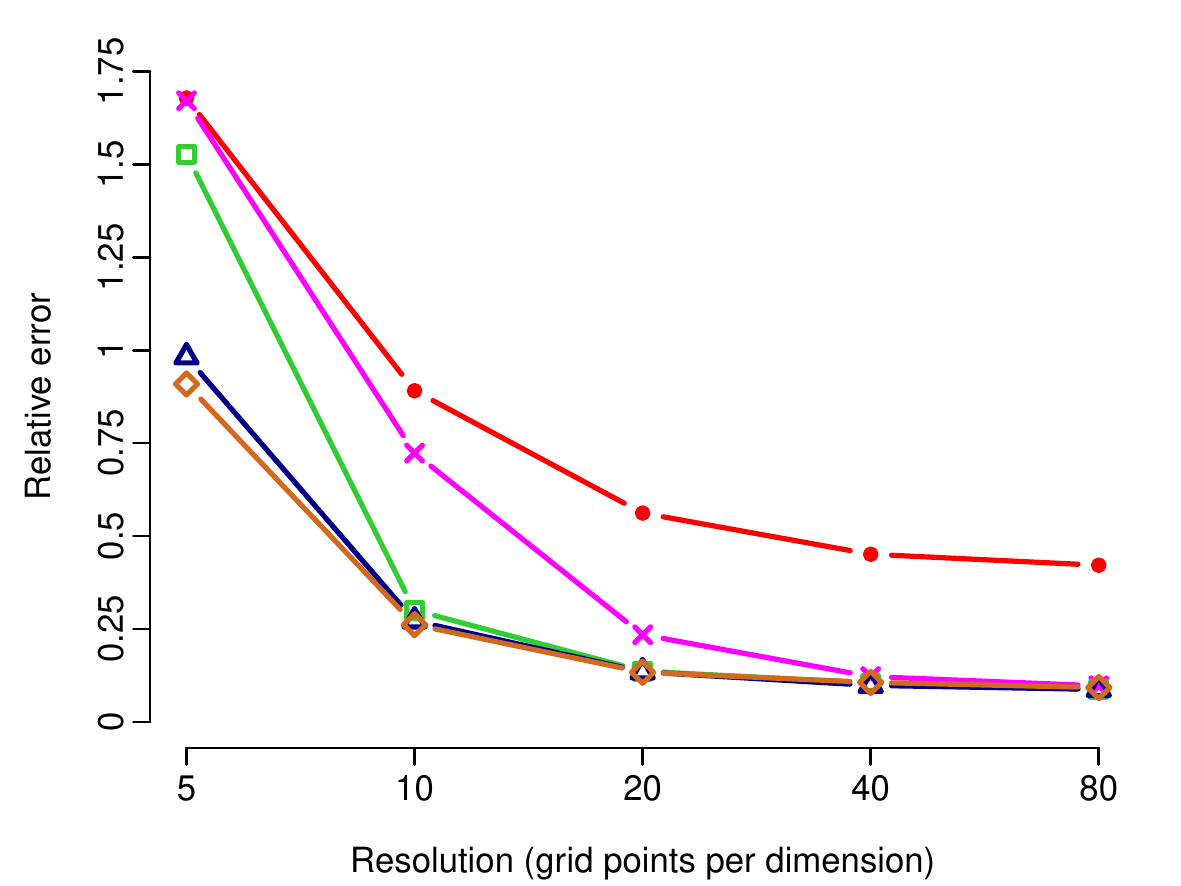} \\
\multicolumn{2}{c}{\textbf{Legend:} Empirical \Line[Emp]\Emp\Line[Emp]~~ Best separable \Line[Sep]\Sep\Line[Sep]~~ Shallow \Line[Sh]\Sh\Line[Sh]~~ Deep \Line[D]\D\Line[D]~~ DeepShared \Line[DS]\DS\Line[DS]}
\end{tabular}
\caption{\label{fig:Matern_2D} Relative errors of different methods for the Mat\'ern covariance model in 2D. In (a), results are for sample size $250$ and resolution $25 \times 25$ with varying smoothness parameter $\nu$. In (b), results are for sample size $500$ and $\nu=0.01$ with varying resolution. The numbers are averages based on $25$ simulation runs.}
\end{figure}

In Figure~\ref{fig:BM_IBM_2D}, we show the results for the first four examples (Ex\,1--4) in 2D with fixed sample size $N=500$ and varying resolution $K=5,10,20,40,80$. The Brownian sheet and the integrated Brownian sheet examples (Ex\,1 and 2) are separable. But even for these examples, the proposed CovNet estimators perform better than the best separable estimator (Fig.~\ref{fig:BM_IBM_2D}(a) and (c)), especially in low resolutions. This shows the ability of the CovNet model to learn the underlying pattern, even when we observe the fields at a rather small number of locations. For the rotated examples, we chose the matrix $\mathrm O$ to be the $45^\circ$-rotation along the $x$-axis:
\[
\mathrm O = \begin{pmatrix} 1/\sqrt{2} & -1/\sqrt{2} \\ 1/\sqrt{2} & 1/\sqrt{2} \end{pmatrix}.
\]
The absence of separability of $\C$ has dire consequence on the performance of the best separable estimator. The other estimators are seemingly unaffected by this, and the CovNet estimators outperform the empirical estimator. Among the CovNet estimators, the deepshared variant performed much better than the others. Recall that the integrated Brownian sheet (both the usual and the rotated) is one order smoother than the (corresponding version of) Brownian sheet. While this added smoothness enhances the performance of the CovNet estimators, we see an opposite effect on the other estimators, especially with small resolutions. The difference between the performance of the different CovNet models is also lesser in the smoother examples.

In Figure~\ref{fig:Matern_2D}, we show the results for the Mat\'ern covariance model (Ex\,5) in 2D. We consider two different setups. In panel (a), we show the results with $N=250$ and $K=25$ with varying $\nu$. Here, the empirical covariance performs very poorly, especially when the surfaces are rougher (i.e., for smaller values of $\nu$). The CovNet estimators perform better than the best separable estimator when $\nu$ is small. When $\nu$ is large, i.e., the surfaces are smoother, the errors of the best separable estimator is almost indistinguishable from those of the CovNet estimators. However, same relative error does not mean that the estimators share the same characteristic. In fact, in this example, the CovNet estimators have an advantage over the other estimators, which is evident from the eigendecomposition of the estimators (see Figure~\ref{fig:eigen_Matern}). Detailed discussion on this is given in Section~\ref{sec:simulation_eigen}. In panel (b), we report the results for $N=500$ and $\nu=0.01$ with varying resolution $K$. Here, we again see the superiority of the CovNet estimators, especially when the resolution is low.

\smallskip

A few words are in order about the cost of storage and manipulation of the estimators. The storage of the empirical covariance estimator becomes prohibitive rather quickly. Although we can compute the estimation error of the empirical covariance relatively easily, it is very costly to manipulate it, e.g., by inverting, for further applications like kriging. The scenario is much better for the best separable estimator. But, both the empirical and the best separable estimators produce a discretized object. Thus, even to evaluate the estimated covariance at a location outside of the observation grid, one needs to interpolate or smooth the estimated covariance. Depending on the smoother used, this can dramatically increase the cost associated with the estimator. The functional form of the CovNet estimators, on the other hand, do not suffer from such problems. After estimation, the storage of the model is quite cheap -- one only needs to store the matrices and vectors associated with the neural network model, which can be done very efficiently. Moreover, using the eigendecomposition methods discussed in Section~\ref{sec:eigendecomposition}, we can easily manipulate the fitted model.

\subsection{Choice of hyperparameters}\label{sec:hyperparameter}

The performance of the proposed method depends on the choice of hyperparameters, namely the number of components $R$ and the depth of the network $L$ (for deep and deepshared models). Thus, it is important to select these hyperparameters from the data, which is quite challenging for neural networks \citep{bengio2012}. We can use $V$-fold cross-validation for this purpose, where we split the data into $V$ parts. One of these $V$ parts is used as the \emph{validation set} and the rest are used as the \emph{training set}. For a particular choice of hyperparameters, the training set is used to fit the model, and its performance is evaluated on the validation set. This procedure is repeated for all the $V$ parts to get the average cross-validation score for a particular set of hyperparameters. Finally, we select the set of hyperparameters that admit the smallest average cross-validation score.

The alternative formulation of the loss function (cf.\ Section~\ref{sec:implementation}) is again useful for the cross-validation. Suppose that our data is split as $\X^{\rm tr}_1,\ldots,\X^{\rm tr}_{N_1}$ and $\X^{\rm va}_1,\ldots,\X^{\rm va}_{N_2}$, constituting the training and the validation sets, respectively. The model is fitted on the training set to produce the estimate $\widehat\G^{\rm tr}$. We evaluate the performance of the model on the validation set by computing the loss $\vertj{\Chat^{\rm va} - \widehat\G^{\rm tr}}_2^2$, where $\Chat^{\rm va}$ is the empirical covariance based on the validation set. Recall that by construction, both $\widehat\G^{\rm tr}$ and $\Chat^{\rm va}$ are of the form
\[
\widehat\G^{\rm tr} = \frac{1}{N_1}\sum_{n=1}^{N_1} \X^{\rm tr,NN}_n \otimes \X^{\rm tr,NN}_n ~~\text{ and }~~ \Chat^{\rm va} = \frac{1}{N_2}\sum_{n=1}^{N_2} \X^{\rm va}_n \otimes \X^{\rm va}_n,
\]
where $\X^{\rm tr,NN}_1,\ldots,\X^{\rm tr,NN}_{N_1}$ are the neural networks fitted to the training sample (see Section~\ref{sec:implementation}). Here, we have assumed w.l.o.g.\ that the observations are centered. Now, it is easy to see that the loss has the explicit form
\[
\vertj{\Chat^{\rm va} - \widehat\G^{\rm tr}}_2^2 = \frac{1}{N_1^2} \sum_{n=1}^{N_1}\sum_{m=1}^{N_1} \langle \X_n^{\rm tr,NN},\X_m^{\rm tr,NN} \rangle^2 + \frac{1}{N_2^2} \sum_{n=1}^{N_2}\sum_{m=1}^{N_2} \langle \X_n^{\rm va},\X_m^{\rm va} \rangle^2 - \frac{2}{N_1N_2}\sum_{n=1}^{N_1}\sum_{m=1}^{N_2} \langle \X_n^{\rm tr,NN},\X_m^{\rm va} \rangle^2,
\]
which depends only on the inner-products. Thus, we can compute the loss efficiently, without forming the high-order covariances $\Chat^{\rm va}$ or $\widehat\G^{\rm tr}$.

\begin{table}[t]
\centering
\caption{\label{table:errors_CV}Relative errors (in \%) of the CovNet models with hyperparameters chosen using $5$-fold cross-validation. Difference from the least observed error over the range of hyperparameters is shown in parentheses. Relative errors for the empirical and the best separable estimators are also reported. The reported numbers are based on one simulation run with $N=500$ and $K=25$.}
\begin{tabular}{lrrrrr}
Example & Empirical & Best separable & Shallow & Deep & Deepshared \\ [3pt]
Brownian sheet 2D  & $9.58$ &	$9.22$ & $9.26\,(0.41)$ & $7.93\,(0.24)$ & $8.72\,(0.35)$ \\
Rotated Brownian sheet 2D  & $11.79$ & $65.99$ & $10.03\,(0.64)$ & $10.36\,(0.69)$ & $9.70\,(0.17)$ \\ [2pt]
Integrated Brownian sheet 2D  & $7.99$ & $7.98$ & $7.34\,(0.06)$ & $9.58\,(3.37)$ & $7.44\,(0.09)$ \\
Rotated integrated Brownian sheet 2D & $11.02$ & $70.53$ & $6.93\,(0.00)$ & $6.69\,(0.04)$ & $6.15\,(0.36)$ \\ [2pt]
Matern 2D $\nu=0.001$ & $51.74$ & $17.65$ & $12.47\,(0.17)$ & $12.50\,(1.43)$ & $13.11\,(1.74)$ \\
Matern 2D $\nu=0.01$ & $51.52$ & $18.40$ & $14.25\,(0.62)$ & $13.95\,(2.56)$ & $13.05\,(0.00)$ \\
Matern 2D $\nu=0.1$ & $17.00$ & $11.68$ & $12.25\,(0.57)$ & $11.55\,(0.19)$ & $11.38\,(0.37)$ \\
Matern 2D $\nu=1$ & $8.23$ & $8.23$ & $8.14\,(0.00)$ & $8.73\,(0.63)$ & $8.00\,(0.00)$
\end{tabular}
\end{table}

The results for the proposed cross-validation strategy are shown in Table~\ref{table:errors_CV} for the examples in 2D. Note that cross-validation is time consuming, and a complete simulation study with cross-validation is rather difficult. So, for each of the examples considered in the previous section, we report relative errors for the three CovNet models (shallow, deep and deepshared) selected via cross-validation based on a single simulation run with $500$ observations at a resolution of $25 \times 25$. For each model, we also show the difference from the least observed relative error over the range of hyperparameters. The relative errors for the empirical and the best separable estimators are also reported to facilitate comparison. In all the examples, the average difference from the best result was less than $1\%$ for all the CovNet models, while the maximum difference was less than $1.75\%$ for the shallow and the deepshared models and less than $3.5\%$ for the deep CovNet model. These results clearly show that the cross-validation method can identify a good set of hyperparameters in practice.

\subsection{Estimated eigenstructure}\label{sec:simulation_eigen}
Here, we demonstrate the usefulness of the eigendecomposition of the CovNet estimators. For this purpose, we consider three examples in 2D, the rotated Brownian sheet (Ex\,2), the rotated integrated Brownian sheet (Ex\,4) and the Mat\'ern covariance (Ex\,5) with $\nu=0.01$. For each example, we plot the eigensurfaces of the CovNet estimators obtained using the method proposed in Section~\ref{sec:eigendecomposition}. The reported results are based on $500$ samples. For the first two examples, we used a resolution of $10 \times 10$. The Mat\'ern example with $\nu=0.01$ is much more rough, and a resolution of $10 \times 10$ was too low for all the methods (see Figure~\ref{fig:Matern_2D}(b)). So, for this example, we used a resolution of $25 \times 25$. For comparison, we have also plotted the leading eigensurfaces of the true covariance $\C$, the empirical estimator and the best separable estimator. For the CovNet estimators, we selected the hyperparameters via the cross-validation strategy described in the previous section.

\begin{figure}[t!]
\centering
\includegraphics[width=0.77\linewidth]{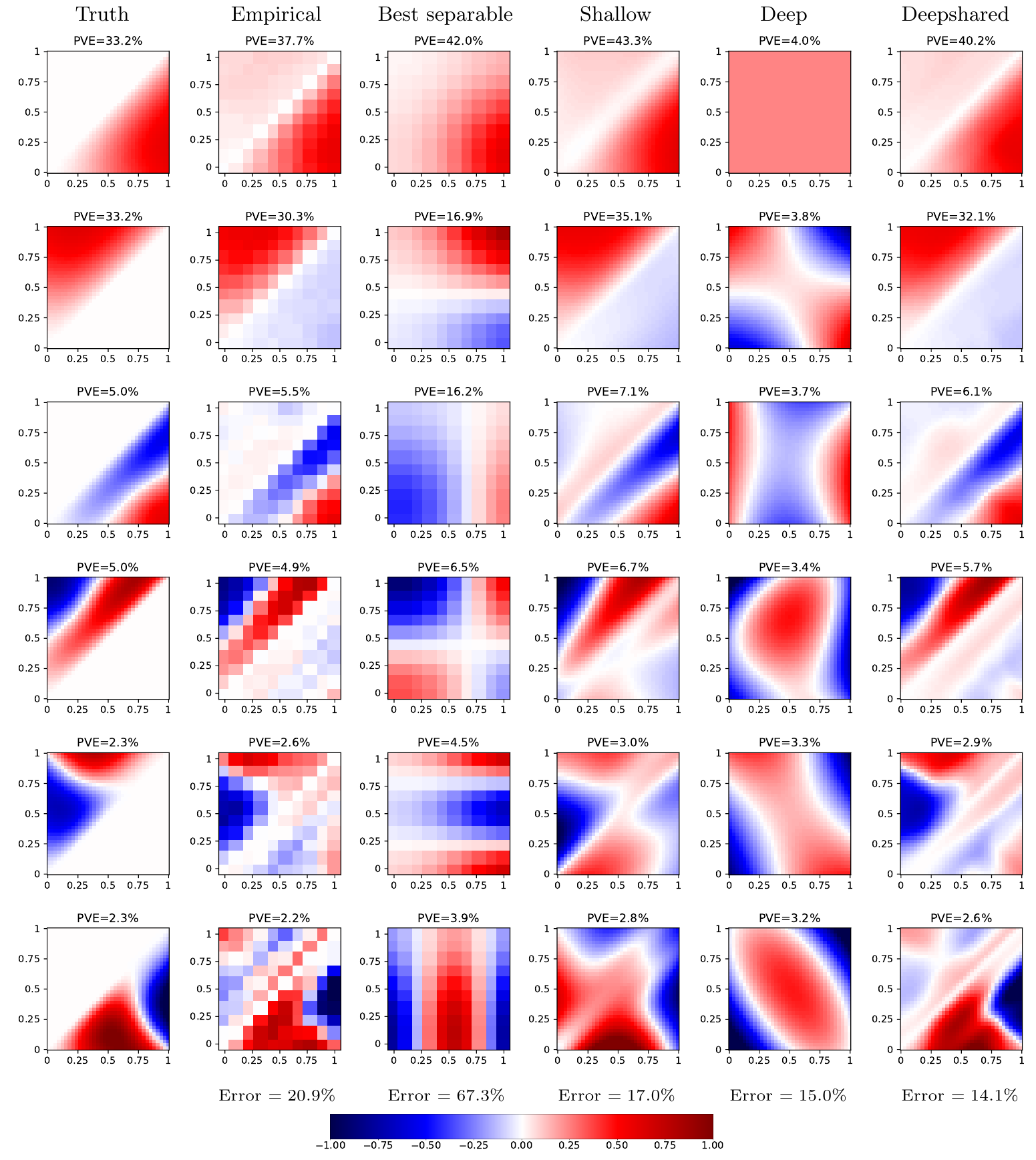}
\caption{\label{fig:eigen_BM_rotated}First six eigensurfaces of the true covariance and different covariance estimators in the rotated Brownian sheet example with $N=500$ and resolution $10 \times 10$. For the CovNet estimators, the eigenfunctions are computed using the methods described in Section~\ref{sec:eigendecomposition}. The plots are heatmaps of the 2D eigenfunctions.}
\end{figure}

In Figure~\ref{fig:eigen_BM_rotated}, we plot the eigensurfaces for the rotated Brownian sheet. Note that in this case, the  eigensurfaces of order seven and beyond explain less than $1\%$ of the total variation of true covariance. Hence, we report the first six eigensurfaces for each covariance (the truth and the estimators). The usefulness of the CovNet estimators is quite evident from these plots. Here, for each surface, we make observations at $100$ ($10 \times 10$) locations. Clearly, this is not enough for the empirical or the best separable estimators. In contrast, the shallow and the deepshared CovNet estimators are able to extract the features of the true covariance.

\begin{figure}[t]
\centering
\includegraphics[width=0.77\linewidth]{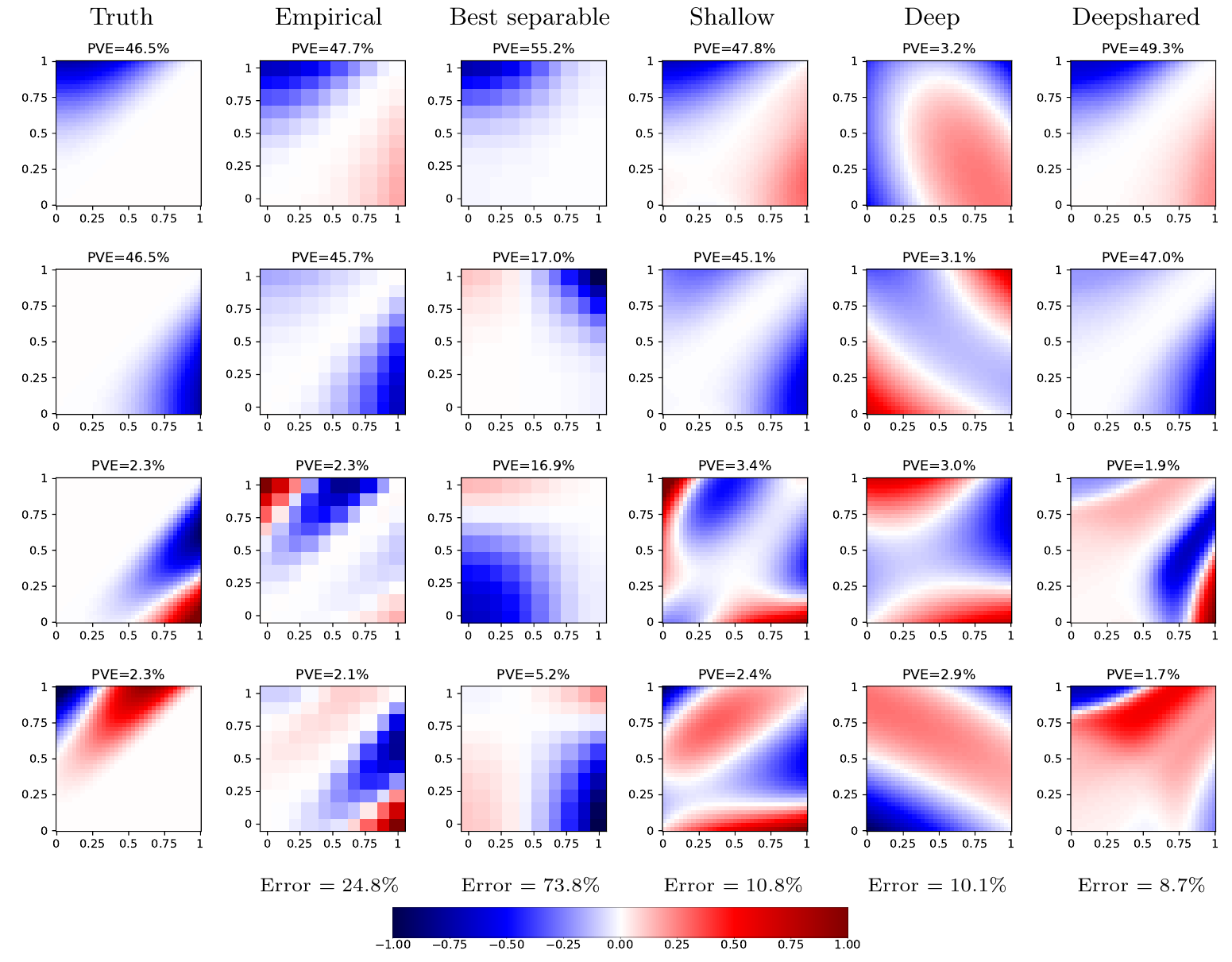}
\caption{\label{fig:eigen_IBM_rotated} First four eigensurfaces of the true covariance and different covariance estimators in the rotated integrated Brownian sheet example with $N=500$ and resolution $10 \times 10$. For the CovNet estimators, the eigenfunctions are computed using the methods described in Section~\ref{sec:eigendecomposition}. The plots are heatmaps of the 2D eigenfunctions.}
\end{figure}

The results for the rotated integrated Brownian sheet are shown in Figure~\ref{fig:eigen_IBM_rotated}. We plot the top four eigensurfaces as the other explain less than $1\%$ of the total variation. The true covariance is quite smooth in this example, as a result the empirical covariance does a better job. Even then, the functional form of the CovNet estimators gives them an edge, which is reflected in the estimation errors.

\begin{figure}[h!]
\centering
\includegraphics[width=0.77\linewidth]{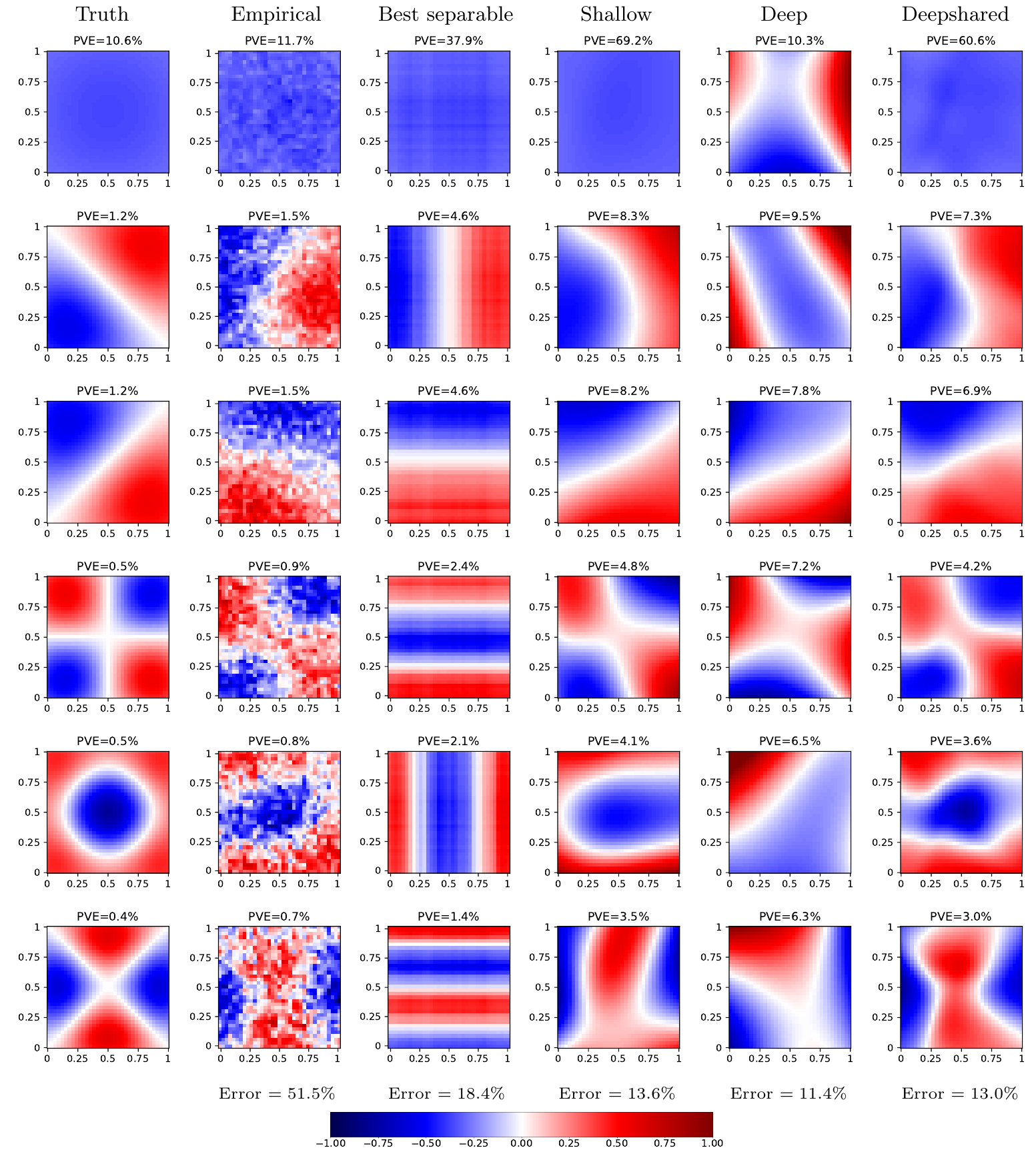}
\caption{\label{fig:eigen_Matern} First six eigensurfaces of the true covariance and different covariance estimators in the Mat\'ern example with $\nu=0.01$, $N=500$, and resolution $25 \times 25$. For the CovNet estimators, the eigenfunctions are computed using the methods described in Section~\ref{sec:eigendecomposition}. The plots are heatmaps of the 2D eigenfunctions.}
\end{figure}

In Figure~\ref{fig:eigen_Matern}, we show the results for the Mat\'ern covariance with $\nu=0.01$. In this case, the underlying process has rough sample paths. This roughness of the observations has a damaging effect on the performance of the empirical covariance. And, although the estimation error of the best separable estimator is relatively low, the estimated eigensurfaces have little resemblance with the true eigensurfaces. In fact, they are not able to capture the underlying features, and the roughness of the observations can be clearly seen to affect the performance. The shallow and the deepshared CovNet estimators do an excellent job in identifying the salient features, even from the rough observations.

A few comments are in order for the deep CovNet estimator. In all three examples, the deep CovNet estimator is seemingly unable to capture the true eigenstructure. However, the estimation error for this model is rather low. This is perhaps due to the high complexity of the model, which allows it to approximate the covariance well enough. But, without any restriction, the model apparently does not really learn \emph{interesting patterns of variation} from the data. The added restriction of the deepshared model (in terms of weight sharing) resolves this problem. The deepshared model is quite rich, but at the same time it is able to extract interesting traits from the data.

We would also like to point out the favorable computational aspect of the CovNet estimators in this context. As already mentioned, the eigendecomposition for the CovNet estimators can be performed without forming the covariance operators. This is sharply in contrast with the empirical covariance, for which we need to apply eigendecomposition on a $K^d \times K^d$-dimensional object, which can be prohibitive depending on $K$ and $d$. The best separable estimator is seemingly immune to this problem. But, even for this estimator, the computation for the eigendecomposition increases with $K$ in the order of $\O(dK^3)$ (eigendecomposition of $d$ matrices, each of the order $K \times K$). The eigendecomposition for the CovNet estimators, on the other hand, is completely free of $K$ (after estimation of the model, of course). There is, however, a Monte-Carlo step involved in the process. But, in all the examples, it took us only a few seconds to obtain the eigendecomposition, which was much faster than the other two estimators. Moreover, the functional form of the CovNet models have additional benefits, as can be seen from the plots.

\subsection{Application to fMRI data}\label{sec:fMRI}
To explore the usefulness of our methodology in a real setting, we consider the fMRI data sets from the \textsl{$1000$ Functional Connectomes Project}. These data sets are available at \url{https://www.nitrc.org/projects/fcon_1000/}, and consist of resting state fMRI scans for more than $1200$ subjects collected at different locations all over the world. For each subject, the data consist of 3D scans of the brain taken at a resolution of $64 \times 64 \times 33$ over $225$ time points. These data sets were previously analyzed by \cite{aston2012,stoehr2021} in the FDA setting, where they checked for the stationarity of the 3D MRI scans over time for each of $197$ individuals from Beijing, China. For our demonstration, we considered \texttt{sub69518} from Beijing, which was identified to exhibit stationarity  by \cite{aston2012}. This gives us $225$ 3D scans on a grid of size $64 \times 64 \times 33$, which we treat as i.i.d.\ observations. We preprocessed the data set by removing a polynomial trend of order $3$ from each voxel as suggested by \cite{aston2012}. Further, we scaled the scans to have voxel-wise unit variance before applying the methodology.

\begin{figure}[t!]
\centering
\begin{tabular}{c}
\small (a) Deepshared CovNet \\
\includegraphics[width=0.9\linewidth]{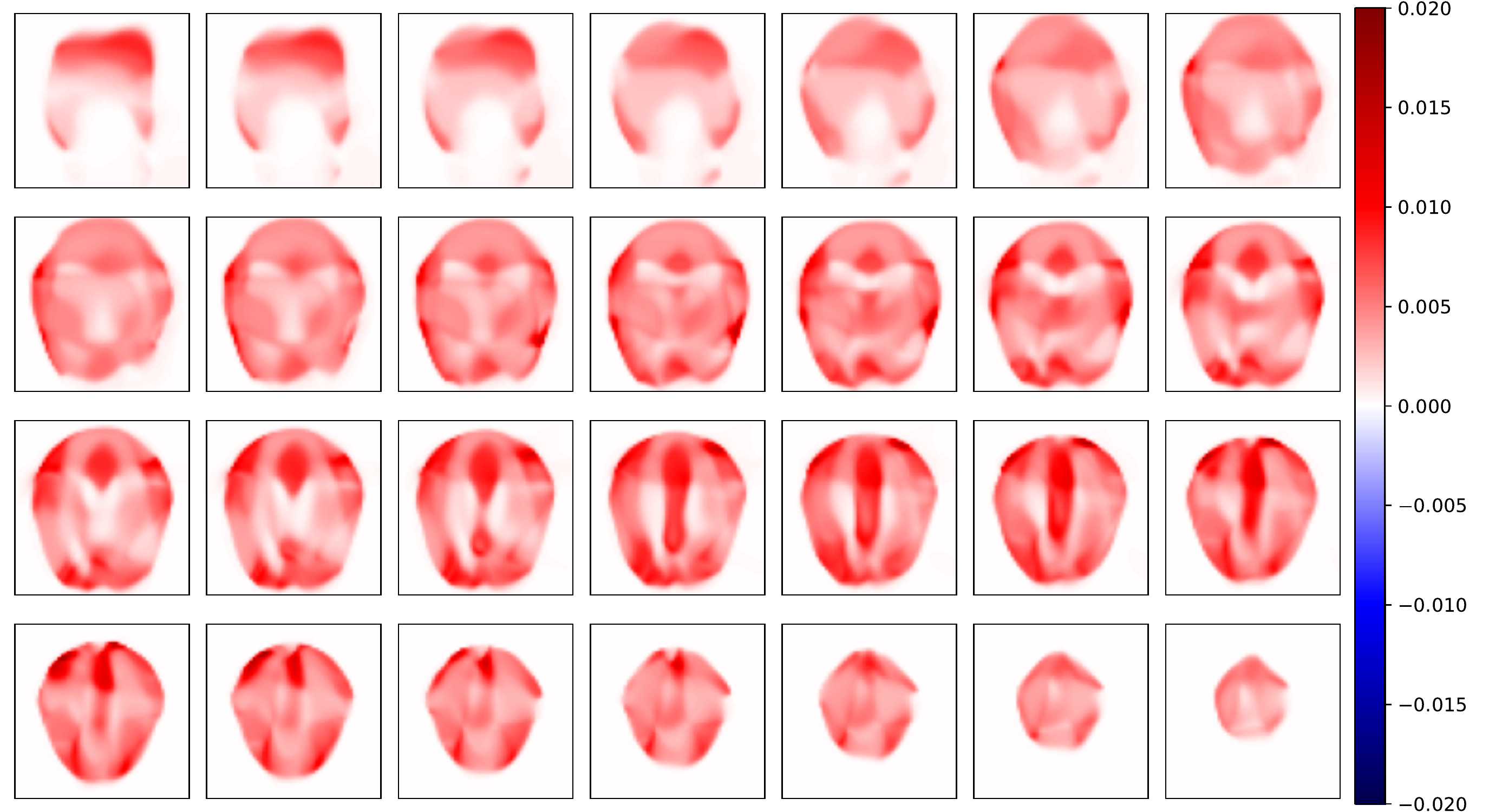}\\ [5pt]
(b) Separable Model \\
\includegraphics[width=0.9\linewidth]{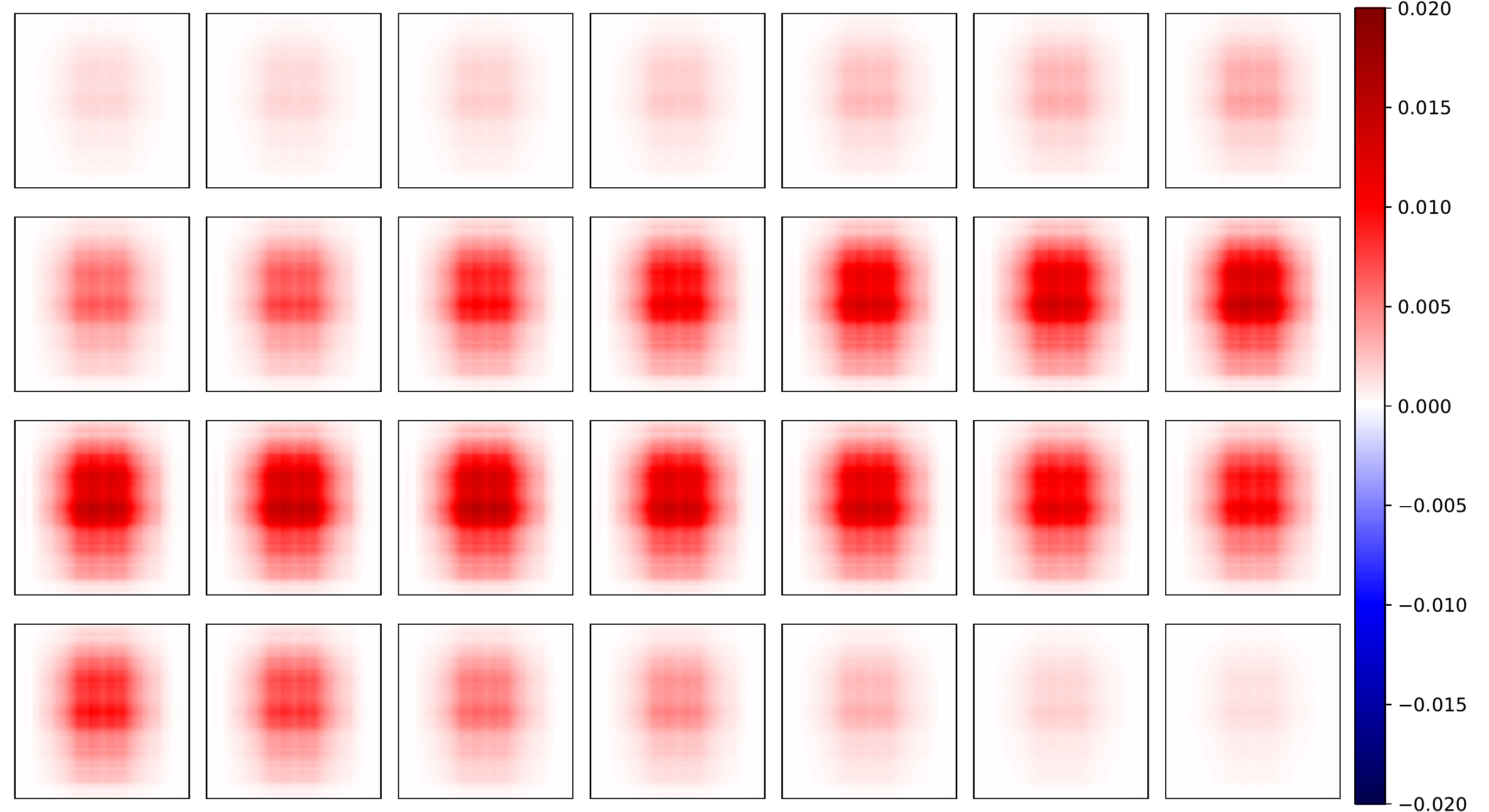}
\end{tabular}
\caption{The leading eigenfunction of (a) the deepshared CovNet and (b) the separable model, fitted to the 3D fMRI data. The hyperparameters of the deepshared CovNet (depth$=6$, $R=30$) were selected via cross-validation described in Section~\ref{sec:hyperparameter}. The plots are heatmaps of 2D slices of the 3D eigenfunctions, where the slices are taken over the z-axis.\label{fig:fMRI_deepshared_1}}
\end{figure}

As already mentioned in the introduction, covariance estimation is one of the most important problems for resting state fMRI data as it enables to understand the connectivity patterns of the brain. At the same time, the high-dimensionality of the problem makes it extremely difficult to achieve. In contrast, the CovNet models can be fitted efficiently to this data. In particular, we applied the deepshared CovNet model to this data owing to its superiority over the other CovNet models in the simulations. We selected the hyperparameters via the cross-validation strategy described in Section~\ref{sec:hyperparameter}. In Figure~\ref{fig:fMRI_deepshared_1}(a), we show the leading eigenfunction of the fitted CovNet model. For comparison, we also fitted a separable covariance model. But, as pointed out by \cite{aston2012}, finding the best separable approximation is also difficult in 3D. So, we used the separable estimator via marginalization as proposed by the authors. The leading eigenfunction of the separable estimator is shown in Figure~\ref{fig:fMRI_deepshared_1}(b).

\begin{figure}[t!]
\centering
\begin{tabular}{c}
\small (a) Second leading eigenfunction \\
\includegraphics[width=0.9\linewidth]{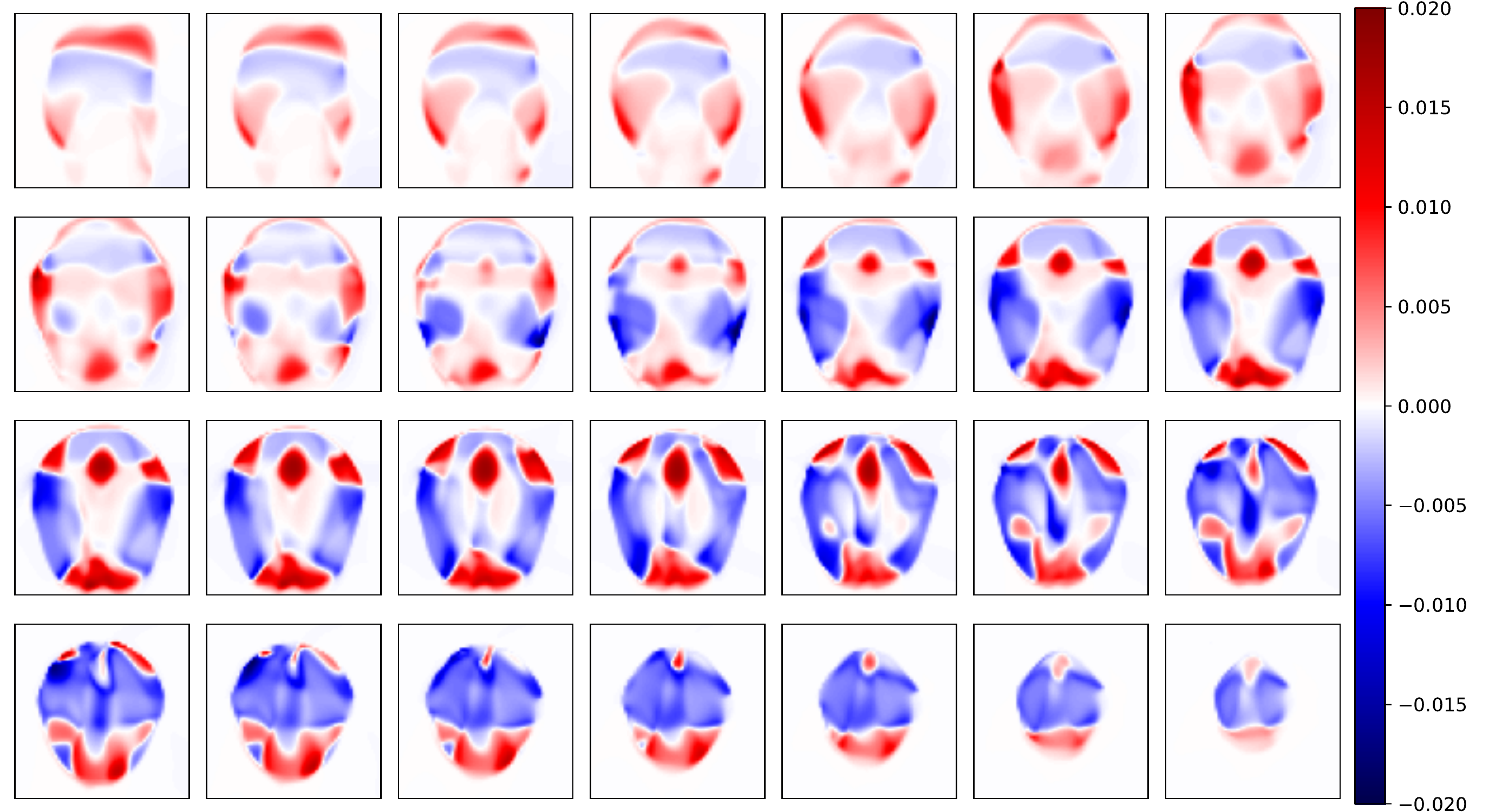}\\ [5pt]
(b) Third leading eigenfunction \\
\includegraphics[width=0.9\linewidth]{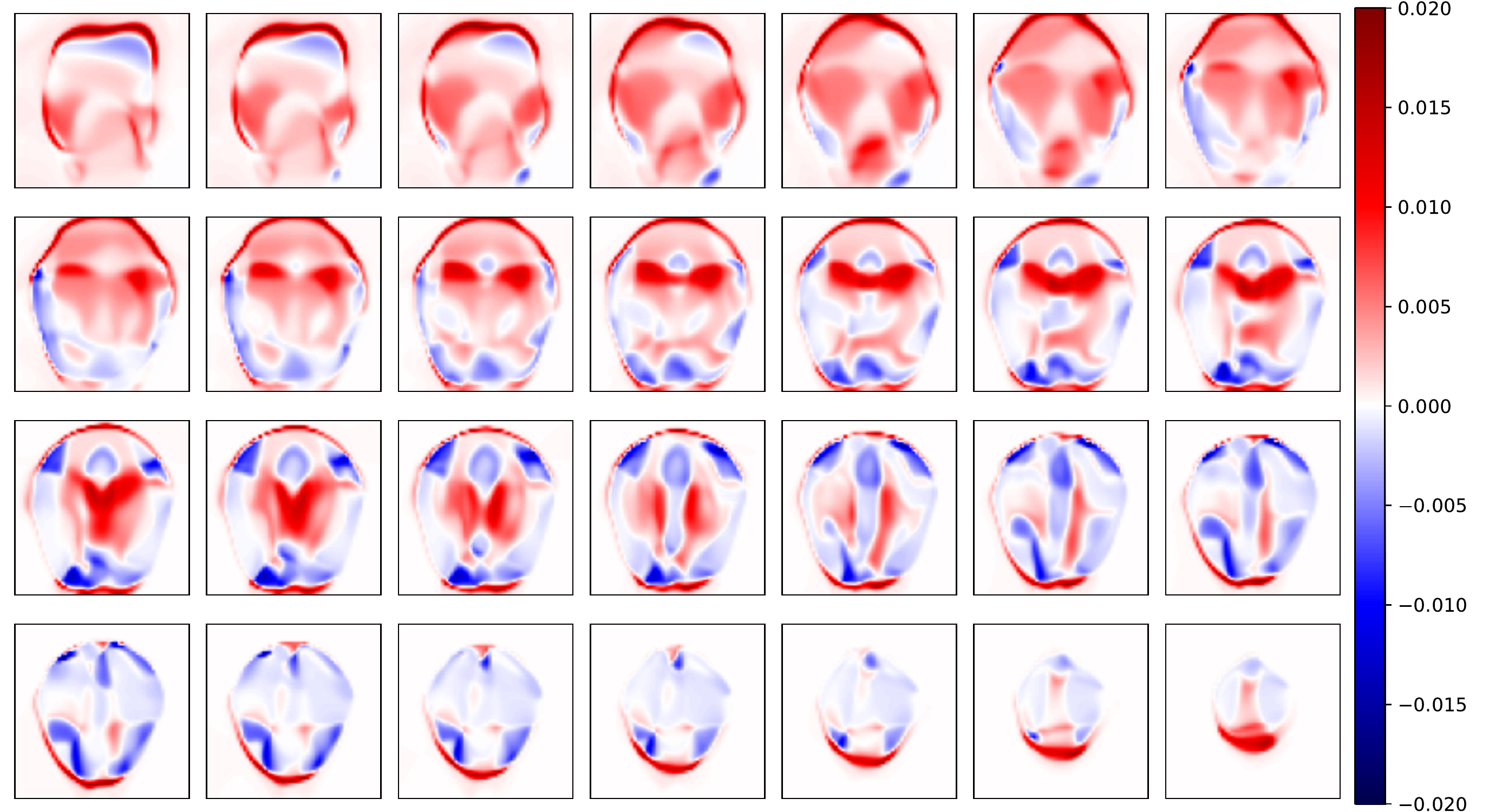}
\end{tabular}
\caption{(a) The second and (b) the third leading eigenfunctions of the deepshared covariance network fitted to the 3D fMRI data. The second and the third eigenfunctions explain $26.85\%$ and $19.05\%$ variability, respectively. The plots are heatmaps of 2D slices of the 3D eigenfunctions, where the slices are taken over the z-axis. \label{fig:fMRI_top3}}
\end{figure}

The plots clearly show that the deepshared CovNet is able to capture a much richer structure than the separable model. While the leading eigenfunction of the CovNet model accounts for $42.74\%$ of the variability explained, the same for the separable model is only $0.0004\%$. In fact, to explain $42.74\%$ variability, we would need $45538$ components in the separable model. Moreover, the fitted CovNet is able to extract the natural symmetry in the brain during resting state \citep{raemaekers2018}, while no such constraints were imposed a priori.

In Figure~\ref{fig:fMRI_top3}, we show the second and the third leading eigenfunctions of the fitted CovNet which, together with the leading eigenfunction, account for $88.64\%$ of the variability. In contrast, for the separable model, we would require $84578$ components to explain $88.64\%$ variability. The plots clearly exhibit the ability of the deepshared CovNet to identify complex structures from the data.

\section{Asymptotic theory}\label{sec:asymptotics}

We conclude the article by developing asymptotic theory for our estimators. In particular, we prove that the CovNet estimators are consistent, and derive their rates of convergence. We will consider two different setups -- (a) when the fields are fully observed and (b) when the fields are observed on a grid with possible noise contamination.

\subsection{The case of fully observed fields}\label{sec:asymptotics_full}
Here, our data consist of i.i.d.\ random fields $\X_1,\ldots,\X_N$ distributed as $\X$. For convenience, we start by assuming that the random field $\X$ is bounded, i.e., there exists $\beta_N > 0$ such that $\|\X\|^2 \le \beta_N$ almost surely. This type of boundedness assumption is quite common in the theoretical analysis of neural networks \citep[e.g.,][]{gyorfi2002,schmidt-hieber2020}. Note that the bound is allowed to grow with $N$, so this can also be seen as a growing truncation level. We will eventually remove this boundedness condition.

We also need to impose some condition on the approximating class, for which we restrict the eigenstructure of the matrix $\Lambda = ((\lambda_{r,s}))$ in \eqref{eq:covnet_kernel_generic}. Specifically, for a constant $\lambda_N > 0$, we enforce that $\Lambda \preceq \lambda_N\mathrm I_R$. Note that for any non-negative definite matrix $\Lambda$, we always have $\Lambda \preceq \lambda_{\max}\mathrm I_R$, where $\lambda_{\max}$ is the largest eigenvalue of $\Lambda$. Thus, our assumption imposes a restriction on the largest eigenvalue of the matrix $\Lambda$ in \eqref{eq:covnet_kernel_generic}. We write $\F_{R,\lambda_N}$ to denote the restricted class of kernels corresponding to the generic class $\F_R$. In particular,
\begin{align}\label{eq:covnet_class_bounded}
\F^{\rm sh}_{R,\lambda_N} &= \big\{c_{\rm sh} \text{ of the form \eqref{eq:shallow_covnet_kernel}}: \mathrm 0 \preceq \Lambda = ((\lambda_{r,s})) \preceq \lambda_N {\rm I}_{R}, \mathbf w_r \in \R^d, b_r \in \R\big\}, \nonumber\\
\F^{\rm d}_{R,L,\lambda_N} &= \big\{c_{\rm d} \text{ of the form \eqref{eq:deep_covnet_kernel}}: \mathrm 0 \preceq \Lambda := ((\lambda_{r,s})) \preceq \lambda_N\,\mathrm I_R, g_1,\ldots,g_R \in \mathcal D_{L,R}\big\} \text{ and } \nonumber\\
\F^{\rm ds}_{R,L,\lambda_N} &= \big\{c_{\rm ds} \text{ of the form \eqref{eq:deepshared_covnet_kernel}}: \mathrm 0 \preceq \Lambda:=((\lambda_{r,s})) \preceq \lambda_N \mathrm I_R, \nonumber\\
&\kern30ex g_1,\ldots,g_R \text{ of the form \eqref{eq:deepshared_neural_network} with }p_1=\cdots=p_L=R\big\},
\end{align}
denote the restricted shallow, deep, and deepshared CovNet classes. In the above, we write $\mathcal D_{L,R}$ to denote the deep neural network class $\mathcal D_{L,\mathbf p}$ when $p_1=\cdots=p_L = R$. We also denote the class of operators corresponding to the class of kernels $\F_{R,\lambda_N}$ by $\widetilde\F_{R,\lambda_N}$. Recall that our estimators are defined as
\begin{equation}\label{eq:covnet_estimator_bounded}
\Chat^{\rm sh}_{R,N} \in \argmin_{\G \in \widetilde\F^{\rm sh}_{R,\lambda_N}} \vertj{\Chat_N - \G}_2^2,~ \Chat^{\rm d}_{R,L,N} \in \argmin_{\G \in \widetilde\F^{\rm d}_{R,L,\lambda_N}} \vertj{\Chat_N - \G}_2^2 ~\text{ and }~ \Chat^{\rm ds}_{R,L,N} \in \argmin_{\G \in \widetilde\F^{\rm ds}_{R,L,\lambda_N}} \vertj{\Chat_N - \G}_2^2.
\end{equation}

In the following, we start by proving two different kinds of results. First, we prove consistency of the estimators under appropriate conditions, and then derive their rates of convergence.

\begin{theorem}\label{thm:covnet_consistency}
Let $\X_1,\ldots,\X_N \overset{\iid}{\sim} \X$, where $\X$ takes values in $\L_2(\Q)$, and $\Q$ is a compact subset of $\R^d$. Also assume that $\|\X\|^2 \le \beta_N$ almost surely, $\E(\X) = 0$ and $\cov(\X) = \C$. Let $\Chat_{R,N}^{\rm sh}$, $\Chat_{R,L,N}^{\rm d}$ and $\Chat_{R,L,N}^{\rm ds}$ be the shallow, the deep, and the deepshared CovNet estimators given by \eqref{eq:covnet_estimator_bounded}. Suppose that $R \to \infty, \lambda_N \to \infty$ as $N \to \infty$, and define $\Delta_N = \max\{\beta_N,|\Q|R\lambda_N\}$.
\begin{enumerate}[(A)]
\item If $dR^2 \Delta_N^4 \log(\Delta_N)/N \to 0$ as $N \to \infty$, then the shallow CovNet estimator is weakly consistent for $\C$, i.e., $\vertj{\Chat^{\rm sh}_{R,N} - \C}_2 \overset{P}{\to} 0$. Additionally, if $\Delta_N^4/N^{1-\delta} \to 0$ for some $\delta \in (0,1)$, then the estimator is strongly consistent for $\C$, i.e., $\vertj{\Chat^{\rm sh}_{R,N} - \C}_2 \overset{a.s.}{\to} 0$ as $N \to \infty$.

\item Let $R>d$. If $L^4R^8\Delta_N^4 \log^2(L \Delta_N)/N \to 0$ as $N \to \infty$, then the deep CovNet estimator is weakly consistent for $\C$, i.e., $\vertj{\Chat^{\rm d}_{R,L,N} - \C}_2 \overset{P}{\to} 0$. Additionally, if $L^4R^8\Delta_N^4 \log^2(L \Delta_N)/N^{1-\delta} \to 0$ for some $\delta \in (0,1)$, then the estimator is strongly consistent for $\C$, i.e., $\vertj{\Chat^{\rm ds}_{R,L,N} - \C}_2 \overset{a.s.}{\to} 0$ as $N \to \infty$.

\item Let $R>d$. If $L^4R^6\Delta_N^4 \log^2(L \Delta_N)/N \to 0$, then the deepshared CovNet estimator is weakly consistent for $\C$, i.e., $\vertj{\Chat^{\rm ds}_{R,L,N} - \C}_2 \overset{P}{\to} 0$. Additionally, if $L^4R^6\Delta_N^4 \log^2(L \Delta_N)/N^{1-\delta} \to 0$ for some $\delta \in (0,1)$, then the estimator is strongly consistent for $\C$, i.e., $\vertj{\Chat^{\rm ds}_{R,L,N} - \C}_2 \overset{a.s.}{\to} 0$ as $N \to \infty$.
\end{enumerate}
\end{theorem}

The proof of the theorem involves \emph{bias-variance-type decompositions} for the estimation error $\vertj{\Chat-\C}_2^2$. To control the bias term, we need the universal approximation property (Theorems~\ref{thm:universal_approximation_shallow} and \ref{thm:universal_approximation_deep}), but now with the additional restriction on the classes \eqref{eq:covnet_class_bounded}. This is ensured by assuming that $R$, $L$ and $\lambda_N$ go to infinity as $N$ diverges (see Remark\,\ref{remark:universal_approximation_bounded}). On the other hand, $dR^2 \Delta_N^4 \log(\Delta_N)/N$, $L^4R^8\Delta_N^4 \log^2(L \Delta_N)/N$, and $L^4R^6\Delta_N^4 \log^2(L \Delta_N)/N$ are linked to the \emph{variance} of the estimators. The conditions of the theorem ensures that the variance also converges to $0$ with the sample size for the different estimators.

\begin{remark}\label{remark:triangular_array}
There is a small technical ambiguity in the statement of Theorem~\ref{thm:covnet_consistency}. The theorem is stated for i.i.d.\ observations distributed as $\X$ when $N$ goes to infinity, whereas the bound on $\X$ is also allowed to evolve with $N$. Thus, the result is to be understood for a triangular sequence of arrays where, for each $N$, the observations are bounded by a constant, which in turn is allowed to diverge keeping up with the assumption of the theorem. However, we avoid stating the theorem in this generality for ease of exposition. The special case of i.i.d.\ observations (i.e., when $\beta_N$ is a constant) follows easily from the theorem.
\end{remark}

Next, we derive the rate of convergence of the estimators.

\begin{theorem}\label{thm:covnet_rate_of_convergence}
Let $\X_1,\ldots,\X_N \overset{\iid}{\sim} \X$, where $\X$ takes values in $\L_2(\Q)$, and $\Q$ is a compact subset of $\R^d$. Suppose that $\|\X\|^2 \le \beta_N$ almost surely, $\E(\X) = 0$ and $\cov(\X) = \C$. Let $\Chat_{R,N}^{\rm sh}$, $\Chat_{R,L,N}^{\rm d}$ and $\Chat_{R,L,N}^{\rm ds}$ be the shallow, the deep, and the deepshared CovNet estimators given by \eqref{eq:covnet_estimator_bounded}. Then,
\begin{align*}
\E\Big(\vertj{\Chat^{\rm sh}_{R,N} - \C}_2^2\Big) &\le 2 \inf_{\G \in \widetilde\F^{\rm sh}_{R,\lambda_N}} \vertj{\G - \C}_2^2 + \O\bigg(\frac{dR^2 \Delta_N^4\log(N)}{N}\bigg), \\
\E\Big(\vertj{\Chat^{\rm d}_{R,L,N} - \C}_2^2\Big) &\le 2 \inf_{\G \in \widetilde\F^{\rm d}_{R,L,\lambda_N}} \vertj{\G - \C}_2^2 + \O\bigg(\frac{L^4R^8\Delta_N^4\log^2(N)}{N}\bigg) \text{, and } \\
\E\Big(\vertj{\Chat^{\rm ds}_{R,L,N} - \C}_2^2\Big) &\le 2 \inf_{\G \in \widetilde\F^{\rm ds}_{R,L,\lambda_N}} \vertj{\G - \C}_2^2 + \O\bigg(\frac{L^4R^6\Delta_N^4\log^2(N)}{N}\bigg).
\end{align*}
Here, $\Delta_N = \max\{\beta_N,|\Q|R\lambda_N\}$ is as defined in Theorem~\ref{thm:covnet_consistency}. In the above, for the deep and the deepshared CovNet estimators, we have assumed that $R>d$.
\end{theorem}

The theorem clearly shows the bias-variance-type decomposition for the proposed estimators. To get the exact rates of convergence, we need to quantify the bias terms, which is an approximation theoretic problem. If, for example, the bias term is zero for some finite $R,L$ and $\lambda_N$, then the derived rate of convergence is the same as that of the empirical estimator, except for the logarithmic term. Thus, in this case, our estimator enjoys a nearly minimax rate of convergence. In general, to get the rate of convergence of the bias, we need to make further assumptions. There are two ways of doing this, either by making assumptions on the eigenstructure of the true covariance or by making assumptions on the smoothness of the underlying field $\X$ (see Appendix\,\ref{supp:approximation_rate} for details). For this line of derivations, the rates depend crucially on the approximation error of the constituents of the CovNet model under consideration. For instance, for the shallow CovNet, if we assume that the underlying field $\X$ takes values in $\mathcal S^{\alpha}(\Q)$, the Sobolev space of functions of order $\alpha$ on $\Q$ \citep[see][]{mhaskar1996}, with almost surely bounded Sobolev norm (i.e., $\|\X\|_{\mathcal S^{\alpha}(\Q)}^2 \le \beta$ a.s.), then by selecting $\lambda_N \asymp R$, we can bound the bias term as $\inf_{\G \in \widetilde\F^{\rm sh}_{R,\lambda_N}} \vertj{\G-\C}_2^2 \lesssim R^{-\alpha/d}$ (see \eqref{eq:bias_convergence_rate_Sobolev}). So, for consistency we need $R^{10} = \o(N/(d\log(N))$, while the optimal rate is achieved for $R \asymp (N/d\log(N))^{d/(10d+\alpha)}$. This leads to the rate of convergence $\O\big((d\log(N)/N)^{\alpha/(10d+\alpha)}\big)$. Similarly, one can use results from \cite{langer2021,ohn2019} to bound the bias of the deep and the deepshared CovNet models.

\begin{remark}\label{remark:approximate_minimizer}
Both Theorems~\ref{thm:covnet_consistency} and \ref{thm:covnet_rate_of_convergence} are derived here for a global minimizer of the loss function. But, in practice, we are not guaranteed to find a global minima. Following the proof of Theorem~\ref{thm:covnet_consistency}, it can be shown that the consistency results hold as long as the estimator is within $\o_P(1)$ of the minimizer, i.e., for estimators $\Chat$ satisfying $\vertj{\Chat_N-\Chat}_2^2 \le \inf_{\G \in \widetilde\F_{R,\lambda_N}} \vertj{\Chat_N-\G}_2^2 + \o_P(1)$. Similarly, from the proof of Theorem~\ref{thm:covnet_rate_of_convergence}, one can check that for an approximate minimizer, the rate of convergence gets inflated by the ``expected minimization gap" $\E\Big(\vertj{\widetilde\C_N-\Chat}_2^2-\inf_{\G \in \widetilde\F_{R,\lambda_N}} \vertj{\widetilde\C_N-\G}_2^2\,\big|\,\mathscr X_N\Big)$, where $\mathscr X_N = \{\X_1,\ldots,\X_N\}$ denotes the observed data, and $\widetilde C_N$ is distributed identically to $\Chat_N$ but independently of $\mathscr X_N$.
\end{remark}

Finally, we prove consistency without the boundedness condition on $\X$. In this case, we need to slightly modify our estimators. To this extent, let $\G$ be a CovNet operator from the unrestricted class $\widetilde\F_R$ with kernel $g(\uvec,\vvec) = \sum_{r=1}^R \sum_{s=1}^R \lambda_{r,s}\,g_r(\uvec)\,g_s(\vvec)$. For $\lambda_N > 0$, define $\Pj_{\lambda_N}\G$ to be the CovNet operator obtained by thresholding the eigenvalues of $\Lambda := ((\lambda_{r,s}))$ to $\lambda_N$. To be precise, if $\Lambda = \sum_{i=1}^R \eta_i\,\mathbf e_i\,\mathbf e_i^\top$ is the eigendecomposition of $\Lambda$, then we define $\Lambda_{\lambda_N} = \sum_{i=1}^R \min\{\eta_i,\lambda_N\}\,\mathbf e_i\,\mathbf e_i^\top$ to be the $\lambda_N$-thresholded version of $\Lambda$. We define $\Pj_{\lambda_N}\G$ to be the operator with kernel $g_{\lambda_N}(\uvec,\vvec) = \sum_{r=1}^R\sum_{s=1}^R \widetilde\lambda_{r,s}\,g_r(\uvec)\,g_s(\vvec)$, where $\widetilde\lambda_{r,s}$ is the $(r,s)$-th element of the matrix $\Lambda_{\lambda_N}$. By construction, $\mathrm 0 \preceq \Lambda_{\lambda_N} \preceq \lambda_N\,\mathrm I_R$ and consequently, $\Pj_{\lambda_N}\G$ is an element of the restricted class $\widetilde\F_{R,\lambda_N}$. We are now ready to re-define the estimator. Define
\begin{equation*}
\Chat^{\rm sh}_{R,N} \in \inf_{\G \in \widetilde\F^{\rm sh}_{R,\sigma}} \vertj{\Chat_N - \G}_2^2,\quad \Chat^{\rm d}_{R,L,N} \in \inf_{\G \in \widetilde\F^{\rm d}_{R,L,\sigma}} \vertj{\Chat_N - \G}_2^2 \quad\text{ and }\quad \Chat^{\rm ds}_{R,L,N} \in \inf_{\G \in \widetilde\F^{\rm ds}_{R,L,\sigma}} \vertj{\Chat_N - \G}_2^2,
\end{equation*}
to be the shallow, the deep and the deepshared CovNet estimators, respectively, but without any restriction on the underlying classes. Now, for a constant $\lambda_N > 0$, we define our modified estimators as
\begin{equation}\label{eq:covnet_estimator_modified}
\widetilde\C^{\rm sh}_{R,N} = \Pj_{\lambda_N}\Chat^{\rm sh}_{R,N}, \quad \widetilde\C^{\rm d}_{R,L,N} = \Pj_{\lambda_N}\Chat^{\rm d}_{R,L,N} \quad\text{ and }\quad \widetilde\C^{\rm ds}_{R,L,N} = \Pj_{\lambda_N}\Chat^{\rm ds}_{R,L,N}.
\end{equation}

These modified estimators are consistent, as shown in the following theorem.
\begin{theorem}\label{thm:covnet_modified_consistency}
Let $\X_1,\ldots,\X_N \overset{\iid}{\sim} \X$, where $\X$ takes values in $\L_2(\Q)$ with $\E(\|\X\|^4)<\infty$. Also assume that $\E(\X) = 0$ and $\cov(\X) = \C$. Let $\widetilde\C_{R,N}^{\rm sh}$, $\widetilde\C_{R,L,N}^{\rm d}$ and $\widetilde\C_{R,L,N}^{\rm ds}$ be the modified shallow, deep, and deepshared CovNet estimators given by \eqref{eq:covnet_estimator_modified}. Assume that $R>d$, and $R \to \infty, \lambda_N \to \infty$ as $N \to \infty$.
\begin{enumerate}[(A)]
\item If $dR^6 \lambda_N^4 \log(R\lambda_N)/N \to 0$ as $N \to \infty$, then the modified shallow CovNet estimator is weakly consistent for $\C$, i.e., $\vertj{\widetilde\C^{\rm sh}_{R,N} - \C}_2 \overset{P}{\to} 0$. Additionally, if $dR^6 \lambda_N^4 \log(R\lambda_N)/N^{1-\delta} \to 0$ for some $\delta \in (0,1)$, then it is strongly consistent for $\C$, i.e., $\vertj{\widetilde\C^{\rm sh}_{R,N} - \C}_2 \overset{a.s.}{\to} 0$ as $N \to \infty$.

\item Let $R>d$. If $L^4R^{12}\lambda_N^4 \log^2(LR\lambda_N)/N \to 0$ as $N \to \infty$, then the modified deep CovNet estimator is weakly consistent for $\C$, i.e., $\vertj{\widetilde\C^{\rm d}_{R,L,N} - \C}_2 \overset{P}{\to} 0$. Additionally, if $L^4R^{12}\lambda_N^4 \log^2(LR\lambda_N)/N^{1-\delta} \to 0$ for some $\delta \in (0,1)$, then it is strongly consistent for $\C$, i.e., $\vertj{\widetilde\C^{\rm d}_{R,L,N} - \C}_2 \overset{a.s.}{\to} 0$ as $N \to \infty$.

\item Let $R>d$. If $L^4R^{10}\lambda_N^4 \log^2(LR\lambda_N)/N \to 0$, then the modified deepshared CovNet estimator is weakly consistent for $\C$, i.e., $\vertj{\widetilde\C^{\rm ds}_{R,L,N} - \C}_2 \overset{P}{\to} 0$. Additionally, if $L^4R^{10}\lambda_N^4 \log^2(LR\lambda_N)/N^{1-\delta} \to 0$ for some $\delta \in (0,1)$, then it is strongly consistent for $\C$, i.e., $\vertj{\widetilde\C^{\rm ds}_{R,L,N} - \C}_2 \overset{a.s.}{\to} 0$ as $N \to \infty$.
\end{enumerate}
\end{theorem}

\begin{remark}\label{remark:nonparametric_model}
Our derived rates are truly nonparametric, with minimal assumptions on the underlying structure. We assumed $\E(\|\X\|^4) < \infty$, which is standard for covariance estimation. Moreover, we made no assumption on the underlying covariance operator $\C$. As a consequence, our derived rates are rather slow in terms of the number of parameters of the models. These can be improved by making further assumptions on the random field $\X$ or the eigenfunctions of $\C$ (e.g., the ones used by \cite{bauer2019} or \cite{schmidt-hieber2020} in the context of nonparametric regression). However, such specialized treatments are beyond the scope of the present article. If the rank of $\C$ is small, which is very often the case in FDA, then a small $R$ is enough to control the bias term (see Appendix~\ref{supp:covnet_bias}). On the other hand, such a small $R$ gives us a considerable gain in terms of the variance, thus reducing the overall estimation error. However, one should also note that the derived rates are only upper bounds, and we do not claim tightness of the bounds.
\end{remark}

\subsection{The case of discretely observed fields}\label{sec:asymptotics_discrete}
The results derived so far are for fully observed random fields. But in practice, we observe the fields on a grid, with possible noise contamination. Here, we develop asymptotic properties of our estimators in this scenario. W.l.o.g., we assume that $\Q = [0,1]^d$, and we observe the data on a $K_1 \times \cdots \times K_d$ regular grid on $[0,1]^d$. To this extent, let $\{T_{1,1}^{K_1},\ldots,T_{1,K_1}^{K_1}\},\ldots,\{T_{d,1}^{K_d},\ldots,T_{d,K_d}^{K_d}\}$ be regular partitions of $[0,1]$ of sizes $K_1,\ldots,K_d$, respectively. Define $V_{i_1,\ldots,i_d}^K = T_{1,i_1}^{K_1} \times \cdots \times T_{d,i_d}^{K_d}$ to be the $(i_1,\ldots,i_d)$-th \emph{voxel} for $1\le i_1\le K_1,\ldots,1\le i_d\le K_d$. The voxels are non-overlapping (i.e., $V_{i_1,\ldots,i_d}^K \cap V_{j_1,\ldots,j_d}^K = \emptyset$ for $(i_1,\ldots,i_d) \ne (j_1,\ldots,j_d)$), and they form a regular partition of $[0,1]^d$. In particular, $|V_{i_1,\ldots,i_d}^K| = (\prod_{i=1}^d K_i)^{-1}$. For each random field $\X_n$, we make a single measurement at each of the voxels. These measurements are assumed to be of the form
\begin{equation}\label{eq:measurement_with_noise}
\widetilde X_n^K[i_1,\ldots,i_d] = X_n^K[i_1,\ldots,i_d] + E_n^K[i_1,\ldots,i_d], \qquad 1\le i_1\le K_1,\ldots,1\le i_d\le K_d,\,n=1,\ldots,N,
\end{equation}
where $X^K_n[i_1,\ldots,i_d]$ is a discretization of $\X_n$ over the $(i_1,\ldots,i_d)$-th voxel and $E_n^K[i_1,\ldots,i_d]$ is the corresponding measurement error or noise. We consider two different measurement schemes which relate the discrete object $\mathbf X_n^K = (X_n^K[i_1,\ldots,i_d])$ to the respective field $\X_n = (X_n(\uvec): \uvec \in [0,1]^d)$.

\begin{enumerate}[(M1)]
\item Point-wise measurement:
\[
X_n^K[i_1,\ldots,i_d] = X_n(u_{i_1},\ldots,u_{i_d}), \qquad 1\le i_1 \le K_1,\ldots,1 \le i_d \le K_d,
\]
where $(u_{i_1},\ldots,u_{i_d}) \in V_{i_1,\ldots,i_d}^K$ is a location within the $(i_1,\ldots,i_d)$-th voxel. For the measurements to be meaningful, we need to assume that $\X$ has continuous sample paths \citep[e.g.,][]{hsing2015}.

\item Voxel-wise average:
\[
X_n^K[i_1,\ldots,i_d] = \frac{1}{|V_{i_1,\ldots,i_d}^K|}\int_{V_{i_1,\ldots,i_d}^K} X_n(\uvec) \diff\uvec, \qquad 1\le i_1 \le K_1,\ldots,1 \le i_d \le K_d.
\]
\end{enumerate}

For the measurement errors $E_n^K[i_1,\ldots,i_d]$, we assume that they are i.i.d.\ with mean $0$ and variance $\sigma_K^2$, and are uncorrelated with the $X_n^K[i_1,\ldots,i_d]$'s. In line with our previous assumptions, we also assume that $\big|E_n^K[i_1,\ldots,i_d]\big|^2 \le \beta^{\rm e}_{N,K}$ almost surely.

We denote the measurements corresponding to the $n$-th field $\X_n$ by $\widetilde{\mathbf X}_n^K = (\widetilde X_n^K[i_1,\ldots,i_d])$. Define $\widetilde{\mathbf C}_N^K = N^{-1}\sum_{n=1}^N \widetilde{\mathbf X}_n^K \otimes \widetilde{\mathbf X}_n^K$ to be the empirical covariance based on the discretely observed data. For a generic class of CovNet operators $\widetilde\F_R$, our estimator is given by
\[
\Chat^K_{R,N} \in \argmin_{\G \in \widetilde\F_R} \big\|\widetilde{\mathbf C}_N^K - \mathbf G^K\big\|_{\mathrm F}^2,
\]
where $\|\cdot\|_{\mathrm F}$ is the Frobenius norm and $\mathbf G^K$ is the discretization of the operator $\G$ over the voxels, defined as 
\[
\mathbf G^K[i_1,\ldots,i_d;j_1,\ldots,j_d] = g(v_{i_1},\ldots,v_{i_d};v_{j_1},\ldots,v_{j_d}),\qquad 1 \le i_1,j_1 \le K_1,\ldots,1 \le i_d,j_d \le K_d,
\]
where $g$ is the kernel corresponding to $\G$ and $(v_{i_1},\ldots,v_{i_d})$ is a location in the $(i_1,\ldots,i_d)$-th voxel $V_{i_1,\ldots,i_d}^K$. If we define $\widetilde\C_N^K$ to be the voxel-wise continuation of $\widetilde{\mathbf C}_N^K$, with kernel
\[
\widetilde c_N^K(\uvec,\vvec) = \sum_{i_1=1}^{K_1}\cdots\sum_{i_d=1}^{K_d}\sum_{j_1=1}^{K_1}\cdots\sum_{j_d=1}^{K_d} \widetilde C_N^K[i_1,\ldots,i_d;j_1,\ldots,j_d]\,\1\{\uvec \in V_{i_1,\ldots,i_d}^K, \vvec \in V_{j_1,\ldots,j_d}^K\},\qquad \uvec, \vvec \in [0,1]^d,
\]
then it is easy to see that
\[
\argmin_{\G \in \widetilde\F_R} \big\|\widetilde{\mathbf C}_N^K - \mathbf G^K\big\|_{\mathrm F}^2 \approx \argmin_{\G \in \widetilde\F_R}\vertj{\widetilde\C_N^K - \G}_2^2,
\]
where the approximation holds when the resolution $K_1\times \cdots\times K_d$ is large (see also \eqref{eq:discrete_approximate_estimator_equivalence} in Appendix~\ref{supp:measurement_with_noise}). We will derive theoretical properties for this approximation. In particular, we define
\begin{align}\label{eq:estimators_discrete}
\Chat^{{\rm sh},K}_{R,N} \in \argmin_{\G \in \widetilde\F_{R,\lambda_N}^{\rm sh}} \vertj{\widetilde\C_N^K - \G}_2^2,\,\Chat^{{\rm d},K}_{R,L,N} \in \argmin_{\G \in \widetilde\F_{R,L,\lambda_N}^{\rm d}} \vertj{\widetilde\C_N^K - \G}_2^2 \text{ and } \Chat^{{\rm ds},K}_{R,L,N} \in \argmin_{\G \in \widetilde\F_{R,L,\lambda_N}^{\rm ds}} \vertj{\widetilde\C_N^K - \G}_2^2,
\end{align}
to be the shallow, deep and deepshared CovNet estimators based on the discrete measurements. The asymptotic behaviour of these estimators is established in the following theorem.

\begin{theorem}\label{thm:discrete_measurement_rate}
Let $\X_1,\ldots,\X_N \overset{\iid}{\sim} \X$, where $\X$ takes values in $\L_2([0,1]^d)$ with $\E(\X) = 0$ and $\cov(\X) = \C$. Let the kernel $c$ of $\C$ is Lipschitz on $[0,1]^{2d}$ with Lipschitz constant $\rho$. Consider the measurement model \eqref{eq:measurement_with_noise}, where the measurement errors $E_n^K[i_1,\ldots,i_d]$ are i.i.d.\ and uncorrelated with $\mathbf X_1^K,\ldots,\mathbf X_N^K$, and satisfy $\big|E_n^K[i_1,\ldots,i_d]\big|^2 \le \beta_{N,K}^{\rm e}$ almost surely, $\E(E_n^K[i_1,\ldots,i_d])=0$, $\var(E_n^K[i_1,\ldots,i_d])=\sigma_K^2$. Suppose that one of the following two hold.
\begin{enumerate}[1.]
\item $\X$ has continuous sample paths, $\|\X\|_{\infty}^2 \le \beta_N$ almost surely, and the measurements $\mathbf X_1^K,\ldots,\mathbf X_N^K$ are obtained from $\X_1,\ldots,\X_N$ via (M1).

\item $\|\X\|^2 \le \beta_N$ almost surely and the measurements $\mathbf X_1^K,\ldots,\mathbf X_N^K$ are obtained from $\X_1,\ldots,\X_N$ via (M2).
\end{enumerate}
Let $\Chat^{{\rm sh},K}_{R,N}$, $\Chat^{{\rm d},K}_{R,L,N}$ and $\Chat^{{\rm ds},K}_{R,L,N}$ be the shallow, deep and deepshared CovNet estimators given by \eqref{eq:estimators_discrete}. Then,
\begin{align*}
\vertj{\Chat^{{\rm sh},K}_{R,N} - \C}_2^2 &\le 18 \inf_{\G \in \widetilde\F_{R,\lambda_N}^{\rm sh}} \vertj{\G -\C}_2^2 + \O\bigg(\frac{dR^2\Delta_{N,K}^4\log(N)}{N}\bigg) + a_K,\\
\vertj{\Chat^{{\rm d},K}_{R,L,N} - \C}_2^2 &\le 18 \inf_{\G \in \widetilde\F_{R,L,\lambda_N}^{\rm d}} \vertj{\G -\C}_2^2 + \O\bigg(\frac{L^4R^8\Delta_{N,K}^4\log^2(N)}{N}\bigg) + a_K, \quad\text{and}\\
\vertj{\Chat^{{\rm ds},K}_{R,L,N} - \C}_2^2 &\le 18 \inf_{\G \in \widetilde\F_{R,L,\lambda_N}^{\rm ds}} \vertj{\G -\C}_2^2 + \O\bigg(\frac{L^4R^6\Delta_{N,K}^4\log^2(N)}{N}\bigg) + a_K,
\end{align*}
where $\Delta_{N,K} = \max\{2(\beta_N+\beta_{N,K}^{\rm e}),R\lambda_N\}$ and $a_K = 21\rho^2\big(K_1^{-2}+\cdots+K_d^{-2}\big) + 21\sigma_K^4/(K_1\cdots K_d)$.
\end{theorem}

The theorem clearly shows the effect of grid size and noise contamination on the estimators. The rates are qualitatively the same as in Theorem~\ref{thm:covnet_rate_of_convergence}, except $\Delta_N$ is replaced by $\Delta_{N,K}$ and a couple of terms depending on $(K_1,\ldots,K_d)$ are added. The term $\Delta_{N,K}$ can be viewed as a noise contaminated version of $\Delta_N$. If we assume that $\beta_{N,K} = \o(\beta_N)$ (which can be seen as assuming that the signal-to-noise ratio diverges), then $\Delta_{N,K}$ is asymptotically equivalent to $\Delta_N$. In this case, we can see a clear separation in the estimation error, one due to the sample size and the other due to the resolution. Among the remaining two terms, $\rho^2(K_1^{-2}+\cdots+K_d^{-2})$ arises due to the discretization of the fields, while $\sigma_K^4/(K_1\cdots K_d)$ is due to noise contamination. Thus, for consistency of our estimators, we require $(K_1^{-2}+\cdots+K_d^{-2}) \to 0$, which is ensured if $\min\{K_1,\ldots,K_d\} \to \infty$. Moreover, the noise level $\sigma_K$ is allowed to diverge, but at a slower rate than $(K_1\cdots K_d)^{1/4}$. The Lipschitz assumption on $c$ in the theorem is convenient, but is by no means necessary. It can be verified that consistency of the estimators holds as long as $c$ is continuous.

\section{Concluding remarks}\label{sec:conclusion}

We have proposed three new classes of neural network models for covariance estimation of functional data observed over multidimensional domains. The advantages of the proposed models include efficient estimation, storage, manipulation and performance guarantees. Our approach is motivated by the demonstrated ability of neural networks in solving complex problems. And indeed, our empirical studies show the superiority of the proposed methods, especially the deepshared CovNet model. At the same time, our methods will also be amenable to the shortcomings of neural networks, e.g., lack of theoretical optimization guarantees for convergence to global minima. But, as is the case with neural networks, despite these limitations, our experimental results appear compelling. Any progress in the study of neural networks will, in principle, translate to a commensurate progress in the understanding of covariance networks.

Throughout the article, we have used the sigmoidal activation function. But, most of the results, especially the ones for the deep CovNet models, can be easily extended to include other activation functions, e.g., the ReLU. In some preliminary numerical studies, we observed similar performance by the sigmoid and the ReLU. We prefer the sigmoid because of the smoothness that it provides, which is often beneficial for functional covariance estimation. 

At the level of generality they are derived, our convergence rates are arguably slow. But, these do not reveal the complete picture and are rather a reflection of our completely nonparametric treatment of the problem. These rates can be improved by considering more structured problems, which is now a topic of interest in theoretical studies of neural networks \citep[][]{bauer2019,schmidt-hieber2020}. Such additional structural assumptions may also allow us to derive approximation errors for the models, which we have not fully addressed here.

\addtocontents{toc}{\protect\setcounter{tocdepth}{0}}
\section*{Acknowledgment}
We are grateful to Prof.\ Sir John A.\ D.\ Aston for providing access to the fMRI data and enlightening us on some of their key aspects.

\appendix
\addtocontents{toc}{\protect\setcounter{tocdepth}{1}}
\section*{Appendices}

In these appendices, we give the proofs omitted from the main text, some further mathematical details and additional simulation results. The organization is as follows. In the next section, we provide some mathematical background useful in subsequent developments. In Section~\ref{supp:covnet_bias}, we discuss the bias of the CovNet models -- we establish the universal approximation property, and sketch two different ways to derive the rate of convergence of the bias term. In particular, we derive the rate of convergence of the bias for the shallow CovNet model. In Section~\ref{supp:implementation}, we provide some details on the estimation of the CovNet models and also describe a way to estimate the mean function from the data using the CovNet models. In Section~\ref{supp:covering_number}, we treat the covering numbers of certain spaces, as these play a fundamental role in the proofs of the asymptotic results. In particular, we derive upper bounds on the covering numbers of the classes of shallow, deep and deepshared CovNet operators. In Section~\ref{supp:asymptotics}, we provide proofs of the asymptotic results. Finally, in Section~\ref{supp:additional_simulations}, we provide some additional numerical results which were left out in the main text.

\setcounter{section}{0}
\section{Mathematical background}\label{supp:math_background}

We start by summarizing some definitions and background concepts. More details can be found in \cite{hsing2015}. Let $\Q \subset \R^d$ be a compact set. We denote by $\L_2(\Q)$ the space of all real-valued square-integrable functions on $\Q$. This is a Hilbert space when equipped with the inner product $\langle f,g \rangle = \int_{\Q} f(\uvec) g(\uvec) \diff \uvec$ for $f,g \in \L_2(\Q)$. A linear map $\mathcal A$ from $\L_2(\Q)$ onto itself is called bounded if there exists a constant $M \ge 0$ such that $\|\mathcal A f\| \le M\|f\|$ for all $f \in \L_2(\Q)$, where $\|\cdot\|$ is the norm induced by the inner product $\langle\cdot,\cdot\rangle$, i.e., $\|f\| = \sqrt{\langle f,f \rangle}$. A bounded linear map is referred to as an \emph{operator}. The minimum value of $M$ for which the boundedness condition holds is called the \emph{operator norm} and is denoted by $\vertj{\cdot}_\infty$. An operator $\mathcal A$ is \emph{compact} if there exist orthonormal bases (ONBs) $(\psi_j)_{j \ge 1}$ and $(\phi_j)_{j \ge 1}$ of $\L_2(\Q)$ such that $\mathcal A f = \sum_{j\ge 1} \lambda_j \langle \psi_j,f \rangle \phi_j$. A compact operator $\mathcal A$ is called \emph{Hilbert-Schmidt} if $\vertj{\mathcal A}_2 = \big\{\sum_{j\ge 1} \|\mathcal A \psi_j\|^2\big\}^{1/2}$ is finite, where $(\psi_j)_{j \ge 1}$ is a ONB of $\L_2(\Q)$. As indicated by the notation, $\vertj{\cdot}_2$ does not depend on the particular choice of the basis, and is called the \emph{Hilbert-Schmidt norm}. We will use $\mathcal B_2(\L_2(\Q))$ to denote the class of all Hilbert-Schmidt operators on $\L_2(\Q)$. An operator $\mathcal A$ is called positive semi-definite if $\langle \mathcal A f, f \rangle \ge 0$ for all $f \in \L_2(\Q)$. A compact, positive semi-definite operator is called \emph{trace class} or \emph{nuclear} if $\vertj{\mathcal A}_1 = \sum_{j\ge 1} \langle \mathcal A \psi_j, \psi_j \rangle$ is finite for some ONB $(\psi_j)_{j \ge 1}$. Again, the sum is independent of the choice of ONB, and $\vertj{\mathcal A}_1$ is called the \emph{trace norm} of $\mathcal A$. One particular type of operators on $\L_2(\Q)$, which is of interest to us is the \emph{integral operator} defined as $\mathcal A f(\uvec) = \int_{\Q} a(\uvec,\vvec)\,f(\vvec)\diff\vvec$ for $\uvec \in \Q, f \in \L_2(\Q)$, where $a \in \L_2(\Q \times \Q)$ is called the \emph{kernel} of the operator $\mathcal A$. An integral operator $\mathcal A$ is positive semi-definite if and only if the associated kernel $c$ is non-negative definite. The operator $\mathcal A$ and the kernel $a$ are linked by an obvious isometry, i.e., $\vertj{\mathcal A}_2 = \|c\|_{\L_2(\Q \times \Q)}$.

Now, let $\X = (X(\uvec):\uvec \in \Q)$ be a random element in $\L_2(\Q)$. For $d=1$, $\X$ is usually referred to as a \emph{random curve}, whereas for $d>1$, it is referred to as a \emph{random field} or \emph{random surface}. We assume that $\X$ has finite second moment, i.e., $\E(\|\X\|^2) < \infty$, which ensures the existence of its mean $m = \E(\X)$ and covariance $\C = \E\{(\X-m) \otimes (\X-m)\}$ of $\X$ (both the expectations are understood in the Bochner sense). The mean $m$ is an element of $\L_2(\Q)$. The covariance $\C$ is the integral operator associated with the \emph{covariance kernel} $c \in \L_2(\Q \times \Q)$, where $c(\uvec,\vvec) = \cov(X(\uvec),X(\vvec))$. Moreover, $\C$ is positive semi-definite and trace class. In this article, we are interested in estimating $\C$ based on independent and identically distributed (i.i.d.) observations $\X_1,\ldots,\X_N \sim \X$. Because of the isomorphism linking the integral operator and the associated kernel, the problem is equivalent to estimating the kernel $c$.

\section{Bias of the CovNet model: universal approximation and rate of convergence}\label{supp:covnet_bias}
Here, we deal with the bias of the CovNet model. We start by proving that all three CovNet structures can approximate any covariance operator up to arbitrary precision, a.k.a.\ the \emph{universal approximation} property. These results are instrumental for the consistency of our CovNet estimators. We start with the shallow CovNet structure and give a detailed proof. The proofs for the deep and the deepshared structures are similar, and we discuss those only briefly.

\subsection{Universal approximation}\label{supp:universal_approximation}
\begin{proof}[Proof of Theorem~\ref{thm:universal_approximation_shallow}]
Recall that $c$ is the kernel of the covariance operator $\C$. So, by the spectral decomposition of $\C$, we get
\begin{equation}\label{eq:Mercer_decompostion_L2}
c(\uvec,\vvec) = \sum_{i=1}^\infty \eta_i\,\psi_i(\uvec)\,\psi_i(\vvec),
\end{equation}
where the sum on the right converges in the $\L_2$ norm on $\Q \times \Q$. Here, $\eta_i$'s are the eigenvalues of $\C$ and $\psi_i$'s are the corresponding eigenfunctions. Since the covariance operator $\C$ is trace-class, we get
\[
\verti{\C}_{1} = \sum_{i=1}^\infty \eta_i = \int_{\Q} c(\uvec,\uvec)\,\diff\uvec < \infty.
\]

Now, fix $\epsilon > 0$ and w.l.o.g.\ let $\epsilon <1$. Since the sum in  \eqref{eq:Mercer_decompostion_L2} converges in the $\L_2$ norm, we can find an integer $I$ (depending on $\epsilon$) such that
\begin{equation}\label{eq:Mercer_bound_L2}
\bigg\|c - \sum_{i=1}^{I} \eta_i\,\psi_i \otimes\,\psi_i\bigg\|_{\L_2(\Q \times \Q)} < \frac{\epsilon}{2}.
\end{equation}
Also, since the functions $\psi_1,\ldots,\psi_I$ are in $\L_2(\Q)$, we can find a positive constant $M$ (depending on $\epsilon$) such that $\max_{i=1,\ldots,I}\big\|\psi_i\big\| \le M$. Since $\sigma$ is a sigmoidal function, using the density of single hidden layer neural networks in the class of $\L_2$ functions \citep[][Theorem~16.2]{gyorfi2002}, for each $i=1,\ldots,I$, we can find $R_i \in \N$, coefficients $a_{i,1},\ldots,a_{i,R_i}$, weights $\mathbf w_{i,1},\ldots,\mathbf w_{i,R_i} \in \R^d$ and biases $b_{i,1},\ldots,b_{i,R_i}$ such that
\begin{equation}\label{eq:universal_approximation_bound_L2}
\bigg\|\psi_i - \sum_{r=1}^{R_i} a_{i,r}\,\sigma(\mathbf w_{i,r}^\top \cdot +\, b_{i,r})\bigg\| < \frac{\epsilon}{\verti{\C}_{1}(4M+2)}.
\end{equation}
Define $\widehat\psi_i(\uvec) = \sum_{r=1}^{R_i} a_{i,r}\,\sigma(\mathbf w_{i,r}^\top \uvec + b_{i,r})$ for $i=1,\ldots,I$, and
\begin{align}\label{eq:universal_approximation_proof_c_hat_L2}
\widehat{c}_{R}(\uvec,\vvec) = \sum_{i=1}^{I} \eta_i\,\widehat\psi_i(\uvec)\,\widehat\psi_i(\vvec) &=\sum_{i=1}^I \eta_i \sum_{r=1}^{R_i} \sum_{s=1}^{R_i} a_{r,i}\,a_{s,i}\,\sigma(\mathbf w_{r,i}^\top \uvec + b_{r,i})\,\sigma(\mathbf w_{s,i}^\top \vvec + b_{s,i}) \nonumber\\
&= \sum_{r=1}^R \sum_{s=1}^R \lambda_{r,s}\,\sigma(\mathbf w_r^\top \uvec + b_r)\,\sigma(\mathbf w_s^\top \vvec + b_s),
\end{align}
where $R = \sum_{i=1}^I R_i$. Here, $\{\mathbf w_1,\ldots,\mathbf w_R\}$ is the collection of all the weights of the $I$ neural networks $\widehat\psi_1,\ldots,\widehat\psi_I$, and $\{b_1,\ldots,b_R\}$ is the collection of all the biases. The associated matrix $\Lambda = ((\lambda_{r,s}))$ is block-diagonal with blocks $\Lambda_i = \eta_i\,(a_{r,i}\,a_{s,i})_{1 \le r,s \le R_i}$. Since $\eta_i > 0$, each of the $\Lambda_i$'s are positive semi-definite, which in turn shows that $\Lambda$ is positive semi-definite. Define $c_I(\uvec,\vvec) = \sum_{i=1}^I \eta_i\,\psi_i(\uvec)\,\psi_i(\vvec)$.
Using 
\begin{equation}\label{eq:product_difference_bound_L2}
\|f \otimes f - g \otimes g\|_{\L_2(\Q \times \Q)} \le 2 \|f\|\,\|f-g\| + \|f-g\|^2,
\end{equation}
we get
\begin{align}\label{eq:universal_approximation_error_bound_L2}
\|c_I - \widehat{c}_R\|_{\L_2(\Q \times \Q)} &= \bigg\|\sum_{i=1}^I \eta_i \big\{\psi_i \otimes \psi_i - \widehat \psi_i \otimes \widehat \psi_i\big\}\bigg\|_{\L_2(\Q \times \Q)} \nonumber ~~~~\text{(by \eqref{eq:universal_approximation_proof_c_hat_L2})} \\
&\le \sum_{i=1}^I \eta_i \big\|\psi_i \otimes \psi_i - \widehat\psi_i \otimes \widehat\psi_i\big\|_{\L_2(\Q \times \Q)} \nonumber \\
&\le \sum_{i=1}^I \eta_i\, \Bigg\{2 M \frac{\epsilon}{\verti{\C}_{1}(4M+2)} + \bigg(\frac{\epsilon}{\verti{\C}_{1}(4M+2)}\bigg)^2\Bigg\}~~~~\text{(by \eqref{eq:universal_approximation_bound_L2} and \eqref{eq:product_difference_bound_L2})} \nonumber \\
&\le \sum_{i=1}^I \eta_i \frac{\epsilon}{\verti{\C}_{1}(4M+2)} (2M+1)~~~~~(\text{since } \epsilon<1) \nonumber\\
&= \frac{\epsilon}{2\verti{C}_{1}} \sum_{i=1}^I \eta_i \le \frac{\epsilon}{2}.
\end{align}
Finally, combining \eqref{eq:Mercer_bound_L2} and \eqref{eq:universal_approximation_error_bound_L2}, we get
\begin{align*}
\|c - \widehat{c}_R\|_{\L_2(\Q \times \Q)} \le \|c - c_I\|_{\L_2(\Q \times \Q)} + \|c_I - \widehat{c}_R\|_{\L_2(\Q \times \Q)} \le \epsilon,
\end{align*}
as intended.

\medskip

Now, if in addition $c$ is continuous, then it admits a similar decomposition as in \eqref{eq:Mercer_decompostion_L2}, where the sum converges absolutely and uniformly \citep[][Theorem~4.6.5]{hsing2015}. So, we can find an integer $I$ such that
\[
\bigg\|c - \sum_{i=1}^{I} \eta_i\,\psi_i \otimes \psi_i\bigg\|_{\L_\infty(\Q \times \Q)} < \frac{\epsilon}{2}.
\]
Also, the eigenfunctions $\psi_1,\ldots,\psi_I$ are now continuous on a compact set $\Q$. So, there exists a positive constant $M$ such that $\max_{i=1,\ldots,I} \|\psi_i\|_{\L_{\infty}(\Q)} \le M$. Since $\sigma$ is a sigmoidal function, we can find neural networks $\widehat\psi_i$ of the form \eqref{eq:universal_approximation_bound_L2} such that
\begin{equation*}
\|\psi_i - \widehat\psi_i\|_{\L_\infty(\Q)} < \frac{\epsilon}{\verti{\C}_{1}(4M+2)},
\end{equation*}
see Lemma~16.1 in \cite{gyorfi2002}. The rest of the proof follows similarly to the previous case upon using a bound similar to \eqref{eq:product_difference_bound_L2} for the $\L_{\infty}$ norm.
\end{proof}

Next, we briefly discuss the universal approximation property of the deep and the deepshared models. 

\begin{proof}[Proof of Theorem~\ref{thm:universal_approximation_deep}]
Recall that the deep CovNet kernel is of the form
\[
c_{\rm d}(\uvec,\vvec) = \sum_{r=1}^R \sum_{s=1}^R \lambda_{r,s}\,g_r(\uvec)\,g_s(\vvec),
\]
where each of the functions $g_1,\ldots,g_R$ are individual deep neural networks. To prove the universal approximation property of this structure, we proceed similarly to the case of the shallow CovNet structure. Namely, we decompose $c$ as
\[
c(\uvec,\vvec) = \sum_{i=1}^I \eta_i\,\psi_i(\uvec)\,\psi_i(\vvec) + \sum_{i=I+1}^\infty \eta_i\,\psi_i(\uvec)\,\psi_i(\vvec) =: c_I(\uvec,\vvec) + e_I(\uvec,\vvec),
\]
where $\|c-c_I\|_{\L_2(\Q \times \Q)} \le \epsilon/2$. Now, for each $i=1,\ldots,I$, we can find deep neural network $\widehat\psi_i = \sum_{r=1}^{R_i} a_{i,r}\,g_{i,r}(\cdot)$ of the required form such that $\|\psi_i - \widehat\psi_i\| < \delta$ (follows from the universal approximation property of the deep neural network with sigmoid activation, see \citet{funahashi1989}). Also, the depth of all these networks can be taken to be the same \citep[see][Corollary~1]{funahashi1989}. Now, by defining $\widehat c_R(\uvec,\vvec) = \sum_{i=1}^I \eta_i\,\widehat\psi_i(\uvec)\,\widehat\psi_i(\vvec)$, it is easy to verify that $\widehat c_R$ has the deep CovNet structure and approximates $c$ up to the desired precision.

Next, consider the deepshared structure. Observe that for any deep CovNet kernel with depth $L$ and number of nodes $R$, we can find a deepshared CovNet kernel with depth $L$ and number of nodes $R^\prime$, such that the two structures are the same (by considering a wider network and deleting some of the connections, see Figures~\ref{fig:deep_covnet} and \ref{fig:deepshared_covnet}). Thus, for a fixed depth, the complexity of the deep and the deepshared structures is the same (when we allow the number of nodes to vary). Thus, the result for the deepshared CovNet model follows from the universal approximation of the deep CovNet model.
\end{proof}

\begin{remark}\label{remark:universal_approximation_bounded}
The results that we have proved establish the universal approximation property of the three CovNet models without any restriction on the parameters. But, to establish consistency of the estimators and their rates of convergence (Section~\ref{sec:asymptotics}), we need to impose restrictions on the eigenstructure of the matrix $\Lambda$. Therefore, we need to control the bias for this restricted class of operators. The universal approximation property for the restricted class follows easily. By the universal approximation property of the unrestricted class, for a given $\C$ and $\epsilon>0$, we can find a CovNet operator $\G$ (of any of the three types) such that $\vertj{\G - \C}_2 < \epsilon$. Now, let $\Lambda$ be the matrix associated with $\G$. Since $\Lambda$ is a positive semi-definite matrix, we can find $\lambda_0 > 0$ such that $\Lambda \preceq \lambda_0\,\mathrm I$ (e.g., by takings $\lambda_0$ to be the largest eigenvalue of $\Lambda$). This shows that for any $\C$, we can find a CovNet operator from the restricted class that can approximate $\C$ up to arbitrary precision, thus establishing the universal approximation property with the additional condition.
\end{remark}

\subsection{Rate of convergence of the bias term}\label{supp:approximation_rate}

Here, we will derive the rate of convergence of the bias for the (possibly restricted) class of CovNet operators. We will describe two possible ways to obtain the rates -- (a) by imposing conditions on the eigenstructure of $\C$ or (b) by imposing conditions on the observation $\X$.

\subsubsection{Restrictions on the covariance}
By the eigendecomposition, we can write $\C = \sum_{i=1}^{\infty} \eta_i\,\psi_i \otimes \psi_i$, where $(\eta_i)_{i \ge 1}$ is the sequence of non-increasing eigenvalues of $\C$ and $(\psi_i)_{i\ge 1}$ is the corresponding sequence of eigenfunctions. For $I \in \N$, define $\C_I = \sum_{i=1}^I \eta_i\,\psi_i \otimes \psi_i$ to be the truncated version of $\C$. Then,
\[
\vertj{\C-\C_I}_2^2 = \sum_{i>I} \eta_i^2 = \O(a_I),
\]
where $a_I$ depends on the \emph{eigen-decay} of $\C$. Now, for each $i=1,\ldots,I$, suppose that we can find (shallow/deep/deepshared) neural networks $\widehat\psi_1,\ldots,\widehat\psi_K$ such that
\[
\|\psi_i - \widehat\psi_i\| = \O(b_{L_i,R_i}),\qquad i=1,\ldots,I,
\]
where $L_i,R_i$ are the parameters (depth and/or width) of $\widehat\psi_i$. The rate of the approximation error $b_{L_i,R_i}$ depends on additional structural assumptions on the eigenfunctions \citep[e.g.,][]{mhaskar1996,bauer2019,ohn2019,schmidt-hieber2020,langer2021}, which can be imposed by means of additional structural assumptions on the kernel $c$. Now, if we define $\widehat\C = \sum_{i=1}^I \eta_i\,\widehat\psi_i \otimes \widehat\psi_i$, then it is easy to see that $\widehat\C$ is a CovNet operator with $R$ nodes (and possibly of depth $L$) such that
\[
\vertj{\C - \widehat\C}_2^2 \le 2\Big\{\vertj{\C - \C_I}_2^2 + \vertj{\C_I - \widehat\C}_2^2\Big\} = \O(a_I) + \O(b_{I,L,R}),
\]
where $b_{I,L,R}$ can be obtained from the $b_{L_i,R_i}$'s using \eqref{eq:product_difference_bound_L2}. Here, the parameters $R$ and $L$ of the CovNet operator depend on $I,L_i,R_i$. Finally, the exact rate of convergence of the bias term can be obtained by carefully scrutinizing the terms $a_I$ and $b_{I,L,R}$ as functions of $I,L,R$, and choosing $I$ appropriately.

\subsubsection{Restrictions on the observations}
The idea here is similar to the one used before. Instead of the eigendecomposition, we will make use of a different type of approximation result. The following lemma will be instrumental in our derivation, which is a suitable adaptation of Lemma~16.7 in \cite{gyorfi2002}. The proof follows easily from the proof of Lemma~16.7 in \cite{gyorfi2002}, so we omit it.

\begin{lemma}\label{lemma:approximation_general} 
Let $\mathcal T$ be an index set, and $\{\phi_t : t \in \mathcal T\}$ be a collection of real-valued functions on a compact domain $\Q$ such that $\|\phi_t\| \le B$ for all $t \in \mathcal T$. Let $f : \Q \to \R$ be a function such that there exists a probability measure $b$ on $\mathcal T$ satisfying
\[
f(\uvec) = \int_{\Q} \phi_t(\uvec) \diff\mu(t) \quad\forall\, \uvec \in \Q.
\]
Then, for every $I \in \N$, there exists a function $f_I(\uvec) = \sum_{i=1}^I w_i \phi_{t_i}(\uvec)$ such that
\[
\|f - f_I\| \le \frac{B}{\sqrt{I}}.
\]
Moreover, the coefficients $w_i$ are non-negative and $\sum_{i=1}^I w_i = 1$.
\end{lemma}

Now, suppose that the random field $\X$ satisfies $\Prob(\|\X\|^2 \le \beta) = 1$. Recall that the covariance kernel $c$ of $\X = (X(\uvec): \uvec \in \Q)$ is defined as
\[
c(\uvec,\vvec) = \cov(X(\uvec),X(\vvec)) = \E(X(\uvec)X(\vvec)) = \int_{\Omega} X(\uvec,\omega) X(\vvec,\omega) \diff\Prob(\omega),\quad\uvec,\vvec \in \Q,
\]
for some set $\Omega$ and the probability measure $\Prob$ on $\Omega$. Thus, by defining $\mathcal T = \Omega$ and $\phi_\omega(\uvec,\vvec) = X(\uvec,\omega) X(\vvec,\omega)$ for $\omega \in \Omega$, we see that
\[
\|\phi_{\omega}\|_{\L_2(\Q \times \Q)}^2 = \iint_{\Q \times \Q} X^2(\uvec,\omega) X^2(\vvec,\omega) \diff\uvec \diff\vvec = \|\X(\omega)\|^4 \le \beta^2~~\text{almost surely}.
\]
Also, $c(\uvec,\vvec) = \int_{\Omega} \phi_{\omega}(\uvec,\vvec) \diff\Prob(\omega)$. So, using Lemma~\ref{lemma:approximation_general}, for every $I \in \N$, we can find $\omega_1,\ldots,\omega_I \in \Omega$ and non-negative constants $\gamma_1,\ldots,\gamma_I$ such that, defining $c_I(\uvec,\vvec) = \sum_{i=1}^I \gamma_i\,\phi_{\omega_i}(\uvec,\vvec)$, we get
\begin{equation}\label{eq:rate_general_c}
\vertj{\C - \C_I}_2^2 = \|c-c_I\|_{\L_2(\Q \times \Q)}^2 = \iint_{\Q \times \Q} \{c(\uvec,\vvec) - c_I(\uvec,\vvec)\}^2 \diff\uvec \diff\vvec \le \frac{\beta^2}{I} = \O(I^{-1}),
\end{equation}
where $\C_I$ is the integral operator associated with the kernel $c_I$. Observe that 
\[
c_I(\uvec,\vvec) = \sum_{i=1}^I \gamma_i \phi_{\omega_i}(\uvec,\vvec) = \sum_{i=1}^I \gamma_i\,X_i(\uvec)\,X_i(\vvec),\quad \uvec,\vvec \in \Q,
\]
where the functions $\X_i:=\X(\omega_i)$. Thus, proceeding as in the previous section, we can obtain (a bound on) the rate of convergence of the bias of the CovNet operator. The exact rate, in this case, will depend on additional structural assumptions on the functions $\X_1,\ldots,\X_I$, or equivalently on $\X$.

\medskip

We demonstrate this by deriving the rate for the restricted shallow CovNet operator. Suppose that $\X$ takes values in $\mathcal S^\alpha(\Q)$, the \emph{Sobolev space} of order $\alpha$ in $\L_2(\Q)$. Further, let $\|\X\|_{\mathcal S^\alpha(\Q)}^2 \le \beta$ almost surely (if $\alpha \ge 2$, this also implies that $\|\X\|^2 \le \beta$ almost surely). Thus, for every $i=1,\ldots,I$, $\X_i \in \mathcal S^\alpha(\Q)$. By Theorem~2.1 in \cite{mhaskar1996}, for every $R \in \N$, we can find weights $\mathbf w_1,\ldots,\mathbf w_R$, bias $b$, and continuous functionals $f_1,\ldots,f_R$ on $\mathcal S^\alpha(\Q)$ such that, defining $\widehat{\X}_i(\uvec) = \sum_{r=1}^R f_r(\X_i)\,\sigma(\mathbf w_r^\top \uvec + b)$, we get 
\[
\|\X_i - \widehat{\X}_i\| \lesssim R^{-\alpha/d} \|\X_i\|_{\mathcal S^\alpha(\Q)} \lesssim R^{-\alpha/d}.
\]
Since the functionals $f_1,\ldots,f_R$ are continuous, there exist finite constants $a_1,\ldots,a_R$ such that 
\[
|f_r(\X_i)| \le a_r \|\X_i\|_{\mathcal S^\alpha(\Q)} \le a_r \sqrt{\beta},\quad r=1,\ldots,R.
\]
Thus, by defining $\widetilde c_{I,R}(\uvec,\vvec) = \sum_{i=1}^I \gamma_i\,\widehat \X_i(\uvec)\,\widehat \X_i(\vvec)$, we get
\begin{align}\label{eq:approx_error_ck}
\big\|c_I - \widetilde c_{I,R}\big\|_{\L_2(\Q \times \Q)} &\le \sum_{i=1}^I \gamma_i \Big\|\widehat \X_i \otimes \widehat \X_i - \X_i \otimes \X_i\Big\|_{\L_2(\Q \times \Q)} \nonumber \\
&\le \sum_{i=1}^I \gamma_i \Big\{\big\|\widehat \X_i - \X_i\big\|^2 + 2 \big\|\X_i\big\| \big\|\widehat \X_i - \X_i\big\|\Big\}~~~~\text{(using \eqref{eq:product_difference_bound_L2})} \nonumber\\
&\lesssim R^{-2\alpha/d} +  R^{-\alpha/d} \asymp R^{-\alpha/d}.
\end{align}
Here, we have used that $\sum_{i=1}^I \gamma_i = 1$, and $\alpha \ge 1, R \ge 1$ implies $R^{-2\alpha/d} = {\scriptstyle\mathcal O}\big(R^{-\alpha/d}\big)$. Now, combining \eqref{eq:rate_general_c} and \eqref{eq:approx_error_ck}, we get
\begin{align*}
\big\|c-\widetilde c_{I,R}\big\|_{\L_2(\Q \times \Q)} \le \big\|c-c_I\big\|_{\L_2(\Q \times \Q)} + \big\|c_I - \widetilde c_{I,R}\big\|_{\L_2(\Q \times \Q)} \lesssim I^{-1/2} + R^{-\alpha/d}.
\end{align*}
By choosing $I \asymp R^{2\alpha/d}$, we get that 
\begin{equation}\label{eq:error_rate_bounded}
\big\|c-\widetilde c_{I,R}\big\|_{\L_2(\Q \times \Q)} \lesssim  R^{-\alpha/d}.
\end{equation}
Now, observe that
\begin{align}\label{eq:neural_structure_bounded}
\widetilde c_{I,R}(\uvec,\vvec) &= \sum_{i=1}^I \gamma_i\,\widehat\X_i(\uvec)\,\widehat\X_i(\vvec) \nonumber\\
&= \sum_{i=1}^I \gamma_i \sum_{r=1}^R \sum_{s=1}^R \beta_{i,r}\,\beta_{i,s}\,\sigma(\mathbf w_r^\top\uvec+b)\,\sigma(\mathbf w_s^\top\vvec+b)~~(\text{where }\beta_{i,r} = f_r(\X_i)) \nonumber\\
&= \sum_{r=1}^R \sum_{s=1}^R \Bigg(\sum_{i=1}^I \gamma_i\,\beta_{i,r}\,\beta_{i,s}\Bigg) \sigma(\mathbf w_r^\top\uvec+b)\,\sigma(\mathbf w_s^\top\vvec+b) \nonumber\\
&=\sum_{r=1}^R \sum_{s=1}^R \lambda_{r,s}\,\sigma(\mathbf w_r^\top\uvec+b)\,\sigma(\mathbf w_s^\top\vvec+b).
\end{align}
Thus, $\widetilde c_{I,R}$ is a shallow CovNet kernel. Since $\gamma_i \ge 0$ for each $i$, the matrix $\Lambda = (\lambda_{r,s})$ is positive semi-definite. Also, $\Lambda \preceq \lambda_R {\rm I}_R$, where $\lambda_R = \beta \sum_{r=1}^R a_r^2$. This follows from the following facts.
\begin{enumerate}[(i)]
\item For a matrix $A = \betavec \betavec^\top \in \R^{R \times R}$, $\xvec^\top A \xvec = (\betavec^\top \xvec)^2 \le \|\betavec\|^2 \|\xvec\|^2$ for every $\xvec$, implying that $A \preceq \|\betavec\|^2 {\rm I}_R$.
\item For a collection of matrices $(A_i)_{i=1}^I$ and $(B_i)_{i=1}^I$ satisfying $0 \preceq A_i \preceq B_i$, and non-negative scalars $\gamma_1,\ldots,\gamma_I$, $\sum_{i=1}^I \gamma_i A_i \preceq \sum_{i=1}^I \gamma_i B_i$.
\item By (i) and (ii), for non-negative scalars $\gamma_1,\ldots,\gamma_I$ and vectors $\betavec_1,\ldots,\betavec_I \in \R^R$, $\sum_{i=1}^I \gamma_i \betavec_i\betavec_i^\top \preceq \big(\sum_{i=1}^I \gamma_i \|\betavec_i\|^2\big)\,{\rm I}_R$.
\item $\Lambda = \sum_{i=1}^I \gamma_i \betavec_i \betavec_i^\top$, where $\betavec_i = (\beta_{i,1},\ldots,\beta_{i,R})^\top$. For each $i=1,\ldots,I$, $\|\betavec_i\|^2 = \sum_{i=1}^R \beta_{i,r}^2 = \sum_{i=1}^R f_r^2(\X_i) \le \beta \sum_{r=1}^R a_r^2$. Also, $\sum_{i=1}^I \gamma_i = 1$ implies that $\sum_{i=1}^I \gamma_i \|\betavec_i\|^2 \le \beta \sum_{r=1}^R a_r^2$.
\end{enumerate}
Using this along with \eqref{eq:error_rate_bounded} and \eqref{eq:neural_structure_bounded}, it follows that if we choose $\lambda_N = \beta \sum_{r=1}^R a_r^2$, then
\begin{equation}\label{eq:bias_convergence_rate_Sobolev}
\inf_{\G \in \widetilde\F_{R,\lambda_N}} \verti{\C - \G}_{2} \lesssim R^{-\alpha/d}.
\end{equation}

\section{Further details on the implementation of the CovNet models}\label{supp:implementation}

In this section, we present further details on the implementation of the CovNet models. We start by proving Proposition~\ref{prop:covnet_optimizer_equivalence}, which justifies the use of the alternative formulation for estimating the CovNet models.
\begin{proof}[Proof of Proposition~\ref{prop:covnet_optimizer_equivalence}]
By construction, it is clear that $\widetilde\F^{\rm NN}_{R,N} \subseteq \widetilde\F_R$ for every $N \ge 1$. So, we will only show that for $N >R$, the other inclusion holds, i.e., $\widetilde\F_{R} \subseteq \widetilde\F_{R,N}^{\rm NN}$. This amounts to showing that for every positive semi-definite matrix $\Lambda = (\lambda_{r,s}) \in \R^{R \times R}$, we can find $\xi_{n,r}$ for $n=1,\ldots,N$, $r=1,\ldots,R$, such that $\lambda_{r,s} = N^{-1}\sum_{n=1}^N (\xi_{n,r}-\bar{\xi}_r)(\xi_{n,s}-\bar{\xi}_s)$.

For $N > R$, we can always find $N$ vectors $\bm{\zeta}_1,\ldots,\bm{\zeta}_N \in \R^R$ such that the rank of the empirical covariance based of these vectors is $R$. For instance, one can take $\bm{\zeta}_i = \mathbf{e}_i$, the $i$-th canonical vector in $\R^R$ for $i=1,\ldots,R$, $\bm{\zeta}_{R+1} = -\1$, and $\bm{\zeta}_i = \mathbf{0}$ for $i=R+2,\ldots,N$, to check that the corresponding empirical covariance matrix is of the form $N^{-1}(\mathrm{I}_R + \1\1^\top)$, which is of rank $R$. Let us denote this empirical covariance by $\widehat\Sigma_N$. Thus, $\widehat\Sigma_N$ is of full rank, and hence positive definite. So, we can find a positive definite matrix $\widehat\Sigma_N^{-1/2}$, so that $\widehat\Sigma_N^{-1/2} \widehat\Sigma_N \widehat\Sigma_N^{-1/2} = \mathrm I_R$. Again, since $\Lambda = (\Lambda_{r,s})$ is positive semi-definite, we can find a positive semi-definite matrix $\Lambda^{1/2}$ such that $\Lambda^{1/2} \Lambda^{1/2} = \Lambda$. Define, $\bm{\xi}_n = \Lambda^{1/2} \widehat\Sigma_N^{-1/2} \bm{\zeta}_n$ for $n=1,\ldots,N$. Then, it is easy to verify that the empirical covariance of $\bm{\xi}_1,\ldots,\bm{\xi}_N$ is $\Lambda$. Now, let $\xi_{n,r}$ be the $r$-th component of $\bm{\xi}_n$. For $n=1,\ldots,N$, define $\X^{\rm NN}_n(\uvec) = \sum_{r=1}^R \xi_{n,r}\,g_r(\uvec)$. It can be easily verified that the empirical covariance of $\X_1^{\rm NN},\ldots,\X_N^{\rm NN}$ is the operator with kernel $\sum_{r=1}^R\sum_{s=1}^R \lambda_{r,s}\,g_r(\uvec)\,g_s(\vvec)$. This establishes the other inclusion.
\end{proof}

Next, we give a detailed derivation of the equivalence between the original loss functions, and the loss functions expressed in terms of the observations $\X_1,\ldots,\X_N$ and the fitted networks $\X_1^{\rm NN},\ldots,\X_N^{\rm NN}$.

\subsection{Detailed derivation of the loss function}
\label{supp:loss_function_detailed}
Let $\X_1,\ldots,\X_N$ be the observed fields and $\X_1^{\rm NN},\ldots,\X_N^{\rm NN}$ be the fitted networks:
\[
\X_n^{\rm NN}(\uvec) = \sum_{r=1}^R \xi_{n,r}\,g_r(\uvec), \qquad n=1,\ldots,N.
\]
To begin with, we assume that all the fields are centered, so that $\sum_{n=1}^N \X_n = 0 = \sum_{n=1}^N \X_n^{\rm NN}$. The loss function is defined as
\begin{align*}
\ell := \ell(\Theta) = \verti{\Chat_N - \frac{1}{N}\sum_{n=1}^N \X_n^{\rm NN} \otimes \X_n^{\rm NN}}_2^2, \\
\end{align*}
where $\Theta$ denotes all the \emph{learnable parameters} of the model. Now, $\ell$ can be written as
\begin{align}\label{eq:loss_function_expansion_L1}
\ell &= \verti{\Chat_N - \frac{1}{N}\sum_{n=1}^N \X_n^{\rm NN} \otimes \X_n^{\rm NN}}_2^2 = \verti{\frac{1}{N}\sum_{n=1}^N (\X_n \otimes \X_n - \X_n^{\rm NN} \otimes \X_n^{\rm NN})}_2^2 \nonumber\\
&= \frac{1}{N^2}\sum_{n=1}^N\sum_{m=1}^N \ip{\X_n \otimes \X_n - \X_n^{\rm NN} \otimes \X_n^{\rm NN}, \X_m \otimes \X_m - \X_m^{\rm NN} \otimes \X_m^{\rm NN}}_2,
\end{align}
where $\ip{\cdot,\cdot}_2$ is the Hilbert-Schmidt inner-product. Now, the inner product in the last step equals
\begin{align*}
&\ip{\X_n \otimes \X_n,\X_m \otimes \X_m}_2 + \ip{\X_n^{\rm NN} \otimes \X_n^{\rm NN}, \X_m^{\rm NN} \otimes \X_m^{\rm NN}}_2 \\
&\kern40ex - \ip{\X_n \otimes \X_n,\X_m^{\rm NN} \otimes \X_m^{\rm NN}}_2 - \ip{\X_n^{\rm NN} \otimes \X_n^{\rm NN}, \X_m \otimes \X_m}_2 \\
&= \langle \X_n,\X_m \rangle^2 + \langle \X_n^{\rm NN}, \X_m^{\rm NN} \rangle^2 - \langle \X_n, \X_m^{\rm NN} \rangle^2 - \langle \X_n^{\rm NN}, \X_m \rangle^2.
\end{align*}
Here, we have used that $\ip{\X_1 \otimes \X_2, \mathcal Y_1 \otimes \mathcal Y_2}_2 = \langle \X_1,\mathcal Y_1\rangle\,\langle \X_2,\mathcal Y_2\rangle$. Plugging this back into \eqref{eq:loss_function_expansion_L1}, we get that
\[
\ell = \frac{1}{N^2}\sum_{n=1}^N\sum_{m=1}^N \langle \X_n, \X_m \rangle^2 + \frac{1}{N^2}\sum_{n=1}^N\sum_{m=1}^N \langle \X_n^{\rm NN}, \X_m^{\rm NN} \rangle^2 - \frac{2}{N^2}\sum_{n=1}^N\sum_{m=1}^N \langle \X_n, \X_m^{\rm NN} \rangle^2.
\]
Thus, the loss $\ell$ can be obtained from the inner products $\langle \X_n,\X_m \rangle$, $\langle \X_n^{\rm NN}, \X_m^{\rm NN} \rangle$ and $\langle \X_n, \X_m^{\rm NN} \rangle$, without forming the high-order objects $\Chat_N$ or $\X_n \otimes \X_n$. When the fields are not centered, one can first center them by subtracting the mean and work with the centered fields. Any mean estimation method can be used for this purpose. We can also use a different approach, which allows us to simultaneously estimate the mean.

\subsection{Simultaneous estimation of the mean}\label{sec:loss_function_mean}

In the previous derivations, we assumed that the fields $\X_1,\ldots,\X_N$ as well as the fitted fields $\X_1^{\rm NN},\ldots,\X_N^{\rm NN}$ are centered, i.e., $\bar\X_N = \bar\X_N^{\rm NN} = 0$. When this is not the case, the loss function can be written as
\begin{align*}
\ell &= \verti{\frac{1}{N}\sum_{n=1}^N \X_n \otimes \X_n - \bar\X_N \otimes \bar\X_N - \frac{1}{N}\sum_{n=1}^N \X_n^{\rm NN} \otimes \X_n^{\rm NN} + \bar\X_N^{\rm NN} \otimes \bar\X_N^{\rm NN}}_2^2 \\
&\le 2\,\Bigg\{\verti{\frac{1}{N}\sum_{n=1}^N \X_n \otimes \X_n - \frac{1}{N}\sum_{n=1}^N \X_n^{\rm NN} \otimes \X_n^{\rm NN}}_2^2 + \verti{\bar\X_N \otimes \bar\X_N - \bar\X_N^{\rm NN} \otimes \bar\X_N^{\rm NN}}_2^2\Bigg\} \\
&= 2\,\big\{\widetilde\ell + \verti{\bar\X_N \otimes \bar\X_N - \bar\X_N^{\rm NN} \otimes \bar\X_N^{\rm NN}}_2^2\big\},
\end{align*}
where $\widetilde\ell$ is the loss function without centering. In this case, instead of minimizing the loss functions $\ell$, one can minimize
\begin{equation}\label{eq:covnet_alternate_formulation}
\widetilde \ell + \verti{\bar\X_N \otimes \bar\X_N - \bar\X^{\rm NN}_N \otimes \bar\X^{\rm NN}_N}_2^2.
\end{equation}
Of course, this is an upper bound to the actual criterion. But, in practice, this gives a reasonable approximation and produces reasonable results. As already demonstrated, $\widetilde\ell$ can be computed efficiently, without forming the high-order objects. Using similar techniques, it can be shown that
\begin{align*}
\verti{\bar\X_N \otimes \bar\X_N - \bar\X_N^{\rm NN} \otimes \bar\X_N^{\rm NN}}_2^2 &= \bigg(\frac{1}{N^2}\sum_{n=1}^N\sum_{m=1}^N \langle \X_n,\X_m\rangle \bigg)^2 + \bigg(\frac{1}{N^2}\sum_{n=1}^N\sum_{m=1}^N \langle \X_n^{\rm NN}, \X_m^{\rm NN} \rangle \bigg)^2 \\
&\kern40ex - 2 \bigg(\frac{1}{N^2}\sum_{n=1}^N\sum_{m=1}^N \langle \X_n, \X_m^{\rm NN} \rangle \bigg)^2,
\end{align*}
which again depends on the inner products, thus allowing for efficient computation of the criterion \eqref{eq:covnet_alternate_formulation}. Moreover, as a by-product, we get an estimate of the mean as follows. As before, we minimize the criterion w.r.t.\ the parameters of $g_1,\ldots,g_R$, and the new parameters $\xi_{n,r}$, to get estimated fields
\begin{equation*}
\X^{\rm NN}_n(\uvec) = \sum_{r=1}^R \hat\xi_{n,r}\,\widehat g_r(\uvec), \qquad n=1,\ldots,N.
\end{equation*}
The mean field is then obtained as the empirical mean based on the fitted fields, i.e.,
\begin{equation*}
\widehat m(\uvec) = \sum_{r=1}^R \bar{\hat\xi}_r\,\widehat g_r(\uvec), \qquad \uvec \in \Q,
\end{equation*}
where $\bar{\hat\xi}_r = N^{-1} \sum_{n=1}^N \hat{\xi}_{n,r}$. Thus, once we estimate the parameters of the model, we can get an estimate of the mean with negligible computational overhead.

\begin{figure}[t]
\centering
\begin{tabular}{c}
\small (a) Ex\,1 with $I=5$ \\
\includegraphics[width=0.9\linewidth]{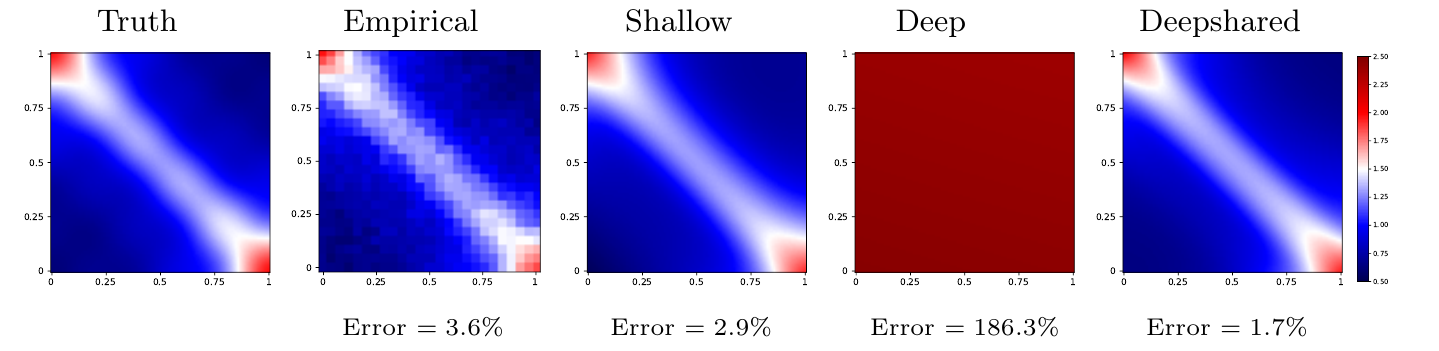} \\ [2pt]
\small (b) Ex\,1 with $I=10$ \\
\includegraphics[width=0.9\linewidth]{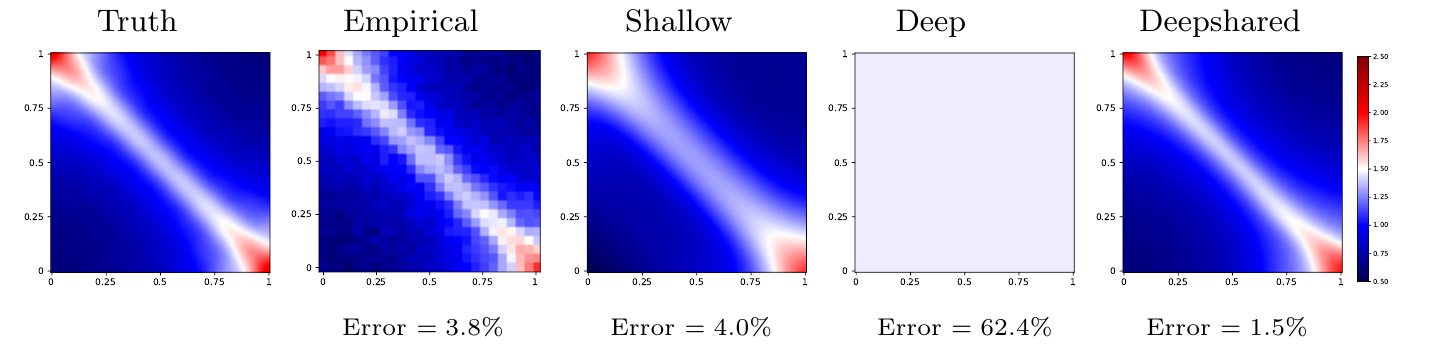} \\ [2pt]
\small (c) Ex\,1 with $I=20$ \\
\includegraphics[width=0.9\linewidth]{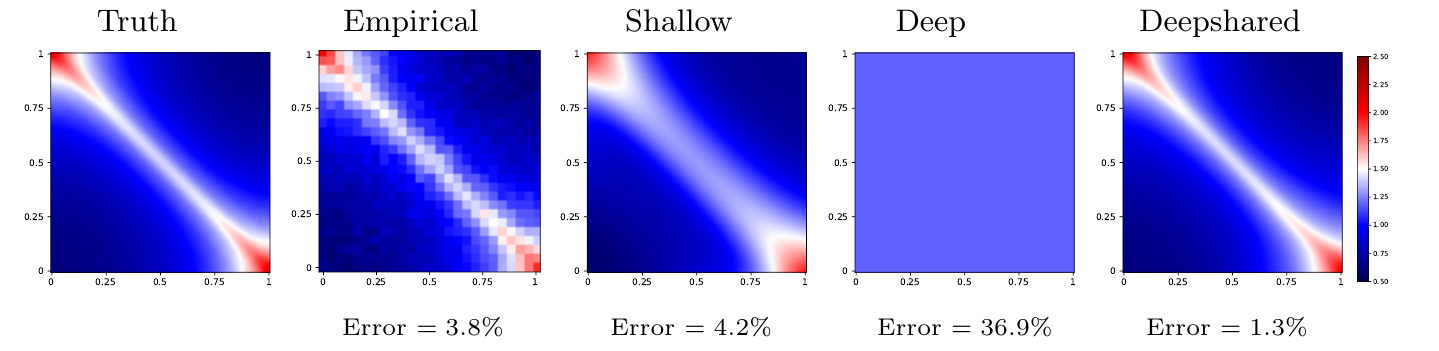}
\end{tabular}
\caption{\label{fig:mean_estimation_fourier}Estimated mean surfaces by different CovNet models in the Fourier basis example with different choices of $I$. The empirical mean surfaces are also shown. The relative estimation errors (in \%) are shown at the bottom of every figure. The plots are heatmaps of the 2D mean surfaces.}
\end{figure}

We demonstrate the usefulness of this method by means of a simulation study. We generated $500$ observations from the Gaussian process on $[0,1]^2$ with mean surface $m$ and covariance kernel $c$ on a regular grid of resolution $25 \times 25$. We took $c$ to be the Mat\'ern covariance kernel (Ex\,5 in Section~\ref{sec:empirical_demonstration}) with $\nu=0.01$. For the mean function, we considered two different setups:
\begin{enumerate}[{Ex}\,1]
\item \emph{Fourier basis}: $m(s,t) = 1 + \sum_{i=1}^I (-1)^i i^{-2} \phi_i(s)\,\phi_i(t)$, where $\phi_i(t) = \sqrt{2} \cos(i\pi t)$ for $i \ge 1$.
\item \emph{Legendre basis}: $m(s,t) = 1 + \sum_{i=1}^I (-1)^i (2i+1) \phi_i(s-0.5)\,\phi_i(t-0.5)$, where $\phi_i$ is the Legendre polynomial of degree $i$.
\end{enumerate}
In both the cases, the number of components $I$ controls the complexity of the mean. The estimated mean surfaces using different CovNet models, along with the true mean and the empirical mean, for different choices of $I$ are shown in Figures~\ref{fig:mean_estimation_fourier} and \ref{fig:mean_estimation_legendre}. For the CovNet models, we used cross-validation to select the hyperparameters. We also calculated the estimation error $\|m - \widehat m\|/\|m\|$ for all the estimators using Monte-Carlo integration, which are shown in the figures. These results clearly exhibit the usefulness of the proposed mean estimation technique. The estimated mean surfaces using the shallow and the deepshared models are very close to the truth. The advantage of having a purely functional form over the discretized empirical estimator is also visible in the figures and is reflected in the estimation errors. The results for the deep CovNet model are, however, not very promising. The deep CovNet model is not able to capture the truth at all. This is perhaps due to the complexity of the deep model, which suggests the need to regularize the deep CovNet model. Further evidence on the need for regularization is furnished by the deepshared model, which has the best performance in all the scenarios. In particular, the deepshared model is able to capture minute details which the shallow model has missed (see Figures~\ref{fig:mean_estimation_fourier}(b), \ref{fig:mean_estimation_fourier}(c), \ref{fig:mean_estimation_legendre}(c)).

One can notice that in \eqref{eq:covnet_alternate_formulation}, the criterion is composed of two components, one accounting for the covariance, while the other accounting for the mean. Instead of simply adding these two components, one may consider a weighted average of the two, depending on their importance. However, we will not pursue this idea further in this article.

\begin{figure}[t]
\centering
\begin{tabular}{c}
\small (a) Ex\,2 with $I=5$ \\
\includegraphics[width=0.9\linewidth]{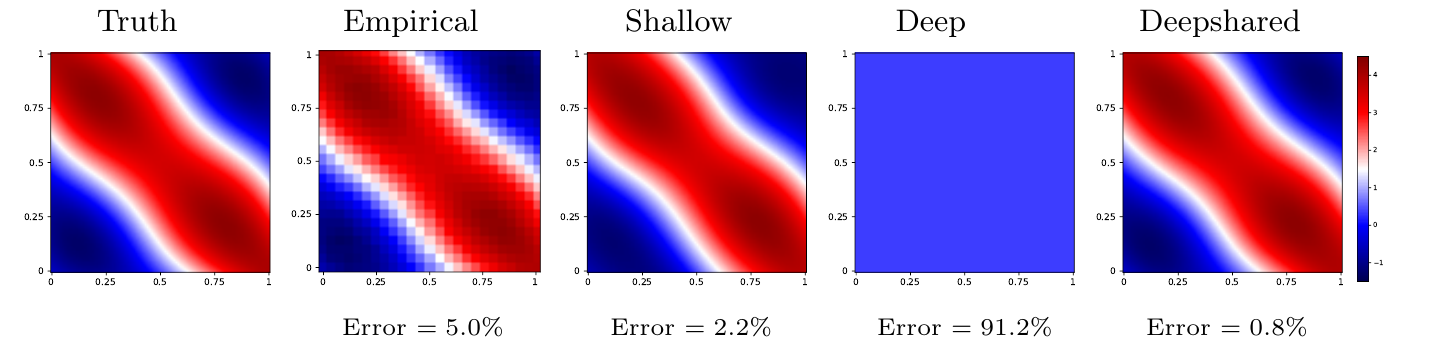} \\ [2pt]
\small (b) Ex\,2 with $I=10$ \\
\includegraphics[width=0.9\linewidth]{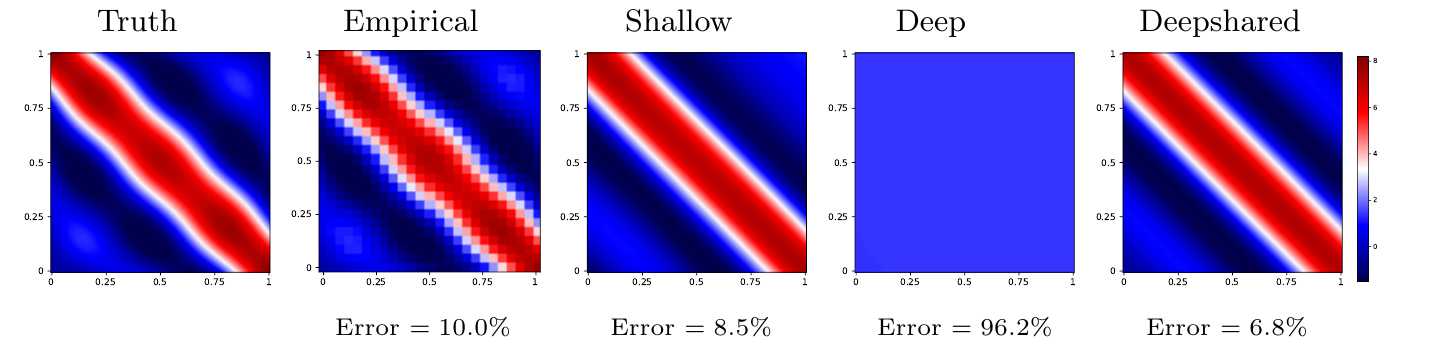} \\ [2pt]
\small (c) Ex\,2 with $I=20$ \\
\includegraphics[width=0.9\linewidth]{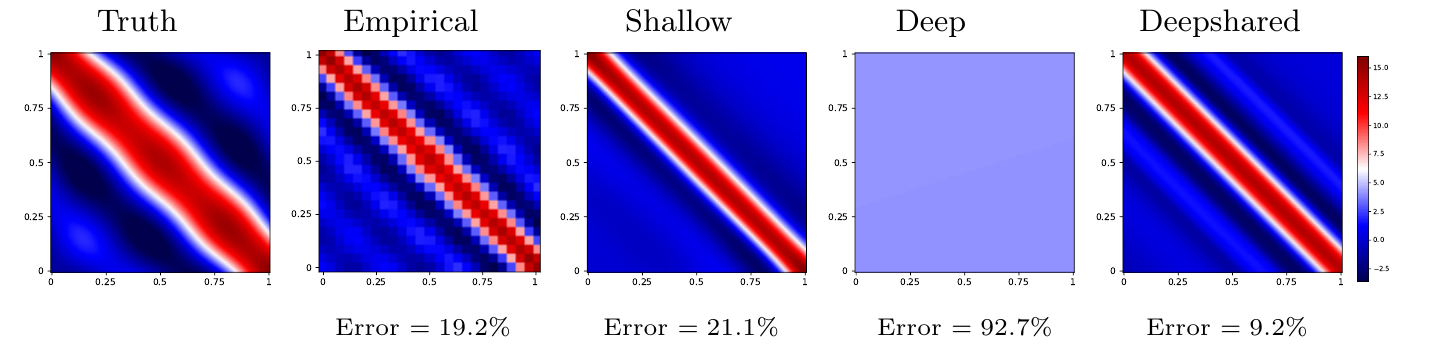}
\end{tabular}
\caption{\label{fig:mean_estimation_legendre}Estimated mean surfaces by different CovNet models in the Legendre basis example with different choices of $I$. The empirical mean surfaces are also shown. The relative estimation errors (in \%) are shown at the bottom of every figure. The plots are heatmaps of the 2D mean surfaces.}
\end{figure}

\subsection{Computing cost}\label{append:implementation_grid}

We end this section with details on implementing the method for measurements on a grid and related computing costs. Recall the generic version of a CovNet kernel \eqref{eq:covnet_kernel_generic}
\[
\sum_{r=1}^R\sum_{s=1}^R \lambda_{r,s}\,g_r(\uvec)\,g_s(\vvec), \qquad \uvec,\vvec \in \Q.
\]
During fitting the CovNet model, using the alternate formulation \eqref{eq:observation_network}--\eqref{eq:covnet_empirical_formulation}, we need to form the $N \times D$ matrix $\mathbf X^{\rm NN} = (\X_n^{\rm NN}(\uvec_i))_{n,i}$, where $\uvec_1,\ldots,\uvec_D$ are the grid points (see the discussion after Remark~\ref{remark:psd_by_design}). This can be done as follows:
\begin{align*}
\underbrace{\mathrm U}_{D \times d} \mapsto \underbrace{\mathrm Z}_{D \times R} = [\underbrace{g_1(\mathrm U)}_{D \times 1},\ldots,\underbrace{g_R(\mathrm U)}_{D \times 1}], \quad \mathbf X^{\rm NN} = \Xi\,\mathrm Z^\top,
\end{align*}
where $\mathrm U = (\uvec_1^\top,\ldots,\uvec_D^\top)^\top$ is the matrix with $\uvec_i$ in its $i$-th row, $g_r(\mathrm U) = (g_r(\uvec_1),\ldots,g_r(\uvec_D))^\top$, i.e., $g_r$ is applied row-wise for a matrix, and $\Xi = (\xi_{n,r})$ is the matrix of coefficients. Also, to compute the eigendecomposition of the fitted CovNet, we need to find the matrix $\widetilde{\mathrm G} = (\widetilde g(r,s))$ with $\widetilde g(r,s) = \int_{\Q} g_r(\uvec)\,g_s(\uvec)\diff\uvec$. As shown in Section~\ref{sec:eigendecomposition}, we approximate
\begin{align*}
\widetilde{\mathrm G} \simeq \bigg(\frac{1}{M} \sum_{j=1}^M g_r(\uvec_j)\,g_s(\uvec_j)\bigg)_{r,s},
\end{align*}
where $\uvec_1,\ldots,\uvec_M$ are i.i.d.\ uniformly distributed over $\Q$. Now, with the conventions above, the approximation to $\widetilde{\mathrm G}$ can be obtained as:
\begin{align*}
\underbrace{\mathrm U}_{M \times d} \mapsto \mathrm Z = \underbrace{[g_1(\mathrm U),\ldots,g_R(\mathrm U)]}_{M \times R}, \quad \widetilde{\mathrm G} = M^{-1} \mathrm Z^\top \mathrm Z,
\end{align*}
where $\mathrm U = (\uvec_1^\top,\ldots,\uvec_M^\top)^\top$ and $g_r(\mathrm U) = (g_r(\uvec_1),\ldots,g_r(\uvec_M))^\top$. Thus, in all practical implementations, the computational cost depends crucially on the cost associated with the evaluation $\mathrm U \mapsto [g_1(\mathrm U),\ldots,g_R(\mathrm U)]$ for a matrix $\mathrm U$.

\smallskip

Recall that, for the shallow CovNet model, $g_r(\uvec) = \sigma(\mathbf w_r^\top \uvec + b_r)$ for $r=1,\ldots,R$. Writing $\mathrm W = (\mathbf w_1^\top,\ldots,\mathbf w_R^\top)^\top$ and $\mathbf b = (b_1,\ldots,b_R)^\top$, it is easy to see that
\begin{align*}
\mathrm U \mapsto \mathrm Z = [g_1(\mathrm U),\ldots,g_R(\mathrm U)] \quad\Leftrightarrow\quad \mathrm Z = \sigma(\mathrm U\,\mathrm W^\top + \mathbf b^\top),
\end{align*}
where the activation $\sigma$ is applied component-wise. This amounts to a computing cost of $\O(DRd)$.

\smallskip

For the deep CovNet model, $g_r$ is defined recursively as in \eqref{eq:deep_neural_network} (the parameters depend on $r$). In this case, $g_r(\mathrm U)$ can be obtained from $\mathrm U$ recursively as
\begin{align*}
\mathrm U_1 &= \sigma(\mathrm U\,\mathrm W_{1,r}^\top + \mathbf b_{1,r}^\top),\\
\mathrm U_{l+1} &= \sigma(\mathrm U_l\,\mathrm W_{l+1,r}^\top + \mathbf b_{l+1,r}^\top), \quad l=1,\ldots,L-1, \\
g_r(\mathrm U) &= \sigma(\mathrm U_L \mathbf w_{L+1,r} + b_{L+1,r}),
\end{align*}
where $\mathrm W_{1,r} \in \R^{p_1 \times d}$, $\mathrm W_{2,r} \in \R^{p_2 \times p_1},\ldots,\mathrm W_{L,r} \in \R^{p_L \times p_{L-1}}$, $\mathbf b_{1,r} \in \R^{p_1},\ldots,\mathbf b_{L,r} \in \R^{p_L}, \mathbf w_{L+1,r} \in \R^{p_L}$, $b_{L+1,r} \in \R$ are the parameters associated with $g_r$. Thus, the computing cost for the evaluation $\mathrm U \mapsto g_r(\mathrm U)$ is $\O\big(D\sum_{l=0}^L p_l p_{l+1}\big)$, where $p_0=d$ and $p_{L+1} = 1$. However, this needs to be done for $r=1,\ldots,R$. Thus, the total cost for computing $\mathrm U \mapsto [g_1(\mathrm U),\ldots,g_R(\mathrm U)]$ for the deep CovNet model is $\O\big(DR\sum_{l=0}^L p_l p_{l+1}\big)$.

\smallskip

Finally, consider the deepshared CovNet model, where $g_1,\ldots,g_R$ are defined recursively as in \eqref{eq:deepshared_neural_network}. In this case $\mathrm Z = [g_1(\mathrm U),\ldots,g_R(\mathrm U)]$ can be computed recursively as follows:
\begin{align*}
\mathrm U_1 &= \sigma(\mathrm U\,\mathrm W_1^\top + \mathbf b_1^\top),\\
\mathrm U_{l+1} &= \sigma(\mathrm U_l\,\mathrm W_{l+1}^\top + \mathbf b_{l+1}^\top), \quad l=1,\ldots,L-1,\\
\mathrm Z &= \sigma(\mathrm U_L\,\mathrm W^\top + \mathbf b^\top),
\end{align*}
where $\mathrm W_1 \in \R^{p_1 \times d},\mathrm W_2 \in \R^{p_2 \times p_1},\ldots,\mathrm W_L \in \R^{p_L \times p_{L-1}}$, $\mathbf b_1 \in \R^{p_1},\ldots,\mathbf b_L \in \R^{p_L}$, $\mathrm W = (\bm\omega_1^\top,\ldots,\bm\omega_R^\top)^\top \in \R^{R \times p_L}$, and $\mathbf b = (\beta_1,\ldots,\beta_R)^\top \in \R^R$ are the parameters associated with $g_1,\ldots,g_R$. Thus, the cost for the evaluation $\mathrm U \mapsto [g_1(\mathrm U),\ldots,g_R(\mathrm U)]$ in this case is $\O\big(D \big(\sum_{l=0}^{L-1} p_lp_{l+1} + Rp_L\big)\big)$, where $p_0 = d$.

\medskip

After obtaining $\mathrm Z = [g_1(\mathrm U),\ldots,g_R(\mathrm U)] \in \R^{D \times R}$, we require an additional $\O(NDR)$ computations to get $\mathbf X^{\rm NN}  = \Xi\,\mathrm Z^\top$. Similarly, for $\widetilde{\mathrm G} = M^{-1}\mathrm Z^\top \mathrm Z$, an additional $\O(MR^2)$ computations are required after obtaining $\mathrm Z \in \R^{M \times R}$ (in this case, all the previous complexities should be understood with $D=M$). We summarize these in Table~\ref{table:computational_complexity} for the special case when $p_1=\cdots=p_L=R$ for the deep and the deepshared CovNet models (as used in our numerical studies). Note that the complexities are always linear in the grid size $D$ or $M$. 

\begin{table}
\centering
\caption{Computational complexities associated with different tasks for different CovNet models.\\ For the deep and the deepshared models, we consider the case when $p_1=\cdots=p_L=R$.\label{table:computational_complexity}}
\begin{tabular}{lccccccc}
&& && Method && \\
Complexity && Shallow CovNet && Deep CovNet && Deepshared CovNet \\[3pt]
Evaluation && $\O\big(DRd\big)$ && $\O\big(DR^2(LR+d)\big)$ && $\O\big(DR(LR+d)\big)$ \\[2pt]
Obtaining $\mathbf X^{\rm NN}$ && $\O\big(DR(N+d)\big)$ && $\O\big(DR(LR^2+dR+N)\big)$ && $\O\big(DR(LR+d+N)\big)$ \\[2pt]
Approximating $\widetilde{\mathrm G}$ && $\O\big(MR(R+d)\big)$ && $\O\big(MR^2(LR+d)\big)$ && $\O\big(MR(LR+d)\big)$
\end{tabular}
\end{table}

\section{Covering numbers}\label{supp:covering_number}

In this section, we derive certain covering numbers that will be instrumental in the proofs of our asymptotic results. We give a brief overview of covering numbers and derive bounds on the covering numbers of some classes of functions useful in our context. The main results of this section are Lemmas~\ref{lemma:cover_covnet_shallow}, \ref{lemma:cover_covnet_deep} and \ref{lemma:cover_covnet_deepshared}, which give upper bounds on the covering numbers for the restricted classes of shallow, deep and deepshared CovNet operators, respectively. The reader can skip this section and go to the next section without loss of continuation.

We start with the definition of covering numbers for general metric spaces. More details can be found in \citet[][Chapter 10]{anthony1999}, \citet[][Chapter 9]{gyorfi2002}, \citet[][Chapter 5]{wainwright2019}. Let $(S,\rho)$ be a metric space, and let $T$ be a subset of $S$. For $\epsilon > 0$, a finite subset $T^\prime$ of $T$ is called an $\epsilon$-cover of $T$ w.r.t.\ the metric $\rho$ if for every $t \in T$ we can find an $t^\prime \in T^\prime$ such that $\rho(t,t^\prime) \le \epsilon$. The $\epsilon$-covering number of $T$ w.r.t.\ $\rho$ is the cardinality of the smallest $\epsilon$-cover of $T$. We denote the covering number by $\cover(\epsilon,T,\rho)$. If no finite $\epsilon$-cover of $T$ exists, then the covering number is defined to be $\infty$.

In our present context, we are interested in the covering numbers of the restricted classes of CovNet operators \eqref{eq:covnet_class_bounded} w.r.t.\ the Hilbert-Schmidt norm. Because of the equivalence between an operator $\G$ and the associated kernel $g$, we get $\vertj{\G}_2 = \|g\|_{\L_2(\Q \times \Q)}$. This shows that if we define $\F$ to be a class of non-negative definite kernels and $\widetilde\F$ to be the corresponding class of integral operators, then 
\begin{equation}\label{eq:operator_kernel_equivalence}
\cover(\epsilon,\widetilde\F,\verti{\cdot}_2) = \cover(\epsilon,\F,\|\cdot\|_{\L_2(\Q \times \Q)}).
\end{equation}
Thus, it is enough to find an upper bound on the covering numbers of the classes of CovNet kernels w.r.t.\ the $\L_2(\Q \times \Q)$ norm. We start by deriving a general bounding result.

\begin{theorem}\label{thm:covnet_kernel_class_general_bound}
Let $G$ be a class of functions from $\Q$ to $\R$ with $\|g\| \le M$ for every $g \in G$. For a positive real number $\lambda_N$ and an integer $R$, define the following class of functions from $\Q \times \Q$ to $\R$:
\[
\F_{R,G,\lambda_N} = \Bigg\{\sum_{r=1}^R\sum_{s=1}^R \lambda_{r,s}\,g_r(\uvec)\,g_s(\vvec): g_r \in G, 0 \preceq \Lambda=(\lambda_{r,s}) \preceq \lambda_N\mathrm{I}_R \Bigg\}.
\]
Then, the $\L_2$ covering number of $\F_{R,G,\lambda_N}$ is bounded as
\[
\cover\left(\epsilon,\F_{R,G,\lambda_N},\|\cdot\|_{\L_2(\Q \times \Q)}\right) \le \left[\frac{M^2R^2\lambda_N+\epsilon}{\epsilon} \times \left\{\frac{2e(4M^2R^2\lambda_N+\epsilon)}{\epsilon} \times \cover\left(\frac{\epsilon}{8MR^2\lambda_N},G,\|\cdot\|\right)\right\}^R\right]^R.
\]
\end{theorem}
We first state and prove a few lemmas, which will be used in the proof of the theorem.
\begin{lemma}\label{lemma:cover_linear_combination_positive}
Let $\mathcal I$ be a compact set and $\mathcal H$ be a collection of functions from $\mathcal I$ to $\R$ such that $\|h\|_{\L_2(\mathcal I)} \le M$ for all $h \in \mathcal H$. Define $\mathcal E_{R,\mathcal H,\gamma} = \big\{\sum_{i=1}^R \alpha_i h_i : h_i \in \mathcal H, 0 \le \alpha_i \le \gamma\big\}$. Then, for any $\epsilon_1,\epsilon_2 > 0$,
\[
\cover(\epsilon_1+\epsilon_2,\mathcal E_{R,\mathcal H,\gamma},\|\cdot\|_{\L_2(\mathcal I)}) \le \left\{\left(\frac{MR\gamma}{2\epsilon_2}+1\right)\times\cover\left(\frac{\epsilon_1}{R\gamma},\mathcal H,\|\cdot\|_{\L_2(\mathcal I)}\right)\right\}^R.
\]
\end{lemma}
\begin{proof}
Without loss of generality, we assume that all the covering numbers defined subsequently are finite. Fix $\epsilon_1,\epsilon_2 > 0$ and let $N_1 = \cover(\epsilon_1,\mathcal H,\|\cdot\|_{\L_2(\mathcal I)})$, $N_2 = \cover(\epsilon_2,[0,\gamma],|\cdot|)$. Thus, we can find a set ${\mathcal H}^\prime_{\epsilon_1} \subset \mathcal H$ of cardinality $N_1$ and a set of real numbers $A^\prime_{\epsilon_2} \subset [0,\gamma]$ of cardinality $N_2$ such that for every $h \in \mathcal H$ and every $\alpha \in [0,\gamma]$, there exists $h^\prime \in {\mathcal H}^\prime_{\epsilon_1}$ and $\alpha^\prime \in A^\prime_{\epsilon_2}$ such that $\|h - h^\prime\|_{\L_2(\mathcal I)} \le \epsilon_1$ and $|\alpha - \alpha^\prime| \le \epsilon_2$. Now, let $\sum_{i=1}^R \alpha_i h_i$ be an element of $\mathcal E_{R,\gamma}$. Let $h_1^\prime,\ldots,h_R^\prime \in {\mathcal H}^\prime_{\epsilon_1}$ and $\alpha_1^\prime,\ldots,\alpha_R^\prime \in A^\prime_{\epsilon_2}$ be such that $\|h_i - h_i^\prime\|_{\L_2(\mathcal I)} \le \epsilon_1$ and $|\alpha_i - \alpha_i^\prime| \le \epsilon_2$ for $i=1,\ldots,R$. Now,
\begin{align*}
\left\|\sum_{i=1}^R \alpha_i h_i - \sum_{i=1}^R \alpha_i^\prime h_i^\prime\right\|_{\L_2(\mathcal I)} &= \left\|\sum_{i=1}^R (\alpha_i - \alpha_i^\prime) h_i + \sum_{i=1}^R \alpha_i^\prime (h_i - h_i^\prime)\right\|_{\L_2(\mathcal I)} \\
&\le \sum_{i=1}^R \underbrace{|\alpha_i-\alpha_i^\prime|}_{\le \epsilon_2}\,\underbrace{\|h_i\|_{\L_2(\mathcal I)}}_{\le M} + \sum_{i=1}^R \underbrace{\alpha_i^\prime}_{\le \gamma}\, \underbrace{\|h_i-h_i^\prime\|_{\L_2(\mathcal I)}}_{\le \epsilon_1} \\
&= R(\gamma\epsilon_1+M\epsilon_2).
\end{align*}
This shows that ${\mathcal E}^\prime_{\epsilon_1,\epsilon_2} := \big\{\sum_{i=1}^R \alpha_i^\prime h_i^\prime : \alpha_i^\prime \in A_{\epsilon_2}^\prime, h_i^\prime \in {\mathcal H}_{\epsilon_1}^\prime\big\}$ is an $\L_2$-cover for $\mathcal E_{R,\gamma}$ of size $R(\gamma\epsilon_1+M\epsilon_2)$. Thus,
\[
\cover\big(R(\gamma\epsilon_1+M\epsilon_2),\mathcal E_{R,\gamma},\|\cdot\|_{\L_2(\mathcal I)}\big) \le |{\mathcal E}^\prime_{\epsilon_1,\epsilon_2}| = (N_1N_2)^R.
\]
Note that $N_2 = \cover(\epsilon_2,[0,\gamma],|\cdot|) \le \gamma/(2\epsilon_2)+1$ \citep[][Example 5.2]{wainwright2019}, which shows that
\[
\cover\big(R(\gamma\epsilon_1+M\epsilon_2),\mathcal E_{R,\gamma},\|\cdot\|_{\L_2(\mathcal I)}\big) \le \left\{\left(\frac{\gamma}{2\epsilon_2}+1\right)\times \cover(\epsilon_1,\mathcal H,\|\cdot\|_{\L_2(\mathcal I)}\right\}^R.
\]
The proof is complete after transforming $\epsilon_1 \mapsto \epsilon_1/(R\gamma)$ and $\epsilon_2 \mapsto \epsilon_2/(RM)$.
\end{proof}
As a corollary, taking $\epsilon_1 = \epsilon_2 = \epsilon/2$, we get the following.
\begin{corollary}\label{cor:cover_linear_combination_positive2}
With $\mathcal E_{R,\mathcal H,\gamma}$ defined as in Lemma~\ref{lemma:cover_linear_combination_positive}, for any $\epsilon > 0$,
\[
\cover(\epsilon,\mathcal E_{R,\mathcal H,\gamma},\|\cdot\|_{\L_2(\mathcal I)}) \le \left\{\left(\frac{MR\gamma}{\epsilon}+1\right)\times\cover\left(\frac{\epsilon}{2R\gamma},\mathcal H,\|\cdot\|_{\L_2(\mathcal I)}\right)\right\}^R.
\]
\end{corollary}
We will also use a modified version of Lemma~16.6 in \cite{gyorfi2002}. The proof follows straightforwardly from the proof in \cite{gyorfi2002}. We give a brief sketch of the proof here for completeness.
\begin{lemma}\label{lemma:cover_linear_combination_l1}
Let $\mathcal I$ be a compact set and $\mathcal H$ be a class of functions from $\mathcal I$ to $\R$ with $\|h\|_{\L_2(\mathcal I)} \le M$ for all $h \in \mathcal H$. For a positive real number $\gamma$ and an integer $R$, define the class of functions
\[
\mathcal E_{R,\mathcal H,\gamma} = \left\{\sum_{i=1}^R \alpha_i\,h_i : h_i \in \mathcal H, \sum_{i=1}^R |\alpha_i| \le \gamma\right\}.
\]
Then, for any $\epsilon_1,\epsilon_2 > 0$, we have
\[
\cover(\epsilon_1 + \epsilon_2, \mathcal E_{R,\mathcal H,\gamma}, \|\cdot\|_{\L_2(\mathcal I)}) \le \left\{\frac{e(M\gamma+2\epsilon_2)}{\epsilon_2} \times \cover\left(\frac{\epsilon_1}{\gamma},\mathcal H,\|\cdot\|_{\L_2(\mathcal I)}\right)\right\}^R.
\]
\end{lemma}
\begin{proof}
Assume without loss of generality that all the covering numbers defined subsequently are finite. Let $N = \cover(\epsilon_1,\mathcal H,\|\cdot\|_{\L_2(\mathcal I)})$. Then, we can find a subset ${\mathcal H}^\prime_{\epsilon_1} \subset \mathcal H$ of cardinality $N$ such that for every $h \in \mathcal H$, we can find $h^\prime \in {\mathcal H}^\prime_{\epsilon_1}$ such that $\|h-h^\prime\|_{\L_2(\mathcal I)} \le \epsilon_1$. Now, let $S_{\gamma} = \{(\alpha_1,\ldots,\alpha_R) \in \R^R: \sum_{i=1}^R |\alpha_i| \le \gamma\}$, and $S_{\gamma,\epsilon_2}^\prime$ be a finite subset of $\R^R$ such that for every $(\alpha_1,\ldots,\alpha_R) \in S_{\gamma}$ we can find $(\alpha_1^\prime,\ldots,\alpha_R^\prime) \in S_{\gamma,\epsilon_2}^\prime$ with $\sum_{i=1}^R |\alpha_i-\alpha_i^\prime| \le \epsilon_2$. Define, ${\mathcal E}_{\epsilon_1,\epsilon_2}^\prime := \big\{\sum_{i=1}^R \alpha_i\,h_i : (\alpha_1,\ldots,\alpha_R) \in S_{\gamma,\epsilon_2}^\prime, h_i \in {\mathcal H}^\prime_{\epsilon_1}\big\}$. Then, for $\sum_{i=1}^R \alpha_i h_i \in \mathcal E_{R,\mathcal H,\gamma}$, we can find $\sum_{i=1}^R \alpha_i^\prime h_i^\prime \in {\mathcal E}_{\epsilon_1,\epsilon_2}^\prime$ such that
\begin{align*}
\left\|\sum_{i=1}^R \alpha_ih_i - \sum_{i=1}^R \alpha_i^\prime h_i^\prime\right\|_{\L_2(\mathcal I)} &= \left\|\sum_{i=1}^R \alpha_i(h_i-h_i^\prime) + \sum_{i=1}^R(\alpha_i-\alpha_i^\prime)h_i^\prime\right\|_{\L_2(\mathcal I)} \\
&\le \underbrace{\sum_{i=1}^R |\alpha_i|}_{\le \gamma}\,\underbrace{\left\|h_i - h_i^\prime\right\|_{\L_2(\mathcal I)}}_{\le \epsilon_1} + \underbrace{\sum_{i=1}^R |\alpha_i - \alpha_i^\prime|}_{\le \epsilon_2}\,\underbrace{\|h_i^\prime\|_{\L_2(\mathcal I)}}_{\le M} \\
&\le \gamma\epsilon_1 + M\epsilon_2.
\end{align*}
This shows that ${\mathcal E}_{\epsilon_1,\epsilon_2}^\prime$ is an $\L_2$ cover for $\mathcal E_{R,\mathcal H,\gamma}$ of size $(\gamma\epsilon_1+M\epsilon_2)$. Thus,
\[
\cover(\gamma\epsilon_1+M\epsilon_2,\mathcal E_{R,\mathcal H,\gamma},\|\cdot\|_{\L_2(\mathcal I)}) \le N^R \times |S_{\gamma,\epsilon_2}^\prime|.
\]
Now, as shown in the proof of Lemma~16.6 in \cite{gyorfi2002},
\[
|S_{\gamma,\epsilon_2}^\prime| \le \left(\frac{e(\gamma+2\epsilon_2)}{\epsilon_2}\right)^R,
\]
which shows that
\[
\cover(\gamma\epsilon_1+M\epsilon_2,\mathcal E_{R,\mathcal H,\gamma},\|\cdot\|_{\L_2(\mathcal I)}) \le \left(\frac{e(\gamma+2\epsilon_2)}{\epsilon_2}\right)^R \times \cover\big(\epsilon_1,\mathcal H,\|\cdot\|_{\L_2(\mathcal I)}\big)^R.
\]
The proof is complete upon transforming $\epsilon_1 \mapsto \epsilon_1/\gamma$ and $\epsilon_2 \mapsto \epsilon_2/M$.
\end{proof}
With $\epsilon_1=\epsilon_2 = \epsilon/2$, we get the following corollary.
\begin{corollary}\label{cor:cover_linear_combination_l1}
With $\mathcal E_{R,\mathcal H,\gamma}$ defined as in Lemma~\ref{lemma:cover_linear_combination_l1}, for any $\epsilon > 0$
\[
\cover(\epsilon, \mathcal E_{R,\mathcal H,\gamma}, \|\cdot\|_{\L_2(\mathcal I)}) \le \left\{\frac{2e(M\gamma+\epsilon)}{\epsilon} \times \cover\left(\frac{\epsilon}{2\gamma},\mathcal H,\|\cdot\|_{\L_2(\mathcal I)}\right)\right\}^R.
\] 
\end{corollary}

We will also need the following result about the covering number of the space of product functions.
\begin{lemma}\label{lemma:cover_product_space}
Let $\mathcal I$ be a compact set and $\mathcal H$ be a collection of functions from $\mathcal I$ to $\R$ such that $\|h\|_{\L_2(\mathcal I)} \le M$ for every $h \in \mathcal H$. Define $\mathcal E = \big\{e: \mathcal I \times \mathcal I \to \R \text{ with }  e(\uvec,\vvec) = h(\uvec)\,h(\vvec), \text{ where } h \in \mathcal H\big\}$. Then,
\[
\cover\big(\epsilon,\mathcal E,\|\cdot\|_{\L_2(\mathcal I \times \mathcal I)}\big) \le \cover\left(\frac{\epsilon}{2M}, \mathcal H, \|\cdot\|_{\L_2(\mathcal I)}\right).
\] 
\end{lemma}
\begin{proof}
Let $e_1 = h_1 \otimes h_1$ and $e_2 = h_2 \otimes h_2$ be two functions in $\mathcal E$. Then,
\begin{align*}
\|e_1 - e_2\|_{\L_2(\mathcal I \times \mathcal I)} &= \|h_1 \otimes h_1 - h_2 \otimes h_2\|_{\L_2(\mathcal I \times \mathcal I)} \\
&= \|h_1 \otimes(h_1 - h_2) + (h_1 - h_2) \otimes h_2\|_{\L_2(\mathcal I \times \mathcal I)} \\
&\le \|h_1 - h_2\|_{\L_2(\mathcal I)} \big(\|h_1\|_{\L_2(\mathcal I)} + \|h_2\|_{\L_2(\mathcal I)}\big) \\
&\le 2M \|h_1 - h_2\|_{\L_2(\mathcal I)}.
\end{align*}
This shows that an $\L_2$-cover for $\mathcal H$ of size $\epsilon$ provides an $\L_2$-cover for $\mathcal E$ of size $2M \epsilon$, proving the lemma.
\end{proof}
We now proceed to prove Theorem~\ref{thm:covnet_kernel_class_general_bound}.
\begin{proof}[Proof of Theorem~\ref{thm:covnet_kernel_class_general_bound}]
Since $\Lambda = (\lambda_{r,s})$ is positive semi-definite, we can find orthonormal vectors $\mathbf{e}_1,\ldots,\mathbf{e}_R \in \R^R$ and non-negative real numbers $\eta_1,\ldots,\eta_R$ such that $\Lambda = \sum_{i=1}^{R} \eta_i\,\mathbf{e}_i\,\mathbf{e}_i^\top$. Moreover, since $\Lambda \preceq \lambda_N {\rm I}_R$, we also get that $\eta_i \le \lambda_N$ for $i=1,\ldots,R$. Using this, we can write
\begin{align*}
\sum_{r=1}^R \sum_{s=1}^R \lambda_{r,s}\,g_r(\uvec)\,g_s(\vvec) &= \sum_{r=1}^R \sum_{s=1}^R \left(\sum_{i=1}^{R} \eta_i\,e_{i,r}\,e_{i,s}\right) g_r(\uvec)\,g_s(\vvec) \\
&= \sum_{i=1}^R \eta_i \left(\sum_{r=1}^R e_{i,r}\,g_r(\uvec)\right)\left(\sum_{s=1}^R e_{i,s}\,g_s(\vvec)\right) = \sum_{i=1}^R \eta_i\,\widetilde g_i(\uvec)\,\widetilde g_i(\vvec),
\end{align*}
where $e_{i,r}$ is the $r$-th coordinate of $\mathbf{e}_i$ and $\widetilde g_i(\uvec) = \sum_{r=1}^R e_{i,r} g_r(\uvec)$. Since $\mathbf{e}_i$'s are orthonormal, $\sum_{r=1}^R e_{ir}^2 = 1$ for every $i=1,\ldots,R$. This shows that we can rewrite $\F_{R,G,\lambda_N}$ as
\begin{equation*}
\F_{R,G,\lambda_N} = \left\{\sum_{i=1}^R \eta_i\,f_i(\uvec)\,f_i(\vvec): f_i \in G_R^{(0)}, 0 \le \eta_i \le \lambda_N\right\},
\end{equation*}
where
\begin{equation*}
G_R^{(0)} = \left\{\sum_{r=1}^R a_r\,g_r(\uvec): g_r \in G, \sum_{r=1}^R a_r^2 = 1\right\}.
\end{equation*}
Now, for any $f \in G_R^{(0)}$, 
\[
\|f\| = \left\|\sum_{r=1}^R a_r\,g_r\right\| \le \sum_{r=1}^R |a_r|\,\underbrace{\|g_r\|}_{\le M} \le M \sum_{r=1}^R |a_r| \le \sqrt{R} M,
\]
where we have used that $\sum_{r=1}^R a_r^2 = 1$ implies $\sum_{r=1}^R |a_r| \le \sqrt{R}$ by the Cauchy-Schwarz inequality. So, for $f \in G_R^{(0)}$, $\|f \otimes f\|_{\L_2(\Q \times \Q)} = \|f\|^2 \le RM^2$. Now, using Corollary~\ref{cor:cover_linear_combination_positive2},
\begin{equation*}\label{eq:cover_general_eq1}
\cover(\epsilon,\F_{R,G,\lambda_N},\|\cdot\|_{\L_2(\Q \times \Q)}) \le \left\{\left(\frac{M^2R^2\lambda_N}{\epsilon}+1\right)\times \cover\left(\frac{\epsilon}{2R\lambda_N},\big\{f \otimes f: f \in G_R^{(0)}\big\}, \|\cdot\|_{\L_2(\Q \times \Q)}\right)\right\}^R.
\end{equation*}
Also, using Lemma~\ref{lemma:cover_product_space}, we get
\[
\cover\left(\epsilon,\big\{f \otimes f: f \in G_R^{(0)}\big\}, \|\cdot\|_{\L_2(\Q \times \Q)}\right) \le \cover\left(\frac{\epsilon}{2\sqrt{R}M},G_R^{(0)},\|\cdot\|\right).
\]
Plugging this into the previous equation, we get
\begin{equation}\label{eq:cover_general_eq2}
\cover(\epsilon,\F_{R,G,\lambda_N},\|\cdot\|_{\L_2(\Q \times \Q)}) \le \left\{\left(\frac{M^2R^2\lambda_N}{\epsilon}+1\right)\times \cover\left(\frac{\epsilon}{4MR^{3/2}\lambda_N},G_R^{(0)},\|\cdot\|\right)\right\}^R.
\end{equation}
Now, $\sum_{r=1}^R a_r^2 = 1$ implies $\sum_{r=1}^R |a_r| \le \sqrt{R}$, and thus
\[
G_R^{(0)} = \left\{\sum_{r=1}^R a_r g_r : g_r \in G, \sum_{r=1}^R a_r^2=1\right\} \subset \left\{\sum_{r=1}^R a_r g_r : g_r \in G, \sum_{r=1}^R |a_r| \le \sqrt{R}\right\}.
\]
Now, using Corollary~\ref{cor:cover_linear_combination_l1}, we bound the covering number of $G_R^{(0)}$ as
\[
\cover(\epsilon,G_R^{(0)},\|\cdot\|) \le \left\{\frac{2e(M\sqrt{R}+\epsilon)}{\epsilon} \times \cover\left(\frac{\epsilon}{2\sqrt{R}},G,\|\cdot\|\right)\right\}^R.
\]
The proof follows by plugging this into \eqref{eq:cover_general_eq2}.
\end{proof}

We will use Theorem~\ref{thm:covnet_kernel_class_general_bound} to derive upper bounds on the covering numbers of the shallow and the deep CovNet kernel classes, respectively. 

\subsection{Covering number of the shallow CovNet class}\label{supp:cover_shallow}
Recall that the (restricted) shallow CovNet class of kernels is defined as
\[
\F_{R,\lambda_N}^{\rm sh} = \left\{\sum_{r=1}^R\sum_{s=1}^R \lambda_{r,s}\,\sigma(\mathbf w_r^\top\uvec+b_r)\,\sigma(\mathbf w_s^\top\vvec+b_s): \mathbf w_r \in \R^d, b_r \in \R, 0 \preceq \Lambda = (\lambda_{r,s}) \preceq \lambda_N \mathrm{I}_R\right\},
\]
which has the form $\F_{R,G,\lambda_N}$, where $G = \big\{\sigma(\mathbf w^\top \uvec+b) : \mathbf w \in \R^d, b \in \R\big\}$. Since the activation function $\sigma$ is sigmoidal, in particular $0 \le \sigma(t) \le 1$ for all $t \in \R$, $\|g\| \le \sqrt{|\Q|}$ for all $g \in G$. Again, since $\sigma$ is non-decreasing, the VC dimension of $G$ is bounded by $d+2$ \citep[see][page 314]{gyorfi2002}. Since $\Q$ is a compact set, the measure $\nu$ defined as $\nu(A) := |\Q \cap A|/|\Q|$ for $A \subset \R^d$, is a probability measure on $\R^d$. So, using Theorem~9.4 in \cite{gyorfi2002}, we get that
\[
\cover\big(\epsilon,G,\|\cdot\|_{\L_2(\nu)}\big) \le 3 \bigg(\frac{2e}{\epsilon^2}\log\frac{3e}{\epsilon^2}\bigg)^{d+2}.
\]
Now, $\|f\| = \sqrt{|\Q|}\,\|f\|_{\L_2(\nu)}$ implies that
\begin{align}\label{eq:cover_prob_to_general}
\cover\big(\epsilon,G,\|\cdot\|\big) &= \cover\bigg(\frac{\epsilon}{\sqrt{|\Q|}}, G, \|\cdot\|_{\L_2(\nu)}\bigg) \le 3 \bigg(\frac{2e|\Q|}{\epsilon^2}\log\frac{3e|\Q|}{\epsilon^2}\bigg)^{d+2}.
\end{align}
Using this in Theorem~\ref{thm:covnet_kernel_class_general_bound} (with $M=\sqrt{|\Q|}$), we get that
\begin{align*}
&\cover\left(\epsilon,\F_{R,\lambda_N}^{\rm sh},\|\cdot\|_{\L_2(\Q \times \Q)}\right) \\
&\le \left[\frac{|\Q|R^2\lambda_N+\epsilon}{\epsilon} \times \left\{\frac{2e(4|\Q|R^2\lambda_N+\epsilon)}{\epsilon} \times \cover\left(\frac{\epsilon}{8\sqrt{|\Q|}R^2\lambda_N},G,\|\cdot\|\right)\right\}^R\right]^R \\
& \le \left[\frac{|\Q|R^2\lambda_N+\epsilon}{\epsilon} \times \left\{\frac{2e(4|\Q|R^2\lambda_N+\epsilon)}{\epsilon} \times 3 \bigg(\frac{128e|\Q|^2R^4\lambda_N^2}{\epsilon^2}\log\frac{192e|\Q|^2R^4\lambda_N^2}{\epsilon^2}\bigg)^{d+2}\right\}^R\right]^R \\
&= \left[\frac{|\Q|R^2\lambda_N+\epsilon}{\epsilon} \times \left\{\frac{6e(4|\Q|R^2\lambda_N+\epsilon)}{\epsilon} \times \bigg(\frac{256e|\Q|^2R^4\lambda_N^2}{\epsilon^2}\log\frac{\sqrt{192e}|\Q|R^2\lambda_N}{\epsilon}\bigg)^{d+2}\right\}^R\right]^R  \\
&\le \left[\frac{|\Q|R^2\lambda_N+\epsilon}{\epsilon} \times \left\{\frac{6e(4|\Q|R^2\lambda_N+\epsilon)}{\epsilon} \times \bigg(\frac{256e\times \sqrt{192e} \times (|\Q|R^2\lambda_N)^3}{\epsilon^3}\bigg)^{d+2}\right\}^R\right]^R~(\text{using } \log(x) \le x) \\
&\le \left(\frac{6e \times (16\sqrt{e}\,|\Q|R^2\lambda_N + \epsilon)}{\epsilon}\right)^{(3d+7)R^2+R}.
\end{align*}
We summarize this in the following Lemma.

\begin{lemma}\label{lemma:cover_covnet_shallow}
For the shallow CovNet class of operators $\widetilde\F_{R,\lambda_N}^{\rm sh}$, the $\vertj{\cdot}_2$-covering number is bounded as
\[
\cover\big(\epsilon,\widetilde\F_{R,\lambda_N}^{\rm sh},\vertj{\cdot}_2\big) \le \left(\frac{c_0 \times (c_1|\Q|R^2\lambda_N + \epsilon)}{\epsilon}\right)^{(3d+7)R^2+R},
\]
where $c_0$ and $c_1$ are constants independent of all the other parameters. In particular, when $R,\lambda_N \to \infty$, we get
\[
\log\cover\big(\epsilon,\widetilde\F_{R,\lambda_N}^{\rm sh},\vertj{\cdot}_2\big) = \O\bigg(dR^2\log\Big(\frac{R^2\lambda_N+\epsilon}{\epsilon}\Big)\bigg).
\]
\end{lemma}

\subsection{Covering number of the deep CovNet class}\label{supp:cover_deep}
The (restricted) deep CovNet class of kernels is defined as
\[
\F_{R,L,\mathbf{p},\lambda_N}^{\rm d} = \left\{\sum_{r=1}^R \sum_{s=1}^R \lambda_{r,s}\,g_r(\uvec)\,g_s(\vvec): g_r \in \mathcal D_{L,\mathbf{p}}, 0 \preceq \Lambda = (\lambda_{r,s}) \preceq \lambda_N \mathrm{I}_R\right\},
\]
which is again of the form $\F_{R,G,\lambda_N}$, where $G = \mathcal D_{L,\mathbf p}$ is the class of deep neural networks with depth $L$ and layer-wise widths $p_1,\ldots,p_L$ \eqref{eq:deep_neural_network_class}.

Let $K = p_1+\cdots+p_L$ be the number of \emph{computation units} of a network from the class $\mathcal D_{L,\mathbf p}$ and $W = \sum_{\ell=1}^L p_{\ell} (p_{\ell-1}+1) + p_L$ be the number of \emph{adjustable parameters} \citep[see][Chapter~6 for details]{anthony1999}. Now, by \citet[][Theorem~18.8]{anthony1999}, we can find constants $\alpha_1,\alpha_2,\alpha_3$ such that for any probability distribution $\nu$ on $\R^d$,
\[
\log\cover(\epsilon,\mathcal D_{L,\mathbf p},\|\cdot\|_{\L_2(\nu)}) \le \alpha_1 \textrm{fat}_{\mathcal D_{L,\mathbf p}}(\alpha_2\epsilon) \log^2\left(\frac{\textrm{fat}_{\mathcal D_{L,\mathbf p}}(\alpha_2\epsilon^2)}{\epsilon}\right),
\]
where $\textrm{fat}_{\F}(\cdot)$ is the \emph{fat-shattering dimension} of the class of functions $\F$ \citep[see][Chapter~11]{anthony1999}. Thus, using \eqref{eq:cover_prob_to_general}
\begin{align*}
\log\cover(\epsilon,\mathcal D_{L,\mathbf p},\|\cdot\|) &= \log\cover\left(\frac{\epsilon}{\sqrt{|\Q|}},\mathcal D_{L,\mathbf p},\|\cdot\|_{\L_2(\nu)}\right) \\
&\le \alpha_1 \textrm{fat}_{\mathcal D_{L,\mathbf p}}\left(\alpha_2\frac{\epsilon}{\sqrt{|\Q|}}\right) \log^2\left(\frac{\textrm{fat}_{\mathcal D_{L,\mathbf p}}(\alpha_2\epsilon^2/|\Q|)\sqrt{|\Q|}}{\epsilon}\right).
\end{align*}
Again, for any $\delta > 0$, $\textrm{fat}_{\F}(\delta) \le \textrm{Pdim}(\F)$, where $\textrm{Pdim}(\F)$ is the \emph{pseudo-dimension} of the class of functions $\F$ \citep[see][Theorem~11.13(i)]{anthony1999}. Using this in the previous equation, we get
\[
\log\cover(\epsilon,\mathcal D_{L,\mathbf p},\|\cdot\|) \le \alpha\, \textrm{Pdim}(\mathcal D_{L,\mathbf p}) \log^2\left(\frac{\textrm{Pdim}(\mathcal D_{L,\mathbf p})\sqrt{|\Q|}}{\epsilon}\right),
\]
for some positive constant $\alpha$. Finally, since the activation function is sigmoidal, using Theorem~14.2 in \cite{anthony1999}, we get
\[
\textrm{Pdim}(\mathcal D_{L,\mathbf p}) \le \big((W+2)K\big)^2 + 11(W+2)K \log_2\big(18(W+2)K^2\big) =: d_{W,K}.
\]
Now, Theorem~\ref{thm:covnet_kernel_class_general_bound} gives us that
\begin{align*}
&\log\cover(\epsilon,\F^{\rm d}_{R,L,\mathbf{p},\lambda_N},\|\cdot\|_{\L_2(\Q \times \Q)}) \\
&\kern5ex \le R\left[\log\bigg(\frac{|\Q|R^2\lambda_N+\epsilon}{\epsilon}\bigg) + R\left\{\log\bigg(\frac{2e(4|\Q|R^2\lambda_N+\epsilon)}{\epsilon}\bigg) + \log\cover\bigg(\frac{\epsilon}{8\sqrt{|\Q|}R^2\lambda_N},\mathcal D_{L,\mathbf{p}},\|\cdot\|\bigg)\right\}\right]\\
&\kern5ex \le R\left[\log\bigg(\frac{|\Q|R^2\lambda_N+\epsilon}{\epsilon}\bigg) + R\left\{\log\bigg(\frac{2e(4|\Q|R^2\lambda_N+\epsilon)}{\epsilon}\bigg) + \alpha\, d_{W,K} \log^2\bigg(\frac{8|\Q|d_{W,K}R^2\lambda_N}{\epsilon}\bigg)\right\}\right] \\
&\kern5ex \le R\Big\{1+R\big(1+\alpha\,d_{W,K}\big)\Big\} \log^2\bigg(\frac{c_0|\Q|R^2\lambda_N+\epsilon}{\epsilon}\bigg) \\
&\kern5ex =\Big\{\big(1+\alpha\,d_{W,K}\big)R^2+R\Big\} \log^2\bigg(\frac{c_0d_{W,K}|\Q|R^2\lambda_N+\epsilon}{\epsilon}\bigg),
\end{align*}
for some constant $c_0$ and $\alpha$. We summarize it in the following lemma.
\begin{lemma}\label{lemma:cover_covnet_deep}
Consider the deep CovNet class of operators $\widetilde\F^{\rm d}_{R,L,\mathbf p,\sigma}$, where $L$ is the number of hidden layers and $\mathbf{p}=(p_1,\ldots,p_L)$ are the number of nodes at each of the hidden layers, $\sigma$ is the sigmoidal activation function, and $R$ is the number of components of the model. Then, there exist constants $\alpha$ and $c_0$ such that
\[
\log\cover\big(\epsilon,\widetilde\F^{\rm d}_{R,L,\mathbf p,\sigma},\vertj{\cdot}_2\big) \le \big\{(1+\alpha\,d_{W,K})R^2+R\big\}\,\log^2\bigg(\frac{c_0d_{W,K}|\Q|R^2\lambda_N+\epsilon}{\epsilon}\bigg),
\]
where $W = \sum_{\ell=1}^L p_{\ell} (p_{\ell-1}+1) + p_L$, $K = \sum_{\ell=1}^L p_{\ell}$, and $d_{W,K} = \big((W+2)K\big)^2 + 11(W+2)K \log_2\big(18(W+2)K^2\big)$.
\end{lemma}

\begin{remark}\label{remark:cover_deep_order}
If we define $p_{\rm max} = \max\{d,p_1,\ldots,p_L\}$, then $W = \O(Lp_{\rm max}^2)$, $K = \O(Lp_{\rm max})$, and hence $d_{W,K} = \O(L^4p_{\max}^6)$. This shows that
\[
\log\cover\big(\epsilon,\widetilde\F^{\rm d}_{R,L,\mathbf p,\sigma},\vertj{\cdot}_2\big) = \O\bigg(L^4R^2p_{\rm max}^6\,\log^2\Big(\frac{L^4R^2\lambda_Np_{\rm max}^6+\epsilon}{\epsilon}\Big)\bigg).
\]
In particular, when $p_1,\ldots,p_L=R$ and $R>d$, we get that
\[
\log\cover\big(\epsilon,\widetilde\F^{\rm d}_{R,L,\mathbf p,\sigma},\vertj{\cdot}_2\big) = \O\bigg(L^4R^8\,\log^2\Big(\frac{L^4R^8\lambda_N+\epsilon}{\epsilon}\Big)\bigg).
\]
\end{remark}

\subsection{Covering number of the deepshared CovNet class}\label{supp:cover_deepshared}

Although the deepshared CovNet kernel has a similar structure to the deep CovNet kernel, unfortunately, Theorem~\ref{thm:covnet_kernel_class_general_bound} is not useful to bound the covering number of the deepshared CovNet class of operators. This is due to the fact that the shared structure of the constituents $g_1,\ldots,g_R$ of a deepshared CovNet kernel $\sum_{r=1}^R\sum_{s=1}^R \lambda_{r,s}\,g_r(\uvec)\,g_s(\vvec)$ cannot be catered for using Theorem~\ref{thm:covnet_kernel_class_general_bound}. So, here, we proceed in a different way.

First, consider a kernel $g$ from the deepshared CovNet class of kernels $\F^{\rm ds}_{R,L,\mathbf p,\sigma}$ (see \eqref{eq:deepshared_covnet_class}). We have demonstrated that the kernel $g$ is a neural network (see Figure~\ref{fig:deepshared_covnet}). For constants $\theta$ and $\mu$, we define another kernel $g^\prime$ associated to $g$ as follows:
\[
g^\prime(\uvec,\vvec) = \1\{\theta\,g(\uvec,\vvec)+\mu > 0\},~~\uvec,\vvec \in \Q.
\]
The kernel $g^\prime$ also has the neural network structure (see Figure~\ref{fig:deepshared_derived}). It has one additional input and one additional layer than $g$, and the output has the \emph{linear threshold} activation function. We define $\F^\prime$ to be the class of all derived kernels from the deepshared CovNet kernel class $\F^{\rm ds}_{R,L,\mathbf p,\sigma}$.

\begin{figure}[t]
\centering
\includegraphics[height=0.5\textheight]{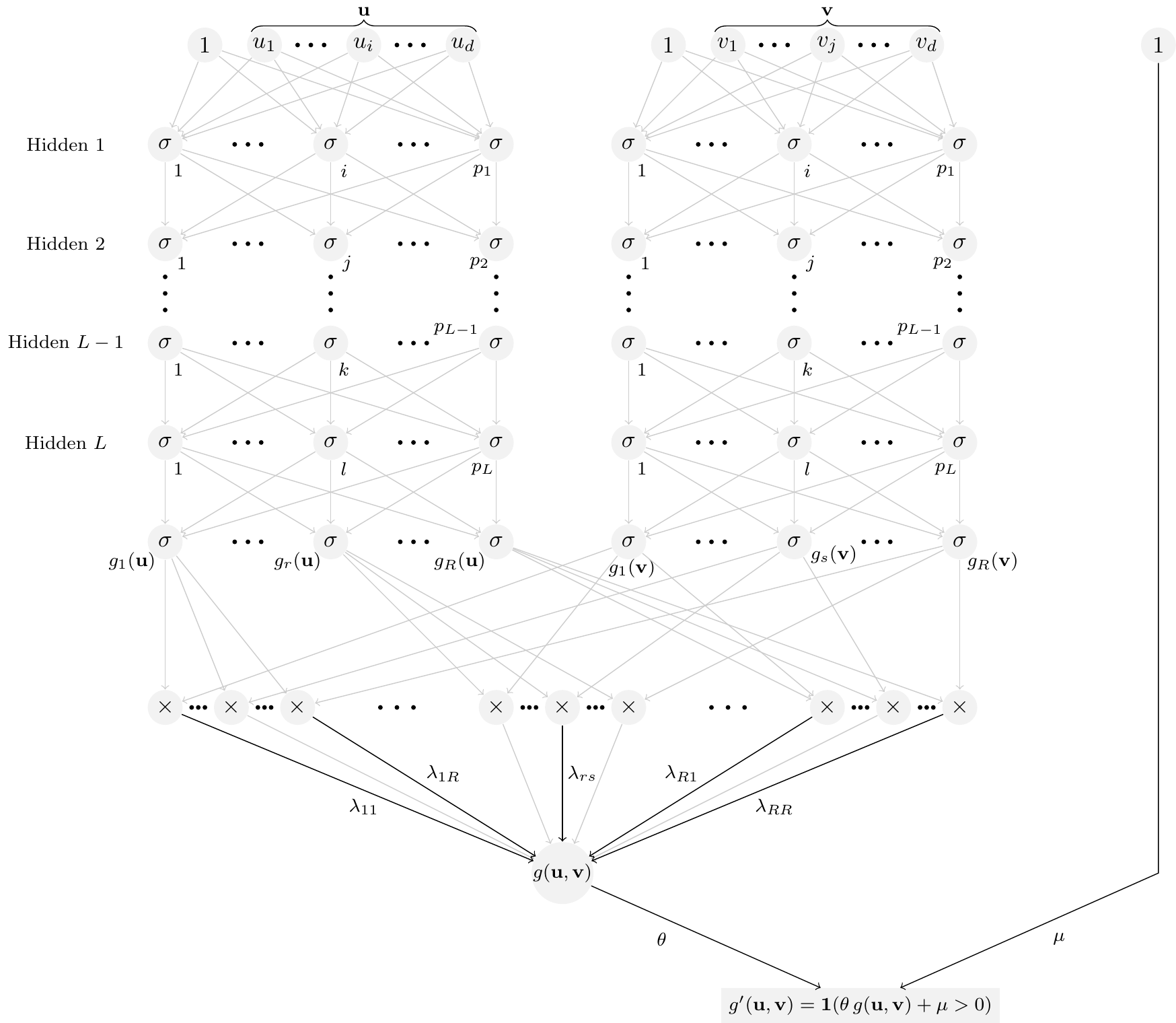}
\caption{\label{fig:deepshared_derived}A schematic representation of the network $g^\prime$ derived from a deepshared network $g$.}
\end{figure}

By construction, using Theorem~14.1 in \citet{anthony1999}, we get that
\begin{align}\label{eq:cover_deepshared_pdim_bound}
\textrm{Pdim}(\F^{\rm ds}_{R,L,\mathbf p,\sigma}) \le \textrm{VCdim}(\F^\prime),
\end{align}
where $\textrm{VCdim}(\F)$ is the Vapnik-Cherbonenkis dimension of the class of functions $\F$. Let $\bm\theta \in \R^W$ denote all the free parameters of the kernel $g^\prime$ (i.e., the weights and biases of the different layers). It is not difficult to see that $W = \sum_{l=1}^L p_l(p_{l-1}+1) + R(p_L+1) + R(R+1)/2 + 2$, where $p_0 = d$. Thus, $g^\prime$ can be viewed as a function from $\R^{2d \times W}$ to $\R$, which maps an element $\wvec = (\uvec^\top,\vvec^\top)^\top \in \R^{2d}$ and $\bm\theta \in \R^W$ to $g^\prime(\uvec,\vvec)$. Also, the function can be evaluated using $t:= 4\sum_{l=1}^L p_l(p_{l-1}+1) + 4R(p_L+1) + 2R^2 + 2$ basic operations of the form:
\begin{itemize}
\item the exponential function $x \mapsto \exp(x)$,
\item the arithmetic operations $+,-,\times,\div$ on real numbers,
\item jumps conditioned on $>,\ge,<,\le,=,\ne$ and comparisons of real numbers,
\end{itemize}
and the output is $\{0,1\}$-valued. Moreover, only $K=2\sum_{l=1}^L p_l + 2R$ of these operations involve the application of the exponential function. Thus, from Theorem~8.14 in \cite{anthony1999}, we get that
\begin{align}\label{eq:cover_deepshared_vcdim_bound}
\mathrm{VCdim}(\F^\prime) \le (W(K+1))^2 + 11 W(K+1)(t+\log_2(9W(K+1)).
\end{align}
Combining \eqref{eq:cover_deepshared_pdim_bound} and \eqref{eq:cover_deepshared_vcdim_bound}, we get that
\begin{align}\label{eq:cover_deepshared_pdim_bound2}
\mathrm{Pdim}(\F^{\rm ds}_{R,L,\mathbf p,\sigma}) \le (W(K+1))^2 + 11 W(K+1)(t+\log_2(9W(K+1)) =: d_{W,K,t}.
\end{align}
Now, using the same derivations used for the deep CovNet class, we get
\begin{align*}
&\log\cover\big(\epsilon,\F^{\rm ds}_{R,L,\mathbf p,\sigma},\|\cdot\|_{\L_2(\Q \times \Q)}\big) = \log\cover\bigg(\frac{\epsilon}{|\Q|},\F^{\rm ds}_{R,L,\mathbf p,\sigma},\|\cdot\|_{\L_2(\nu)} \bigg) \\
&\le \alpha_1 \mathrm{fat}_{\F^{\rm ds}_{R,L,\mathbf p,\sigma}}\bigg(\alpha_2\frac{\epsilon}{|\Q|}\bigg)\,\log^2\Bigg(\frac{\mathrm{fat}_{\F^{\rm ds}_{R,L,\mathbf p,\sigma}}\Big(\frac{\alpha_2\epsilon^2}{|\Q|^2}\Big)|\Q|}{\epsilon}\Bigg) ~~\text{\citep[][Theorem~18.8]{anthony1999}} \\
&\le \alpha_1 \mathrm{Pdim}(\F^{\rm ds}_{R,L,\mathbf p,\sigma})\,\log^2\bigg(\frac{|\Q|\,\mathrm{Pdim}(\F^{\rm ds}_{R,L,\mathbf p,\sigma})}{\epsilon}\bigg) ~~\text{\citep[][Theorem~11.13(i)]{anthony1999}} \\
&\le \alpha_1\,d_{W,K,t}\,\log^2\bigg(\frac{|\Q|\,d_{W,K,t}}{\epsilon}\bigg). ~~~~\text{(by \eqref{eq:cover_deepshared_pdim_bound2})}
\end{align*}
We summarize this in the following lemma.

\begin{lemma}\label{lemma:cover_covnet_deepshared}
Consider the deepshared CovNet class of operators $\widetilde\F^{\rm ds}_{R,L,\mathbf p,\sigma}$, where $L$ is the number of hidden layers and $\mathbf{p}=(p_1,\ldots,p_L)$ are the number of nodes at each of the hidden layers, $\sigma$ is the sigmoidal activation function, and $R$ is the number of components of the model. Then, there exists $\alpha > 0$ such that
\[
\log\cover\big(\epsilon,\widetilde\F^{\rm ds}_{R,L,\mathbf p,\sigma},\vertj{\cdot}_2\big) \le \alpha\,d_{W,K,t}\,\log^2\bigg(\frac{|\Q|\,d_{W,K,t}}{\epsilon}\bigg),
\]
where $W = \sum_{l=1}^L p_l(p_{l-1}+1) + R(p_L+1) + R(R+1)/2 + 2$, $K=2\sum_{l=1}^L p_l + 2R$,  $t= 4\sum_{l=1}^L p_l(p_{l-1}+1) + 4R(p_L+1) + 2R^2 + 2$, and $d_{W,K,t} = \big(W(K+1)\big)^2 + 11 W(K+1)\,\big(t+\log_2(9W(K+1)\big)$.
\end{lemma}

\begin{remark}\label{remark:cover_deepshared_order}
By defining $p_{\rm max} = \max\{d,p_1,\ldots,p_L\}$, we get $W = \O(Lp_{\rm max}^2 + Rp_{\max} + R^2)$, $K = \O(Lp_{\rm max}+R)$ and $t = \O(Lp_{\rm max}^2 + Rp_{\rm max} + R^2)$. Thus, $d_{W,K,t} = \O(L^4p_{\rm max}^6 + R^6 + L^2R^2p_{\rm max}^4 + L^2R^4p_{\rm max}^2 + LR^3p_{\rm max}^3)$. In particular, when $p_1=\ldots=p_L = R$ and $R \ge d$, we get that $d_{W,K,t} = \O(L^4R^6)$, and hence
\[
\log\cover\big(\epsilon,\widetilde\F^{\rm ds}_{R,L,\mathbf p,\sigma},\vertj{\cdot}_2\big) = \O\bigg(L^4R^6\,\log^2\Big(\frac{LR}{\epsilon}\Big)\bigg).
\]
\end{remark}

\section{Proofs of the asymptotic results}\label{supp:asymptotics}

Here, we provide detailed proofs of the asymptotic properties of the CovNet estimators, as laid out in Section~\ref{sec:asymptotics}. We will first derive our results for a general class of models, and then obtain the corresponding results for the CovNet structures as special cases. To be precise, consider a general class of operators $\widetilde\F_N$ (depending on $N$ and possibly on other parameters, which we suppress for ease of exposition) satisfying $\vertj{\G}_2 \le \gamma_N$ for every $\G \in \widetilde\F_N$. Define the following estimator based on $\widetilde\F_N$:
\[
\Chat_{\widetilde\F_N} \in \argmin_{\G \in \widetilde\F_N} \vertj{\Chat_N - \G}_2^2,
\]
where $\Chat_N = N^{-1}\sum_{n=1}^N \X_n \otimes \X_n$ is the empirical covariance operator based on $\X_1,\ldots,\X_N$. We will prove two types of results for this estimator based on two \emph{bias-variance-type} decomposition.

\begin{enumerate}[(A)]
\item \textbf{Consistency}: We will show that $\vertj{\Chat_{\widetilde\F_N} - \C}_2^2 \le B_N + 2 V_N$, where $B_N := \inf_{\G \in \widetilde\F_N} \verti{\G-\C}_2^2$ is the \emph{bias} of the estimator and $V_N$ is the \emph{variance} term, defined as $V_N = \sup_{\G \in \widetilde\F_N} \Big|\vertj{\G - \Chat_N}_2^2 - \E\vertj{\G - \Chat_N}_2^2\Big|$.
We will derive conditions for convergence of $V_N$ to $0$ in terms of $N$, $\gamma_N$, $\beta_N$ and $\cover(\cdot,\widetilde\F_N,\vertj{\cdot}_2)$, the covering number of $\widetilde\F_N$ w.r.t.\ the $\vertj{\cdot}_2$ norm (see Appendix~\ref{supp:covering_number} for details). This will be used to establish weak (in probability) or strong (almost sure) convergence of the estimator.

\vspace{0.02in}

\item \textbf{Rate of convergence}: Here, we will show that $\vertj{\Chat_{\widetilde\F_N} - \C}_2^2 \le \widetilde B_N + \widetilde V_N$, where the bias term satisfies $\E\widetilde B_N = 2 \inf_{\G \in \widetilde\F_N} \vertj{\G-\C}_2^2$. The variance term in this case is $\widetilde V_N = \E\Big(\vertj{\Chat_{\widetilde\F_N} - \Chat_N}_2^2 - \vertj{\C - \Chat_N}_2^2\Big) - 2 \Big(\vertj{\Chat_{\widetilde\F_N} - \Chat_N}_2^2 - \vertj{\C - \Chat_N}_2^2\Big)$. We will derive an upper bound on $\Prob(\widetilde V_N > \epsilon)$ in terms of $N$, $\gamma_N$, $\beta_N$ and the covering number of $\widetilde\F_N$. This will be used to derive an upper bound for $\E\widetilde V_N$ and subsequently, the rate of convergence.
\end{enumerate}

\subsection{Consistency}\label{append:consistency_detailed}

We start with the consistency. In what follows, we denote by $\mathscr X_N$ the data at hand, i.e., $\mathscr X_N = \{\X_1,\ldots,\X_N\}$. Since $\Chat_N$ is an unbiased estimator of $\C$, for any operator $\G$,
\begin{equation}\label{eq:HS_equality_covnet1}
\E\vertj{\G - \Chat_N}_2^2 - \E\vertj{\C - \Chat_N}_2^2 = \verti{\G-\C}_2^2.
\end{equation}
Now, let $\widetilde\C_N$ be a random element distributed identically to $\Chat_N$ and independent of $\mathscr X_N$ (e.g., generated as the empirical covariance of i.i.d.\ observations distributed identically to $\X_1,\ldots,\X_N$, but independent of $\mathscr X_N$). Then, using \eqref{eq:HS_equality_covnet1} we obtain
\begin{align*}
\vertj{\Chat_{\widetilde\F_N} - \C}_2^2 &= \E\Big(\vertj{\Chat_{\widetilde\F_N} - \widetilde\C_N}_2^2 \,\vert\, \mathscr X_N\Big) - \E\Big(\vertj{\C - \widetilde\C_N}_2^2\Big) \\
&= \E\Big(\vertj{\Chat_{\widetilde\F_N} - \widetilde\C_N}_2^2 \,\vert\, \mathscr X_N\Big) - \inf_{\G \in \widetilde\F_N} \E\Big(\vertj{\G-\widetilde\C_N}_2^2\Big) + \inf_{\G \in \widetilde\F_N} \E\Big(\vertj{\G-\widetilde\C_N}_2^2\Big) - \E\Big(\vertj{\C-\widetilde\C_N}_2^2\Big).
\end{align*}
Now,
\begin{align*}
\inf_{\G \in \widetilde\F_N} \E\Big(\vertj{\G-\widetilde\C_N}_2^2\Big) - \E\Big(\vertj{\C - \widetilde\C_N}_2^2\Big) = \inf_{\G \in \widetilde\F_N} \Big\{\E\vertj{\G-\widetilde\C_N}_2^2 - \E\vertj{\C-\widetilde\C_N}_2^2\Big\} = \inf_{\G \in \widetilde\F_N} \vertj{\G-\C}_2^2 = B_N,
\end{align*}
where we have used \eqref{eq:HS_equality_covnet1}. For the other part, we write
\begin{align*}
&\E\Big(\vertj{\Chat_{\widetilde\F_N} - \widetilde\C_N}_2^2 \,\vert\, \mathscr X_N\Big) - \inf_{\G \in \widetilde\F_N} \E\Big(\vertj{\G-\widetilde\C_N}_2^2\Big) \\
&= \sup_{\G \in \widetilde\F_N} \bigg\{\E\Big(\vertj{\Chat_{\widetilde\F_N} - \widetilde\C_N}_2^2 \,\vert\, \mathscr X_N\Big) - \E\Big(\vertj{\G-\widetilde\C_N}_2^2\Big)\bigg\} \\
&\le \sup_{\G \in \widetilde\F_N} \bigg\{\E\Big(\vertj{\Chat_{\widetilde\F_N} - \widetilde\C_N}_2^2 \,\vert\, \mathscr X_N\Big) - \vertj{\Chat_{\widetilde\F_N} - \Chat_N}_2^2 + \vertj{\G-\Chat_N}_2^2 - \E\Big(\vertj{\G-\widetilde\C_N}_2^2\Big)\bigg\} \\
&\le 2 \sup_{\G \in \widetilde\F_N} \bigg|\vertj{\G - \Chat_N} - \E\vertj{\G - \Chat_N}_2^2\bigg| =: 2 V_N,
\end{align*}
where we have used that $\Chat_{\widetilde\F_N} \in \widetilde\F_N$, $\vertj{\Chat_{\widetilde\F_N} - \Chat_N}_2^2 \le \vertj{\G - \Chat_N}_2^2$ and $\E\vertj{\G - \widetilde\C_N}_2^2 = \E\vertj{\G - \Chat_N}_2^2$ for $\G \in \widetilde\F_N$. Using these, we get the \emph{bias-variance} type decomposition
\begin{align}\label{eq:bias_variance_type1}
\vertj{\Chat_{\widetilde\F_N} - \C}_2^2 \le B_N + 2V_N,
\end{align}
where $B_N = \inf_{\G \in \widetilde\F_N} \vertj{\G - \C}_2^2$ and $V_N = \sup_{\G \in \widetilde\F_N} \big|\vertj{\G - \Chat_N} - \E\vertj{\G - \Chat_N}_2^2\big|$. The bias term $B_N$ converges to $0$, which follows from the universal approximation property of the different CovNet models. To control the variance term $V_N$, we use the following lemma which gives conditions under which it converges to $0$.

\begin{lemma}\label{lemma:covnet1_type1_VN_convergence}
Let $\widetilde\F_N$ be a class of operators with $\verti{\G}_2 \le \gamma_N$ for every $\G \in \widetilde\F_N$. Let $\X_1,\ldots,\X_N \overset{\iid}{\sim} \X$, with $\Prob(\|\X\|^2 \le \beta_N)=1$ and $\E(\X) = 0$. Define $V_N = \sup_{\G \in \widetilde\F_N} \big|\vertj{\G-\Chat_N}_2^2 - \E\vertj{\G-\Chat_N}_2^2\big|$. If for every $u > 0$,
\[
\frac{(\beta_N+\gamma_N)^4}{N} \times \log\cover\bigg(\frac{u}{4(\beta_N+\gamma_N)},\widetilde\F_N,\verti{\cdot}_2\bigg) \to 0 \text{ as } N \to \infty,
\]
then $V_N \overset{P}{\to} 0$ as $N \to \infty$. If in addition $(\beta_N+\gamma_N)^4/N^{1-\delta} \to 0$ for some $\delta \in (0,1)$, then $V_N \overset{a.s.}{\to} 0$ as $N \to \infty$.
\end{lemma}

\begin{proof}
The proof is divided into two parts. First, we derive conditions under which $V_N - \E V_N$ converges to $0$, in probability or almost surely (Corollary~\ref{cor:covnet1_type1_variance_convergence}). Then, we derive conditions under which $\E V_N$ converges to $0$ (Corollary~\ref{cor:covnet1_type1_EV1_convergence}). We start with the following lemma.

\begin{lemma}\label{lemma:covnet1_type1_variance_bound}
Consider the setup of Lemma~\ref{lemma:covnet1_type1_VN_convergence}. Then, for all $t \ge 0$,
\[
\Prob\big(|V_N - \E V_N| > t\big) \le 2 \exp\bigg\{-\frac{Nt^2}{8\beta_N^2(\beta_N+\gamma_N)^2}\bigg\}.
\]
\end{lemma}
As a consequence, we get the following.
\begin{corollary}\label{cor:covnet1_type1_variance_convergence}
Consider the setup of Lemma~\ref{lemma:covnet1_type1_VN_convergence}.
\begin{enumerate}[(a)]
\item If $\beta_N^2(\beta_N+\gamma_N)^2/N \to 0$, then $V_N - \E V_N \overset{P}{\to} 0$ as $N \to \infty$. 
\item If $\beta_N^2(\beta_N+\gamma_N)^2/N^{1-\delta} \to 0$ for some $\delta \in (0,1)$, then $V_N - \E V_N \overset{a.s.}{\to} 0$ as $N \to \infty$.
\end{enumerate}
\end{corollary}
\begin{proof}
Part (a) about convergence in probability is easy. For part (b), note that for any $t \ge 0$,
\[
\sum_{n=1}^\infty \Prob\big(|V_N - \E V_N| \ge t \big) \le 2 \sum_{n=1}^\infty \exp\bigg\{-\frac{Nt^2}{8\beta_N^2(\beta_N+\gamma_N)^2}\bigg\} = 2 \sum_{n=1}^\infty \exp\bigg\{-N^\delta \frac{N^{1-\delta}t^2}{8\beta_N^2(\beta_N+\gamma_N)^2}\bigg\} < \infty.
\]
The result now follows from the Borel-Cantelli lemma.
\end{proof}

\begin{proof}[Proof of Lemma~\ref{lemma:covnet1_type1_variance_bound}]
For any $\G \in \widetilde\F_N$,
\begin{equation}\label{eq:V1_different_formulation}
\vertj{\G - \Chat_N}_2^2 = \frac{1}{N^2} \sum_{n=1}^N\sum_{m=1}^N \ip{\G - \X_n\otimes \X_n,\G - \X_m\otimes \X_m}_2. 
\end{equation}
For $\G \in \widetilde\F_N$, define $h_{\G} : \L_2(\Q) \times \L_2(\Q) \to \R$ as
\[
h_{\G}(x,y) = \ip{\G - x \otimes x,\G - y\otimes y}_2,~ x,y \in \L_2(\Q).
\]
Then, for any $x,y \in \L_2(\Q)$ with $\|x\|^2, \|y\|^2 \le \beta_N$, we get
\begin{align*}
\big|h_{\G}(x,x) - h_{\G}(y,y)\big| &= \Big|\vertj{\G - x \otimes x}_2^2 - \vertj{\G - y\otimes y}_2^2\Big| = \Big|\ip{y \otimes y - x \otimes x, 2\G - x \otimes x - y \otimes y}_2\Big| \\
&\le \big(\|x\|^2 + \|y\|^2\big)\,\big(2\vertj{\G}_2 + \|x\|^2 + \|y\|^2\Big) \le 4\beta_N(\beta_N+\gamma_N).
\end{align*}
Also, for any $x,y,z \in \L_2(\Q)$ with $\|x\|^2, \|y\|^2, \|z\|^2 \le \beta_N$
\begin{align*}
\big|h_{\G}(x,z) - h_{\G}(y,z)\big| &= \Big|\ip{\G - x \otimes x,\G - z \otimes z}_2 - \ip{\G - y\otimes y, \G - z \otimes z}_2\Big| \\
&= \Big|\ip{y \otimes y - x \otimes x, \G - z\otimes z}_2\Big| \le \big(\|x\|^2+\|y\|^2\big)\,\big(\vertj{\G}_2 + \|z\|^2\big) \le 2\beta_N(\beta_N+\gamma_N).
\end{align*}
Now, for $\G \in \widetilde\F_N$, define $f_{\G}: \L_2(\Q)^N \to \R$ as $f_{\G}(x_1,\ldots,x_N) = \sum_{n=1}^N\sum_{m=1}^N h_{\G}(x_n,x_m)$. The function $f_{\G}$ is symmetric in its arguments, and for any $x,x^\prime,x_2,\ldots,x_N \in \L_2(\Q)$ with squared norm bounded by $\beta_N$,
\[
\big|f_{\G}(x,x_2,\ldots,x_N) - f_{\G}(x^\prime,x_2,\ldots,x_N)\big| \le 4N\beta_N(\beta_N+\gamma_N).
\]
Then, by defining $g: \L_2(\Q)^N \to \R$ as
\[
g(x_1,\ldots,x_N) = \sup_{\G \in \widetilde\F_N} \Big|\sum_{n=1}^N\sum_{m=1}^N \big\{h_{\G}(x_n,x_m) - \E h_{\G}(\X_n,\X_m)\big\} \Big| = \sup_{\G \in \widetilde\F_N} \Big|f_{\G}(x_1,\ldots,x_N) - \E f_{\G}(\X_1,\ldots,\X_N)\Big|,
\]
we get that $g$ is a symmetric function in its arguments, and for any $\G \in \widetilde\F_N$,
\begin{align*}
&\Big|f_{\G}(x,x_2,\ldots,x_N) - \E f_{\G}(\X_1,\ldots,\X_N)\Big| - g(x^\prime,x_2,\ldots,x_N) \\
&\le \Big|f_{\G}(x,x_2,\ldots,x_N) - \E f_{\G}(\X_1,\ldots,\X_N)\Big| - \Big|f_{\G}(x^\prime,x_2,\ldots,x_N) - \E f_{\G}(\X_1,\ldots,\X_N)\Big| \\
&\le \Big|f_{\G}(x,x_2,\ldots,x_N) - f_{\G}(x^\prime,x_2,\ldots,x_N)\Big| \le 4N\beta_N(\beta_N+\gamma_N).
\end{align*}
This, upon taking supremum over $\G \in \widetilde\F_N$ and interchanging the roles of $x,x^\prime$, gives us 
\[
|g(x,x_2,\ldots,x_N) - g(x^\prime,x_2,\ldots,x_N)| \le 4N\beta_N(\beta_N+\gamma_N).
\]
Now, using the method of bounded difference \citep[][Corollary 2.21]{wainwright2019},
we get that, for all $t \ge 0$,
\[
\Prob\big(|g(\X_1,\ldots,\X_N) - \E g(\X_1,\ldots,\X_N)| > t\big) \le 2 \exp\bigg\{-\frac{t^2}{8N^3\beta_N^2(\beta_N+\gamma_N)^2}\bigg\}.
\]
The proof is complete upon noting that $V_N = g(\X_1,\ldots,\X_N)/N^2$ (cf. \eqref{eq:V1_different_formulation}).
\end{proof}

Next, we derive bounds on $\E V_N$.
\begin{lemma}\label{lemma:covnet1_type1_EV1_bound}
Consider the setup of Lemma~\ref{lemma:covnet1_type1_VN_convergence}. Then, for every $u > 0$,
\begin{align*}
\E V_N \le 4u &+ \frac{16(\beta_N+\gamma_N)^2}{N^3u} \times \cover\bigg(\frac{Nu}{4(\beta_N+\gamma_N)},\widetilde\F_N,\verti{\cdot}_2\bigg) \times \exp\bigg\{-\frac{N^3u^2}{4(\beta_N+\gamma_N)^2}\bigg\} \\
&+ \frac{128(\beta_N+\gamma_N)^4}{Nu} \times \cover\bigg(\frac{u}{4(\beta_N+\gamma_N)},\widetilde\F_N,\verti{\cdot}_2\bigg) \times \exp\bigg\{-\frac{Nu^2}{32(\beta_N+\gamma_N)^4}\bigg\}.
\end{align*}
\end{lemma}
This gives us the following condition for the convergence of $\E V_N$ to $0$.
\begin{corollary}\label{cor:covnet1_type1_EV1_convergence}
Consider the setup of Lemma~\ref{lemma:covnet1_type1_VN_convergence}. If for every $u>0$,
\[
\frac{(\beta_N+\gamma_N)^4}{N} \times \log\cover\bigg(\frac{u}{4(\beta_N+\gamma_N)},\widetilde\F_N,\verti{\cdot}_2\bigg) \to 0 \text{ as } N \to \infty,
\]
then $\E V_N \to 0$ as $N \to \infty$.
\end{corollary}
\begin{proof}
The condition ensures that for every $u > 0$, $\lim_{N \to \infty} \E V_N \le 4u$. The proof is completed upon taking limit as $u$ goes to $0$.
\end{proof}
\begin{proof}[Proof of Lemma~\ref{lemma:covnet1_type1_EV1_bound}]
Let $\X_1^\prime,\ldots,\X_N^\prime \overset{\iid}{\sim} \X$ be independent of $\X_1,\ldots,\X_N$. Now,
\begin{align}\label{eq:EV1_convergence_eq1}
\E V_N &= \E\Bigg\{\sup_{\G \in \widetilde\F_N} \bigg|\frac{1}{N^2}\sum_{n=1}^N\sum_{m=1}^N h_{\G}(\X_n,\X_m) - \frac{1}{N^2}\sum_{n=1}^N\sum_{m=1}^N \E h_{\G}(\X_n^\prime,\X_m^\prime)\bigg|\Bigg\} \nonumber\\
& =\E\Bigg[\sup_{\G \in \widetilde\F_N} \Bigg\{\E\bigg|\frac{1}{N^2}\sum_{n=1}^N\sum_{m=1}^N h_{\G}(\X_n,\X_m) - \frac{1}{N^2}\sum_{n=1}^N\sum_{m=1}^N h_{\G}(\X_n^\prime,\X_m^\prime)\bigg|\,\Big\vert\mathcal D_N\Bigg\}\Bigg] \nonumber\\
& \le \E\Bigg[\E\Bigg\{\sup_{\G \in \widetilde\F_N} \bigg|\frac{1}{N^2}\sum_{n=1}^N\sum_{m=1}^N h_{\G}(\X_n,\X_m) - \frac{1}{N^2}\sum_{n=1}^N\sum_{m=1}^N h_{\G}(\X_n^\prime,\X_m^\prime)\bigg|\,\Big\vert\mathcal D_N\Bigg\}\Bigg]~~(\text{by Fatou's lemma}) \nonumber\\
& = \E\Bigg\{\sup_{\G \in \widetilde\F_N} \bigg|\frac{1}{N^2}\sum_{n=1}^N\sum_{m=1}^N h_{\G}(\X_n,\X_m) - \frac{1}{N^2}\sum_{n=1}^N\sum_{m=1}^N h_{\G}(\X_n^\prime,\X_m^\prime)\bigg|\Bigg\} \nonumber\\
&\le \frac{1}{N^2} \E\Bigg[\sup_{\G \in \widetilde\F_N} \bigg|\sum_{n=1}^N \Big\{h_{\G}(\X_n,\X_n) - h_{\G}(\X_n^\prime,\X_n^\prime)\Big\}\bigg| + \sup_{\G \in \widetilde\F_N} \bigg|\sumsum_{1 \le n \ne m \le N} \Big\{h_{\G}(\X_n,\X_m) - h_{\G}(\X_n^\prime,\X_m^\prime)\Big\}\bigg|\Bigg] \nonumber\\
&=: E_1 + E_2.
\end{align}
By symmetry, we can show that
\begin{align}\label{eq:EV1_convergence_eq2}
E_1 \le \frac{2}{N^2} \E\bigg[\sup_{\G \in \widetilde\F_N} \bigg|\sum_{n=1}^N \zeta_n h_{\G}(\X_n,\X_n)\bigg|\bigg] \text{ and } E_2 \le \frac{2}{N^2} \E\bigg[\sup_{\G \in \widetilde\F_N} \bigg|\sumsum_{1 \le n \ne m \le N} \zeta_n\zeta_m h_{\G}(\X_n,\X_m)\bigg|\bigg],
\end{align}
where $\zeta_1,\ldots,\zeta_N$ are i.i.d.\ \emph{Rademacher} random variables, which take the values $\pm 1$ with equal probability, independent of all the other variables. Next, we derive upper bounds on
\[
\E\bigg[\sup_{\G \in \widetilde\F_N} \bigg|\sum_{n=1}^N \zeta_n h_{\G}(\X_n,\X_n)\bigg|\bigg] \text{ and } \frac{2}{N^2} \E\bigg[\sup_{\G \in \widetilde\F_N} \bigg|\sumsum_{1 \le n \ne m \le N} \zeta_n\zeta_m h_{\G}(\X_n,\X_m)\bigg|\bigg].
\]

For any $\G_1,\G_2 \in \widetilde\F_N$,
\begin{align*}
\bigg|\frac{1}{N}\sum_{n=1}^N \zeta_n \big\{h_{\G_1}(\X_n,\X_n) - h_{\G_2}(\X_n,\X_n)\big\}\bigg| &=\bigg|\frac{1}{N}\sum_{n=1}^N \zeta_n \ip{\G_1-\G_2,\G_1+\G_2-2\X_n\otimes \X_n}_2\bigg| \\
&\le \vertj{\G_1-\G_2}_2 \times \frac{1}{N}\sum_{n=1}^N \vertj{\G_1+\G_2 - 2\X_n \otimes \X_n}_2 \\
&\le 2(\beta_N+\gamma_N) \times \vertj{\G_1-\G_2}_2.
\end{align*}
Now, for $\delta > 0$, let $\widetilde\F_N(\delta)$ be an $\delta$-cover of the least possible size for $\widetilde\F_N$ w.r.t.\ the $\vertj{\cdot}_2$ norm. Then, using the previous equation, for all $t > 0$,
\begin{align*}
&\Prob\Bigg\{\sup_{\G \in \widetilde\F_N} \bigg|\frac{1}{N}\sum_{n=1}^N \zeta_n h_{\G}(\X_n,\X_n)\bigg| > t \Bigg\} \nonumber\\
&= \E\,\Prob\Bigg\{\sup_{\G \in \widetilde\F_N} \bigg|\frac{1}{N} \sum_{n=1}^N \zeta_n h_{\G}(\X_n,\X_n)\bigg| > t \,\Big\vert\, \mathscr X_N  \Bigg\} \nonumber\\
&\le \E\,\Prob\Bigg\{\sup_{\G \in \widetilde\F_N\big(\frac{t}{4(\beta_N+\gamma_N)}\big)} \bigg|\frac{1}{N} \sum_{n=1}^N \zeta_n h_{\G}(\X_n,\X_n)\bigg| > \frac{t}{2} \,\Big\vert\, \mathscr X_N  \Bigg\} \nonumber \\
&\le \Big|\widetilde\F_N\Big(\frac{t}{4(\beta_N+\gamma_N)}\Big)\Big| \times \E\Bigg[\sup_{\G \in \widetilde\F_N\big(\frac{t}{4(\beta_N+\gamma_N)}\big)}\Prob\bigg\{\bigg|\frac{1}{N} \sum_{n=1}^N \zeta_n h_{\G}(\X_n,\X_n)\bigg| > \frac{t}{2} \,\Big\vert\, \mathscr X_N \bigg\}\Bigg]. \nonumber
\end{align*}
Note that the cardinality of $\widetilde\F_N(\delta)$ is $\cover(\delta,\widetilde\F_N,\vertj{\cdot}_2)$. Also, $\sum_{n=1}^N h_{\G}^2(\X_n,\X_n) \le N(\beta_N+\gamma_N)^2$ for all $\G \in \widetilde\F_N(t/(4(\beta_N+\gamma_N)))$. So, using the Bernstein's inequality we get that for all $\G \in \widetilde\F_N(t/(4(\beta_N+\gamma_N)))$,
\begin{align*}
\Prob\bigg\{\bigg|\frac{1}{N} \sum_{n=1}^N \zeta_n h_{\G}(\X_n,\X_n)\bigg| > \frac{t}{2} \,\Big\vert\, \mathscr X_N \bigg\} \le 2\,\exp\bigg\{-\frac{Nt^2}{4(\beta_N+\gamma_N)^2}\bigg\},
\end{align*}
where the bound is free of $\mathscr X_N$. All these give us that for all $t > 0$,
\begin{align}\label{eq:tail_bound_E1}
&\Prob\Bigg\{\sup_{\G \in \widetilde\F_N} \bigg|\frac{1}{N}\sum_{n=1}^N \zeta_n h_{\G}(\X_n,\X_n)\bigg| > t \Bigg\} \le 2 \times \cover\bigg(\frac{t}{4(\beta_N+\gamma_N)},\widetilde\F_N,\verti{\cdot}_2\bigg) \times \exp\bigg\{-\frac{Nt^2}{4(\beta_N+\gamma_N)^2}\bigg\}.
\end{align}
Again, for any $\G_1,\G_2 \in \widetilde\F_N$,
\begin{align*}
&\bigg|\frac{1}{N^2} \sumsum_{1 \le m \ne n \le N} \zeta_n\zeta_m\big\{h_{\G_1}(\X_n,\X_m) - h_{\G_2}(\X_n,\X_m)\big\}\bigg| \\
&= \bigg|\frac{1}{N^2} \sumsum_{1 \le m \ne n \le N} \zeta_n\zeta_m\ip{\G_1-\G_2,\G_1+\G_2-\X_n\otimes\X_n-\X_m\otimes\X_m}_2\bigg| \\
&\le \vertj{\G_1-\G_2}_2 \times \frac{1}{N^2}\sumsum_{1 \le n \ne m \le N} \vertj{\G_1+\G_2-\X_n\otimes\X_n-\X_m\otimes\X_m}_2 \\
&\le 2\,(\beta_N+\gamma_N)\,\vertj{\G_1-\G_2}_2.
\end{align*}
Using this, and proceeding in a similar way as before, we can show that for all $t>0$,
\begin{align*}
&\Prob\Bigg\{\sup_{\G \in \widetilde\F_N} \bigg|\frac{1}{N^2}\sumsum_{1 \le m \ne n \le N} \zeta_n\zeta_m h_{\G}(\X_m,\X_n)\bigg| > t \Bigg\} \\
&\le \cover\bigg(\frac{t}{4(\beta_N+\gamma_N)},\widetilde\F_N,\vertj{\cdot}_2\bigg) \times \E\Bigg[\sup_{\G \in \widetilde\F_N\big(\frac{t}{4(\beta_N+\gamma_N)}\big)}\Prob\bigg\{\bigg|\frac{1}{N^2}\sumsum_{1 \le m \ne n \le N} \zeta_n\zeta_m h_{\G}(\X_m,\X_n)\bigg| > \frac{t}{2} \,\Big\vert\, \mathscr X_N\bigg\}\Bigg].
\end{align*}
Note that $|h_{\G}(\X_n,\X_n)| \le (\beta_N+\gamma_N)^2$ almost surely for all $\G \in \widetilde\F_N$. Thus, using the method of bounded difference \citep[][Corollary~2.21]{wainwright2019}, we can show that for all $\G \in \widetilde\F_N(t/(4(\beta_N+\gamma_N)))$,
\begin{align*}
\Prob\bigg\{\bigg|\frac{1}{N^2}\sumsum_{1 \le m \ne n \le N} \zeta_n\zeta_m h_{\G}(\X_m,\X_n)\bigg| > \frac{t}{2} \,\Big\vert\, \mathscr X_N\bigg\} \le 2\,\exp\bigg\{-\frac{Nt^2}{32(\beta_N+\gamma_N)^4}\bigg\},
\end{align*}
for all $t>0$. Thus, we finally get that for all $t>0$,
\begin{align}\label{eq:tail_bound_E2}
&\Prob\Bigg\{\sup_{\G \in \widetilde\F_N} \bigg|\frac{1}{N^2}\sumsum_{1 \le n \ne m \le N} \zeta_n\zeta_m h_{\G}(\X_n,\X_m)\bigg| > t\Bigg\} \nonumber\\
&\kern30ex\le 2 \times \cover\bigg(\frac{t}{4(\beta_N+\gamma_N)},\widetilde\F_N,\verti{\cdot}_2\bigg) \times \exp\bigg\{-\frac{Nt^2}{32(\beta_N+\gamma_N)^4}\bigg\}.
\end{align}
Now, for a non-negative random variable $Y$,
\[
\E(Y) = \int_{0}^{\infty} \Prob(Y>t) \diff t \le u + \int_{u}^\infty \Prob(Y>t) \diff t,
\]
holds for every $u > 0$. Using this with \eqref{eq:tail_bound_E1}, we get that
\begin{align*}
& \E\Bigg[\frac{1}{N^2} \sup_{\G \in \widetilde\F_N} \bigg|\sum_{n=1}^N \zeta_n h_{\G}(\X_n,\X_n)\bigg| > t \Bigg] \\
&\kern5ex \le u + 2 \int_{u}^\infty \cover\bigg(\frac{Nt}{4(\beta_N+\gamma_N)},\widetilde\F_N,\verti{\cdot}_2\bigg) \times \exp\bigg\{-\frac{N^3t^2}{4(\beta_N+\gamma_N)^4}\bigg\} \diff t \\
&\kern5ex \le u + 2 \times \cover\bigg(\frac{Nu}{4(\beta_N+\gamma_N)},\widetilde\F_N,\verti{\cdot}_2\bigg) \int_{u}^\infty \exp\bigg\{-\frac{N^3ut}{4(\beta_N+\gamma_N)^4}\bigg\} \diff t\\
&\kern5ex = u + 2 \times \frac{4(\beta_N+\gamma_N)^2}{N^3u} \times \cover\bigg(\frac{Nu}{4(\beta_N+\gamma_N)},\widetilde\F_N,\verti{\cdot}_2\bigg) \times \exp\bigg\{-\frac{N^3u^2}{4(\beta_N+\gamma_N)^4}\bigg\}.
\end{align*}
Similarly, from \eqref{eq:tail_bound_E2}, we get
\begin{align*}
&\E\Bigg[\frac{1}{N^2} \sup_{\G \in \widetilde\F_N} \bigg|\sumsum_{1 \le n \ne m \le N} \zeta_n\zeta_m h_{\G}(\X_n,\X_m)\bigg| > t\Bigg] \\
&\kern5ex \le u + 2 \times \frac{32(\beta_N+\gamma_N)^4}{Nu} \times \cover\bigg(\frac{u}{4(\beta_N+\gamma_N)},\widetilde\F_N,\verti{\cdot}_2\bigg) \times \exp\bigg\{-\frac{Nu^2}{32(\beta_N+\gamma_N)^4}\bigg\}.
\end{align*}
The proof is complete upon substituting these upper bounds in \eqref{eq:EV1_convergence_eq1} in addition to \eqref{eq:EV1_convergence_eq2}.
\end{proof}

The proof of Lemma~\ref{lemma:covnet1_type1_VN_convergence} now follows upon noting that the stipulated assumptions imply that the conditions of Corollaries~\ref{cor:covnet1_type1_variance_convergence} and \ref{cor:covnet1_type1_EV1_convergence} are satisfied.
\end{proof}

Next, we derive an upper bound on the Hilbert-Schmidt norm of an operator from the shallow, deep and deepshared CovNet classes. In particular, we get the following.
\begin{proposition}\label{prop:HS_norm_bound}
Let $\widetilde\F_N$ be any of the three CovNet classes of operators (i.e., shallow/deepshared/deep). Then, for any $\G \in \widetilde\F_N$, $\vertj{\G}_2 \le R\lambda_N|\Q|$.
\end{proposition}
\begin{proof}
We start with the class of shallow CovNet operators $\widetilde\F_{R,\lambda_N}^{\rm sh}$. Let $\G \in \widetilde\F_{R,\lambda_N}^{\rm sh}$ be an operator with kernel of the form
\[
g(\uvec,\vvec) = \sum_{r=1}^R \sum_{s=1}^R \lambda_{r,s}\,\sigma(\mathbf w_r^\top\uvec+b_r)\,\sigma(\mathbf w_s^\top\vvec+b_s),\quad \uvec,\vvec \in \Q,
\]
where $\Lambda = (\lambda_{r,s})$ satisfies $\mathrm 0 \preceq \Lambda \preceq \lambda_N \mathrm{I}_R$. Since $g$ is a non-negative definite kernel 
\[
\|g\|_{\L_{\infty}(\Q \times \Q)} := \sup_{\uvec,\vvec \in \Q} |g(\uvec,\vvec)| = \sup_{\uvec \in \Q} g(\uvec,\uvec).
\]
It is easy to see that for any $\uvec \in \Q$, 
\[
g(\uvec,\uvec) = \sum_{r=1}^{R} \sum_{s=1}^{R} \lambda_{r,s}\, \sigma(\mathbf w_r^\top \uvec + b_r)\,\sigma(\mathbf w_s^\top \uvec + b_s) = \sum_{r=1}^{R}\sum_{s=1}^{R} \lambda_{r,s}\,a_r\,a_s = {\bf a}^\top \Lambda {\bf a},
\]
where $a_r = \sigma(\mathbf w_r^\top \uvec + b_r), r=1,\ldots,R$, and ${\bf a} = (a_1,\ldots,a_R)^\top$. Since $0 \preceq \Lambda \preceq \lambda_N {\rm I}_R$, ${\bf a}^\top \Lambda {\bf a} \le \lambda_N {\bf a}^\top{\bf a} = \lambda_N \sum_{r=1}^{R} a_r^2$. Also, since $\sigma$ is a sigmoidal activation function, $0 \le a_r \le 1$ for $r=1,\ldots,R$. This gives us $\|g\|_{\L_{\infty}(\Q \times \Q)} \le R \lambda_N$, which in turn shows 
\begin{equation*}
\|g\|^2_{\L_2(\Q \times \Q)} = \iint_{\Q \times \Q} g^2(\uvec,\vvec)\diff\uvec\diff\vvec \le R^2 \lambda_N^2 |\Q|^2.
\end{equation*}
Thus, for every $\G \in \widetilde\F_{R,\lambda_N}^{\rm sh}$, $\verti{\G}_2 = \|g\|_{\L_2(\Q \times \Q)} \le R\lambda_N|\Q|$, proving the result for the shallow CovNet class. For the deep and the deepshared CovNet classes, consider a kernel of the form 
\[
g(\uvec,\vvec) = \sum_{r=1}^R \sum_{s=1}^R \lambda_{r,s}\,g_r(\uvec)\,g_s(\vvec),\quad\uvec,\vvec \in \Q,
\]
where $\|g_r\|_{\L_{\infty}(\Q)} \le M$ for all $r=1,\ldots,R$, and $\mathrm 0 \preceq \Lambda = (\lambda_{r,s}) \preceq \lambda_N\,\mathrm I_R$. Then, similarly to the previous derivations, it can be shown that $\|g\|_{\L_{\infty}(\Q \times \Q)} \le R\lambda_N M^2$ and $\|g\|_{\L_2(\Q \times \Q)} \le R\lambda_N M^2 |\Q|$. Since our activation function is sigmoidal, in particular since $|\sigma(t)| \le 1$ for all $t$, it follows that $M \le 1$. Thus, for the deep as well as the deepshared CovNet class, we get that $\sup_{\G \in \widetilde\F_N} \vertj{\G}_2 \le R\lambda_N |\Q|$.
\end{proof}

Now, we combine the pieces together. For the shallow CovNet estimator, using \eqref{eq:bias_variance_type1}, we get that
\[
\vertj{\Chat^{\rm sh}_{R,N} - \C}_2^2 \le B_N + 2V_N,
\]
where $B_N = \inf_{\G \in \widetilde\F_{R,\lambda_N}^{\rm sh}} \vertj{\G - \C}_2^2$ and $V_N = \sup_{\G \in \widetilde\F_{R,\lambda_N}^{\rm sh}} \big|\vertj{\G-\Chat_N}_2^2 - \E\vertj{\G-\Chat_N}_2^2\big|$. By Theorem~\ref{thm:universal_approximation_shallow} and Remark~\ref{remark:universal_approximation_bounded}, $B_N \to 0$ as $R,\lambda_N \to \infty$. From Lemma~\ref{lemma:covnet1_type1_VN_convergence}, we know that $V_N \overset{P}{\to} 0$ if, for all $u > 0$,
\[
\frac{(\beta_N+\gamma_N)^4}{N} \times \log\cover\bigg(\frac{u}{4(\beta_N+\gamma_N)},\widetilde\F_{R,\lambda_N}^{\rm sh},\verti{\cdot}_2\bigg) \to 0 \text{ as } N \to \infty,
\]
and $V_N \overset{a.s.}{\to} 0$ if additionally $(\beta_N+\gamma_N)^4/N^{1-\delta} \to 0$ for some $\delta \in (0,1)$, where $\gamma_N = \argmax_{\G \in \widetilde\F_{R,\lambda_N}^{\rm sh}}\vertj{\G}_2$. By Proposition~\ref{prop:HS_norm_bound}, $\gamma_N \le R\lambda_N|\Q|$. Also, Lemma~\ref{lemma:cover_covnet_shallow} shows that
\[
\log\cover\big(\epsilon,\widetilde\F^{\rm sh}_{R,\lambda_N},\vertj{\cdot}_2\big) = \O\bigg(dR^2\log\Big(\frac{R^2\lambda_N+\epsilon}{\epsilon}\Big)\bigg),~~\text{ as } R, \lambda_N \to \infty.
\]
Thus, when $R,\lambda_N \to \infty$ as $N \to \infty$ in such a way that $dR^2 \Delta_N^4 \log(\Delta_N)/N \to 0$, where $\Delta_N = \max\{\beta_N,\gamma_N\} = \max\{\beta_N,|\Q|R\lambda_N\} $, it is easy to see that the condition for convergence of $V_N$ to $0$ in probability is satisfied. This shows the convergence of $\vertj{\Chat_{R,N}^{\rm sh} - \C}_2^2$ to $0$ in probability as $N \to \infty$. For almost sure convergence, we additionally require $(\beta_N+\gamma_N)^4/N^{1-\delta} \to 0$ for some $\delta \in (0,1)$, or equivalently $\Delta_N^4/N^{1-\delta} \to 0$ as $N \to \infty$. This proves part (A) of Theorem~\ref{thm:covnet_consistency}.

The proof of parts (B) and (C) follow similarly upon noting that by Proposition~\ref{prop:HS_norm_bound}, $\vertj{\G}_2 \le R\lambda_N|\Q|$ for all $\G \in \widetilde\F_{R,L,\lambda_N}^{\rm d}$ or $\widetilde\F_{R,L,\lambda_N}^{\rm ds}$, and the following about the covering numbers:
\begin{align*}
\log\cover\big(\epsilon,\widetilde\F^{\rm d}_{R,L,\lambda_N},\vertj{\cdot}_2\big) &= \O\bigg(L^4R^8\,\log^2\Big(\frac{L^4R^8\lambda_N+\epsilon}{\epsilon}\Big)\bigg) \quad \text{(Lemma~\ref{lemma:cover_covnet_deep} and Remark~\ref{remark:cover_deep_order}), and }\\
\log\cover\big(\epsilon,\widetilde\F^{\rm ds}_{R,L,\lambda_N},\vertj{\cdot}_2\big) &= \O\bigg(L^4R^6\,\log^2\Big(\frac{LR}{\epsilon}\Big)\bigg) \quad \text{(Lemma~\ref{lemma:cover_covnet_deepshared} and Remark~\ref{remark:cover_deepshared_order})}.
\end{align*}

\subsection{Rates of convergence}\label{append:rates_of_convergence_detailed}

To derive the rate of convergence, we use a different bias-variance type decomposition. Recall that $\widetilde\C_N$ is a random element, distributed identically to $\Chat_N$ independently of the data $\mathscr X_N$. Using \eqref{eq:HS_equality_covnet1}, we get
\begin{align}\label{eq:bias_variance_type2}
\vertj{\Chat_{\widetilde\F_N} - \C}_2^2 &= \E\Big(\vertj{\Chat_{\widetilde\F_N} - \widetilde\C_N}_2^2 \,\big\vert\, \mathscr X_N\Big) - \E\Big(\vertj{\C - \widetilde\C_N}_2^2\Big) \nonumber\\
&= \E\Big(\vertj{\Chat_{\widetilde\F_N} - \widetilde\C_N}_2^2 \,\big\vert\, \mathscr X_N\Big) - \E\Big(\vertj{\C - \Chat_N}_2^2\Big) - 2\Big(\vertj{\Chat_{\widetilde\F_N} - \Chat_N}_2^2 - \vertj{\C - \Chat_N}_2^2\Big) \nonumber\\
&\kern46ex + 2\Big(\vertj{\Chat_{\widetilde\F_N} - \Chat_N}_2^2 - \vertj{\C - \Chat_N}_2^2\Big) \nonumber\\
&=: \widetilde V_{N} + \widetilde B_{N}.
\end{align}
For the second component, we can write 
\begin{align*}
\widetilde B_{N} = 2\bigg(\inf_{\G \in \widetilde\F_N} \vertj{\G - \Chat_N}_2^2 - \vertj{\C - \Chat_N}_2^2\bigg) = 2 \inf_{\G \in \widetilde\F_N} \Big(\vertj{\G - \Chat_N}_2^2 - \vertj{\C - \Chat_N}_2^2\Big),
\end{align*}
so that
\begin{align}\label{eq:rate_of_convergence_covnet1_bias}
\E\widetilde B_N &= 2 \E\bigg\{\inf_{\G \in \widetilde\F_N} \Big(\vertj{\G - \Chat_N}_2^2 - \vertj{\C - \Chat_N}_2^2\Big)\bigg\} \nonumber\\
&\le 2 \inf_{\G \in \widetilde\F_N} \bigg\{\E\Big(\vertj{\G - \Chat_N}_2^2 - \vertj{\C - \Chat_N}_2^2\Big)\bigg\} = 2 \inf_{\G \in \widetilde\F_N} \vertj{\G - \C}_2^2,
\end{align}
where in the last step we have used \eqref{eq:HS_equality_covnet1}. This is the bias of the estimator which, according to the universal approximation theorem, converges to $0$.

Next, we focus on the variance term $\widetilde V_N$. For $t>0$, we get
\begin{align}\label{eq:covnet1_type2_variance_tail_bound}
&\Prob\big(\widetilde V_N > t\big) = \Prob\bigg\{\E\Big(\vertj{\Chat_{\widetilde\F_N} - \widetilde\C_N}_2^2 \,\big\vert\, \mathscr X_N\Big) - \E\Big(\vertj{\C - \widetilde\C_N}_2^2\Big) - 2\Big(\vertj{\Chat_{\widetilde\F_N} - \Chat_N}_2^2 - \vertj{\C - \Chat_N}_2^2\Big) > t \bigg\} \nonumber\\
&\kern5ex \le \Prob\bigg\{\exists\,\G \in \widetilde\F_N : \E\Big(\vertj{\G - \widetilde\C_N}_2^2 - \vertj{\C - \widetilde\C_N}_2^2\Big) - 2\Big(\vertj{\G - \Chat_N}_2^2 - \vertj{\C - \Chat_N}_2^2\Big) > t \bigg\} \nonumber\\
&\kern5ex = \Prob\bigg\{\exists\,\G \in \widetilde \F_N : \E\Big(\vertj{\G - \Chat_N}_2^2 - \vertj{\C - \Chat_N}_2^2\Big) - 2\Big(\vertj{\G - \Chat_N}_2^2 - \vertj{\C - \Chat_N}_2^2\Big) > t \bigg\}~~(\text{since }\widetilde\C_N \eqdist \Chat_N) \nonumber\\
&\kern5ex = \Prob\bigg\{\exists\,\G \in \widetilde\F_N : \E\Big(\vertj{\G-\Chat_N}_2^2 - \vertj{\C-\Chat_N}_2^2\Big) - \Big(\vertj{\G-\Chat_N}_2^2 - \vertj{\C-\Chat_N}_2^2\Big) \nonumber\\
&\kern50ex > \frac{1}{2} \bigg(\frac{t}{2} + \frac{t}{2} + \E\Big(\vertj{\G-\Chat_N}_2^2 - \vertj{\C-\Chat_N}_2^2\Big)\bigg)\bigg\}.
\end{align}
To bound the last probability, we will use the following lemma, which is a modified version of Theorem~11.4 in \cite{gyorfi2002}.

\begin{lemma}\label{lemma:covnet1_type2_general}
Let $\X_1,\ldots,\X_N \overset{\mathrm{i.i.d.}}{\sim} \X \in \L_2(\Q)$ be such that $\Prob(\|\X\|^2 \le \beta) = 1$, $\E(\X) = 0$ and $\E(\X \otimes \X) = \C$. Define $\Chat_N = N^{-1}\sum_{n=1}^N \X_n \otimes \X_n$ to be the empirical covariance operator based on $\X_1,\ldots,\X_N$. Let $\widetilde\F$ be a class of operators satisfying $\verti{\G}_2 \le \gamma$ for all $\G \in \widetilde\F$. Then, for every $a,b > 0$ and $\delta \in \big(0,\frac{1}{2}\big]$,
\begin{align*}
&\Prob\Bigg[\exists\,\G \in \widetilde\F : \E\Big(\vertj{\G-\Chat_N}_2^2 - \vertj{\C-\Chat_N}_2^2\Big) - \Big(\vertj{\G-\Chat_N}_2^2 - \vertj{\C-\Chat_N}_2^2\Big) \\
&\kern50ex > \delta\bigg\{a + b + \E\Big(\vertj{\G-\Chat_N}_2^2 - \vertj{\C-\Chat_N}_2^2\big) \bigg\}\Bigg] \\
&\le 14 \times \cover\bigg(\frac{\epsilon b}{80(\beta+\gamma)^3},\widetilde\F,\verti{\cdot}\bigg) \times \exp\bigg\{-\frac{\epsilon^2(1-\epsilon) a N}{214(\beta+\gamma)^4\,(1+\epsilon)}\bigg\}.
\end{align*}
\end{lemma}

\begin{proof}
Note that for any $\G \in \widetilde\F_N$,
\begin{align*}
\vertj{\G - \Chat_N}_2^2 - \vertj{\C - \Chat_N}_2^2 &= \vertj{\G}_2^2 - \vertj{\C}_2^2 - \frac{2}{N} \sum_{n=1}^N \ip{\G-\C, \X_n\otimes \X_n}_2^2 \\
&=\frac{1}{N}\sum_{n=1}^N \big(\verti{\G}_2^2 - \verti{\C}_2^2 - 2 \ip{\G-\C,\X_n\otimes \X_n}_2\big) \\
&= \frac{1}{N}\sum_{n=1}^N \big\{\verti{\G - \X_n \otimes \X_n}_2^2 - \verti{\C - \X_n \otimes \X_n}_2^2\big\}.
\end{align*}
Thus, the probability in question equals
\begin{align*}
&\Prob\Bigg[\exists\,\G \in \widetilde\F_N : \E\Big(\vertj{\G - \X \otimes \X}_2^2 - \vertj{\C - \X \otimes \X}_2^2\Big) - \frac{1}{N}\sum_{n=1}^N \big\{\verti{\G - \X_n \otimes \X_n}_2^2 - \verti{\C - \X_n \otimes \X_n}_2^2\big\} \\
&\kern50ex > \delta\bigg\{a + b + \E\Big(\vertj{\G - \X \otimes \X}_2^2 - \vertj{\C - \X \otimes \X}_2^2\Big)\bigg\}\Bigg],
\end{align*}
and we will derive an upper bound for this. For $\G \in \widetilde\F_N$, define $h_{\G} : \L_2(\Q) \to \R$ by
\[
h_{\G}(x) = \verti{\G - x \otimes x}_2^2 - \verti{\C - x \otimes x}_2^2, ~~x \in \L_2(\Q).
\]
With this, the probability in question can be written as
\begin{align*}
\Prob\Big\{\exists\,\G \in \widetilde\F_N: \E h_{\G}(\X) - \frac{1}{N}\sum_{n=1}^N h_{\G}(\X_i) \ge \epsilon\big(a+b+\E h_{\G}(\X)\big)\Big\}.
\end{align*}
We first note a few facts about the function $h_{\G}$. Note that
\[
- \verti{\C - x \otimes x}_2^2 \le h_{\G}(x) \le \verti{\G - x \otimes x}_2^2, \text{ for all } x \in \L_2(\Q).
\]
Thus, using $\vertj{\G_1 + \G_2}_2^2 \le 2 (\vertj{\G_1}_2^2 + \vertj{\G_2}_2^2)$, with $\Prob\big(\|\X\|^2 \le \beta_N\big)=1$, $\vertj{\C}_2 \le \beta_N$ and $\vertj{\G}_2 \le \gamma_N$, we get
\begin{equation}\label{eq:lemma_covnet2_type2_general_fact1}
|h_{\G}(\X)| \le \max\{4\beta_N^2,2(\beta_N^2+\gamma_N^2)\} =: \eta_N \text{ almost surely for all } \G \in \widetilde\F_N.
\end{equation}
Again,
\begin{equation}\label{eq:lemma_covnet2_type2_general_fact2}
\E h_{\G}(\X) = \E\big(\verti{\G - \X \otimes \X}_2^2 - \verti{\C - \X \otimes \X}_2^2\big) = \verti{\G - \C}_2^2 \le 2(\beta_N^2+\gamma_N^2).
\end{equation}
Using the alternative form $h_{\G}(x) = \verti{\G}_2^2 - \verti{\C}_2^2 - 2 \ip{\G-\C,x \otimes x}_2$, we get
\begin{align}\label{eq:lemma_covnet2_type2_general_fact3}
\var\big(h_{\G}(\X)\big) = 4\var\big(\ip{\G-\C,\X \otimes \X}_2\big) &\le 4 \E\big(\ip{\G-\C,\X \otimes \X}_2^2\big) \le 4\,\verti{\G-\C}_2^2\,\E\big(\|\X\|^4\big) \le 4\beta_N^2\,\E h_{\G}(\X).
\end{align}
Finally, using \eqref{eq:lemma_covnet2_type2_general_fact2} and \eqref{eq:lemma_covnet2_type2_general_fact3}, we get
\begin{align}\label{eq:lemma_covnet2_type2_general_fact4}
\E h_{\G}^2(\X) = \var\big(h_{\G}(\X)\big) + \big(\E h_{\G}(\X)\big)^2 \le (6\beta_N^2+2\gamma_N^2) \E h_{\G}(\X). 
\end{align}

We are now ready to prove the lemma. In our proof, we borrow heavily from the proof of Lemma~11.4 in \cite{gyorfi2002}. In what follows, $\X_1^\prime,\ldots,\X_N^\prime \overset{\iid}{\sim} \X$ are distributed independently of $\X_1,\ldots,\X_N$. Also, we denote the data at hand by $\mathscr X_N$, i.e., $\mathscr X_N := \{\X_1,\ldots,\X_N\}$. Define $\G_N$ depending on $\mathscr X_N$ such that
\[
\E h_{\G_N}(\X) - \frac{1}{N} \sum_{n=1}^N h_{\G_N}(\X_n) \ge \epsilon\big(a+b+\E h_{\G_N}(\X)\big),
\]
if such a $\G_N$ exists in $\widetilde\F_N$; otherwise choose an arbitrary $\G \in \widetilde\F_N$. Then, using Chebyshev's inequality
\begin{align*}
&\Prob\Big\{\E\big(h_{\G_N}(\X) \,\vert\, \mathscr X_N\big) - \frac{1}{N} \sum_{n=1}^N h_{\G_N}(\X_n^\prime) > \frac{\epsilon}{2} (a + b) + \frac{\epsilon}{2} \E\big(h_{\G_N}(\X) \,\vert\, \mathscr X_N\big) \,\vert\, \mathscr X_N\Big\} \\
&\le \frac{\var\Big\{\frac{1}{N}\sum_{n=1}^N h_{\G_N}(\X_n^\prime) \,\vert\, \mathscr X_N\Big\}}{\Big\{\frac{\epsilon}{2} (a + b) + \frac{\epsilon}{2} \E\big(h_{\G_N}(\X) \,\vert\, \mathscr X_n\big)\Big\}^2} = \frac{\var\big(h_{\G_N}(\X) \,\vert\, \mathscr X_N\big)}{N\Big\{\frac{\epsilon}{2} (a + b) + \frac{\epsilon}{2} \E\big(h_{\G_N}(\X) \,\vert\, \mathscr X_N\big)\Big\}^2} \\
&\le \frac{4\beta_N^2\E\big(h_{\G_N}(\X) \,\vert\, \mathscr X_N\big)}{N\Big\{\frac{\epsilon}{2} (a + b) + \frac{\epsilon}{2} \E\big(h_{\G_N}(\X) \,\vert\, \mathscr X_N\big)\Big\}^2} = \frac{16\beta_N^2}{N\epsilon^2} \times \frac{\E\big(h_{\G_N}(\X) \,\vert\, \mathscr X_N\big)}{\Big\{(a + b) + \E\big(h_{\G_N}(\X) \,\vert\, \mathscr X_N\big)\Big\}^2}\\
&\le \frac{16\beta_N^2}{N\epsilon^2} \times \frac{1}{4(a+b)}~~\bigg(\text{using } \frac{x}{(\alpha+x)^2} \le \frac{1}{4\alpha} \text{ for all } x \ge 0 \text{ and } \alpha \ge 0\bigg) \\
&= \frac{4\beta_N^2}{N\epsilon^2(a+b)},
\end{align*}
where, we have used \eqref{eq:lemma_covnet2_type2_general_fact3} in the third line. Hence, for $N \ge 32 \beta_N^2/(\epsilon^2(a+b))$,
\begin{equation}\label{eq:lemma_covnet2_type2_general_eq1}
\Prob\Big\{\E\big(h_{\G_N}(\X) \,\vert\, \mathscr X_N\big) - \frac{1}{N} \sum_{n=1}^N h_{\G_N}(\X_n^\prime) \le \frac{\epsilon}{2} (a + b) + \frac{\epsilon}{2} \E\big(h_{\G_N}(\X) \,\vert\, \mathscr X_N\big) \,\vert\, \mathscr X_N\Big\} \ge \frac{7}{8}.
\end{equation}
Now,
\begin{align*}
&\Prob\bigg\{\exists\,\G \in \widetilde\F_N: \frac{1}{N}\sum_{n=1}^N h_{\G}(\X_n^\prime) - \frac{1}{N} \sum_{n=1}^N h_{\G}(\X_n) \ge \frac{\epsilon}{2} (a + b) + \frac{\epsilon}{2} \E h_{\G}(\X)\bigg\} \\
&\kern5ex \ge \Prob\bigg\{\frac{1}{N} \sum_{n=1}^N h_{\G_N}(\X_n^\prime) - \frac{1}{N} \sum_{n=1}^N h_{\G_N}(\X_n) \ge \frac{\epsilon}{2} (a + b) + \frac{\epsilon}{2} \E\big(h_{\G_N}(\X) \,\vert\, \mathscr X_N\big)\bigg\}\\
&\kern5ex\ge \Prob\bigg\{\E\big(h_{\G_N}(\X) \,\vert\, \mathscr X_N\big) - \frac{1}{N} \sum_{n=1}^N h_{\G_N}(\X_n) \ge \epsilon(a + b) + \epsilon \E\big(h_{\G_N}(\X) \,\vert\, \mathscr X_N\big), \\
&\kern12ex \E\big(h_{\G_N}(\X) \,\vert\, \mathscr X_N\big) - \frac{1}{N} \sum_{n=1}^N h_{\G_N}(\X_n^\prime) \ge \frac{\epsilon}{2} (a + b) + \frac{\epsilon}{2} \E\big(h_{\G_N}(\X) \,\vert\, \mathscr X_N\big) \bigg\}\\
&\kern5ex \ge \frac{7}{8} \times \Prob\bigg\{\E\big(h_{\G_N}(\X) \,\vert\, \mathscr X_N\big) - \frac{1}{N} \sum_{n=1}^N h_{\G_N}(\X_n) \ge \epsilon(a + b) + \epsilon \E\big(h_{\G_N}(\X) \,\vert\, \mathscr X_N\big)\bigg\} \\
&\kern60ex (\text{conditioning on } \mathscr X_N \text{ and using \eqref{eq:lemma_covnet2_type2_general_eq1}}) \\
&\kern5ex = \frac{7}{8} \times \Prob\bigg\{\exists\,\G \in \widetilde\F_N: \E h_{\G}(\X) - \frac{1}{N} \sum_{n=1}^N h_{\G}(\X_n) \ge \epsilon(a + b) + \epsilon \E h_{\G}(\X) \bigg\}.
\end{align*}
So, for $N \ge 32 \beta_N^2/(\epsilon^2(a+b))$,
\begin{align}\label{eq:lemma_covnet2_type2_general_eq2}
&\Prob\bigg\{\exists\,\G \in \widetilde\F_N: \E h_{\G}(\X) - \frac{1}{N} \sum_{n=1}^N h_{\G}(\X_n) \ge \epsilon(a + b) + \epsilon \E h_{\G}(\X) \bigg\} \nonumber\\
&\kern10ex \le \frac{8}{7} \times \Prob\bigg\{\exists\,\G \in \widetilde\F_N: \frac{1}{N} \sum_{n=1}^N h_{\G}(\X_n^\prime) - \frac{1}{N} \sum_{n=1}^N h_{\G}(\X_n) \ge \frac{\epsilon}{2} (a + b) + \frac{\epsilon}{2} \E h_{\G}(\X) \bigg\}.
\end{align}
Now, for events $E_1,E_2,E_3$, $\Prob(E_1) \le \Prob(E_1 \cap E_2 \cap E_3) + \Prob(E_2^c) + \Prob(E_3^c)$, 
using which we write
\begin{align}\label{eq:lemma_covnet2_type2_general_eq3}
&\Prob\bigg\{\exists\,\G \in \widetilde\F_N: \frac{1}{N}\sum_{n=1}^N h_{\G}(\X_n^\prime) - \frac{1}{N}\sum_{n=1}^N h_{\G}(\X_n) \ge \frac{\epsilon}{2} (a + b) + \frac{\epsilon}{2} \E h_{\G}(\X) \bigg\} \nonumber\\
&\kern5ex \le \Prob\bigg\{\exists\,\G \in \widetilde\F_N: \frac{1}{N}\sum_{n=1}^N h_{\G}(\X_n^\prime) - \frac{1}{N}\sum_{n=1}^N h_{\G}(\X_n) \ge \frac{\epsilon}{2} (a + b) + \frac{\epsilon}{2} \E h_{\G}(\X),\nonumber\\
&\kern12ex \frac{1}{N}\sum_{n=1}^N h_{\G}^2(\X_n) - \E h_{\G}^2(\X) \le \epsilon\Big(a+b+\frac{1}{N}\sum_{n=1}^N h_{\G}^2(\X_n) + \E h_{\G}^2(\X) \Big),\nonumber\\
&\kern12ex \frac{1}{N}\sum_{n=1}^N h_{\G}^2(\X_n^\prime) - \E h_{\G}^2(\X) \le \epsilon\Big(a+b+\frac{1}{N}\sum_{n=1}^N h_{\G}^2(\X_n^\prime) + \E h_{\G}^2(\X) \Big) \bigg\} \nonumber\\
&\kern5ex + 2 \Prob\bigg\{\exists\,\G \in \widetilde\F_N: \frac{\frac{1}{N}\sum_{n=1}^N h_{\G}^2(\X_n) - \E h_{\G}^2(\X)}{a+b+\frac{1}{N}\sum_{n=1}^N h_{\G}^2(\X_n) + \E h_{\G}^2(\X)} > \epsilon\bigg\}.
\end{align}
Using Theorem 11.6 of \cite{gyorfi2002}, and noting that $h_{\G}^2(\X) \le \eta_N^2$ almost surely (cf. \eqref{eq:lemma_covnet2_type2_general_fact1}), we get
\begin{align}\label{eq:lemma_covnet2_type2_general_part2}
&\Prob\bigg\{\exists\,\G \in \widetilde\F_N: \frac{\frac{1}{N}\sum_{n=1}^N h_{\G}^2(\X_n) - \E h_{\G}^2(\X)}{a+b+\frac{1}{N}\sum_{n=1}^N h_{\G}^2(\X_n) + \E h_{\G}^2(\X)} > \epsilon\bigg\} \nonumber\\
&\kern10ex \le 4 \times \E \cover_1\bigg(\frac{(a+b)\epsilon}{5}, \{h_{\G}^2: \G \in \widetilde \F_N\}, \mathscr X_N\bigg) \times \exp\bigg\{-\frac{3\epsilon^2(a+b) N}{40\eta_N^2}\bigg\},
\end{align}
where $\cover_1(\epsilon,F,\mathscr X_N)$ is the \emph{random $L_1$-covering number} of the class of functions $F$ on $\mathscr X_N$, see \citet[][Chapter~9]{gyorfi2002} for details. This provides an upper bound for the second part in \eqref{eq:lemma_covnet2_type2_general_eq3}. For the first part, the second inequality inside the probability implies
\[
\E h_{\G}^2(\X) \ge \frac{1-\epsilon}{1+\epsilon} \times \frac{1}{N}\sum_{n=1}^N h_{\G}^2(\X_n) - \frac{\epsilon(a+b)}{1+\epsilon}.
\]
Also, $\E h_{\G}^2(\X) \le (6\beta_N^2+2\gamma_N^2) \E h_{\G}(\X)$ (cf. \eqref{eq:lemma_covnet2_type2_general_fact4}). 
So, the first probability can be bounded by
\begin{align}\label{eq:lemma_covnet2_type2_general_eq4}
&\Prob\bigg[\exists\,\G \in \widetilde\F_N: \frac{1}{N}\sum_{n=1}^N h_{\G}(\X_n^\prime) - \frac{1}{N}\sum_{n=1}^N h_{\G}(\X_n) \ge \frac{\epsilon}{2}(a+b) \nonumber\\
&\kern20ex + \frac{\epsilon}{2} \times \frac{1}{4(3\beta_N^2+\gamma_N^2)}\bigg\{\frac{(1-\epsilon)}{(1+\epsilon)} \frac{1}{N}\sum_{n=1}^N h_{\G}^2(\X_n^\prime) + \frac{(1-\epsilon)}{(1+\epsilon)} \frac{1}{N}\sum_{n=1}^N h_{\G}^2(\X_n) - 2 \frac{\epsilon(a+b)}{(1+\epsilon)}\bigg\}\Bigg]\nonumber\\
&= \Prob\bigg[\exists\,\G \in \widetilde\F_N: \frac{1}{N}\sum_{n=1}^N \zeta_n \big(h_{\G}(\X_n^\prime) - h_{\G}(\X_n)\big) \ge \frac{\epsilon}{2}(a+b) - \frac{\epsilon^2 (a+b)}{4(3\beta_N^2+\gamma_N^2)\,(1+\epsilon)} \nonumber\\
&\kern50ex + \frac{\epsilon(1-\epsilon)}{8(3\beta_N^2+\gamma_N^2)\,(1+\epsilon)} \frac{1}{N}\sum_{n=1}^N \big(h_{\G}^2(\X_n^\prime) + h_{\G}^2(\X_n)\big)\bigg]\nonumber\\
&\le 2 \times \Prob\bigg[\exists\,\G \in \widetilde\F_N: \bigg|\frac{1}{N}\sum_{n=1}^N \zeta_n h_{\G}(\X_n)\bigg| \ge \frac{1}{2}\bigg(\frac{\epsilon}{2}(a+b) - \frac{\epsilon^2(a+b)}{4(3\beta_N^2+\gamma_N^2)\,(1+\epsilon)}\bigg) \nonumber\\
&\kern50ex +\frac{\epsilon(1-\epsilon)}{8(3\beta_N^2+\gamma_N^2)\,(1+\epsilon)} \frac{1}{N}\sum_{n=1}^N h_{\G}^2(\X_n)\bigg].
\end{align}
Here, $\zeta_n$'s are $\iid$ Rademacher random variables, which take values $\pm 1$ with equal probability and are independent of all other variables. Next, we derive an upper bound for the probability in \eqref{eq:lemma_covnet2_type2_general_eq4}.

Given $\mathscr X_N$ and $\delta>0$, let $\mathcal H_{\delta}(\mathscr X_N)$ be the smallest subset of $\{h_{\G} : \G \in \widetilde\F_N\}$ such that for every $\G \in \widetilde\F_N$ we can find $h \in \mathcal H_{\delta}(\mathscr X_N)$ satisfying
\[
\frac{1}{N} \sum_{n=1}^N \big|h_{\G}(\X_n) - h(\X_n)\big| < \delta.
\]
Then, for each $\G \in \widetilde\F_N$, we can find $h \in \mathcal H_{\delta}(\mathscr X_N)$ such that
\[
\bigg|\frac{1}{N}\sum_{n=1}^N \zeta_n h_{\G}(\X_n)\bigg| = \bigg|\frac{1}{N}\sum_{n=1}^N \zeta_n\big\{h(\X_n)+h_{\G}(\X_n) - h(\X_n)\big\}\bigg| \le \bigg|\frac{1}{N}\sum_{n=1}^N \zeta_n h(\X_n)\bigg| + \delta.
\]
Also,
\begin{align*}
\frac{1}{N}\sum_{n=1}^N h_{\G}^2(\X_n) &= \frac{1}{N}\sum_{n=1}^N h^2(\X_n) + \frac{1}{N}\sum_{n=1}^N \big\{h_{\G}^2(\X_n) - h^2(\X_n)\big\} \\
&= \frac{1}{N}\sum_{n=1}^N h^2(\X_n) + \frac{1}{N}\sum_{n=1}^N \big(h_{\G}(\X_n)+h(\X_n)\big)\big(h_{\G}(\X_n)-h(\X_n)\big) \\
&\ge \frac{1}{N}\sum_{n=1}^N h^2(\X_n) - 2\,\eta_N\,\frac{1}{N}\sum_{n=1}^N \big|h_{\G}(\X_n)-h(\X_n)\big|~~(\text{since } |h_{\G}(\X)|, |h(\X)| \le \eta_N) \\
&\ge \frac{1}{N}\sum_{n=1}^N h^2(\X_n) - 2\delta\eta_N.
\end{align*}
Thus,
\begin{align}\label{eq:lemma_covnet2_type2_general_eq5}
& \Prob\bigg[\exists\,\G \in \widetilde\F_N: \bigg|\frac{1}{N}\sum_{n=1}^N \zeta_n h_{\G}(\X_n)\bigg| \ge \frac{1}{2}\bigg(\frac{\epsilon}{2}(a+b) - \frac{\epsilon^2(a+b)}{4(3\beta_N^2+\gamma_N^2)\,(1+\epsilon)}\bigg) \nonumber\\
&\kern50ex + \frac{\epsilon(1-\epsilon)}{8(3\beta_N^2+\gamma_N^2)\,(1+\epsilon)} \frac{1}{N}\sum_{n=1}^N h_{\G}^2(\X_n) \,\Big\vert\, \mathscr X_N\bigg] \nonumber\\
&\le \Prob\bigg[\exists\,h \in \mathcal H_{\delta}(\mathscr X_N): \bigg|\frac{1}{N}\sum_{n=1}^N \zeta_n h(\X_n)\bigg| + \delta \ge \frac{1}{2}\bigg(\frac{\epsilon}{2}(a+b) - \frac{\epsilon^2(a+b)}{4(3\beta_N^2+\gamma_N^2)\,\eta_N(1+\epsilon)}\bigg) \nonumber\\
&\kern50ex + \frac{\epsilon(1-\epsilon)}{8(3\beta_N^2+\gamma_N^2)\,(1+\epsilon)} \bigg\{\frac{1}{N}\sum_{n=1}^N h^2(\X_n) - 2\delta\eta_N\bigg\} \,\Big\vert\, \mathscr X_N\bigg] \nonumber\\
&= \Prob\bigg[\exists\,h \in \mathcal H_{\delta}(\mathscr X_N): \bigg|\frac{1}{N}\sum_{n=1}^N \zeta_n h(\X_n)\bigg| \ge \frac{\epsilon}{4}(a+b) - \frac{\epsilon^2(a+b)}{8(3\beta_N^2+\gamma_N^2)\,(1+\epsilon)} - \delta - \delta \frac{\epsilon(1-\epsilon)\eta_N}{4(3\beta_N^2+\gamma_N^2)\,(1+\epsilon)} \nonumber\\
&\kern55ex + \frac{\epsilon(1-\epsilon)}{8(3\beta_N^2+\gamma_N^2)\,(1+\epsilon)} \frac{1}{N}\sum_{n=1}^N h^2(\X_n) \,\Big\vert\, \mathscr X_N\bigg] \nonumber\\
&\le \big|\mathcal H_\delta(\mathscr X_N)\big| \sup_{h \in \mathcal H_\delta(\mathscr X_N)} \Prob\Bigg[\bigg|\frac{1}{N}\sum_{n=1}^N \zeta_n h(\X_n)\bigg| \ge \frac{\epsilon}{4}(a+b) - \frac{\epsilon^2(a+b)}{8(3\beta_N^2+\gamma_N^2)\,(1+\epsilon)} - \delta - \delta \frac{\epsilon(1-\epsilon)\eta_N}{4(3\beta_N^2+\gamma_N^2)\,(1+\epsilon)} \nonumber\\
&\kern55ex + \frac{\epsilon(1-\epsilon)}{8(3\beta_N^2+\gamma_N^2)\,(1+\epsilon)} \frac{1}{N}\sum_{n=1}^N h^2(\X_n) \,\Big\vert\, \mathscr X_N\Bigg].
\end{align}
Set $\delta = \epsilon b/5$, so that
\[
\frac{\epsilon b}{4} - \frac{\epsilon^2 b}{8(3\beta_N^2+\gamma_N^2)\,(1+\epsilon)} - \delta - \delta \frac{\epsilon(1-\epsilon)\eta_N}{4(3\beta_N^2+\gamma_N^2)\,(1+\epsilon)} \ge 0.
\]
With this choice, the right side of \eqref{eq:lemma_covnet2_type2_general_eq5} can be bounded by
\begin{align*}
&\Big|\mathcal H_{\frac{\epsilon b}{5}}(\mathscr X_N)\Big| \sup_{h \in \mathcal H_{\frac{\epsilon b}{5}}(\mathscr X_N)} \Prob\Bigg[\bigg|\frac{1}{N}\sum_{n=1}^N \zeta_n h(\X_n)\bigg| \ge \frac{\epsilon}{4} a - \frac{\epsilon^2 a}{8(3\beta_N^2+\gamma_N^2)\,(1+\epsilon)} \\
&\kern50ex + \frac{\epsilon(1-\epsilon)}{8(3\beta_N^2+\gamma_N^2)\,(1+\epsilon)} \frac{1}{N}\sum_{n=1}^N h^2(\X_n) \,\Big\vert\, \mathscr X_N\Bigg].
\end{align*}
Let $V_n = \zeta_n h(\X_n)$ for $n=1,\ldots,N$ and $\sigma^2 = N^{-1}\sum_{n=1}^N \var\big(V_n \,\vert\, \mathscr X_N\big) = N^{-1}\sum_{n=1}^N h^2(\X_n)$. The probability on the last equation equals $\Prob\big[|N^{-1}\sum_{n=1}^N V_n| \ge C_1 + C_2 \sigma^2 \,\vert\, \mathscr X_N\big]$, with 
\[
C_1 = \frac{\epsilon}{4} a - \frac{\epsilon^2 a}{8(3\beta_N^2+\gamma_N^2)\,(1+\epsilon)} \text{ and } C_2 = \frac{\epsilon(1-\epsilon)}{8(3\beta_N^2+\gamma_N^2)\,(1+\epsilon)}.
\]
Given $\mathscr X_N$, $V_1,\ldots,V_N$ are independent random variables with $|V_n| = |h_{\G}(\X_n)| \le \eta_N$ and $\E(V_n \,\vert\, \mathscr X_N) = 0$ for $n=1,\ldots,N$. Also, both $C_1$ and $C_2$ are non-negative. So, using Bernstein's inequality \citep[][Lemma~A.2]{gyorfi2002}, we get
\begin{align*}
\Prob\Bigg[\bigg|\frac{1}{N}\sum_{n=1}^n V_n\bigg| \ge C_1 + C_2 \sigma^2 \,\Big\vert\, \mathscr X_N\Bigg] \le 2 \times \exp\bigg\{-18N\frac{C_1 C_2}{(2C_2\eta_N+3)^2}\bigg\},
\end{align*}
see \citet[][pages 217--218]{gyorfi2002}. Plugging in the expressions for $C_1$ and $C_2$ and noting that
\[
C_1 = \frac{\epsilon}{4} a - \frac{\epsilon^2 a}{8(3\beta_N^2+\gamma_N^2)\,(1+\epsilon)} \ge \frac{\epsilon}{4}a - \frac{\epsilon}{32}a = \frac{7\epsilon a}{32},
\]
we get
\[
18 N \frac{C_1 C_2}{(2C_2\eta_N+3)^2} \ge \frac{3\epsilon^2(1-\epsilon)aN}{10(3\beta_N^2+\gamma_N^2)\,(1+\epsilon)}.
\]
All these finally give us
\begin{align*}
&\Prob\Bigg[\bigg|\frac{1}{N}\sum_{n=1}^N \zeta_n h(\X_n)\bigg| \ge \frac{\epsilon}{4} a - \frac{\epsilon^2 a}{8(3\beta_N^2+\gamma_N^2)\,(1+\epsilon)} + \frac{\epsilon(1-\epsilon)}{8(3\beta_N^2+\gamma_N^2)\,(1+\epsilon)} \frac{1}{N}\sum_{n=1}^N h^2(\X_n) \,\Big\vert\, \mathscr X_N\Bigg] \\
&\kern65ex \le 2 \times \exp\bigg\{- \frac{3\epsilon^2(1-\epsilon)aN}{10(3\beta_N^2+\gamma_N^2)\,(1+\epsilon)}\bigg\}.
\end{align*}
This upper bound does not depend on $\mathscr X_N$. So, using $\Prob(A) = \E\,\Prob(A \,\vert\, B)$, we get
\begin{align}\label{eq:lemma_covnet2_type2_general_eq6}
& \Prob\Bigg[\exists\,\G \in \widetilde\F_N: \bigg|\frac{1}{N}\sum_{n=1}^N \zeta_n h_{\G}(\X_n)\bigg| \ge \frac{1}{2}\bigg(\frac{\epsilon}{2}(a+b) - \frac{\epsilon^2(a+b)}{4(3\beta_N^2+\gamma_N^2)\,(1+\epsilon)}\bigg) \nonumber\\
&\kern55ex + \frac{\epsilon(1-\epsilon)}{8(3\beta_N^2+\gamma_N^2)\,(1+\epsilon)} \frac{1}{N}\sum_{n=1}^N h_{\G}^2(\X_n)\Bigg] \nonumber \\
&\le 2 \times \E\Big|\mathcal H_{\frac{\epsilon b}{5}}(\mathscr X_N)\Big| \times \exp\bigg\{- \frac{3\epsilon^2(1-\epsilon)aN}{10(3\beta_N^2+\gamma_N^2)\,(1+\epsilon)}\bigg\}.
\end{align}

Next, we derive upper bounds on $\cover_1(\epsilon,\{h_{\G}^2 : \G \in \widetilde\F_N\},\mathscr X_N)$ and $\big|\mathcal H_{\delta}(\mathscr X_N)\big|$. For any $\G_1,\G_2 \in \widetilde\F_N$,
\begin{align*}
\frac{1}{N}\sum_{n=1}^N \big|h_{\G_1}(\X_n) - h_{\G_2}(\X_n)\big| &= \frac{1}{N}\sum_{n=1}^N \big|\verti{\G_1}_{2}^2 - \verti{\G_2}_{2}^2 - 2\ip{\G_1-\G_2,\X_n \otimes \X_n}_2\big| \\
&= \frac{1}{N}\sum_{n=1}^N \big|\ip{\G_1-\G_2,\G_1+\G_2-2\X_n \otimes \X_n}_2\big| \\
&\le \verti{\G_1 - \G_2}_{2} \times \frac{1}{N}\sum_{n=1}^N \verti{\G_1 + \G_2 - 2\X_n \otimes \X_n}_2 \\
&\le 2(\beta_N+\gamma_N) \times \verti{\G_1-\G_2}_2. 
\end{align*}
This shows that a $\delta$-cover for $\widetilde\F_N$ w.r.t.\ the $\verti{\cdot}_2$ norm is equivalent to a random $L_1$-cover for $\{h_{\G} : \G \in \widetilde\F_N\}$ on $\mathscr X_N$ of size $2(\beta_N+\gamma_N)\delta$. Thus, 
\begin{equation}\label{eq:lemma_covnet2_type2_general_cover1}
\big|\mathcal H_{\delta}(\mathscr X_N)\big| \le \cover\bigg(\frac{\delta}{2(\beta_N+\gamma_N)},\widetilde\F, \verti{\cdot}_2\bigg).
\end{equation}
Note that this upper bound does not depend on $\mathscr X_N$. In fact, it is not a random quantity. Again,
\begin{align*}
\frac{1}{N}\sum_{n=1}^N \big|h_{\G_1}^2(\X_n) - h_{\G_2}^2(\X_n)\big| &= \frac{1}{N}\sum_{n=1}^N \big|h_{\G_1}(\X_n)-h_{\G_2}(\X_n)\big|\,\big|h_{\G_1}(\X_n)+h_{\G_2}(\X_n)\big| \\
&\le 2\eta_N \times \frac{1}{N}\sum_{n=1}^N \big|h_{\G_1}(\X_n)-h_{\G_2}(\X_n)\big| \\
&\le 4\eta_N(\beta_N+\gamma_N) \times \verti{\G_1 - \G_2}_2.
\end{align*}
This gives us
\begin{equation}\label{eq:lemma_covnet2_type2_general_cover2}
\cover_1\big(\delta,\{h_{\G}^2: \G \in \widetilde\F_N\},\mathscr X_N\big) \le \cover\bigg(\frac{\delta}{4\eta_N(\beta_N+\gamma_N)},\widetilde\F_N,\verti{\cdot}_2\bigg).
\end{equation}
This, again, is a non-random upper bound. We now assemble all the pieces together. By \eqref{eq:lemma_covnet2_type2_general_eq2} and \eqref{eq:lemma_covnet2_type2_general_eq3}, for every $N \ge 32\beta_N^2/(\epsilon^2(a+b))$,
\begin{align*}
\Prob\bigg[\exists\,\G \in \widetilde\F_N : \E h_{\G}(\X) - \frac{1}{N}\sum_{n=1}^N h_{\G}(\X) \ge \epsilon\big(a+b+\E h_{\G}(\X)\big)\bigg] \le \frac{8}{7} \times (P_1 + 2 P_2),
\end{align*}
where, by \eqref{eq:lemma_covnet2_type2_general_eq4} and \eqref{eq:lemma_covnet2_type2_general_eq6},
\begin{align*}
P_1 \le 4 \times \E\Big|\mathcal H_{\frac{\epsilon b}{5}}(\mathscr X_N)\Big| \times \exp\bigg\{- \frac{3\epsilon^2(1-\epsilon)aN}{10(3\beta_N^2+\gamma_N^2)\,(1+\epsilon)}\bigg\},
\end{align*}
and by \eqref{eq:lemma_covnet2_type2_general_part2},
\begin{align*}
P_2 \le 4 \times \E \cover_1\bigg(\frac{(a+b)\epsilon}{5}, \{h_{\G}^2: \G \in \widetilde\F_N\}, \mathscr X_N\bigg) \times \exp\bigg\{-\frac{3\epsilon^2(a+b) N}{40\eta_N^2}\bigg\}.
\end{align*}
Using these, with the bounds on the covering numbers \eqref{eq:lemma_covnet2_type2_general_cover1}, \eqref{eq:lemma_covnet2_type2_general_cover2}, we get
\begin{align*}
&\Prob\bigg[\exists\,\G \in \widetilde\F_N : \E h_{\G}(\X) - \frac{1}{N}\sum_{n=1}^N h_{\G}(\X) \ge \epsilon\big(a+b+\E h_{\G}(\X)\big)\bigg] \\
&\kern5ex \le \frac{8}{7} \times 4 \times \E\Big|\mathcal H_{\frac{\epsilon b}{5}}(\mathscr X_N)\Big| \times \exp\bigg\{-\frac{3\epsilon^2(1-\epsilon) a N}{10(3\beta_N^2+\gamma_N^2)\,(1+\epsilon)}\bigg\} \\
&\kern15ex + \frac{16}{7} \times 4 \times \E \cover_1\bigg(\frac{(a+b)\epsilon}{5}, \{h_{\G}^2: \G \in \widetilde\F_N\}, \mathscr X_N\bigg) \times \exp\bigg\{-\frac{3\epsilon^2(a+b) N}{40\eta_N^2}\bigg\} \\
&\kern5ex \le \frac{32}{7} \times \cover\bigg(\frac{\epsilon b}{10(\beta_N+\gamma_N)},\widetilde\F_N,\verti{\cdot}_2\bigg) \times \exp\bigg\{-\frac{3\epsilon^2(1-\epsilon) a N}{10(3\beta_N^2+\gamma_N^2)\,(1+\epsilon)}\bigg\} \\
&\kern15ex + \frac{64}{7} \times \cover\bigg(\frac{(a+b)\epsilon}{20\eta_N(\beta_N+\gamma_N)},\widetilde\F_N,\verti{\cdot}_2\bigg) \times \exp\bigg\{-\frac{3\epsilon^2(a+b) N}{40\eta_N^2}\bigg\} \\
&\kern5ex \le 14 \times \cover\bigg(\frac{\epsilon b}{20\eta_N(\beta_N+\gamma_N)},\widetilde\F_N,\verti{\cdot}\bigg) \times \exp\bigg\{-\frac{3\epsilon^2(1-\epsilon) a N}{40\eta_N^2\,(1+\epsilon)}\bigg\} \\
&\kern5ex \le 14 \times \cover\bigg(\frac{\epsilon b}{80(\beta_N+\gamma_N)^3},\widetilde\F_N,\verti{\cdot}\bigg) \times \exp\bigg\{-\frac{\epsilon^2(1-\epsilon) a N}{214(\beta_N+\gamma_N)^4\,(1+\epsilon)}\bigg\},
\end{align*}
where we have used that $\eta_N := \max\{4\beta_N^2, 2(\beta_N^2+\gamma_N^2)\} \le 4(\beta_N+\gamma_N)^2$. This proves the result for $N \ge 32\beta_N^2/(\epsilon^2(a+b))$. For $N < 32\beta_N^2/(\epsilon^2(a+b))$,
\begin{align*}
\exp\bigg\{-\frac{\epsilon^2(1-\epsilon) a N}{214(\beta_N+\gamma_N)^4\,(1+\epsilon)}\bigg\} \ge \frac{1}{14},
\end{align*}
so the inequality holds trivially. This completes the proof.
\end{proof}

\begin{remark}
When $\beta_N=\gamma_N$, we can mimic the same proof to get the following upper bound with slightly better constants:
\[
14 \times \cover\bigg(\frac{\epsilon b}{160\beta_N^3},\widetilde\F_N,\verti{\cdot}\bigg) \times \exp\bigg\{-\frac{\epsilon^2(1-\epsilon) a N}{214\beta_N^4(1+\epsilon)}\bigg\}.
\]
\end{remark}

\medskip

We use \eqref{eq:covnet1_type2_variance_tail_bound} and Lemma~\ref{lemma:covnet1_type2_general} with $\epsilon = 1/2, a = b = t/2$, to get
\begin{align*}
\Prob\big(\widetilde V_N > t\big) \le 14 \times \cover\bigg(\frac{t}{320(\beta_N+\gamma_N)^3},\widetilde\F_N,\vertj{\cdot}_2\bigg) \times \exp\bigg\{-\frac{Nt}{5136(\beta_N+\gamma_N)^4}\bigg\}.
\end{align*}
Now, for any non-negative random variable $Y$, 
\begin{align*}
\E(Y) = \int_0^\infty \Prob(Y > t) \diff t \le u + \int_u^\infty \Prob(Y > t) \diff t,
\end{align*}
for every $u > 0$. Using this, we get that for all $u>0$,
\begin{align*}
&\E \widetilde V_N \le u + \int_u^\infty \Prob\big(\widetilde V_N > t\big) \diff t \nonumber\\
&\kern9ex \le u + 14 \times \int_u^\infty \cover\bigg(\frac{t}{320(\beta_N+\gamma_N)^3},\widetilde \F,\verti{\cdot}_2\bigg) \times \exp\bigg\{-\frac{Nt}{5136(\beta_N+\gamma_N)^4}\bigg\} \nonumber\\
&\kern9ex \le u + 14 \times \cover\bigg(\frac{u}{320(\beta_N+\gamma_N)^3},\widetilde \F,\verti{\cdot}_2\bigg) \times \int_u^\infty \exp\bigg\{-\frac{Nt}{5136(\beta_N+\gamma_N)^4}\bigg\} \nonumber\\
&\le u + \frac{14 \times 5136(\beta_N+\gamma_N)^4}{N} \times \cover\bigg(\frac{u}{320(\beta_N+\gamma_N)^3},\widetilde \F,\verti{\cdot}_2\bigg) \times \exp\bigg\{-\frac{Nu}{5136(\beta_N+\gamma_N)^4}\bigg\}.
\end{align*}
To obtain the rate of convergence of $\E\widetilde V_N$, we (approximately) minimize the above quantity w.r.t.\ $u$. In particular, by choosing
\[
u = \frac{5136(\beta_N+\gamma_N)^4}{N} \Bigg\{\log(14) + \log\cover\bigg(\frac{5136(\beta_N+\gamma_N)}{320N},\widetilde\F_N,\vertj{\cdot}_2\bigg)\Bigg\},
\]
we get that
\begin{align}\label{eq:rate_of_convergence_covnet1_variance}
\E\widetilde V_N &\le \frac{5146(\beta_N+\gamma_N)^4}{N} \Bigg\{1+\log(14)+\log\cover\bigg(\frac{5136(\beta_N+\gamma_N)}{320N},\widetilde\F_N,\vertj{\cdot}_2\bigg)\Bigg\} \nonumber \\
&= \O\bigg(\frac{(\beta_N+\gamma_N)^4}{N} \times \log\cover\Big(\frac{\beta_N+\gamma_N}{N},\widetilde\F_N,\vertj{\cdot}_2\Big)\bigg).
\end{align}

Finally, combining \eqref{eq:rate_of_convergence_covnet1_bias} and \eqref{eq:rate_of_convergence_covnet1_variance} with \eqref{eq:bias_variance_type2}, we get the following lemma.
\begin{lemma}\label{lemma:covnet_rate_of_convergence}
Let $\widetilde\F_N$ be a class of operators with $\verti{\G}_2 \le \gamma_N$ for every $\G \in \widetilde\F_N$. Let $\X_1,\ldots,\X_N \overset{\iid}{\sim} \X$, with $\Prob(\|\X\|^2 \le \beta_N)=1$, $\E(\X)=0$ and $\var(\X)=\C$. Define $\Chat_{\widetilde\F_N} = \inf_{\G \in \widetilde\F_N} \vertj{\Chat_N - \G}_2^2$, where $\Chat_N = N^{-1} \sum_{n=1}^N \X_n \otimes \X_n$ is the empirical covariance operator. Then,
\begin{align*}
\E\Big(\vertj{\Chat_{\widetilde\F_N}-\C}_2^2\Big) &\le 2 \inf_{\G \in \widetilde\F_N} \vertj{\G - \C}_2^2 + \O\bigg(\frac{\Delta_N^4}{N} \times \log\cover\bigg(\frac{\Delta_N}{N},\widetilde\F_N,\vertj{\cdot}_2\bigg)\bigg),
\end{align*}
where $\Delta_N = \max\{\beta_N,\gamma_N\}$.
\end{lemma}

As in the case of consistency, we use this lemma in conjunction with the fact that $\vertj{\G}_2 \le |\Q|R\lambda_N$ for an CovNet operator $\G$ from any of the three CovNet classes (shallow, deep or deepshared) and with the bounds on the covering numbers of these classes to get the rates given in Theorem~\ref{thm:covnet_rate_of_convergence}.

\subsection{Consistency without boundedness}\label{supp:consistency_unbounded}

We now extend the consistency results by removing the boundedness condition on $\X$. Recall that our estimators are re-defined in this case. As before, we prove a general result and then show the particular cases of shallow, deep and deepshared CovNet models. To this effect, consider a kernel of the form
\begin{equation}\label{eq:covnet_kernel_general}
g(\uvec,\vvec) = \sum_{r=1}^R \sum_{s=1}^R \lambda_{r,s}\,g_r(\uvec)\,g_s(\vvec),~~\uvec,\vvec \in \Q,
\end{equation}
where $\Lambda:=(\lambda_{r,s})$ is positive semi-definite. We denote the corresponding integral operator by $\G$. For $\lambda_N > 0$, define $\Pj_{\lambda_N}\G$ to be the operator obtained by thresholding the eigenvalues of $\Lambda$ to $\lambda_N$. That is, if $\Lambda = \sum_{i=1}^R \eta_i\,\mathbf e_i\,\mathbf e_i^\top$ is the eigendecomposition of $\Lambda$, then we define $\Lambda_{\lambda_N} = \sum_{i=1}^R \max\{\eta_i,\lambda_N\}\,\mathbf e_i\,\mathbf e_i^\top$ to be the $\lambda_N$-thresholded version of $\Lambda$. We define $\Pj_{\lambda_N}\G$ as the integral operator with kernel $g_{\lambda_N}(\uvec,\vvec) = \sum_{r=1}^R\sum_{s=1}^R \widetilde\lambda_{r,s}\,g_r(\uvec)\,g_s(\vvec)$, where $\widetilde\lambda_{r,s}$ is the $(r,s)$-th element of the thresholded matrix $\Lambda_{\lambda_N}$. By construction, $\mathrm 0 \preceq \Lambda_{\lambda_N} \preceq \lambda_N\,\mathrm I_R$. Let $\widetilde\F_R$ be a class of operators with kernels of the form \eqref{eq:covnet_kernel_general} and we denote the corresponding class of restricted operators by $\widetilde\F_{R,\lambda_N}$. Define
\[
\Chat_{R,N} = \inf_{\G \in \widetilde\F_R} \vertj{\Chat_N - \G}_2^2,
\]
to be the estimator without any restriction on the underlying class. Now, for a constant $\lambda_N > 0$, our modified estimator is defined as
\begin{equation}\label{eq:covnet_estimator_modified_general}
\widetilde\C_{R,N} = \Pj_{\lambda_N}\Chat_{R,N}.
\end{equation}
The following theorem illustrates the conditions for consistency of the modified estimator. The result is similar to Theorem~10.2 of \cite{gyorfi2002}.

\begin{theorem}\label{thm:consistency_unbounded_general}
Let $\X_1,\ldots,\X_N \overset{\iid}{\sim} \X$, where $\E(\|\X\|^4) < \infty$, $\E(\X) = 0$ and $\var(\X) = \C$. Suppose that $R_N,\lambda_N \to \infty$ as $N \to \infty$. Also, suppose that the underlying class of operators $\widetilde\F_{R,\lambda_N}$ is an universal approximator, i.e., $\inf_{\G \in \widetilde\F_{R,\lambda_N}} \verti{\G - \C}_2^2 \to 0$ as $N \to \infty$. If for every $u > 0$,
\[
\frac{R_N^4 \lambda_N^4}{N} \times \log\cover\bigg(\frac{u}{R_N\lambda_N},\widetilde\F_{R,\lambda_N},\verti{\cdot}_2\bigg) \to 0 \text{ as } N \to \infty,
\]
then $\vertj{\widetilde\C_{R,N} - \C}_2^2$ converges in probability to $0$ as $N \to \infty$. If in addition $(R_N\lambda_N)^4/N^{1-\delta} \to 0$ for some $\delta \in (0,1)$, then the previous convergence holds almost surely.
\end{theorem}
\begin{proof}
Note that $\vertj{\widetilde\C_{R,N} - \C}_2^2 = \E\big(\vertj{\widetilde\C_{R,N} - \widetilde\C_N}_2^2 \,\vert\, \mathscr X_N\big) - \E\big(\vertj{\C - \widetilde\C_N}_2^2\big)$, where $\widetilde\C_N$ is distributed identically to $\Chat_N$, independently of $\mathscr X_N$ (cf.~\eqref{eq:HS_equality_covnet1}). Now, $\E\big(\vertj{\C - \widetilde\C_N}_2^2\big) = N^{-1} \E\big(\vertj{\C - \X \otimes \X}_2^2\big)$ converges to $0$ as $N \to \infty$. Thus, it is enough to show that $\Big\{\E\big(\vertj{\widetilde\C_{R,N} - \widetilde\C_N}_2^2 \,\vert\, \mathscr X_N\big)\Big\}^{1/2} - \Big\{\E\big(\vertj{\C - \widetilde\C_N}_2^2\big)\Big\}^{1/2}$ converges to $0$ as $N \to \infty$ in the appropriate notion (i.e., in probability or almost surely). 
We write
\begin{align*}
0 &\le \Big\{\E\Big(\vertj{\widetilde\C_{R,N} - \widetilde\C_N}_2^2 \Big\vert\, \mathscr X_N\Big)\Big\}^{1/2} - \Big\{\E\Big(\vertj{\C - \widetilde\C_N}_2^2\Big)\Big\}^{1/2} \nonumber\\
&= \Big\{\E\Big(\vertj{\widetilde\C_{R,N} - \widetilde\C_N}_2^2 \Big\vert\, \mathscr X_N\Big)\Big\}^{1/2} - \inf_{\G \in \widetilde\F_{R,\lambda_N}} \Big\{\E\Big(\vertj{\G - \widetilde\C_N}_2^2\Big)\Big\}^{1/2} \nonumber \\
&\kern10ex + \inf_{\G \in \widetilde\F_{R,\lambda_N}} \Big\{\E\Big(\vertj{\G - \widetilde\C_N}_2^2\Big)\Big\}^{1/2} - \Big\{\E\Big(\vertj{\C - \widetilde\C_N}_2^2\Big)\Big\}^{1/2}\\
&=: A_N + B_N.
\end{align*}
For the second term, 
\begin{align*}
B_N & =\inf_{\G \in \widetilde\F_{R,\lambda_N}} \Big\{\E\Big(\vertj{\G - \widetilde\C_N}_2^2\Big)\Big\}^{1/2} - \Big\{\E\Big(\vertj{\C - \widetilde\C_N}_2^2\Big)\Big\}^{1/2} \\
&\le \inf_{\G \in \widetilde\F_{R,\lambda_N}} \bigg|\Big\{\E\Big(\vertj{\G - \widetilde\C_N}_2^2\Big)\Big\}^{1/2} - \Big\{\E\Big(\vertj{\C - \widetilde\C_N}_2^2\Big)\Big\}^{1/2}\bigg| \\
&\le \inf_{\G \in \widetilde\F_{R,\lambda_N}} \Big\{\vertj{\G - \C}_2^2\Big\}^{1/2}~~(\text{by  \eqref{eq:HS_equality_covnet1}}) \\
&= \inf_{\G \in \widetilde\F_{R,\lambda_N}} \vertj{\G - \C}_2,
\end{align*}
which by assumption converges to $0$ as $N \to \infty$. So, for convergence in probability, it is enough to show that $\Prob\big(A_N \le 0\big) \to 1$ as $N \to \infty$. Similarly, for almost sure convergence, it is enough to show that $\Prob\big(\limsup_{N \to \infty} A_N \le 0\big) = 1$.

\medskip

Recall that for $\G \in \widetilde\F_{R,\lambda_N}$, $\vertj{\G}_2 \le \gamma_N := R\lambda_N|\Q|M^2$, where $M = \sup_{g \in \mathcal H} \|g\|$. Let $L>0$ be arbitrary. Since $R_N, \lambda_N \to \infty$, we can assume w.l.o.g.\ that $L \le \gamma_N$. Define $\X_L$ to be the projection of $\X$ onto $\{y \in \L_2(\Q): \|y\|^2 \le L\}$. Similarly, for $n=1,2,\ldots$, define $\X_{n,L}$ to be the projection of $\X_n$ onto $\{y \in \L_2(\Q): \|y\|^2 \le L\}$. We define $\C_L = \E(\X_L \otimes \X_L)$ to be the covariance operator of $\X_L$ and $\Chat_{N,L} = N^{-1} \sum_{n=1}^N \X_{n,L} \otimes \X_{n,L}$ to be its empirical counterpart based on $\X_{1,L},\ldots,\X_{N,L}$. Then,
\begin{align*}
A_N & =\bigg(\Big\{\E\Big(\vertj{\widetilde\C_{R,N} - \widetilde\C_N}_2^2 \,\big\vert\, \mathscr X_N\Big)\Big\}^{1/2} - \inf_{\G \in \widetilde\F_{R,\lambda_N}} \Big\{\E\Big(\vertj{\G - \widetilde\C_N}_2^2\Big)\Big\}^{1/2}\bigg) \\
&= \sup_{\G \in \widetilde\F_{R,\lambda_N}} \bigg(\Big\{\E\Big(\vertj{\widetilde\C_{R,N} - \widetilde\C_N}_2^2 \,\big\vert\, \mathscr X_N\Big)\Big\}^{1/2} - \Big\{\E\Big(\vertj{\G - \widetilde\C_N}_2^2\Big)\Big\}^{1/2}\bigg) \\
& = \sup_{\G \in \widetilde\F_{R,\lambda_N}} \bigg(\Big\{\E\Big(\vertj{\widetilde\C_{R,N} - \widetilde\C_N}_2^2 \,\big\vert\, \mathscr X_N\Big)\Big\}^{1/2} - \Big\{\E\Big(\vertj{\widetilde\C_{R,N} - \widetilde\C_{N,L}}_2^2 \,\big\vert\, \mathscr X_N\Big)\Big\}^{1/2} \\
&\kern13ex + \Big\{\E\Big(\vertj{\widetilde\C_{R,N} - \widetilde\C_{N,L}}_2^2 \,\big\vert\, \mathscr X_N\Big)\Big\}^{1/2} - \Big\{\vertj{\widetilde\C_{R,N} - \Chat_{N,L}}_2^2\Big\}^{1/2} \\
&\kern13ex + \vertj{\widetilde\C_{R,N} - \Chat_{N,L}}_2 - \vertj{\Chat_{R,N} - \Chat_{N,L}}_2 \\
&\kern13ex + \vertj{\Chat_{R,N} - \Chat_{N,L}}_2 - \vertj{\Chat_{R,N} - \Chat_N}_2 \\
&\kern13ex + \vertj{\Chat_{R,N} - \Chat_N}_2 - \vertj{\G - \Chat_N}_2 \\
&\kern13ex + \vertj{\G - \Chat_N}_2 - \vertj{\G - \Chat_{N,L}}_2 \\
&\kern13ex + \Big\{\vertj{\G - \Chat_{N,L}}_2^2\Big\}^{1/2} - \Big\{\E\Big(\vertj{\G - \Chat_{N,L}}_2^2\Big)\Big\}^{1/2} \\
&\kern13ex + \Big\{\E\Big(\vertj{\G - \Chat_{N,L}}_2^2\Big)\Big\}^{1/2} - \Big\{\E\Big(\vertj{\G - \Chat_N}_2^2\Big)\Big\}^{1/2} \bigg).
\end{align*}
By definition, the fifth term is non-positive. Also, by construction, the third term is non-positive. Of the remaining, the first and the eighth terms are bounded above by $\big\{\E\big(\vertj{\Chat_N - \Chat_{N,L}}_2^2\big) \big\}^{1/2}$, while the fourth and the sixth terms are bounded above by $\vertj{\Chat_N - \Chat_{N,L}}_2$. Finally, the second and the seventh terms are bounded above by $\sup_{\G \in \widetilde\F_{R,\lambda_N}} \Big|\vertj{\G-\Chat_{N,L}}_2 - \big\{\E\big(\vertj{\G - \Chat_{N,L}}_2^2\big)\big\}^{1/2}\Big|$. Under the assumption, $\sup_{\G \in \widetilde\F_{R,\lambda_N}} \big|\vertj{\G-\Chat_{N,L}}_2^2 - \E\big(\vertj{\G - \Chat_{N,L}}_2^2\big)\big| \to 0$ in probability as $N \to \infty$ (cf. Lemma~\ref{lemma:covnet1_type1_VN_convergence}). Using this and the uniform continuity of $x \mapsto \sqrt{x}$ on $[0,\infty)$, we get
\[
\sup_{\G \in \widetilde\F_{R,\lambda_N}} \bigg|\vertj{\G-\Chat_{N,L}}_2 - \Big\{\E\big(\vertj{\G - \Chat_{N,L}}_2^2\big)\Big\}^{1/2}\bigg| \overset{P}{\to} 0 \text{ as } N \to \infty.
\]
For the other two terms, observe that $\Chat_N - \Chat_{N,L} = N^{-1} \sum_{n=1}^N Z_{n,L}$,
where $Z_{n,L} = \X_n \otimes \X_n - \X_{n,L} \otimes \X_{n,L}$. Note that $\E(Z_{n,L}) = \C - \C_L$. Now, $\vertj{\Chat_N - \Chat_{N,L}}_2 \le N^{-1} \sum_{n=1}^N \vertj{Z_{n,L}}_2$, which converges almost surely to $\E(\vertj{Z_L}_2) = \E\big(\vertj{\X \otimes \X - \X_L \otimes \X_L}_2\big)$ as $N \to \infty$. On the other hand,
\begin{align*}
\E\vertj{\Chat_N - \Chat_{N,L}}_2^2 = \E\verti{\frac{1}{N}\sum_{n=1}^N Z_{n,L}}_2^2 &= \frac{1}{N} \E \vertj{Z_L}_2^2 + \frac{N(N-1)}{N^2} \vertj{\E(Z_L)}_2^2 \\
&= \frac{1}{N} \E \vertj{\X \otimes \X - \X_L \otimes \X_L}_2^2 + \frac{N(N-1)}{N^2} \vertj{\C - \C_L}_2^2. 
\end{align*}
Since $\E \|\X\|^4 < \infty$, this last term converges to $\vertj{\C - \C_L}_2^2$ as $N \to \infty$ for all $L > 0$. Thus, for all $L > 0$,
\begin{align*}
\Prob\Big(A_N \le 2\E\vertj{\X \otimes \X - \X_L \otimes \X_L}_2 + 2\vertj{\C - \C_L}_2\Big) \to 1 \text{ as } N \to \infty.
\end{align*}
Note that $\vertj{\C - \C_L}_2 \le \E(\vertj{\X \otimes \X - \X_L \otimes \X_L}_2)$. Now, using \eqref{eq:product_difference_bound_L2},
\begin{align*}
\E\big(\vertj{\X \otimes \X - \X_L \otimes \X_L}_2\big) &\le 2 \E\big(\|\X\|\,\|\X - \X_L\|\big) + \E\big(\|\X - \X_L\|^2\big) \\
&\le 2 \Big\{\E\big(\|\X\|^2\big)\,\E\big(\|\X - \X_L\|^2\big)\Big\}^{1/2} + \E\big(\|\X - \X_L\|^2\big).
\end{align*}
By the dominated convergence theorem, the last quantity converges to $0$ as $L \to \infty$. Thus, taking limit as $L \to \infty$, we get that $\Prob(A_N \le 0) \to 1$ as $N \to \infty$, proving the convergence of $\vertj{\widetilde\C_{R,N} - \C}_2^2$ to $0$ in probability.

\medskip

For the almost sure convergence, all the steps remain the same, except now with the additional condition
\[
\sup_{\G \in \widetilde\F_{R,\lambda_N}} \bigg|\vertj{\G-\Chat_{N,L}}_2 - \Big\{\E\big(\vertj{\G - \Chat_{N,L}}_2^2\big)\Big\}^{1/2}\bigg| \overset{a.s.}{\to} 0 \text{ as } N \to \infty,
\]
see Lemma~\ref{lemma:covnet1_type1_VN_convergence}. Thus, by the same steps we get $\Prob\big(\limsup_{N \to \infty} A_N \le 0) = 1$, completing the proof.
\end{proof}

\subsection{The case of discretely observed data}\label{supp:measurement_with_noise}
Here, we derive the asymptotic properties of our estimators in the discrete measurement regime (Section~\ref{sec:asymptotics_discrete}). Again, we will derive the results for a general class of operators $\widetilde\F_N$. We define $\widetilde\X_n^K = (\widetilde X_n^K(\uvec): \uvec \in [0,1]^d)$ to be the voxel-wise continuation of the $n$-th measurement $\widetilde{\mathbf X}_n^K = (\widetilde X_n^K[i_1,\ldots,i_d])$, with 
\[
\widetilde X_n^K(\uvec) = \sum_{i_1=1}^{K_1}\cdots\sum_{i_d=1}^{K_d} \widetilde X_n^K[i_1,\ldots,i_d]\,\1\big\{\uvec \in V_{i_1,\ldots,i_d}^K\big\}.
\]
Our estimator in this case is defined as
\begin{equation}\label{eq:discrete_estimator_general}
\Chat_{\widetilde\F_N}^K \in \argmin_{\G \in \widetilde\F_N} \vertj{\widetilde\C_N^K - \G}_2^2,
\end{equation}
where $\widetilde\C_N^K = N^{-1}\sum_{n=1}^N \widetilde\X_n^K \otimes \widetilde\X_n^K$ is the empirical covariance based on $\widetilde\X_1^K,\ldots,\widetilde\X_N^K$. In the following, for an operator $\A$, we write $\mathbf A^K \in \R^{K_1\times\cdots\times K_d \times K_1\times\cdots\times K_d}$ to denote its discretized version w.r.t.\ the voxels $V_{i_1,\ldots,i_d}^K, 1 \le i_1 \le K_1,\ldots, 1 \le i_d \le K_d$, and $\A^K$ to denote the voxel-wise continuation of $\mathbf A^K$. The Frobenius norm of $\mathbf A^K$ is defined as
\[
\|\mathbf A^K\|_{\rm F}^2 = \sum_{i_1=1}^{K_1}\cdots\sum_{i_d=1}^{K_d}\sum_{j_1=1}^{K_1}\cdots\sum_{j_d=1}^{K_d} A^2[i_1,\ldots,i_d;j_1,\ldots,j_d].
\]
It is easy to verify that $\|\mathbf A^K\|_{\rm F}^2 = (K_1\cdots K_d)^2\verti{\A^K}_2^2$. Thus, for a class of covariance operators $\widetilde\F$,
\begin{equation}\label{eq:discrete_approximate_estimator_equivalence}
\argmin_{\G \in \widetilde\F} \big\|\widehat{\mathbf C}_N^K - \mathbf G^K\big\|_{\rm F}^2 = \argmin_{\G \in \widetilde\F} \verti{\Chat_N^K - \G^K}_2^2 \approx \argmin_{\G \in \widetilde\F} \verti{\Chat_N^K - \G}_2^2,
\end{equation}
where the last assertion holds when $K$ is large. This justifies the definition of our estimator in \eqref{eq:discrete_estimator_general}.

\medskip

By \eqref{eq:measurement_with_noise}, it follows that $\widetilde\X_n^K = \X_n^K + \mathcal E_n^K$, where $\X_n^K = (X_n^K(\uvec): \uvec \in [0,1]^d)$ and $\mathcal E_n^K = (E_n^K(\uvec): \uvec \in [0,1]^d)$ are the voxel-wise continuations of $\mathbf X_n^K = (X_n^K[i_1,\ldots,i_d])$ and $\bm{\mathcal E}_n^K = (\epsilon_n^K[i_1,\ldots,i_d])$, respectively, i.e.,
\begin{align*}
X_n^K(\uvec) = \sum_{i_1=1}^{K_1}\cdots\sum_{i_d=1}^{K_d} X_n^K[i_1,\ldots,i_d]\,\1\big\{\uvec \in V_{i_1,\ldots,i_d}^K\big\} \,\text{ and }\, E_n^K(\uvec) = \sum_{i_1=1}^{K_1}\cdots\sum_{i_d=1}^{K_d} E_n^K[i_1,\ldots,i_d]\,\1\big\{\uvec \in V_{i_1,\ldots,i_d}^K\big\}.
\end{align*}
It is easy to see that $\widetilde\X_1^K,\ldots,\widetilde\X_N^K$ are i.i.d.\ with mean zero. Let $\widetilde\C^K = \E(\widetilde\X_1^K \otimes \widetilde\X_1^K)$ be the covariance of $\widetilde\X_1^K$ (the existence of $\C^K$ will be proved shortly). Since $\mathbf X_n^K$ and $\bm{\mathcal E}_n^K$ are uncorrelated, it follows that $\widetilde\C^K = \C^K + \Sigma^K$, where $\C^K = \E(\X^K \otimes \X^K)$ is the covariance of $\X^K$ and $\Sigma^K = \E(\mathcal E^K \otimes \mathcal E^K)$ is the covariance of $\mathcal E^K$. To quantify the rate of convergence of the estimator \eqref{eq:discrete_estimator_general}, we use the following decomposition
\begin{align}\label{eq:discrete_plus_noise_upper_bound}
\verti{\Chat_{\widetilde\F_N}^K  - \C}_2^2 &= \verti{\Chat_{\widetilde\F_N}^K  - \widetilde\C^K + \widetilde\C^K - \C}_2^2 \nonumber\\
&= \verti{\widetilde\C_{\widetilde\F_N}^K  - \widetilde\C^K + \C^K - \C + \Sigma^K}_2^2 \nonumber\\
&\le 3 \verti{\widetilde\C_{\widetilde\F_N}^K  - \widetilde\C^K}_2^2 + 3 \verti{\C^K - \C}_2^2 + 3 \verti{\Sigma^K}_2^2.
\end{align}
We will separately bound each term in the above.

\medskip

Firstly, note that $\|\widetilde\X_n^K\|^2 \le 2 \|\X_n^K\|^2 + 2 \|\mathcal E_n^K\|^2$. Now,
\[
\|\X_n^K\|^2 = \int_{[0,1]^d} \big(X_n^K(\uvec)\big)^2 \diff\uvec = \frac{1}{K_1\cdots K_d}\sum_{i_1=1}^{K_1}\cdots\sum_{i_d=1}^{K_d} \big(X_n^K[i_1,\ldots,i_d]\big)^2.
\]
Under (M1), since $\|\X\|_{\infty}^2 = \sup_{\uvec \in [0,1]^d} X^2(\uvec) \le \beta_N$ almost surely, it follows that $\|\X_n^K\|^2 \le \beta_N$ almost surely. Under (M2), if we define $\mathcal I_{i_1,\ldots,i_d}^K(\uvec) = |V_{i_1,\ldots,i_d}^K|^{-1/2}\1\big\{\uvec \in V_{i_1,\ldots,i_d}^K\big\}$, then it follows that $X_n^K[i_1,\ldots,i_d] = \sqrt{K_1\cdots K_d}\,\langle \X_n,\mathcal I_{i_1,\ldots,i_d}^K\rangle$. Thus,
\[
\|\X_n^K\|^2 = \frac{1}{K_1\cdots K_d}\sum_{i_1=1}^{K_1}\cdots\sum_{i_d=1}^{K_d} \big(X_n^K[i_1,\ldots,i_d]\big)^2 = \sum_{i_1=1}^{K_1}\cdots\sum_{i_d=1}^{K_d} \langle \X,\mathcal I_{i_1,\ldots,i_d}^K\rangle^2.
\]
Now, the functions $\mathcal I_{i_1,\ldots,i_d}^K, 1 \le i_1 \le K_1,\ldots,1 \le i_d \le K_d$ are orthonormal, and hence by Parseval's inequality it follows that $\|\X_n^K\|^2 \le \|\X\|^2 \le \beta_N$ almost surely. In a similar way, it follows that
\[
\|\mathcal E_n^K\|^2 = \frac{1}{K_1\cdots K_d} \sum_{i_1=1}^{K_1}\cdots\sum_{i_d=1}^{K_d} \big(E_n^K[i_1,\ldots,i_d]\big)^2 \le \beta_{N,K}^{\rm e} \text{ almost surely}.
\]
Thus, under both measurement schemes, we get that $\|\widetilde\X_n^K\|^2 \le \widetilde\beta_{N,K}:=2(\beta_N+\beta_{N,K}^{\rm e})$ almost surely. This also ensures that $\widetilde\C^K,\C^K$ and $\Sigma^K$ are well-defined. Now, since $\widetilde\C_N^K$ is the empirical covariance based on $\widetilde\X_1^K,\ldots,\widetilde\X_N^K$, which are i.i.d.\ with $\E(\X_1^K) = 0$, $\var(\X_1^K) = \widetilde\C^K$ and $\|\widetilde\X_n^K\|^2 \le \widetilde\beta_{N,K}$ almost surely, it follows from Lemma~\ref{lemma:covnet_rate_of_convergence} that
\begin{equation}\label{eq:discrete_plus_noise_rate_eq1}
\E\Big(\verti{\widetilde\C_{\widetilde\F_N}^K - \widetilde\C^K}_2^2\Big) \le 2 \inf_{\G \in \widetilde\F_N} \verti{\G - \widetilde\C^K}_2^2 + \O\bigg(\frac{\widetilde\Delta_{N,K}^4}{N} \times \log\cover\bigg(\frac{\widetilde\Delta_{N,K}}{N},\widetilde\F_N,\vertj{\cdot}_2\bigg)\bigg),
\end{equation}
where $\widetilde\Delta_{N,K} = \max\{\widetilde\beta_{N,K},\gamma_N\}$. For the first term on the right-hand side, we use the upper-bound
\begin{align}\label{eq:discrete_plus_noise_rate_eq2}
\inf_{\G \in \widetilde\F_N} \verti{\G - \widetilde\C^K}_2^2 &= \inf_{\G \in \widetilde\F_N} \verti{\G - \C + \C - \widetilde\C^K}_2^2 \nonumber\\
&= \inf_{\G \in \widetilde\F_N} \verti{\G - \C + \C - \C^K - \Sigma^K}_2^2 \nonumber\\
&\le 3 \inf_{\G \in \widetilde \F_N} \verti{\G - \C}_2^2 + 3 \verti{\C - \C^K}_2^2 + 3 \verti{\Sigma^K}_2^2.
\end{align}

Secondly, because of the discretized nature of $\C^K$, we get that
\[
\vertj{\C^K - \C}_2^2 = \frac{1}{(K_1\cdots K_d)^2}\sum_{i_1=1}^{K_1}\cdots\sum_{i_d=1}^{K_d}\sum_{j_1=1}^{K_1}\cdots\sum_{j_d=1}^{K_d} \iint_{V_{i_1,\ldots,i_d}^K \times V_{j_1,\ldots,j_d}^K} \big\{c^K(\uvec,\vvec) - c(\uvec,\vvec)\big\}^2 \diff\uvec \diff\vvec,
\]
where $c^K$ is the kernel corresponding to $\C^K$, defined as $c^K(\uvec,\vvec) = \cov\{\X_1^K(\uvec),\X_1^K(\vvec)\}$. It follows that, under (M1), $c^K(\uvec,\vvec) = c(u_{i_1},\ldots,u_{i_d};u_{j_1},\ldots,u_{j_d})$ for $\uvec \in V_{i_1,\ldots,i_d}^K, \vvec \in V_{j_1,\ldots,j_d}^K$. Thus, for $\uvec \in V_{i_1,\ldots,i_d}^K, \vvec \in V_{j_1,\ldots,j_d}^K$, using the Lipschitz property of $c$, we get
\begin{align*}
\big|c^K(\uvec,\vvec) - c(\uvec,\vvec)\big\|^2 &\le \rho^2 \|(\uvec,\vvec) - (u_{i_1},\ldots,u_{i_d};u_{j_1},\ldots,u_{j_d})\big\|^2 \\
&\le \rho^2 \bigg(\frac{1}{K_1^2} + \cdots + \frac{1}{K_d^2}\bigg).
\end{align*}
Again, under (M2), $\X_1^K(\uvec) = \langle X, |V_{i_1,\ldots,i_d}^K|^{-1} \1\big\{\uvec \in V_{i_1,\ldots,i_d}^K\big\}\rangle$ for $\uvec \in V_{i_1,\ldots,i_d}^K$, so that 
\begin{align*}
c^K(\uvec,\vvec) &= \frac{1}{|V_{i_1,\ldots,i_d}^K|\,|V_{j_1,\ldots,j_d}^K|}\langle \C\,\1\big\{\uvec \in V_{i_1,\ldots,i_d}^K\big\},\big\{\vvec \in V_{j_1,\ldots,j_d}^K\big\}\rangle_2 \\
&= \frac{1}{|V_{i_1,\ldots,i_d}^K|\,|V_{j_1,\ldots,j_d}^K|} \iint_{V_{i_1,\ldots,i_d}^K \times V_{j_1,\ldots,j_d}^K} c(\uvec^\prime,\vvec^\prime) \diff\uvec^\prime \diff\vvec^\prime.
\end{align*}
Thus, under this scheme again, for $\uvec \in V_{i_1,\ldots,i_d}^K, \vvec \in V_{j_1,\ldots,j_d}^K$, we get
\begin{align*}
\big|c^K(\uvec,\vvec) - c(\uvec,\vvec)\big\|^2 &= \bigg|\frac{1}{|V_{i_1,\ldots,i_d}^K|\,|V_{j_1,\ldots,j_d}^K|}\iint_{V_{i_1,\ldots,i_d}^K \times V_{j_1,\ldots,j_d}^K} c(\uvec^\prime,\vvec^\prime) \diff\uvec^\prime \diff\vvec^\prime - c(\uvec,\vvec)\bigg|^2 \\
&=\bigg|\frac{1}{|V_{i_1,\ldots,i_d}^K|\,|V_{j_1,\ldots,j_d}^K|}\iint_{V_{i_1,\ldots,i_d}^K \times V_{j_1,\ldots,j_d}^K} \underbrace{\big|c(\uvec^\prime,\vvec^\prime) - c(\uvec,\vvec)\big|}_{\le \rho\sqrt{\frac{1}{K_1^2}+\cdots+\frac{1}{K_d^2}}}\diff\uvec^\prime\diff\vvec^\prime\bigg| \\
&\le \rho^2 \bigg(\frac{1}{K_1^2} + \cdots + \frac{1}{K_d^2}\bigg).
\end{align*}
Hence, under both (M1) and (M2), we get that
\begin{equation}\label{eq:discrete_plus_noise_rate_eq3}
\vertj{\C^K-\C}_2^2 \le \rho^2 \bigg(\frac{1}{K_1^2} + \cdots + \frac{1}{K_d^2}\bigg).
\end{equation}

Thirdly, by the assumptions of Theorem~\ref{thm:discrete_measurement_rate}, it follows that
\begin{align}\label{eq:discrete_plus_noise_rate_eq4}
\verti{\Sigma^K}_2^2 &= \frac{1}{(K_1\cdots K_d)^2} \sum_{i_1=1}^{K_1}\cdots\sum_{i_d=1}^{K_d}\sum_{j_1=1}^{K_1}\cdots\sum_{j_d=1}^{K_d} \big\{\cov(E_n^K[i_1,\ldots,i_d],E_n^K[j_1,\ldots,j_d])\big\}^2 \nonumber\\
&= \frac{1}{(K_1\cdots K_d)^2} \sum_{i_1=1}^{K_1}\cdots\sum_{i_d=1}^{K_d} \big\{\var(E_n^K[i_1,\ldots,i_d])\big\}^2 \nonumber\\
&= \frac{\sigma_K^4}{K_1\cdots K_d}.
\end{align}

Combining \eqref{eq:discrete_plus_noise_rate_eq1}--\eqref{eq:discrete_plus_noise_rate_eq4} with \eqref{eq:discrete_plus_noise_upper_bound}, we get that
\begin{align*}
\E\Big(\verti{\Chat_{\widetilde\F_N}^K  - \C}_2^2\Big) &\le 18 \inf_{\G \in \widetilde\F_N} \verti{\G - \C}_2^2 + 21 \rho^2\bigg(\frac{1}{K_1^2}+\cdots+\frac{1}{K_d^2}\bigg) + \frac{21\sigma_K^4}{K_1\cdots K_d} \\
&\kern30ex + \O\bigg(\frac{\widetilde\Delta_{N,K}^4}{N} \times \log\cover\bigg(\frac{\widetilde\Delta_{N,K}}{N},\widetilde\F_N,\vertj{\cdot}_2\bigg)\bigg),
\end{align*}
where $\widetilde\Delta_{N,K} = \max\{2(\beta_N+\beta_{N,K}^{\rm e}),\gamma_N\}$. Using this, with $\gamma_N = R\lambda_N$ and the derived bounds on the covering numbers, we get the rates as given in Theorem~\ref{thm:discrete_measurement_rate}.

\section{Additional simulation results}\label{supp:additional_simulations}

In this section, we provide some additional simulation results which were skipped in the main text. We start with the estimation errors for Ex\,1--5 in the main text in 2D for the fixed resolution and varying sample size regime, which are shown in Figure~\ref{fig:all_2D_N}. The results are as expected -- the estimation error for all the methods decreases as the sample size increases. 

\begin{figure}[h!t]
\centering
\begin{tabular}{cc}
(a) Brownian sheet & (b) Rotated Brownian sheet \\
\includegraphics[width=0.4\linewidth]{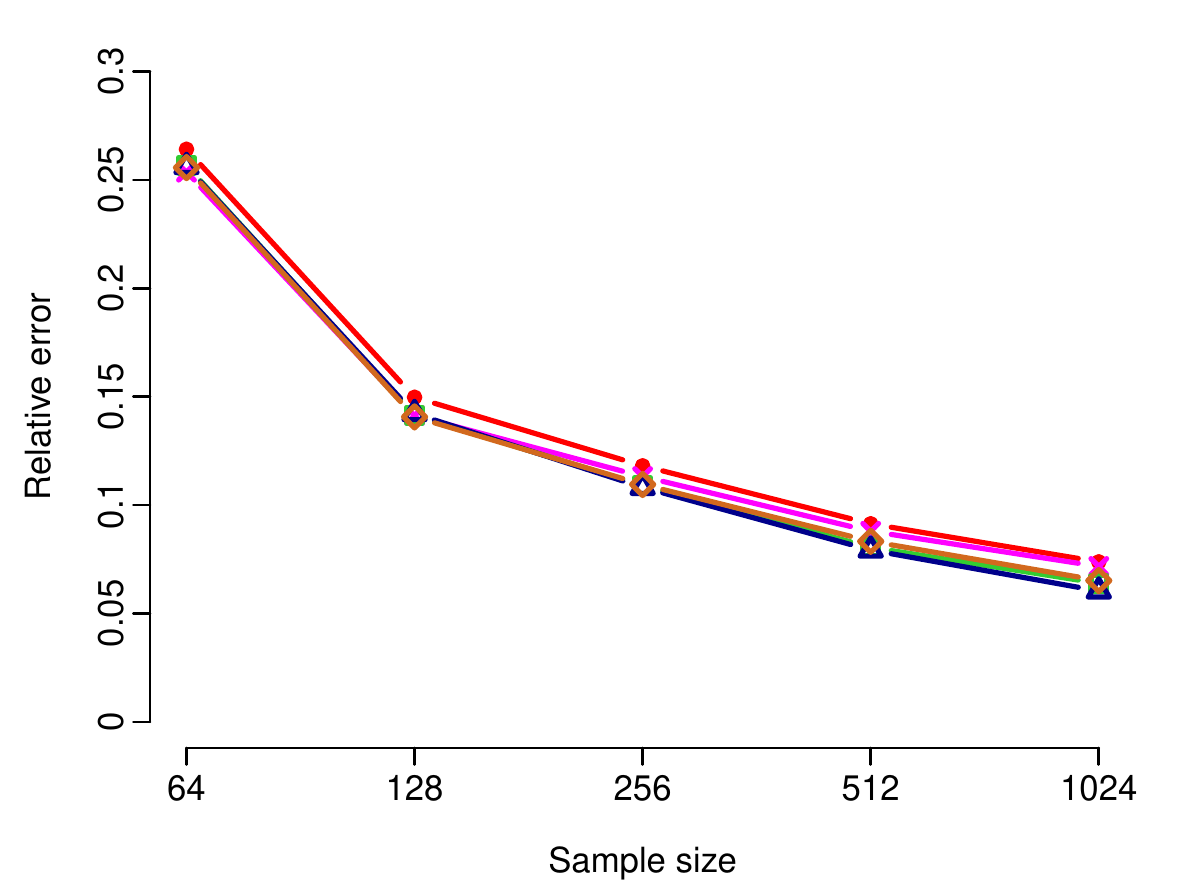} &
\includegraphics[width=0.4\linewidth]{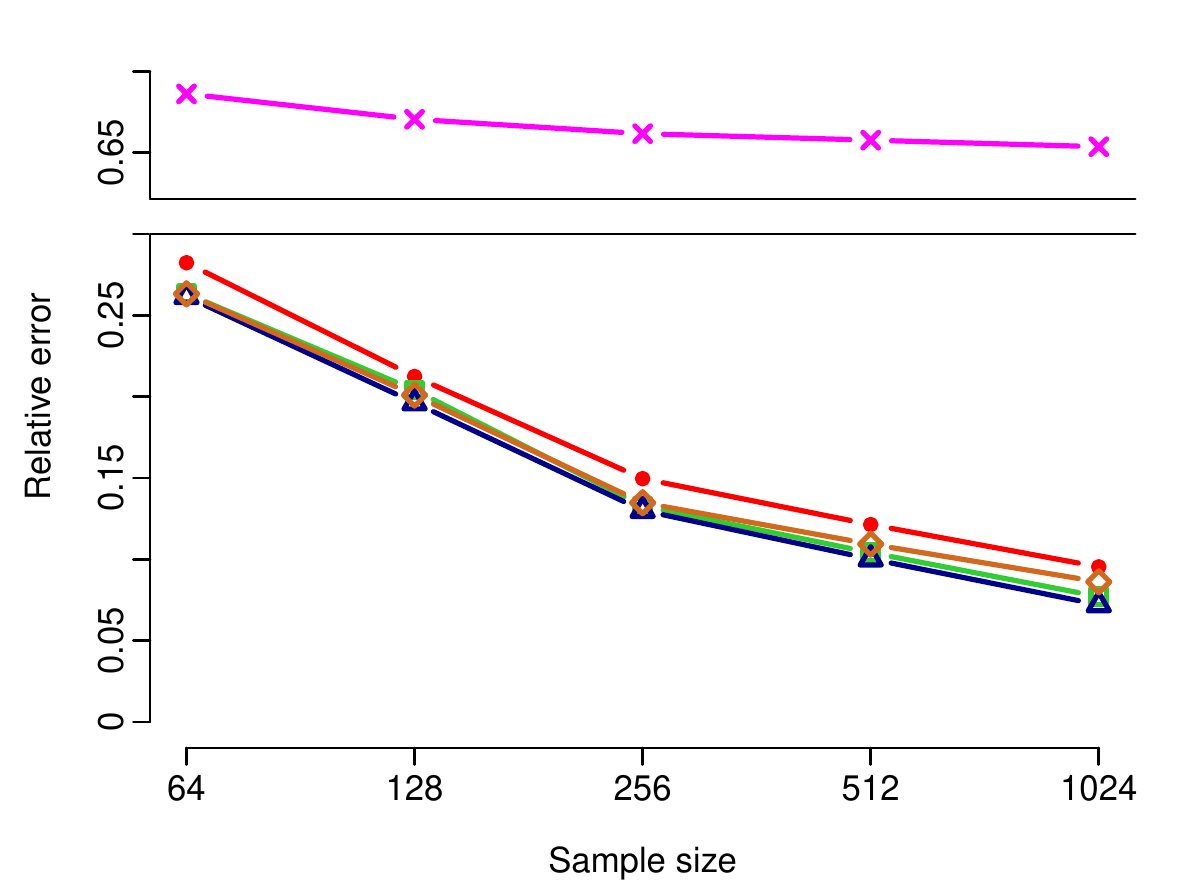} \\ [5pt]
(c) Integrated Brownian sheet & (d) Rotated integrated Brownian sheet \\
\includegraphics[width=0.4\linewidth]{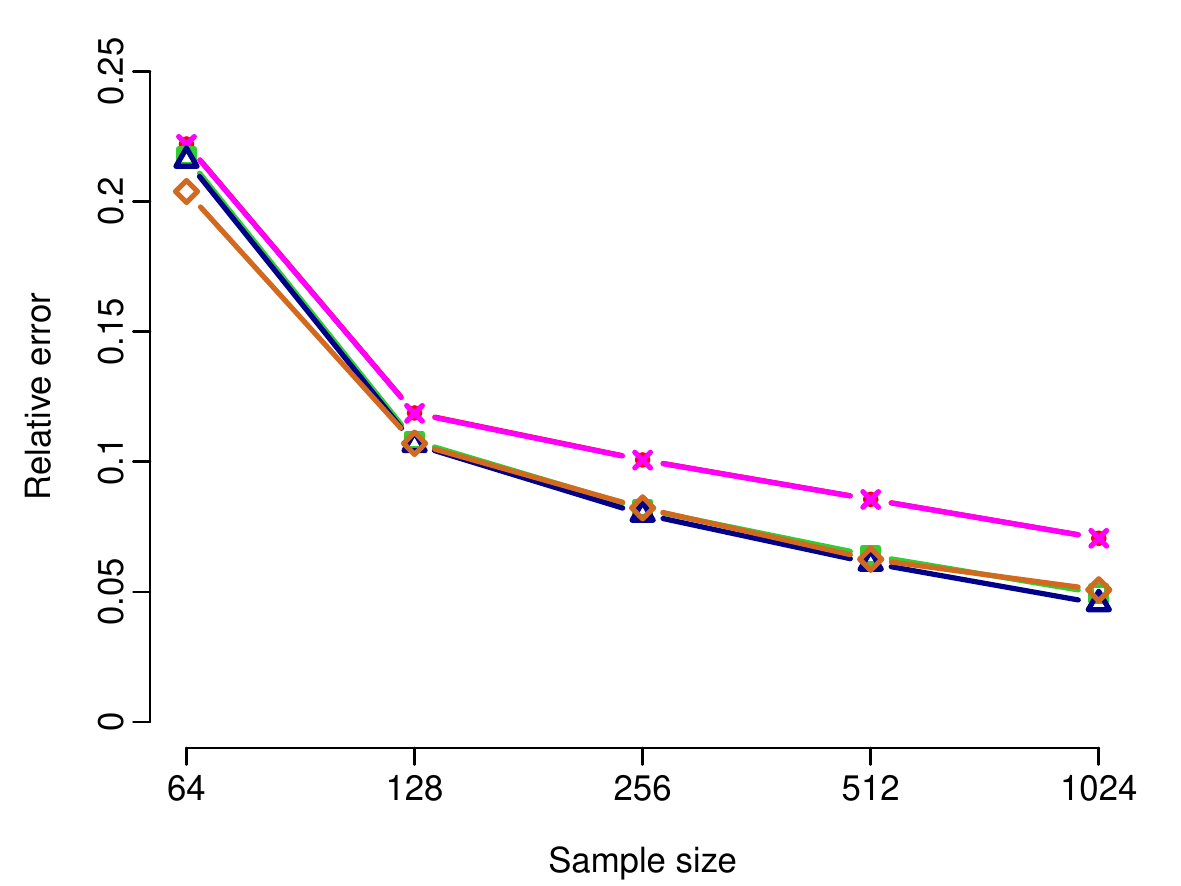} &
\includegraphics[width=0.4\linewidth]{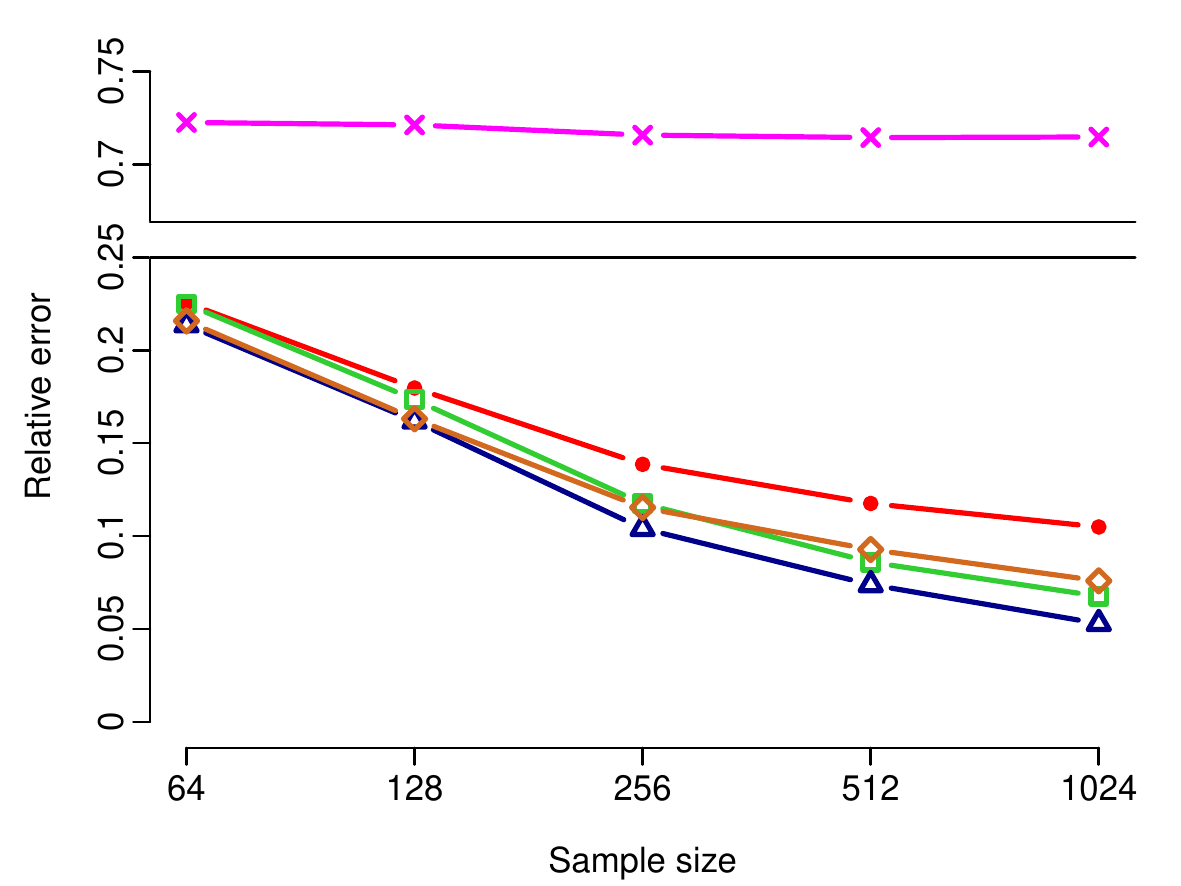} \\ [5pt]
\multicolumn{2}{c}{(e) Matern ($\nu=0.01$)} \\
\multicolumn{2}{c}{\includegraphics[width=0.4\linewidth]{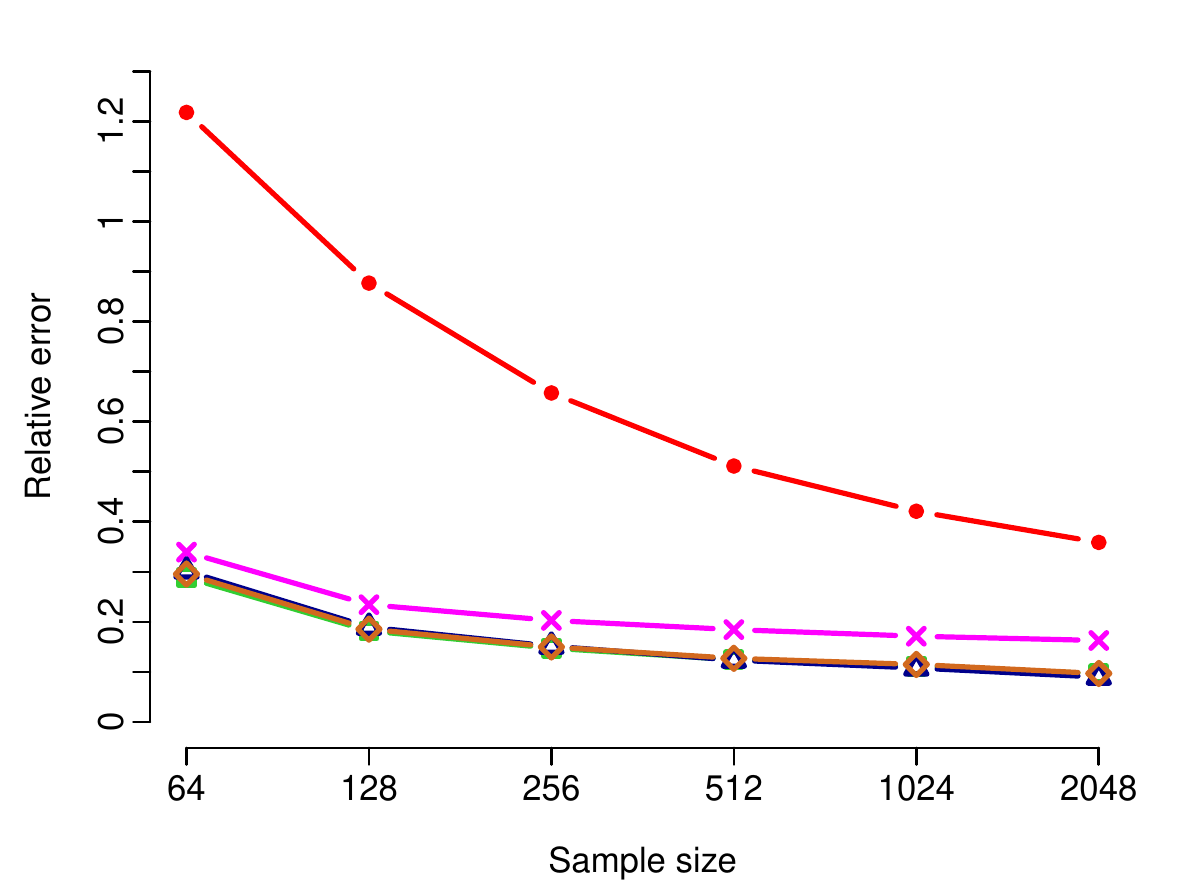}} \\
\multicolumn{2}{c}{\textbf{Legend:} Empirical \Line[Emp]\Emp\Line[Emp]~~ Best separable \Line[Sep]\Sep\Line[Sep]~~ Shallow \Line[Sh]\Sh\Line[Sh]~~ Deep \Line[D]\D\Line[D]~~ DeepShared \Line[DS]\DS\Line[DS]}
\end{tabular}
\caption{\label{fig:all_2D_N} Relative errors of different methods for different examples in 2D. Results are reported for a fixed resolution of $25 \times 25$ and varying sample size. The numbers are averages based on $25$ simulation runs.}
\end{figure}

In Ex\,1 and 3, when the true covariance is separable, although all the methods perform almost similarly, the CovNet estimators seem to have an edge over the others. In Ex\,2 and 4, the separability of the true covariance is broken by rotation, which has a dire consequence on the performance of the best separable estimator, but not on the CovNet estimators. The advantage of the CovNet estimators is more prominent for larger sample sizes, as can be seen in the results for the integrated Brownian motion (both the usual and the rotated versions) and the Mat\'ern covariance models (Figures~\ref{fig:all_2D_N}(c)--(e)).

\begin{figure}[t]
\centering
\begin{tabular}{cc}
(a) Rotated Brownian sheet: $N=500$ & (b) Rotated integrated Brownian sheet: $N=500$ \\
\includegraphics[width=0.4\linewidth]{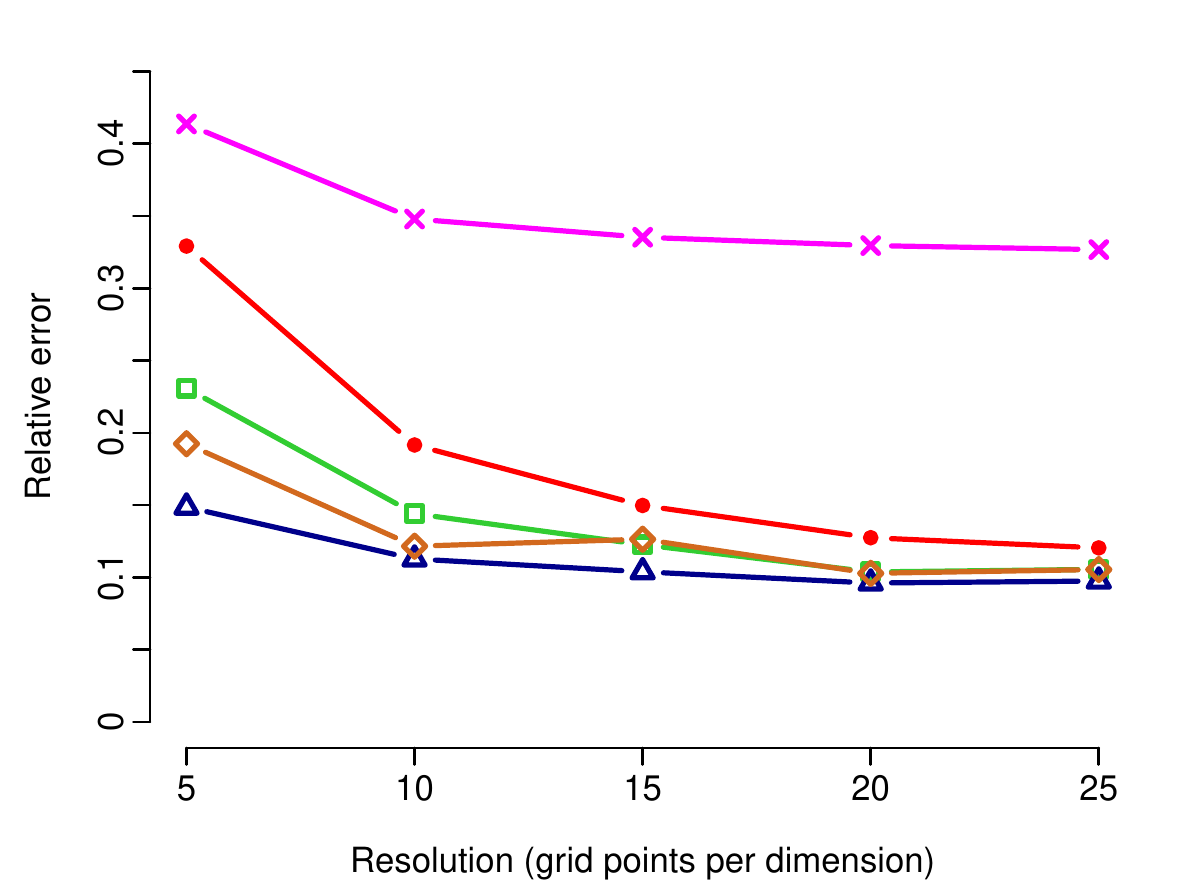} &
\includegraphics[width=0.4\linewidth]{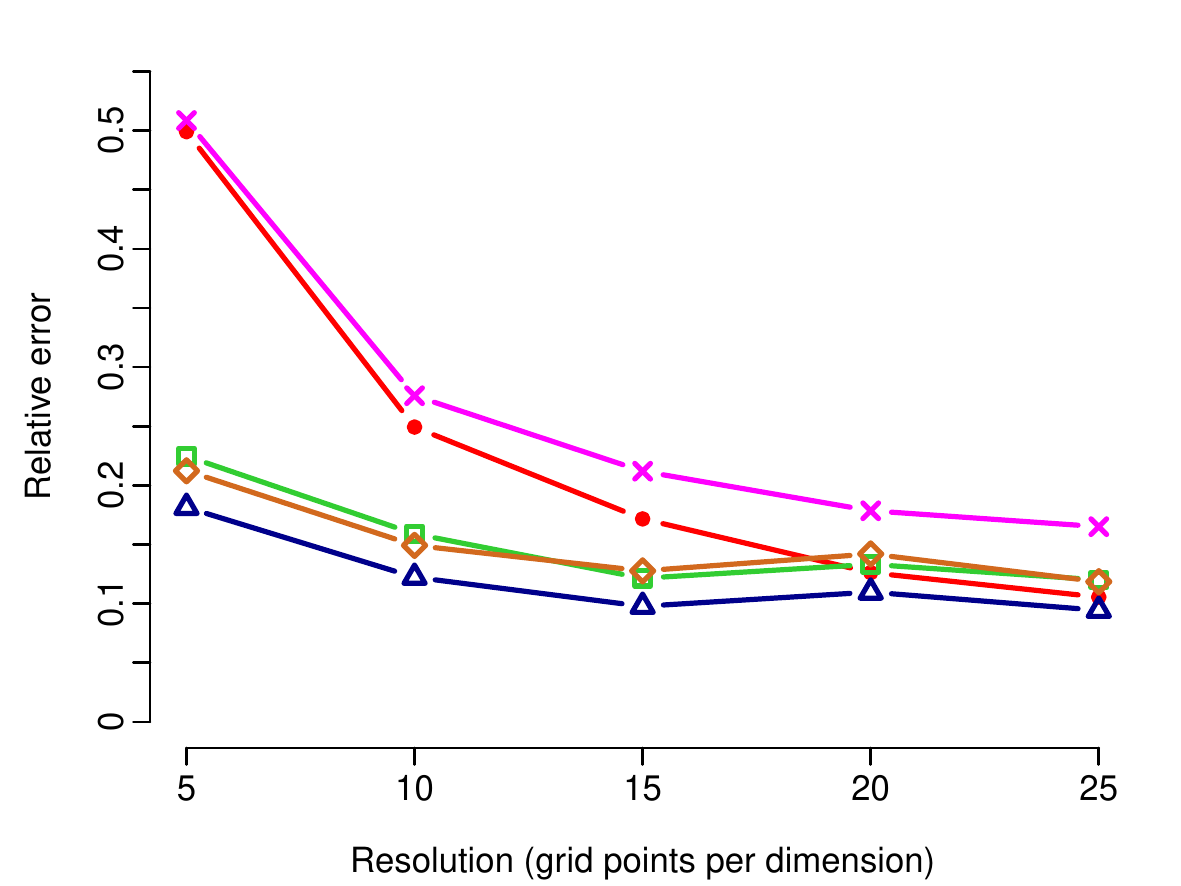} \\[5pt]
(c) Mat\'ern covariance: $N=250$, resolution $15 \times 15 \times 15$ & (d) Mat\'ern covariance: $N=500$, $\nu=0.01$  \\
\includegraphics[width=0.4\linewidth]{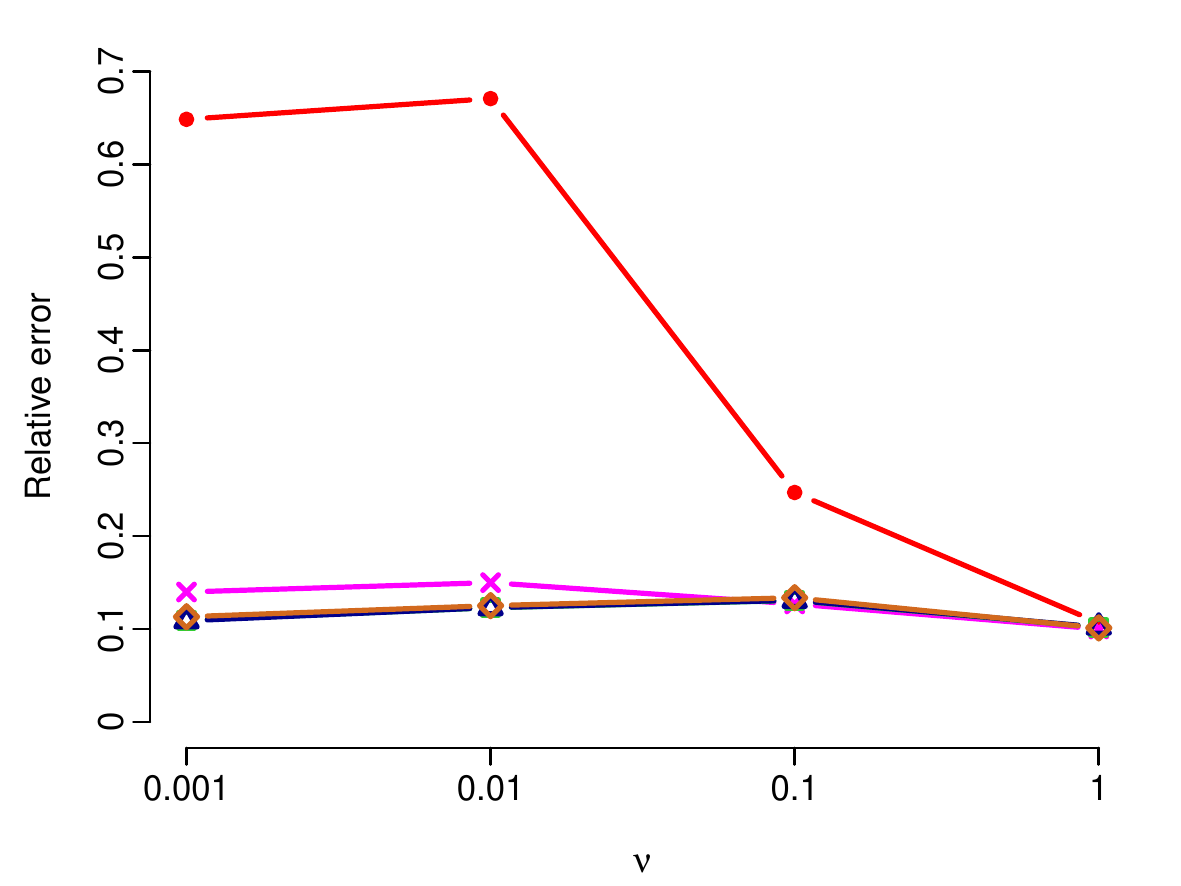} &
\includegraphics[width=0.4\linewidth]{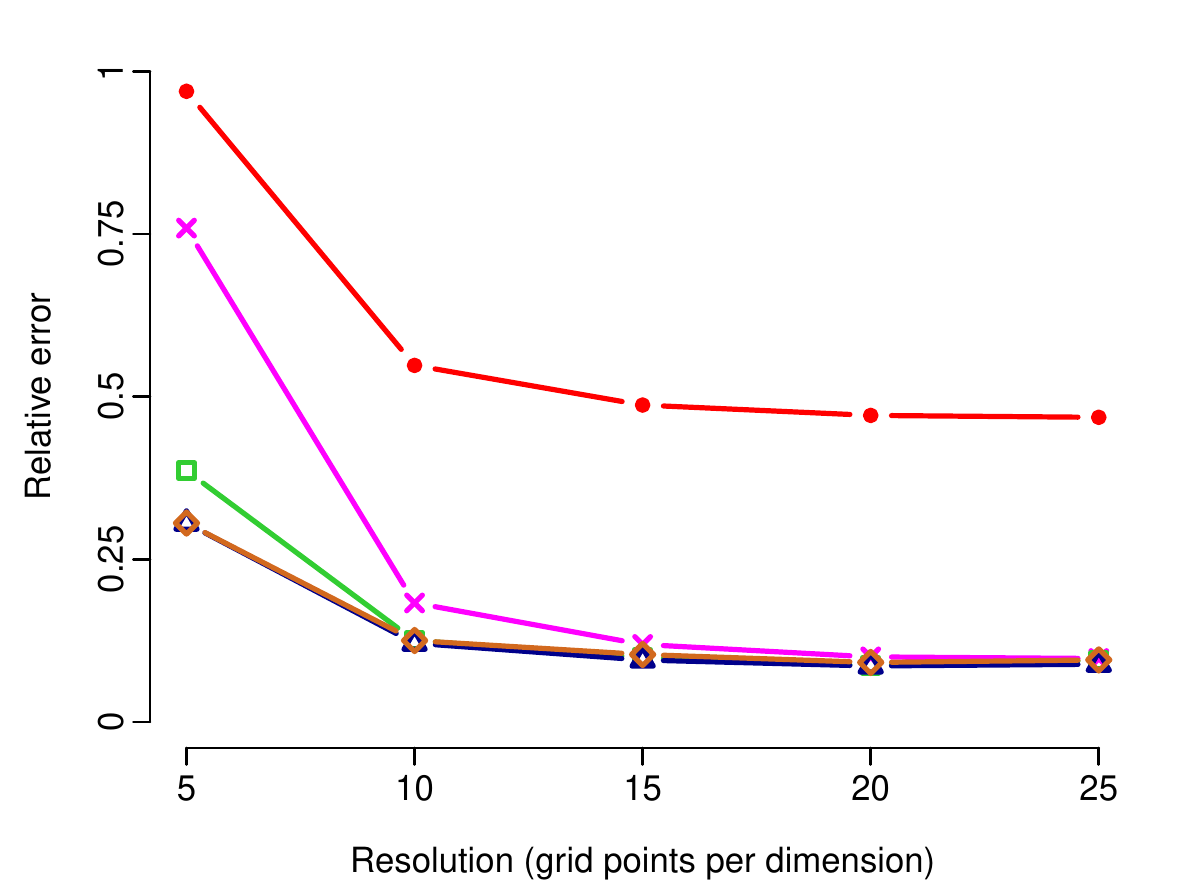} \\
\multicolumn{2}{c}{\textbf{Legend:} Empirical \Line[Emp]\Emp\Line[Emp]~~ Best separable \Line[Sep]\Sep\Line[Sep]~~ Shallow \Line[Sh]\Sh\Line[Sh]~~ Deep \Line[D]\D\Line[D]~~ DeepShared \Line[DS]\DS\Line[DS]}
\end{tabular}
\caption{\label{fig:errors_3D} Relative errors of different methods for different examples in 3D. In (a) and (b), results are for a fixed sample size of $500$ and varying resolutions. In (c), results are for sample size $250$ and resolution $15 \times 15 \times 15$, with varying smoothness parameter $\nu$. In (d), results are for sample size $500$ and $\nu=0.01$ with varying resolutions. The numbers are averages based on $20$ simulation runs.}
\end{figure}

\medskip

Next, we consider the results in 3D, which are reported in Figure~\ref{fig:errors_3D}. Here, we only report the results for the non-separable examples (Ex\,2, 4 and 5). For the rotated examples (Ex\,2 and 4), we take $\mathrm O$ to be the composition of the basic rotations by $45^\circ$ along the $x,y$, and $z$-axes, respectively. Formally, $\mathrm O = \mathrm O_z \mathrm O_y \mathrm O_x$, where
\[
\mathrm O_x = \begin{pmatrix} 1 & 0 & 0 \\ 0 & 1/\sqrt{2} & -1/\sqrt{2} \\ 0 & 1/\sqrt{2} & 1/\sqrt{2} \end{pmatrix}, 
\mathrm O_y = \begin{pmatrix} 1/\sqrt{2} & 0 & 1/\sqrt{2} \\ 0 & 1 & 0 \\ -1/\sqrt{2} & 0 & 1/\sqrt{2} \end{pmatrix} \text{ and }
\mathrm O_z = \begin{pmatrix} 1/\sqrt{2} & -1/\sqrt{2} & 0 \\ 1/\sqrt{2} & 1/\sqrt{2} & 0 \\ 0 & 0 & 1 \end{pmatrix}.
\]
The best separable estimator is designed specifically for problems which have the 2D structure (specifically, for spatio-temporal problems), with no straightforward extension to the 3D scenario. To accommodate for this, for the best separable estimator, we combined two of the three dimensions together to transform the data into 2D. This gives us three different results depending on which two dimensions are combined. In Figure~\ref{fig:errors_3D}, we report the best (minimum error) among these three results.

In 3D, our basic findings remain the same as in 2D. The CovNet estimators outperform the empirical and the best separable estimators, especially when the resolution is low. Also, the smoothness of the integrated Brownian sheet enhances the performance of the CovNet operators, as opposed to the other two estimators. For the Mat\'ern example, the empirical estimator performs very poorly for smaller values of $\nu$. One interesting observation here is that the effect of non-separability is not so severe on the best separable estimator, particularly for the integrated Brownian sheet. This may be due to the fact that during the construction of the best separable estimator in 3D, we merged two dimensions together, which somehow caters for the non-separability. Among the different CovNet estimators, the deepshared model had the best performance.

\begin{table}[t]
\centering
\caption{\label{table:errors_CV_3D}Relative errors (in \%) of the CovNet models with hyperparameters chosen using $5$-fold cross-validation. Difference from the least observed error over the range of hyperparameters is shown in parentheses. Relative errors for the empirical and the best separable estimators are also reported. The reported numbers are based on one simulation run with $N=500$ and $K=15$.}
\begin{tabular}{lrrrrr}
Example & Empirical & Best separable & Shallow & Deep & Deepshared \\ [3pt]
Rotated Brownian sheet 3D & $15.38$ & $33.83$ & $12.84\,(0.00)$ & $11.64\,(0.00)$ & $10.67\,(0.35)$ \\
Rotated integrated Brownian sheet 3D & $15.35$ & $19.70$ & $16.41\,(1.54)$ & $12.57\,(0.73)$ & $11.16\,(0.81)$ \\ [2pt]
Matern 3D $\nu=0.001$ & $47.52$ & $12.38$ & $10.01\,(0.00)$ & $10.94\,(1.23)$ & $10.10\,(0.00)$ \\
Matern 3D $\nu=0.01$ & $49.14$ & $12.38$ & $11.13\,(0.00)$ & $11.83\,(0.36)$ & $11.88\,(0.98)$ \\
Matern 3D $\nu=0.1$ & $18.94$ & $11.65$ & $12.76\,(1.72)$ & $12.21\,(0.09)$ & $11.99\,(1.07)$ \\
Matern 3D $\nu=1$ & $9.94$ & $9.66$ & $9.64\,(0.00)$ & $9.61\,(0.00)$ & $9.45\,(0.55)$
\end{tabular}
\end{table}

\medskip

Finally, in Table~\ref{table:errors_CV_3D}, we show the results for the proposed cross-validation strategy on the examples in 3D. As before, for each example, we report relative errors for the three CovNet models (shallow, deep and deepshared) with hyperparameters selected via cross-validation based on a single simulation run with $500$ observations at a resolution of $15 \times 15 \times 15$. For each model, we also show the difference from the least observed relative error over the range of hyperparameters. The relative errors for the empirical and the best separable estimators are also reported to facilitate comparison. Similar to 2D, in all the examples, the average difference from the best result was less than $1\%$ for all the CovNet models, while the maximum difference was less than $1.75\%$. So, here too, the cross-validation method identifies a good set of hyperparameters.

\bibliographystyle{imsart-nameyear}
\bibliography{biblio}

\end{document}